\documentclass[11pt]{article}

\usepackage{amsmath, bm, amssymb, amsthm, graphicx, wrapfig, verbatim, mathtools, bbm, tikz, hyperref, caption, subcaption}
\usepackage[inline]{enumitem}
\usetikzlibrary{backgrounds}

\pgfdeclarelayer{background}
\pgfsetlayers{background,main}

\usetikzlibrary{intersections}
\usepackage{pgfplots}
\usepackage{adjustbox}

\usepgfplotslibrary{fillbetween, groupplots}
\pgfplotsset{compat=1.18}

\newcommand{\R}{\mathbb{R}}

\newcommand{\E}{\mathbb{E}}
\newcommand{\Var}{\mathrm{Var}}

\usepackage[margin=1in]{geometry}

\usetikzlibrary{calc} 

\usepackage{booktabs}
\usepackage{longtable}

\title{Overlap Analysis of the Shortest Path Problem: \\ Local Search, Landscapes, and Franz--Parisi Potential}
\author{Frederic Koehler\thanks{University of Chicago, \texttt{fkoehler@uchicago.edu}.} \and Joonhyung Shin\thanks{University of Chicago, \texttt{joons@uchicago.edu}.}}

\newcommand{\Binom}{\operatorname{Binom}}
\newcommand{\Multinom}{\operatorname{Multinomial}}
\newcommand{\Pois}{\operatorname{Poisson}}
\newcommand{\Unif}{\operatorname{Uniform}}
\newcommand{\pto}{\overset{p}\to}
\newcommand{\dto}{\overset{d}\to}
\newcommand{\deq}{\overset{d}=}
\newcommand{\KL}{\operatorname{KL}}
\newcommand{\plim}{\operatorname*{plim}}

\newcommand{\Bernoulli}{\operatorname*{Bernoulli}}

\newcommand{\ZTB}{\operatorname{ZTB}}
\newcommand{\ZTP}{\operatorname{ZTP}}

\numberwithin{equation}{section}

\newtheorem{theorem}{Theorem}[section]

\newtheorem{lemma}[theorem]{Lemma}

\newtheorem{proposition}[theorem]{Proposition}

\newtheorem{corollary}[theorem]{Corollary}

\theoremstyle{definition}
\newtheorem{definition}{Definition}[section]

\theoremstyle{remark}
\newtheorem*{remark}{Remark}

\theoremstyle{remark}
\newtheorem{example}{Example}

\newtheorem{claim}{Claim}

\begin{document}

\maketitle
\begin{abstract}
Two fundamental directions in algorithms and complexity involve: (1) classifying which optimization problems can be (approximately) solved in polynomial time, and (2) understanding which computational problems are hard to solve \emph{on average} in addition to the worst case. 
For many average-case problems for which we lack efficient algorithms, there does not currently exist strong evidence from reduction-based hardness that these problems are computationally intractable. However, we can still attempt to predict their polynomial time tractability \emph{heuristically}, by (rigorously) proving unconditional lower bounds against restricted classes of algorithms. 

Geometric approaches to predicting tractability and hardness of nonconvex optimization problems typically study the \emph{optimization landscape} of the problem. For optimization problems with random objectives or constraints, ideas originating in statistical physics suggest that a key property to study is the \emph{overlap} (equivalently, distance) between approximately-optimal solutions. Formally, properties of \emph{Gibbs measures} and the corresponding \emph{Franz--Parisi potential} can be used to prove rigorous lower bounds against natural local search algorithms, such as Langevin dynamics, based on free energy barriers. A related theory, called the \emph{Overlap Gap Property (OGP)}, proves rigorous lower bounds against classes of algorithms which are sufficiently stable/Lipschitz functions of their input.

A remarkable recent work of Li and Schramm showed that the shortest path problem in random graphs admits lower bounds against a class of stable algorithms by establishing the OGP. On the other hand, this problem is polynomial time tractable. In order to better understand why this happens, we revisit the shortest path problem and analyze its optimization landscape further. Our key finding is that both the OGP and the Franz--Parisi potential predict that: (1) local search will fail in the optimization landscape of shortest paths, but (2) local search should succeed in the optimization landscape for shortest path \emph{trees}, which is true. Using the Franz--Parisi potential, we explain an analogy with results from combinatorial optimization --- submodular minimization is algorithmically tractable via local search on the Lov\'asz extension, even though ``naive'' local search over sets or the multilinear extension provably fails.
\end{abstract}
\thispagestyle{empty}
\newpage

\tableofcontents
\thispagestyle{empty}
\newpage

\setcounter{page}{1}

\section{Introduction}
A fundamental problem in the theory of computation is to understand the complexity and tractability of natural optimization problems. In particular, from the standpoint of complexity theory we would like to understand which problems admit polynomial-time algorithms and which do not. The theory of reduction-based hardness has been very successful at this task for categorizing the complexity of different \emph{worst-case optimization} problems; for example, many problems are known to be equivalent to 3-SAT (3-satisfiability) under polynomial time (Karp) reductions, which is the celebrated theory of NP-completeness. Since Karp reductions allow for arbitrary polynomial time computation, the reductions in this theory are allowed to be relatively intricate and complex, which has allowed for many deep and nontrivial hardness results to be proven, such as in hardness of approximation/PCP theory. 

However, the situation with the complexity of natural random/average-case problems, as arise when we consider problems over random graphs or random data generating processes (as often considered in, e.g., ML and statistics) is quite different. In this setting, reduction-based hardness has seen some successes, but not as comprehensively as in the worst-case setting. A key difficulty is that reductions which work to reduce the worst-case hardness of Problem A from Problem B rarely preserve/map the natural distribution of random problem instances. See, e.g., \cite{brennan2020reducibility,bresler2025computational} and references within for discussion and some positive results.

For this reason, there is great interest in alternative approaches to study the complexity of random optimization problems. In general, one of the alternatives to studying reduction-based hardness is to \emph{unconditionally rule out natural subclasses} of polynomial time algorithms. This is also a familiar strategy in the worst-case setting, for example in the context of lower bounds against decision trees, concrete families of circuits like $AC_0$, families of convex programs, etc. In recent years, there have been remarkable successes in proving unconditional lower bounds for random problems: for example, against the powerful Sum-Of-Squares/Laserre hierarchy of semidefinite programs, against low-degree polynomials, against natural families of gradient-based optimization methods, against spectral algorithms, against classes of message passing algorithms, etc. (see, e.g., \cite{barak2019nearly,wein2022optimal,gamarnik2024hardness} and references within).

A major difference between lower bounds against a concrete class like Sum-of-Squares programs versus classical reduction-based hardness is sensitivity to \emph{polynomial-time equivalent reformulations} of the same problem. For example, although there are reductions between problems like Planted Clique and Sparse PCA \cite{brennan2020reducibility}, an SoS lower bound against Planted Clique does not logically imply an SoS lower bound against natural SDP formulations of Sparse PCA.  

\paragraph{Landscape analysis.} Another popular type of analysis for the tractability of random optimization problems is via various types of studies of \emph{optimization geometry} and \emph{landscape analysis}. 

Landscape analysis, in both worst-case and average-case problems, can potentially be sensitive to small changes in the formulation of the problem. For example: Auer, Herbster, and Warmuth showed that for the task of learning a single sigmoidal neuron the squared loss can have exponentially many bad local minima \cite{auer1995exponentially}, even though if we use the logistic loss\footnote{With an $\ell_2$ constraint or regularization to prevent divergence to infinity.} for the same problem gradient descent will be globally convergent, because it is convex. That being said, being sensitive is not necessarily a bad thing, since it also means we can detect differences between polynomial-time equivalent formulations of the problem --- for example, for algorithm design it is indeed helpful to know that gradient descent on one loss function will work better than another one, even if, e.g., the global minimizer of both problems is the same (so finding the global minima of either function is polynomial time equivalent). For another example, there are many natural optimization formulations we can propose for finding the minimum of a submodular function, and it is quite valuable to know that optimization over the Lov\'asz extension more promising than optimizing over, e.g., the multilinear extension (which can be highly nonconvex). 
We will return to this example later.

In this paper, we compare two combinatorial optimization problems --- one is over paths, and one is over trees. We can easily write both of them down as follows.




\begin{center}
  \begin{minipage}{0.46\linewidth}\centering
    \[
      \textbf{(P1)} \quad \min_{P \in \mathcal{P}_{s,t}} \; |P|
    \]
  \end{minipage}\hspace{2em}%
  \begin{minipage}{0.46\linewidth}\centering
    \[
      \textbf{(P2)} \quad \min_{T \in \mathcal{T}} \; \sum_{v \in V} \mathsf{d}_T(s,v)
    \]
  \end{minipage}
\end{center}

\medskip

Here, $\mathcal{P}_{s,t}$ denotes the set of paths from $s$ to $t$, and $|P|$ is the length of a path $P$. 
Meanwhile, $\mathcal{T}$ denotes the set of spanning trees of the graph on vertex set $V$, and $\mathsf{d}_T(s,v)$ is the distance from $s$ to a vertex $v$ in the tree $T$.
They are both natural formulations of the shortest path problem, and there are very simple polynomial time reductions between the two problems. While (P1) is the most obvious formulation, (P2) probably appears more in the literature (see, e.g. 
Appendix~\ref{sec:dca}, for how it appears in \cite{murota2014dijkstra}, Appendix~\ref{apdx:dijkstra-continuous} for standard optimization references and physical justification, and the approximation algorithms literature \cite{khuller1995balancing,awerbuch1990cost,bharath2003routing}).
As we will see, there is a lot which can be said about the landscape of both problems.

\paragraph{Surprises in shortest path optimization (P1).} Despite the great success so far of the OGP framework,  in a remarkable recent work it was shown that the shortest path problem, which is definitely computationally tractable, satisfies the OGP for typical random graphs \cite{LS2024}. 
To be precise, problem (P1) above 
is proven to exhibit the ensemble overlap-gap property asymptotically almost surely, when $G \sim G(n, c\log(n)/n)$ is a sparse Erd\"os--R\'enyi graph.
As they showed, this arises from the fact (Theorem~2.2 of \cite{LS2024}) that shortest paths in two correlated Erd\"os--R\'enyi random graphs are either identical or have little overlap with high probability. 
By standard arguments in the OGP literature (see Corollary~2.3 of \cite{LS2024}), this implies that there is no stable\footnote{In this particular case, any algorithm whose output changes by at most $o(n)$ edges when we resample the edges of a graph one-by-one in random order. See also Section~\ref{sec:ensemble-ogp} below.} algorithm that computes a shortest path from $1$ to $2$. This is worrisome because it is similar to the type of argument which is typically used in the OGP literature to give evidence that a problem should be computationally hard, but this problem is not actually hard.  

As we discuss in more detail in Appendix~\ref{apdx:path-barrier}, the core structural result proved in \cite{LS2024} can also naturally be interpreted as a \emph{free energy barrier} in a sense which is made precise by the \emph{Franz--Parisi potential} (FPP) \cite{franz1995recipes}. We will explain more about this later, but in essence the FPP is part of the general framework of studying \emph{finite-temperature Gibbs measures} associated to the optimization problem. In the case of (P1), this would come down to studying the geometry of probability measures of the form
\[ \nu_{G,\beta}(P) \propto \exp\left(-\overline{\beta} |P| \right)\]
where $P$ ranges over the set of $s$ to $t$ paths as before, and $\overline{\beta} > 0$ is an ``inverse temperature'' parameter that controls the level of randomness/stochasticity in related local search algorithms. 
Again, as we formally show in the appendix, we can easily prove the existence of a free energy barrier from the ideas in \cite{LS2024} and this implies rigorous lower bounds against natural local MCMC algorithms, which could be interpreted as evidence that the problem is hard.

\paragraph{New insights from studying (P2).} The goal of this paper is to give a complete analysis of the optimization landscape of (P2) from both the perspective of the overlap gap (OGP) theory as well as Gibbs measures \& the Franz--Parisi potential (FPP). Informally, our key findings are that:
\begin{enumerate}
    \item Formulation (P2) \emph{does not} have the ensemble overlap gap property\footnote{One of several reasons this is interesting is that the model does satisfy various types of \emph{disorder chaos}, which is a shared feature with (P1) that sometimes comes together with the OGP. We discuss this more later.}. In fact, we can analytically compute an explicit formula for the limiting overlap which smoothly varies between $1$ and $0$. 
    \item Formulation (P2) \emph{does not} have free energy barriers in the sense of the Franz--Parisi potential (FPP). In fact, we obtain precise asymptotic descriptions of the finite temperature Gibbs measures and derive explicit \emph{quasiconvex} expressions for the limiting FPP in all regimes.
\end{enumerate}
Using the Franz--Parisi potential, we can make a clear analogy between these findings and the aforementioned results for submodular minimization with different extensions. See Appendix~\ref{sec:example-lovasz}.

\paragraph{Summary: shortest path as a model organism.} 
The ultimate goal of our work is \emph{not} to prove that the shortest path problem in random graphs is polynomial tractable, which is obvious from our existing knowledge. Instead, our goal is to deepen our understanding of fundamental frameworks such as the OGP and FPP by working out one example, the shortest path problem, in great detail. This ends up to be quite mathematically interesting, but we also mention a few other reasons why we feel this is important:
\begin{itemize}
    \item In the OGP and FPP literature, a major focus has been on proving evidence for \emph{intractability} for problems which we currently have no algorithms for. Only in relatively few cases (e.g., \cite{abbe2022binary}) has the stability/overlap of an algorithm which \emph{succeeds} been studied from the same perspective. This raises a worry that we could have misinterpreted some OGP lower bounds --- we needs examples to see in what kinds of special cases a problem which ``appears'' OGP-hard might actually be easy. 
    \item There remain uncertainties about the interpretation of these heuristics in sparse models. The shortest path problem would not be as unstable in path space (P1) if the random graph was much denser.
    To quote \cite{LS2024}, ``sparse average-case models are known to sometimes exhibit anomalous behavior.''
    \item Despite the clear historical and conceptual connections between the OGP and FPP approaches, there are few examples where both analyses have been worked out in complete detail. One reason is that for many problems, it is difficult.\footnote{As a reminder, even in the well-studied Sherrington-Kirkpatrick model, it is not rigorously known that the Gibbs measure is replica-symmetric up to the ``Almeida-Thouless line'' (as is expected, c.f. \cite{thouless1977solution,brennecke2022replica}). So even more basic properties of the overlap are not fully understood.}
    \item All commonly studied heuristics have apparent downsides and disadvantages. For example, the low-degree heuristic works reasonably well in the shortest path example from \cite{LS2024}, but its predictions are ``wrong'' (and/or, difficult to interpret) in other settings \cite{buhai2025quasi,huang2025optimal,huang2024low,koehler2022reconstruction}. Having several heuristics with well-understood pros and cons can help us make more reliable predictions in the future.
    \item Although we \emph{do} know that the shortest path problem is tractable, this does actually does not tell us much\footnote{This is true even though we know for independent reasons (see Appendix~\ref{sec:dca}) that a natural ``local search'' procedure \emph{does} converge for (P2). Perhaps unintuitively, even unimodal optimization problems can admit free energy barriers which trap finite temperature MCMC chains --- see \cite{bandeira2023free} for a detailed example and explanation.} about overlap concentration, the structure of the Gibbs measures, etc. We had to work it out ourselves to see what these heuristics actually predict about (P2).
\end{itemize}

To summarize, we consider the shortest path problem and its two landscapes (P1) and (P2) to be ``model organisms'' for understanding how overlap analysis works. With a better understanding of what happens in this ``simple'' example, we believe we will be able to better interpret these heuristics in harder problems where the accuracy of our predictions is most crucial and useful.

\subsection{Background on the OGP and analogy with Gibbs measures}\label{sec:ensemble-ogp}
To begin to make precise some of the concepts mentioned in the introduction, 
We first begin by recalling the definition of the ensemble-OGP, often abbreviated as e-OGP. See the survey \cite{gamarnik2021overlap} for a more detailed introduction.
\begin{definition}[\cite{gamarnik2021overlap,chen2019suboptimality}]
Let $(\Sigma_n,\mathsf{d}_n)$ be a metric space and
consider a collection of minimization problems, defined over $\Sigma_n$ and parameterized by elements of a set $\Xi_n$, specified by a function $L : \Sigma_n  \times \Xi_n \to \mathbb R$. We say the \emph{ensemble}-OGP is satisfied with parameters $\mu > 0, 0 \le \nu_1 < \nu_2$ if for any $\xi,\psi \in \Xi_n$, and for any
\[ \sigma \in \{ \sigma \in \Sigma_n : L(\sigma, \xi) \le \min_{s} L(s,\xi) + \mu \},\qquad \tau \in \{ \tau \in \Sigma_n : L(\tau,\psi) \le \min_t L(t, \psi) + \mu \}, \]
it holds that 
\[ \mathsf{d}_n(\sigma,\tau) \in [0,\nu_1] \cup [\nu_2, \infty). \]
\end{definition}
In the special case where $\Xi_n$ is a singleton set, the e-OGP corresponds to a restriction on the geometric properties of the near-global maxima of a single optimization problem: no pair of near-maxima exist at ``intermediate'' distance $(\nu_1,\nu_2)$. In general, the e-OGP requires the same property hold for the union of all near-maxima among several different optimization problems.
The e-OGP and further generalizations of the OGP have been useful for proving lower bounds against stable algorithms in a wide variety of models.

The index $n$ plays no formal role in the definition, but it conventionally corresponds to a natural notion of dimension of the state space $\Sigma_n$, and OGP is generally applied to understand high-dimensional limits where $n \to \infty$.
 Generally speaking, that the distance metric $\mathsf{d}_n$ and loss $L$ should be appropriately scaled for the parameters $\mu,\nu_1,\nu_2$ to be consistent as $n \to \infty$; in particular, in combinatorial optimization $\mathsf{d}_n$ is typically a normalized Hamming distance so that its maximum value is $1$. 
 
\paragraph{Ensemble-OGP for correlated random graphs.}
One of the main applications of the e-OGP has been to optimization problems on random Erd\"os--R\'enyi graphs. In this setting, e-OGP can be applied by starting with a single random graph and then step-by-step resampling single random vertex pairs in the graph: this generates a sequence of correlated Erd\"os--R\'enyi graphs on vertex set $V$ as elements of $\Sigma_n = \binom{V}{2}$  \cite{gamarnik2021overlap}. This could equivalently be thought of as simulating the trajectory of a random graph under the single-scan Glauber dynamics (see, e.g., \cite{wilmer2009markov}). With these definitions, satisfying the OGP (or not) becomes an event dependent on the particular realization of the stochastic process --- 
typically the ensemble $\Xi_n$ (here, the collection of correlated random graphs) is \emph{random} and we are interested in understanding whether the OGP occurs for a typical realization of the ensemble. 

\paragraph{Aside: OGP analysis and Gibbs measures from a single perspective.} We will shortly proceed to discussion of our formal results and findings, which will involve an extensive discussion of both the OGP and Gibbs measure approaches. We believe the following simple observation helps clarify the similarities and differences between the two approaches. Starting with Gibbs measures, we know via the ``Gumbel trick'' (see Appendix~\ref{apdx:gumbel}) that the analysis of Gibbs measures is equivalent to that of a \emph{randomly perturbed optimization problem}. In particular, for any objective function $f(x)$, $\beta > 0$, and finite set $\mathcal X$,
\[ \log \sum_{x \in \mathcal X} e^{-\beta f(x)} = \gamma + \mathbb{E}_G \min_{x \in \mathcal X}[\beta f(x) - G(x)] \]
where the expectation is over independent standard Gumbels $G(x)$ at each site $x \in \mathcal X$. On the other hand, the ensemble-OGP studies how the optimizer of the objective $L$ changes under \emph{partial resampling of the underlying randomness} (random graph in our setting). So in fact, these are both natural approaches to studying \emph{the stability of the solution of an optimization problem under random perturbation} --- a kind of ``sensitivity analysis'' for optimization. Perhaps the key \emph{difference} between the two is that the Gibbs measure approach is agnostic to the internal structure of $f$.  

\subsection{Our results}
\subsubsection{Basic models and notations}\label{sec:models-notations}

We denote by $\mathcal{G}(n,q)$ the Erd\"os--R\'enyi model of random graphs, with $n$ vertices and the edge probability $q$. 
Throughout this paper, as in \cite{LS2024}, we consider the asymptotic regime where $q = \Theta(\log(n)/n)$. This means that the graph is relatively sparse and the giant component of the graph contains at least a $1 - o(1)$ fraction of the vertices (see, e.g., \cite{alon2016probabilistic} for background on random graphs). As in the literature on the diameter of random graphs \cite[e.g.]{chung2001diameter,riordan2010diameter}, the vertices outside of the giant component play no role in the analysis and can be ignored. 

\paragraph{Single graph setup.} We will always consider a fixed sequence of $\alpha_n \in (0,\infty)$ and 
define in terms of them the edge probability
\begin{equation} 
q_n = \frac{\alpha_n \log n}{n}.
\end{equation}
To study the asymptotic behavior of sparse random graphs, we consider a sequence $G_1,G_2,\cdots$ of independent Erd\"os--R\'enyi graphs sampled from $G_n\sim\mathcal{G}(n,q_n)$, where the vertices are set to be integers
\[
    V_n:=V(G_n)=\{1,\cdots,n\}\,.
\]
Most of our single graph asymptotics results concern certain functions of a graph, which is naturally a sequence of random variables whose law is dictated by our fixed sequence $\{\alpha_n\}$. To guarantee $q_n=\Theta(\log (n)/n)$, we assume throughout the paper that
\begin{equation}\label{eqn:alpha-bound}
    0<\liminf_{n\to\infty}\alpha_n\leq\limsup_{n\to\infty}\alpha_n<\infty\,.
\end{equation}
We do not require that $\{\alpha_n\}$ converges or is above the connectivity threshold $\liminf\alpha_n>1$; our results hold without those assumptions.

\paragraph{Graphs with fixed source.} Throughout this paper, we mostly consider single source paths, i.e., paths starting from a fixed vertex, and $1$ will always be our source vertex. Thus, the connected component of $G_n$ containing $1$ will be our primary focus, which we denote by $\overline{G}_n$. As noted above, we only consider Erd\"os--R\'enyi graphs with $\liminf\alpha_n>0$ in which there is one giant component of size $n-o(n)$ with high probability. Hence, $1$ will asymptotically almost surely be contained in this giant component, so $\overline{G}_n$ can be thought of as the giant component of $G_n$. Whenever we talk about shortest path trees or spanning trees, we are actually referring to those of $\overline{G}_n$. Moreover, these spanning trees are naturally rooted at the source vertex $1$ so every pair of vertices adjacent in the trees have an induced parent-child relationship; we always take this into account when comparing two trees.

\paragraph{Correlated graphs setup.} In order to study the overlap of two correlated instances, we need an additional sequence of $\rho_n\in(0,1)$ specifying the correlation strength. We define two graphs with correlation $\rho_n$ in the usual way. We begin with a sample $G_n^{(1)}\sim\mathcal{G}(n,q_n)$ for each $n$, same with the single graph case. Then for each pair of distinct vertices of $G_n^{(1)}$, we resample it with probability $1-\rho_n$. This gives a graph $G_n^{(2)}$ which is also marginally $\mathcal{G}(n,q_n)$. Our results concern the (independent) sequence of correlated pairs $\{(G_n^{(1)},G_n^{(2)})\}$ obtained this way.

\paragraph{Proxies for describing asymptotic behaviors.} The sequences $\{\alpha_n\}$ (which determines the edge probabilities $\{q_n\}$) and $\{\rho_n\}$ (in the correlated setup) are the only prescribed objects that completely dictate the law of the random graphs. However, many asymptotic behaviors are often more conveniently and concisely described by certain quantities that depend on $\{\alpha_n\}$. Among those, we introduce a few important quantities that play a significant role throughout this paper.

First and foremost, we define
\begin{equation} \label{eqn:lstar}
\ell^*_n=\frac{\log n}{\log(nq_n)} = \frac{\log n}{\log \alpha_n + \log \log n}
\end{equation} 
or equivalently, $(nq_n)^{\ell_n^*}=n$. This can be understood as an approximate distance of a typical pair of vertices in sparse Erd\"os--R\'enyi graphs and is the same quantity used in \cite{LS2024}. Also, we define
\begin{equation} \label{eqn:dstar}
    \begin{split}
        d_n^*&:=\min\left\{d\in\mathbb{Z}:(nq_n)^d\geq\frac{n}{(\log\log n)^2}\right\}\\
        &=\left\lceil\ell^*-\frac{2\log\log\log n}{\log\alpha_n+\log\log n}\right\rceil
    \end{split}
\end{equation}
to be the smallest integer satisfying $(nq_n)^{d_n^*}\geq n/(\log\log n)^2$. Under the assumption \eqref{eqn:alpha-bound}, $d_n^*$ is one of the two closest integers to $\ell^*_n$ for every large $n$. The denominator $(\log\log n)^2$ is somewhat arbitrary and is especially relevant to our analysis of Gibbs measures; see Section~\ref{sec:to-gibbs} and the accompanying footnotes for a brief explanation of our choice of denominator.

The gap between $d_n^*$ and $\ell_n^*$ is what turns out to be the key to understanding the dynamics of the graphs. Perhaps the simplest form of such a gap is their difference, which we denote by
\begin{equation}\label{eqn:Delta}
    \Delta_n:=d_n^*-\ell_n^*\,.
\end{equation}
Note that $\Delta_n$ might be negative, but \eqref{eqn:alpha-bound} implies that
\[
    0\leq\liminf_{n\to\infty}\Delta_n\leq\limsup_{n\to\infty}\Delta_n\leq 1\,.
\]
Another form, which we denote by $\lambda_n$, is nontrivial and less intuitive but is equally important. This is defined by the equation
\begin{equation}\label{eqn:lambda}
    \log(1/\lambda_n)=(nq_n)^{\Delta_n}=\frac{(nq_n)^{d_n^*}}{n}\,.
\end{equation}
By definition, we have $0<\lambda_n<1$ and thus
\[
    0\leq\liminf_{n\to\infty}\lambda_n\leq\limsup_{n\to\infty}\lambda_n\leq1\,.
\]
We remind the reader that the sequences $\{\Delta_n\}$ and $\{\lambda_n\}$ are functions of $n$ and $\{\alpha_n\}$. It turns out that the limiting behavior of $\Delta_n$ and $\lambda_n$ as $n \to \infty$ is what is critically important for the asymptotics, rather than $\alpha_n$.


\paragraph{Convergence assumptions on the proxies.} As a final remark, we note that many of our results are stated under the assumption that proxy sequences such as $\lambda_n$ and $\Delta_n$ converge (or diverge to either $+\infty$ or $-\infty$), which might seem to put restrictions on $\{\alpha_n\}$. However, this is enough to describe the asymptotic behavior for \emph{any} $\{\alpha_n\}$ (satisfying \eqref{eqn:alpha-bound}). For instance, assume we have proved that some property $\mathcal{P}_{n}(\alpha_{n})$ holds with high probability for any $\{\alpha_{n}\}$ such that $\{\lambda_{n}\}$ and $\{\Delta_{n}\}$ converge. Then in fact $\mathcal{P}_n(\alpha_n)$ holds with high probability for any $\{\alpha_n\}$; otherwise, there is a subsequence of indices $k_1,k_2,\cdots,$ where $\mathcal{P}_{k_n}(\alpha_{k_n})$ does not hold with probability bounded away from zero, and from that we can find a further subsequence $\{k_{i_n}\}$ where the proxies converge, which is a contradiction. An example of this can be found in Section~\ref{sec:to-overlap} where we prove the nonexistence of OGP. In case the property depends on the limiting values and the proxy sequences have multiple limit points, several asymptotic behaviors may coexist.


\subsubsection{Asymptotics for a single random graph}

We first establish the distance asymptotics for $\{G_n\}$. The proof and related discussions are detailed in Section~\ref{sec:single}.

\begin{theorem}[Distance asymptotics]\label{thm:single-intro}
Suppose that $\lambda_n\to\lambda\in[0,1]$. Let $N_{d_n^*}$ and $N_{d_n^*+1}$ denote the number of vertices of $G_n$ with distance $d_n^*$ and $d_n^*+1$ from the source vertex $1$, respectively. Then as $n \to \infty$, we have
    \[ N_{d_n^*}/n\pto 1-\lambda, \qquad \text{and} \qquad N_{d_n^*+1}/n\pto \lambda. \]
    So with probability $1 - o(1)$ all vertices except for at most $o(n)$ of them have distance either $d_n^*$ or $d_n^*+1$ from vertex $1$, and the proportion of vertices with distance $d_n^*$ converges to $1-\lambda$ in probability.
\end{theorem}

\subsubsection{Overlap asymptotics for correlated graphs}



We have two $\rho_n$-correlated graphs $G_n^{(1)}$ and $G_n^{(2)}$ on the common vertex set $V_n=\{1,\cdots,n\}$. Suppose we have shortest path trees $T_n^{(1)}$ of $\overline{G}_n^{(1)}$ and $T_n^{(2)}$ of $\overline{G}_n^{(2)}$. Their overlap is defined by
\begin{equation}\label{eqn:defn-overlap}
    \mathsf{R}(T_n^{(1)},T_n^{(2)}):=\frac{|T_n^{(1)}\cap T_n^{(2)}|}{\sqrt{|T_n^{(1)}||T_n^{(2)}|}}
\end{equation}
where $|T^{(i)}|$ denotes the number of the edges of $T^{(i)}$ and $|T_n^{(1)}\cap T_n^{(2)}|$ denotes the number of common edges, respecting their parent-child relationships. Since the giant component has size $n-o_p(n)$, we may analyze the overlap through the quantities
\[
    R_n:=\frac{1}{n}|T_n^{(1)}\cap T_n^{(2)}|\,.
\]
Let $\mu_n^{(1)}$ and $\mu_n^{(2)}$ be the uniform measures over the set of shortest path trees of $\overline{G}_n^{(1)}$ and $\overline{G}_n^{(2)}$, respectively, and suppose we draw $T_n^{(1)}\sim\mu_n^{(1)}$ and $T_n^{(2)}\sim\mu_n^{(2)}$ (conditioned on the graphs). We study the overlap under the optimal coupling $\Pi$ of $\mu_n^{(1)}$ and $\mu_n^{(2)}$
\begin{equation}\label{eqn:parconc}
    \tilde{R}_n:=\inf_{\Pi} \mathbb{E}_{\Pi}\left[R_n\mid G_n^{(1)},G_n^{(2)}\right] = n - W_1\left(\mu_n^{(1)},\mu_n^{(2)}\right) 
\end{equation}
where $W_1$ is the Wasserstein distance under the (directed) Hamming metric. This definition naturally guarantees that $\tilde{R}_n = 1$ when $G_n^{(1)} = G_n^{(2)}$. Note that $\tilde{R}_n$ are still random variables, since each measure $\mu_n^{(i)}$ is a function of the random graph $G_n^{(i)}$. Before we state the result, we define an additional proxy sequence $\{\gamma_n\}$ by
\begin{equation}\label{eqn:defgamma}
    \gamma_n := \rho_n^{d^*_n}\in(0,1)\,.
\end{equation}
This plays an important role when $\rho_n$ is close to $1$; if $\rho_n$ is bounded away from $1$, then $\gamma_n \to 0$ as $n \to \infty$ since $d^*_n \to \infty$.
When $\rho_n$ is close to $1$, we have by Taylor expansion that $1-\rho_n \approx -\log(\rho_n) = \log(1/\gamma_n)/d^*_n$, so the expected number of edges resampled along a path of length $d^*_n$ is approximately $\log(1/\gamma_n)$.

\begin{theorem}[Correlated graph asymptotics]\label{thm:correlated-intro}
Suppose that $\lambda_n\to\lambda\in[0,1]$ and
\begin{equation}\label{eq:coupled-asymptotics-1}
    \rho_n \to \rho \in [0,1]\,.
\end{equation}
    If $\rho < 1$, set $\gamma = 0$ because $\gamma_n \to 0$ necessarily. Otherwise if $\rho = 1$,
    require additionally that 
    \begin{equation} \gamma_n \to \gamma \in [0,1]. \end{equation} 
    Then: 
    \begin{enumerate}
        \item Let $N_{d_n^*}^I$ denote the number of vertices with distance $d_n^*$ from the source $1$ in both $G_n^{(1)}$ and $G_n^{(2)}$. We have that
    \[\frac{N_{d^*_n}^I}{n} \pto 1-2\lambda+\lambda^{2-\gamma}, \]
    i.e., the proportion of the vertices that have distance $d^*_n$ for both $G_n^{(1)}$ and $G_n^{(2)}$ converges to $1-2\lambda+\lambda^{2-\gamma}$ in probability. Thus, the proportion of the vertices that have the same distance for $G_n^{(1)}$ and $G_n^{(2)}$ converges to $1-2\lambda+2\lambda^{2-\gamma}$ in probability.
        \item The overlap of uniformly random shortest path trees satisfies
            \[
        \tilde{R}_n\pto\begin{dcases*}
            1 & \text{if $0<\lambda<1$ and $\gamma = 1$,}\\
            \frac{1-2\lambda+\lambda^{2-\gamma}}{1-\lambda}\cdot\lambda^{2-\gamma}+f(\lambda,\gamma) & \text{if $0<\lambda<1$ and $0 < \gamma < 1$,}\\
            \rho (1 - \lambda)\lambda^{2} & \text{if $0<\lambda<1$ and $\gamma = 0$,}\\
            \gamma & \text{if $\lambda \in \{0,1\}$.} 
        \end{dcases*}
    \]

    where we define
    $f(\lambda,\gamma) = g((1-\gamma)\log(1/\lambda),\gamma\log(1/\lambda))\cdot(1-\lambda^{\gamma})$
    and for $a,b > 0$ we let 
    \[
        g(a,b)= \mathbb{E}\left[\frac{Z}{\max(X_1,X_2)+Z}\,\middle|\,Z>0\right]
    \]
    where the expectation is over independent $X_1,X_2\sim\Pois(a)$ and $Z\sim\Pois(b)$.

        \end{enumerate}
\end{theorem}

The proof is detailed in Section~\ref{sec:correlated}. To help illustrate the results, we observe a few direct consequences. First of all, our results are more than strong enough to show that there is no e-OGP. This is because by tuning the correlation between two correlated graphs, we can achieve any desired overlap with high probability. This is proved in the following corollary when $\lambda_n\to\lambda$ converges, but this assumption can be removed --- see Section~\ref{sec:to-overlap}.
\begin{corollary}[No e-OGP for shortest path tree]\label{cor:no-eogp}
For any choice of limiting $\lambda$, and any choice of parameters $\mu > 0$ and $0 \le \nu_1 < \nu_2$, the ensemble-OGP property, with $\mathsf{d}_n$ the normalized Hamming metric, is violated with probability $1 - o(1)$.
\end{corollary}
\begin{proof}
Suppose that the ensemble-OGP is satisfied with parameters $\mu > 0, 0 \le \nu_1 < \nu_2$  with probability $\Omega(1)$ as $n \to \infty$. Let $x$ be any number between $\nu_1$ and $\nu_2$; Theorem~\ref{thm:correlated-intro} tells us that with high probability, we can identify a pair of $(\rho,\lambda)$ such that asymptotically almost surely, the correlated graph $G_2$ will have overlap $(1 + o(1))x$, so asymptotically almost surely the overlap is $(\nu_1,\nu_2)$. Since the conditional law of the graph $G_2$ given $G_1$ is identical to the conditional law of one of the elements of the ensemble $\Xi$ defining the e-OGP, we have a contradiction. 
\end{proof}
Our result can also be used to show that the uniform measure over spanning trees satisfies a type of Wasserstein disorder chaos as defined in \cite{el2022sampling}. Intuitively, this means that if we resample a vanishing proportion of edges of the graph, the distribution of spanning trees will already change significantly.
As discussed further in Remark~\ref{rmk:reasonable-stability}, this implies a fairly restrictive lower bound result: 
$O(1)$-Lipschitz algorithms cannot sample from the measure (c.f. \cite{el2022sampling,LS2024,alaoui2023sampling,ma2025polynomial}).
\begin{corollary}[Wasserstein Disorder Chaos]\label{cor:wasserstein-disorder}
Suppose $\{\rho_n\}$ is a constant sequence $\rho_n=\rho$. Then
\[ \liminf_{\rho \to 1^-} \operatorname*{plim}_{n \to \infty} \tilde{R}_n \le 4/27 < 1. \]
\end{corollary}
\begin{proof}
We use the theorem and the fact that $\max_{\lambda \in [0,1]} (1 - \lambda)\lambda^2 = 4/27$.
\end{proof}

We also show that the uniform measure over shortest path trees exhibits disorder chaos in a non-Wasserstein sense (see e.g. \cite{chatterjee2009disorder,el2022sampling}). A more general result is proved in Theorem~\ref{thm:indepcoupling}.

\begin{theorem}[Disorder chaos]\label{thm:d-chaos-intro}
    If $\rho_n=\rho<1$ is fixed, then the overlap of uniformly random shortest path trees sampled independently from $G_n^{(1)}$ and $G_n^{(2)}$ is $o(1)$ asymptotically almost surely. 
\end{theorem}

\begin{remark}[Two phases of evolution]
    As we have shown, when $0<\lambda<1$, the overlap of the uniformly random shortest path trees drops in two phases: first for $\rho_n = 1 - o(1)$ as we go from $\gamma=1$ to $\gamma=0$, and then in a second phase down to $\rho=0$. When $0<\gamma<1$, the set of vertices having the same distance is nontrivial and its proportion decays as $\gamma$ shrinks (see the first result of the theorem, and also Theorem~\ref{thm:d-chaos-intro}). Once $\gamma=0$, the set of vertices having the same distance stays roughly the same size. Instead the overlap $\rho$ of the edge sets governs the overlap of the shortest path trees, which is why the overlap is proportional to $\rho$.
    See Figure~\ref{fig:both_heatmaps} for visual illustration of both regimes. 
\end{remark}

\subsubsection{Numerical simulations match our predictions}

\begin{figure}[t]
    \centering
    \begin{subfigure}[b]{0.48\textwidth}
        \centering
        \includegraphics[width=\textwidth]{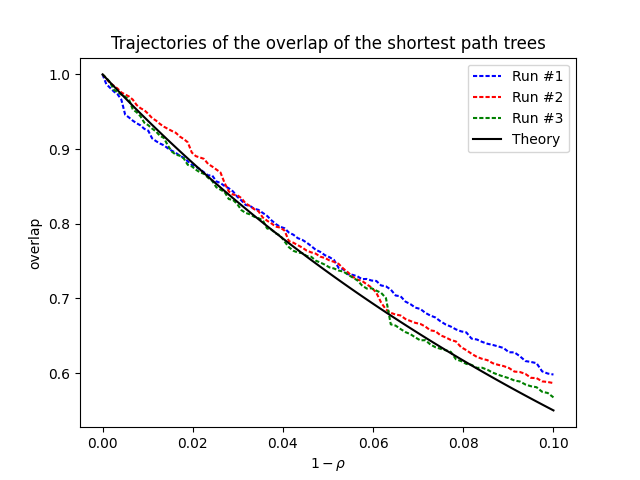}
        \caption{Overlap of the uniformly randomly sampled shortest path trees from the optimal coupling, plotted with the theoretical curve from Theorem~\ref{thm:correlated-intro}.}
        \label{fig:overlap_spt}
    \end{subfigure}
    \hfill
    \begin{subfigure}[b]{0.48\textwidth}
        \centering
        \includegraphics[width=\textwidth]{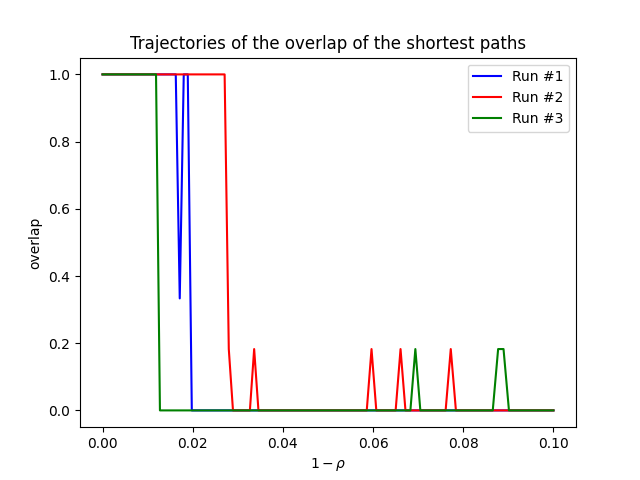}
        \caption{Overlap of the shortest paths from the root to a fixed vertex, where the paths are taken from the sampled shortest path trees.}
        \label{fig:overlap_sp}
    \end{subfigure}
    \caption{Numerical simulation of the overlap of uniformly random shortest path trees vs. shortest paths. It was conducted independently three times with parameters $n=10^5$ and $q=10^{-4}$.}
    \label{fig:overlap}
\end{figure}

To visualize the distinct differences in the behavior of the overlap, we conducted a numerical simulation  (Figure~\ref{fig:overlap}). The implementation details of the simulation are given in Section~\ref{sec:simulation-details}.

Figure~\ref{fig:overlap_spt} shows the simulation result for shortest path trees. The black curve indicates the theoretical curve which is computed using the formula
\[
    \frac{1-2\lambda_n+\lambda_n^{2-\gamma_n}}{1-\lambda_n}\cdot\lambda_n^{2-\lambda_n}+f(\lambda_n,\gamma_n)
\]
from Theorem~\ref{thm:correlated-intro},
which is the most accurate when $\rho_n$ is close to $1$. As expected, the simulated curve is very close to the theoretical curve 
and the overlap smoothly decreases as the correlation drops.
On the other hand, Figure~\ref{fig:overlap_sp} shows that the overlap of the shortest paths (contained in the corresponding trees) suddenly drops to zero, which agrees with the result of \cite{LS2024}.

\subsubsection{Gibbs measures, energy landscape, and Franz--Parisi potential}
We extend the analysis to compute the asymptotic free energy of the natural finite temperature analogue of shortest path trees of the component $\overline{G}_n$ containing $1$. This is a natural way from the statistical physics perspective to study the structure of the set of approximate optima (in this case, approximate shortest path trees).
To state the results, we first must define the natural finite temperature model at inverse temperature $\beta \ge 0$ which is
\begin{equation}\label{eqn:finite-temp-intro}
\mu_{G_n,\beta}(T) = \frac{1}{Z_{G_n,\beta}} \exp\left(-\beta \log\log(n)  \sum_{v\in\overline{V}_n} \mathsf{d}_T(1,v) \right)
\end{equation}
where $\overline{V}_n$ is the set of vertices reachable from $1$. This definition will be justified further in the preliminaries.
\begin{remark}
    In the case $\beta=0$, the measure $\mu_{G_n,0}$ is the uniform measure over the set of the \emph{spanning trees} of $\overline{G}_n$. In the other extreme $\beta=\infty$ (zero temperature limit), $\mu_{G_n,\infty}$ is the uniform measure over the set of the \emph{shortest path trees} of $\overline{G}_n$.
\end{remark}
We denote the ground state energy (minimum energy) of this system by
\[
    E_{G_n}:=\sum_{v\in\overline{V}_n}\mathsf{d}_G(1,v)
\]
which is simply the sum of distances from $1$ to other vertices in the same component $\overline{V}_n$. We will compute several thermodynamic quantities such as free energy density and the Franz--Parisi potential relative to this ground state.


\paragraph{The free energy density.} We define the \emph{relative free energy density} by
\[
    \mathcal{F}_{G_n,\beta}:=\frac{1}{n}\left(-\frac{1}{\beta\log\log n}\log Z_{G_n,\beta}-E_{G_n}\right)\,.
\]
The following theorem computes the limiting value of $\mathcal{F}_{G_n,\beta}$, which turns out to exhibit distinct behaviors depending on the proxy sequences $\{\lambda_n\}$ and $\{\Delta_n\}$. See Section~\ref{sec:free-energy-density} for the proof.

\begin{theorem}[Limiting free energy density]\label{thm:intro-gibbs-free-energy}
    Suppose that $\Delta_n\to\Delta\in[0,1]$ and $\lambda_n\to\lambda\in[0,1]$, and let $\beta>0$ be a constant. Then the free energy density converges in probability to
    \[
        \mathcal{F}_{G_n,\beta}\pto\mathcal{F}_{\lambda,\Delta}(\beta)
    \]
    where $\mathcal{F}_{\lambda,\Delta}$ is a continuous function of $\beta$ and is one of the following.
    \begin{enumerate}[label=(\Alph*)]
        \item\label{item:intro-free-regime1} If $\Delta=1$, then $\lambda=0$ and
        \[
            \mathcal{F}_{\lambda,\Delta}(\beta)=-\frac{1}{\beta}\,.
        \]
        \item\label{item:intro-free-regime2} If $\Delta\in[0, 1)$ and $\lambda\in[0, 1)$, then
        \[
            \mathcal{F}_{\lambda,\Delta}(\beta)=\begin{dcases}
                -\frac{\lambda+\Delta}{\beta}&\text{if $\beta\geq1-\Delta$,}\\
                \frac{1-(
                \lambda+\Delta)}{1-\Delta}-\frac{1}{\beta}&\text{if $\beta\leq1-\Delta$.}
            \end{dcases}
        \]
        
        \item\label{item:intro-free-regime3} If $\lambda=1$, then $\Delta=0$ and
        \[
            \mathcal{F}_{\lambda,\Delta}(\beta)=-\frac{1}{\beta}\,.
        \]

    \end{enumerate}
\end{theorem}

\begin{remark}[Approximating the log partition function]
    Computing $\mathcal{F}_{G,\beta}$ is equivalent to computing $\log Z_{G,\beta}$ up to an error of $o_p(n\log\log n)$. Through a more careful and delicate analysis, we established a better formula for the log-partition function accurate up to an $o(n)$ error, which is needed to prove the following results. Since the formula is rather sophisticated and requires a few definitions, we refer the readers directly to Section~\ref{sec:partition-function} for the precise statement.
\end{remark}

\begin{figure}
    \centering






    \begin{tikzpicture}

    \begin{axis}[
        at={(0,0.6cm)},
        axis x line=middle,
        xlabel=$1-\Delta$,
        ylabel=$\beta$,
        xlabel style={
            font=\footnotesize,
            at={(axis description cs:0.5,0)},
            anchor=north
        }, 
        ylabel style={font=\footnotesize}, 
        tick label style={font=\footnotesize}, 
        xmin=0, xmax=1,
        ymin=0, ymax=1.6,
        domain=0:1,
        samples=100,
        width=4cm,
        height=5cm,
        xtick={0,1},
        ytick={1},
    ]
    
    \addplot[draw=none, name path=zero] {0};
    \addplot[draw=none, domain=0:1, samples=50, name path=top] {1.5 + 0.02*sin(deg(x*20))};
    \addplot[black, thick, domain=0:1, samples=50, name path=f] {x};
    
    \addplot[red, opacity=0.3] fill between[
        of=f and zero,
        soft clip={domain=0:1},
    ];
    
    \addplot[blue, opacity=0.3] fill between[
        of=f and top,
        soft clip={domain=0:1},
    ];
    
    
    
    \addplot[black, dotted, thin] coordinates {(0,1) (1,1)};
    
    \coordinate (Ab) at (1,0);
    \coordinate (Am) at (1,1);
    \coordinate (At) at (1,1.5);
    
    \end{axis}
    
    \begin{axis}[
        at={(3cm,0.2cm)},
        axis x line=middle,
        xlabel=$\lambda$,
        xlabel style={
            font=\footnotesize,
            at={(axis description cs:0.5,0)},
            anchor=north
        }, 
        ylabel style={font=\footnotesize}, 
        tick label style={font=\footnotesize}, 
        xmin=0, xmax=1,
        ymin=0, ymax=1.7,
        domain=0:1,
        samples=100,
        width=4cm,
        height=5.3cm,
        xtick={0,1},
        ytick=\empty,
    ]
    
    \addplot[draw=none, name path=zero] {0};
    \addplot[draw=none, domain=0:1, samples=50, name path=top] {1.6 + 0.02*sin(deg(x*20))};
    \addplot[black, thick, domain=0:1, samples=50, name path=f] {1};
    
    \addplot[red, opacity=0.3] fill between[
        of=f and zero,
        soft clip={domain=0:1},
    ];
    
    \addplot[blue, opacity=0.3] fill between[
        of=f and top,
        soft clip={domain=0:1},
    ];
    
    
    
    \coordinate (Bb) at (0,0);
    \coordinate (Bm) at (0,1);
    \coordinate (Bt) at (0,1.5);
    
    \coordinate (Cb) at (1,0);
    \coordinate (Cbf) at (0,0.14);
    \coordinate (Cm) at (1,1);
    \coordinate (Ct) at (1,1.5);
    \coordinate (Ctf) at (0,1.6);
    
    \end{axis}
    
    \begin{axis}[
        at={(6.5cm,-0.5cm)},
        axis x line=middle,
        xlabel=$\kappa^{-1}$,
        xlabel style={
            font=\footnotesize,
            at={(axis description cs:0.5,0)},
            anchor=north
        }, 
        ylabel style={font=\footnotesize}, 
        tick label style={font=\footnotesize}, 
        xmin=0, xmax=2.5,
        ymin=0, ymax=1.8,
        domain=0:2.5,
        samples=100,
        width=8cm,
        height=5.7cm,
        xtick={0,1},
        ytick=\empty,
    ]
    
    \addplot[draw=none, name path=zero] {0};
    \addplot[draw=none, domain=0:2.5, samples=100, name path=top] {1.7 + 0.02*sin(deg(x*20))};
    \addplot[black, dashed, thick, domain=0:1, samples=50, name path=f] {1-x};
    
    \addplot[red, opacity=0.3] fill between[
        of=f and zero,
        soft clip={domain=0:1},
    ];
    
    \addplot[blue, opacity=0.3] fill between[
        of=f and top,
        soft clip={domain=0:1},
    ];
    
    \addplot[blue, opacity=0.3] fill between[
        of=zero and top,
        soft clip={domain=1:2.5},
    ];
    
    
    \coordinate (Db) at (0,0);
    \coordinate (Dm) at (0,1);
    \coordinate (Dt) at (0,1.5);
    
    \coordinate (Eb) at (0,0.25);
    \coordinate (Em) at (2.5,1);
    \coordinate (Et) at (0,1.7);
    
    \end{axis}

    \draw[black, dotted, thin] (At) -- (Bt);
    \draw[black, dotted, thin] (Am) -- (Bm);
    \draw[black, dotted, thin] (At) -- (Ctf);
    \draw[black, dotted, thin] (Ab) -- (Bb);
    \draw[black, dotted, thin] (Ab) -- (Cbf);
    
    \draw[black, dotted, thin] (Ct) -- (Dt);
    \draw[black, dotted, thin] (Cm) -- (Dm);
    \draw[black, dotted, thin] (Ct) -- (Et);
    \draw[black, dotted, thin] (Cb) -- (Db);
    \draw[black, dotted, thin] (Cb) -- (Eb);

    \draw[|-|, thick] (0,-1.5) -- (2.4,-1.5) node[midway, below=4pt]{(A)};
    \draw[|-|, thick] (3,-1.5) -- (5.4,-1.5) node[midway, below=4pt]{(B)};
    \draw[|-|, thick] (6.45,-1.5) -- node[midway, below=4pt]{(C.i)} (6.5,-1.5);
    \draw[|-|, thick] (6.5,-1.5) -- node[midway, below=4pt]{(C.ii)} (9.1,-1.5);
    \draw[|-|, thick] (9.1,-1.5) -- node[midway, below=4pt]{(C.iii)} (12.9,-1.5);

\end{tikzpicture}

    \caption{Phase diagram in the limit $n\to\infty$. The second diagram is obtained by ``zooming in'' the first diagram at $1-\Delta=1$, and the third diagram is from zooming in the second diagram at $\lambda=1$. The labels for the regions correspond to the regimes \ref{item:intro-regime1} through \ref{item:intro-regime3-3}. The \textcolor{blue}{blue region} is the low-temperature phase where the Gibbs measure is close to the uniform measure in Wasserstein distance, and the \textcolor{red}{red region} is the high-temperature phase where they are far apart. Solid black lines indicate a discontinuous (``first-order'') phase transition, and dashed lines indicate a continuous (``second-order'') phase transition. The point at which the line changes from solid to dashed is precisely where the parameter regime transitions from $\lambda<1$ to $\lambda=1$.}
    \label{fig:phase-diagram}
\end{figure}
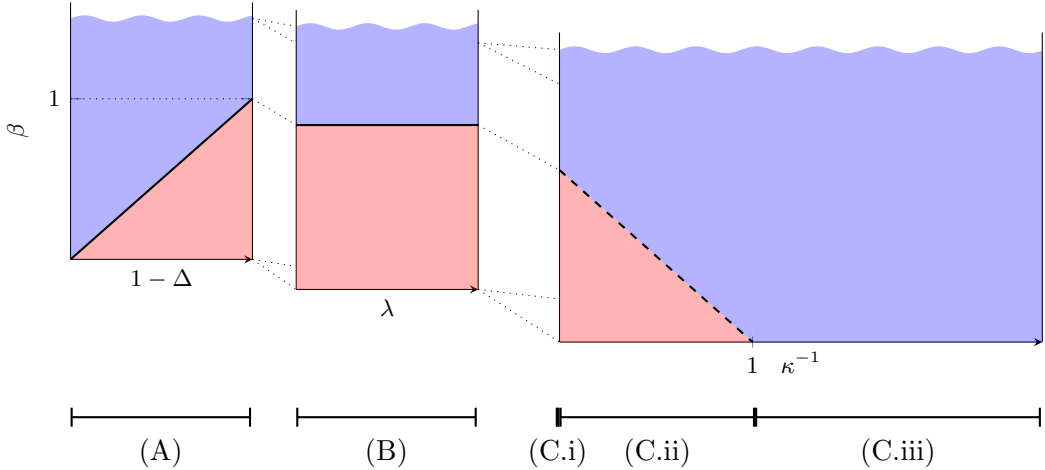

\paragraph{Phase transitions.} A particularly interesting regime in Theorem~\ref{thm:intro-gibbs-free-energy} is \ref{item:intro-free-regime2}, where the free energy density has a non-differential point at $\beta=1-\Delta$. At this point, the model exhibits what is called a \emph{first-order} or \emph{discontinuous phase transition} in the inverse temperature $\beta$ (see, e.g., \cite{chaikin1995principles,friedli2017statistical} for more background on this terminology). Here, we describe the phase transitions in our system from a different perspective; namely, we analyze how the Gibbs measure $\mu_{G_n,\beta}$ changes with $\beta$ as an element of the space of probability measures equipped with the $1$-Wasserstein metric. Interestingly, this viewpoint reveals phase transition phenomena observable in the seemingly bland regime \ref{item:intro-free-regime3}, depending on the additional proxy
\[
    \kappa_n:=(1-\lambda_n)\log\log n\,.
\]
This is reflected in the following formal result. 

\begin{theorem}[Phase transitions]\label{thm:intro-phase-transition}
    Suppose that $\Delta_n\to\Delta\in[0,1]$, $\lambda_n\to\lambda\in[0,1]$, and $\kappa_n\to\kappa\in[0,\infty]$. Let $\beta>0$ be a constant and define
    \begin{equation}\label{eqn:intro-critical-temp}
        \beta_c=1-\Delta-\kappa^{-1}\in[-\infty,1]\,.
    \end{equation}
    Then we have the following.
    \begin{itemize}
        \item (Low temperature phase) If $\beta>\beta_c$, then
        \[
            \frac{1}{n}W_1(\mu_{G_n,\infty},\mu_{G_n,\beta})\pto0\,.
        \]
        In other words, the Gibbs measure $\mu_{G_n,\infty}$ is asymptotically almost surely close to the uniform measure over shortest path trees in $1$-Wasserstein metric under the Hamming metric\footnote{As a result, we have an approximate sampler for the entire low temperature phase in Wasserstein by sampling a uniformly random shortest path tree. This is the same notion of Wasserstein sampling as studied in, e.g., \cite{el2022sampling,alaoui2023sampling}. The same type of approximation also appears in \cite{eldan2018decomposition,austin2019structure}.}
        \item (High temperature phase) If $\beta<\beta_c$, then there is a constant $C_{\delta}>0$ depending only on $\delta=\beta_c-\beta$ such that
        \[
            \frac{1}{n}W_1(\mu_{G_n,\infty},\mu_{G_n,\beta})>C_\delta
        \]
        with probability at least $1-o(1)$.
    \end{itemize}
\end{theorem}

The proof can be found in Sections~\ref{sec:gibbs-low-phase} and \ref{sec:gibbs-high-phase}. Theorem~\ref{thm:intro-phase-transition} gives the precise value of the (inverse) \emph{critical temperature} $\beta_c$, above which (lower temperature) the system behaves like the uniform measure in Wasserstein sense. We apply Theorem~\ref{thm:intro-phase-transition} to each of the following regimes, which extends the list in Theorem~\ref{thm:intro-gibbs-free-energy} by splitting \ref{item:intro-free-regime3} depending on $\kappa$.
\begin{enumerate}[label=(\Alph*)]
    \item\label{item:intro-regime1} $\Delta=1$. This implies $\lambda=0$ and $\kappa=\infty$, so $\beta_c=0$. Thus, $\mu_{G_n,\beta}$ always behaves like the uniform measure.
    \item\label{item:intro-regime2} $\Delta\in[0,1)$ and $\lambda\in[0,1)$. Here, phase transition occurs at $\beta_c=1-\Delta$.
    \item\label{item:intro-regime3} $\lambda=1$. This gives $\Delta=0$, and we further split this into three sub-regimes depending on $\kappa$.
    \begin{enumerate}[label=(C.\roman*), ref=(C.\roman*), align=left]
        \item\label{item:intro-regime3-1} $\kappa=\infty$. Phase transition occurs at $\beta_c=1$.
        \item\label{item:intro-regime3-2} $\kappa\in(1,\infty)$. Phase transition occurs at $\beta_c=1-\kappa$.
        \item\label{item:intro-regime3-3} $\kappa\in[0,1]$. We always have $\beta_c\leq0$, so $\mu_{G_n,\beta}$ is always close to the uniform measure, similar to the regime \ref{item:intro-regime1}.
    \end{enumerate}
\end{enumerate}
We see that phase transition in Wasserstein space occurs in the regimes \ref{item:intro-regime2}, \ref{item:intro-regime3-1}, and \ref{item:intro-regime3-2}. It turns out that, the phase transition in \ref{item:intro-regime3-2} is a \emph{continuous phase transition}, as opposed to \ref{item:intro-regime2} and \ref{item:intro-regime3-1} which exhibit a \emph{discontinuous phase transition}. This is formalized in the following theorem.
\begin{theorem}
    Suppose that $\lambda_n\to\lambda\in[0,1]$, $\Delta_n\to\Delta\in[0,1]$, and $\kappa_n\to\kappa\in[0,\infty]$.
    \begin{itemize}
        \item (Discontinuous phase transition) In regimes \ref{item:intro-regime2} and \ref{item:intro-regime3-1}, there exists a universal constant $C>0$ such that for all constant $\delta>0$, we can find a constant $\beta>0$ with $|\beta-\beta_c|<\delta$ satisfying
        \[
            \frac{1}{n}W_1(\mu_{G_n,\beta},\mu_{G_n,\beta_c})\geq C
        \]
        with probability at least $1-o(1)$.

        \item (Continuous phase transition) In regime \ref{item:intro-regime3-2}, for any constant $\delta>0$, any $\beta>0$ with $|\beta-\beta_c|<\delta$ satisfies
        \[
            \frac{1}{n}W_1(\mu_{G_n,\beta},\mu_{G_n,\beta_c})\leq\kappa\sqrt{12\delta}
        \]
        with probability at least $1-o(1)$.
    \end{itemize}
\end{theorem}
The proof is presented in Section~\ref{sec:phase-transition}. See Figure~\ref{fig:phase-diagram} for the phase diagrams and Figure~\ref{fig:potential-well} for an overall picture of the potential well, which summarize our results. Figure~\ref{fig:potential-well} is in a similar spirit to Figure 1 of \cite{zdeborova2010generalization} used to illustrate ``state following'', which as they explain is closely related to the Franz--Parisi potential.  

\begin{figure}
    \centering

    \begin{tikzpicture}[scale=0.7]
    

\begin{scope}[xshift=-7cm]

\node[anchor=south] at (2.8,4.1) {Regime (A)};
\draw[thick]
  (0,4) -- (0,0) -- (5.6,0) -- (5.6,4);

\draw[blue,dotted,thick] (1.1,0) -- (4.5,0) node[midway,below,blue]{LT (Ground) state};

\end{scope}


\node[anchor=south] at (2.8,4.1) {Regimes (B) and (C.i)};
\draw[thick]
  (0,4) -- (0,2.4) .. controls (0,2) and (0.16,1.8) .. (0.32,1.6) -- (1.5,0.2) --
  (4.1,0.2) -- (5.28,1.6) .. controls (5.44,1.8) and (5.6,2) .. (5.6,2.4) -- (5.6,4);

\draw[blue,dotted,thick] (1.1,0.2) -- (4.5,0.2) node[midway,below,blue]{LT (Ground) state};

\draw[red,dash dot,thick] (0,2.4) -- (5.6,2.4) node[midway,below,red]{HT state};
\draw[red,dash dot,thick] (0.32,1.6) -- (5.28,1.6);

\fill[red,opacity=0.1]
  (5.6,2.4) .. controls (5.6,2) and (5.44,1.8) .. (5.28,1.6) --
  (0.32,1.6) .. controls (0.16,1.8) and (0,2) .. (0,2.4) -- cycle;

\fill[gray,opacity=0.2]
  (0.32,1.6) -- (1.5,0.2) -- (4.1,0.2) -- (5.28,1.6) -- cycle;


\begin{scope}[xshift=7cm]
\node[anchor=south] at (2.8,4.1) {Regime (C.ii)};
\draw[thick]
  (0,4) -- (0,2.4) .. controls (0,2) and (0.16,1.8) .. (0.32,1.6) --
  (5.28,1.6) .. controls (5.44,1.8) and (5.6,2) .. (5.6,2.4) -- (5.6,4);

\draw[blue,dotted,thick] (0,1.6) -- (5.6,1.6) node[midway,below,blue]{LT (Ground) state};

\draw[red,dash dot,thick] (0,2.4) -- (5.6,2.4) node[midway,below,red]{HT state};

\fill[red,opacity=0.1]
  (5.6,2.4) .. controls (5.6,2) and (5.44,1.8) .. (5.28,1.6) --
  (0.32,1.6) .. controls (0.16,1.8) and (0,2) .. (0,2.4) -- cycle;

\end{scope}

\end{tikzpicture}

    \caption{Visualization of the potential well for different regimes. In the middle well, the high temperature state and the low temperature state are separated by a gray region with edges of constant slope $\beta_c^{-1}$, resulting in a discontinuous phase transition. For the rightmost one, the two states touch each other, leading to a continuous phase transition.}
    \label{fig:potential-well}
\end{figure}
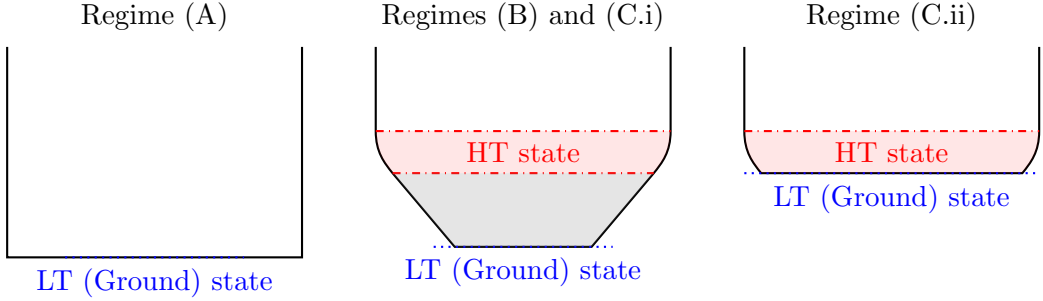


\paragraph{The Franz--Parisi potential.} For a fixed inverse temperature $\beta>0$, suppose that the random spanning trees $T_n$ are (independently) drawn from the Gibbs measures $\mu_{G_n,\beta}$ conditioned on $\{G_n\}$. Recall the definition \eqref{eqn:defn-overlap} of the overlap function for trees denoted by $\mathsf{R}$. We define the \emph{relative Franz--Parisi potential} by
\[
    \mathcal{F}_{\beta}^{\text{FP}}(r):=\lim_{\epsilon\to0^+}\plim_{n\to\infty}\frac{1}{n}\left(-\frac{1}{\beta\log\log n}\log Z_{n,r,\epsilon}-E_{G_n}\right)
\]
where
\[
    Z_{n,r,\epsilon}:=\sum_{T:\mathsf{R}(T_,T_n)\in[r-\epsilon,r+\epsilon]}\exp\left(-\beta\log\log (n)\sum_{v\in\overline{V}_n}\mathsf{d}_T(1,v)\right)
\]
is the partition function for trees overlapping with $T_n$ by approximately $r$. We precisely computed the Franz--Parisi potential except for all regimes where it exists; Figure~\ref{fig:fpp} depicts a rough landscape of $\mathcal{F}_{\beta}^{\text{FP}}$ for different regimes. More precise statements computing $\mathcal{F}_{\beta}^{\text{FP}}$ and justifying the landscape can be found Section~\ref{sec:franz-parisi}. 

\paragraph{Replica symmetry.} One simple consequence of our analysis is that away from the critical temperature, the model is always \emph{replica symmetric} in the following sense: to leading order, the normalized overlap between two samples from the Gibbs measure concentrates around its expectation. A formal statement is given in the following theorem.
\begin{theorem}[Replica symmetry]\label{thm:intro-replica-symmetry}
    Let $\beta>0$ be fixed and assume $\lambda_n\to\lambda\in[0,1]$. Suppose that for given $\{G_n\}$ we independently sample $T_n,T_n'\sim\mu_{G_n,\beta}$. Then
    \[
        \mathsf{R}(T_n,T_n')\pto\begin{dcases}
            0&\text{if $\lambda\in\{0,1\}$ or $\beta<1$,}\\
            f(\lambda)&\text{if $\lambda\in(0,1)$ and $\beta>1$}
        \end{dcases}
    \]
    where
    \[
        f(\lambda)=(1-\lambda)\E[1/X\mid X>0]\,,\quad X\sim\Pois(\log(1/\lambda))\,.
    \]
\end{theorem}
This is proved in Section~\ref{sec:replica-symmetry}. We remark that the expected overlap is nonzero if $\lambda\in(0,1)$ and $\beta>1$, which corresponds to the low-temperature Regime (B) in Figure~\ref{fig:fpp}.

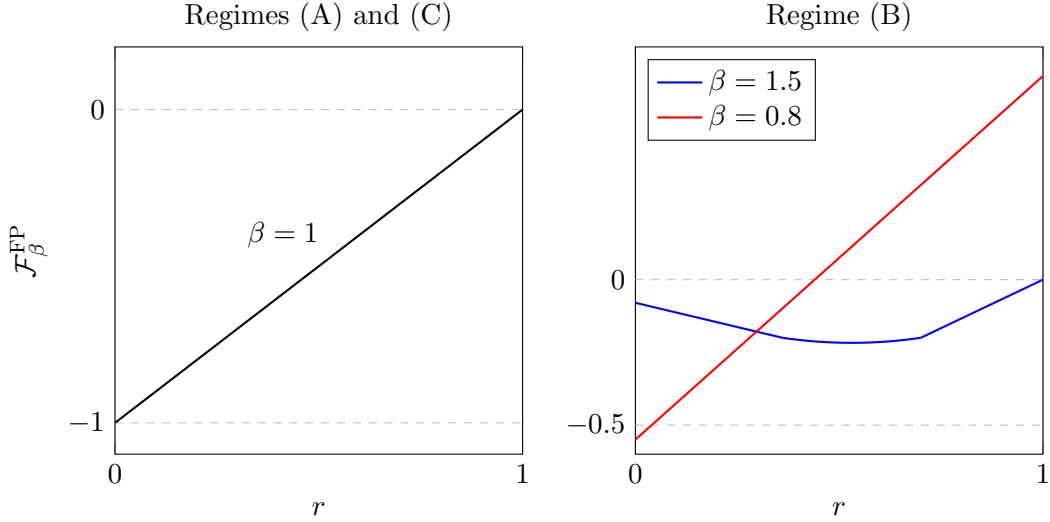
\begin{figure}
    \centering

    \begin{tikzpicture}
\begin{groupplot}[
  group style = {group size=2 by 1, horizontal sep=1.5cm},
  width=7cm, height=7cm,
  xmin=0, xmax=1,
  axis lines=box,
  xtick={0,1},
  xlabel={$r$},
  ytick=\empty,
  grid=both,
  title style={yshift=-4pt}
]

\nextgroupplot[
  title={Regimes (A) and (C)},  
  ymin=-1.1, ymax=0.2,
  ylabel={$\mathcal{F}_{\beta}^{\text{FP}}$},
  extra y ticks={-1,0},
  extra y tick style={
    major tick length=0pt,
    grid style={dashed}
  },
  legend pos=north west
]
\pgfmathsetmacro{\b}{1}
\addplot[thick, domain=0:1, samples=50] {(x-1)/\b};
\node[anchor=west] at (axis cs:0.3,-0.4) {$\beta=1$};

\pgfmathsetmacro{\lamb}{0.3}
\pgfmathsetmacro{\bl}{1.5}
\pgfmathsetmacro{\bh}{0.8}
\nextgroupplot[
  title={Regime (B)},
  ymin=-0.6, ymax=0.8,
  extra y ticks={-0.5,0},
  extra y tick style={
    major tick length=0pt,
    grid style={dashed}
  },
  legend pos=north west
]
\pgfmathsetmacro{\lamb}{0.3}
\pgfmathsetmacro{\bl}{1.5}
\pgfmathsetmacro{\bh}{0.8}
\addplot[blue, thick, domain=1-\lamb:1, samples=2, forget plot] {-1/\bl+x/\bl};
\addplot[blue, thick, domain=\lamb*ln(1/\lamb):1-\lamb, samples=50, forget plot] {-\lamb/\bl+0.6*(x-1+\lamb)*(x-\lamb*ln(1/\lamb))};
\addplot[blue, thick, domain=0:\lamb*ln(1/\lamb), samples=2] {(-\lamb+(\bl-1)*\lamb*ln(1/\lamb))/\bl-(\bl-1)*x/\bl};
\addlegendentry{$\beta=\bl$}
\addplot[red, thick, domain=0:1, samples=2] {1-\lamb-1/\bh+x/\bh};
\addlegendentry{$\beta=\bh$}


\end{groupplot}
\end{tikzpicture}

    \caption{A plot of the Franz--Parisi potential against the correlation $r\in[0,1]$ for different regimes. The $y$-axes (for $\mathcal{F}_\beta^{\text{FP}}$) are scaled differently to better visualize the landscape. For the right plot we choose $\lambda=0.3$ (and $\Delta=0$). A low temperature phase $\beta=1.5$ in the right plot has its unique \emph{quasiconvex} landscape which is quite different from the other cases which simply increase in $r$.}
    \label{fig:fpp}
\end{figure}

\begin{figure}[p]

    \centering
    \begin{subfigure}{\textwidth}
 \hspace*{1.75cm}
        \includegraphics[width=0.9\textwidth]{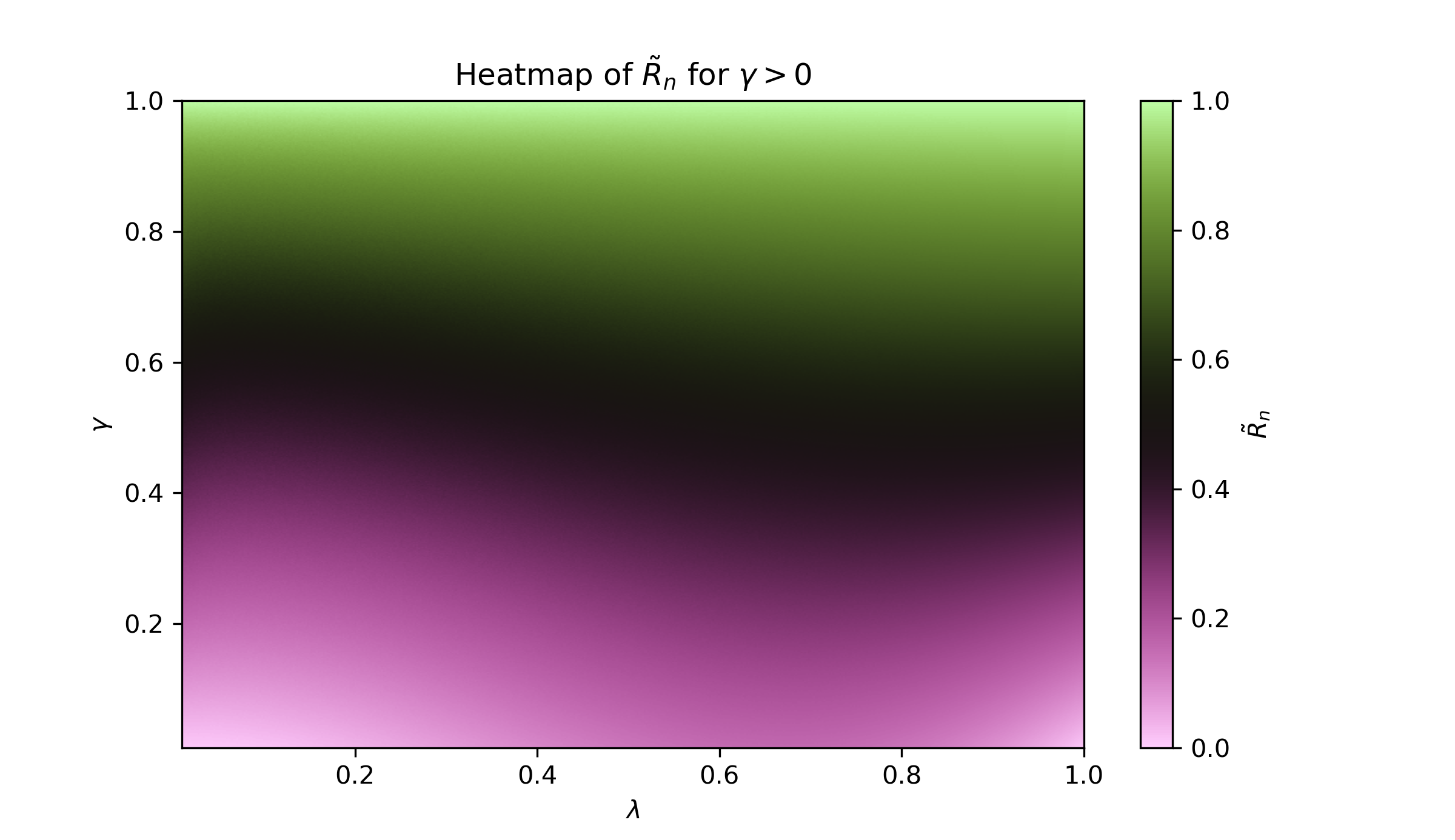}
        \label{fig:heatmap_gamma}
    \end{subfigure}

    
    \begin{subfigure}{\textwidth}
 \hspace*{1.75cm}
        \includegraphics[width=0.9\textwidth]{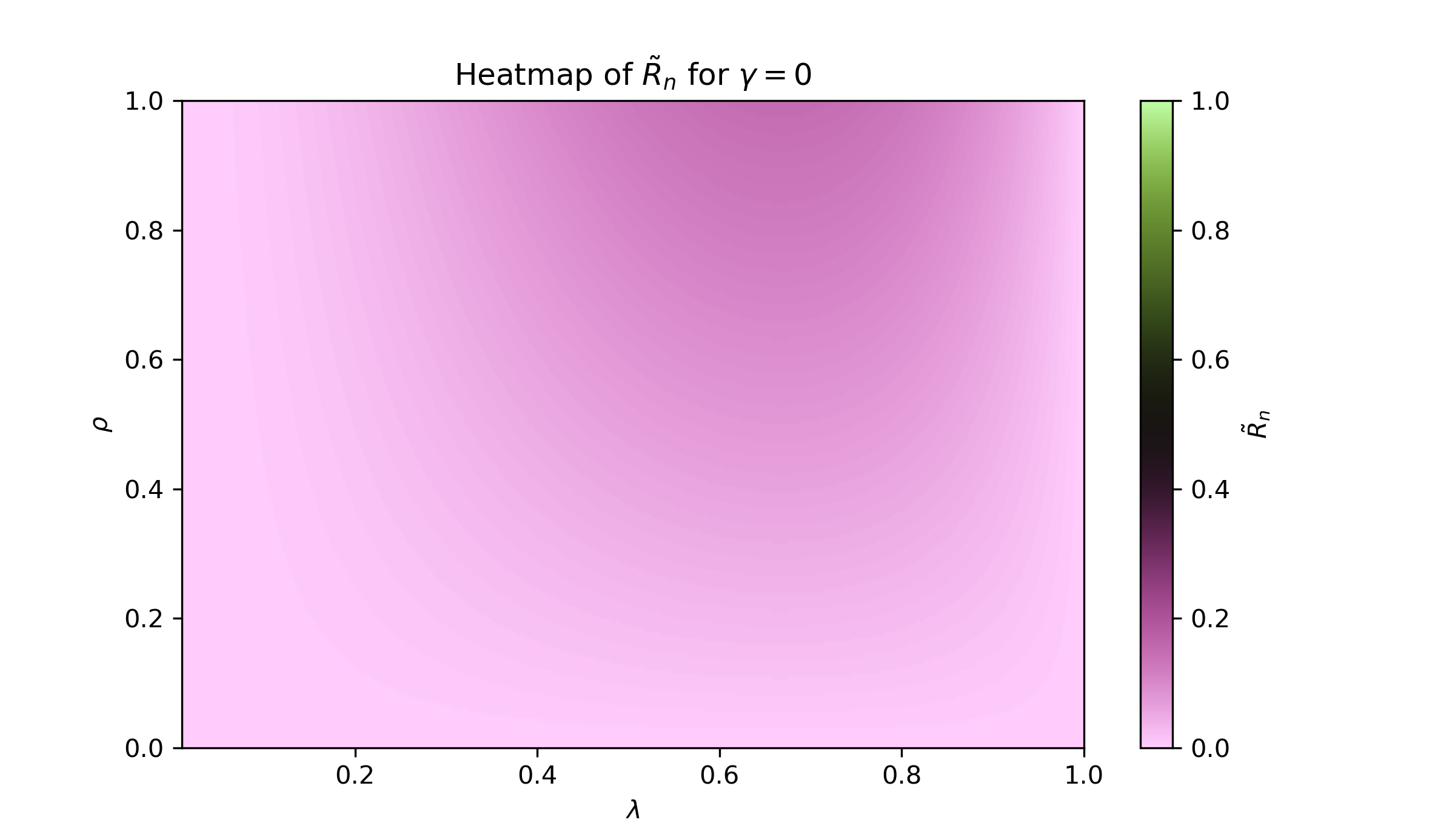}
        \label{fig:heatmap_rho}
    \end{subfigure}
    \vspace{-1em}
    \caption{Uniformly random spanning tree overlap $\tilde{R}_n$ in its two regimes. Top: $\gamma > 0$, so $\rho = 1$. Bottom: $\gamma = 0$, so $\rho \in [0,1]$. The bottom of the top figure and top of the bottom figure coincide.} 
    \label{fig:both_heatmaps}
\end{figure}

\subsection{Technical overview}\label{sec:tech-overview}


In this section, we give a high level overview of our analysis along with some informal arguments. Throughout this section, in the interest of simplicity, we ignore any subtleties that might result from the graph $G_n$ being disconnected. 


\begin{figure}
    \centering
    \includegraphics[width=0.5\linewidth]{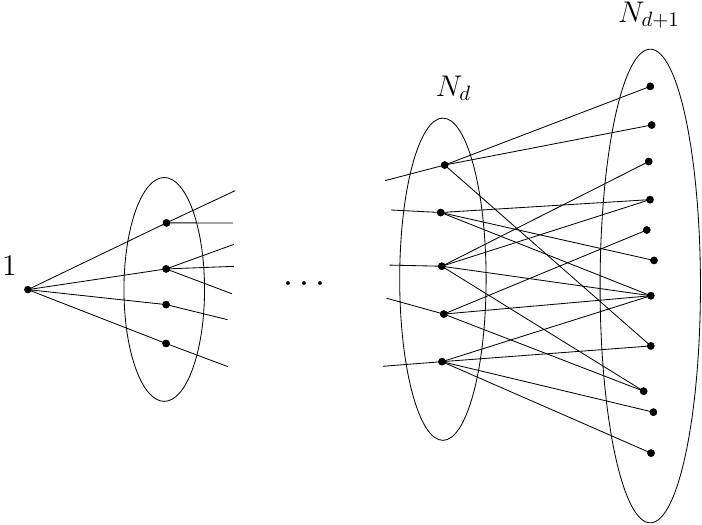}
    \caption{An illustration of the BFS process.}
    \label{fig:bfs}
\end{figure}


\subsubsection{Asymptotics of correlated random graphs}
\paragraph{The distribution of distances and BFS.} 
Before analyzing the correlated graph setting relevant to the OGP, it helps to carefully understand what happens with a single graph.
As illustrated in Figure~\ref{fig:bfs}, we sample/reveal the Erd\"os--R\'enyi graph $G=G_n$ progressively via Breadth-First Search (BFS), as in the literature regarding the diameter of random graphs \cite{chung2001diameter,riordan2010diameter}. 
Conditioned on the set of vertices of distance up to $d$ from a fixed source vertex $1$, we next reveal which vertices are distance $d+1$. The edges between the remaining vertices and the vertices of distance $d$ are sampled independently of all the observed edges/non-edges so far. Therefore, the number $N_d$ of the vertices of distance $d$ follows the stochastic process
\begin{equation}\label{eqn:nd_process}
    \begin{split}
        N_0 &= 1\\
        N_1 &\sim\Binom(n-N_0, q_n)\\
        &\vdots\\
        N_{d} &\sim\Binom\left(n-\sum_{i=0}^{d-1}N_i, 1-(1-q_n)^{N_{d-1}}\right)\\
        &\vdots
    \end{split}
\end{equation}
which terminates once the BFS has exhausted its connected component.

Based on well-known properties of sparse random graphs, we know that for small depths $d$ the graph is locally tree-like, so $N_d$ will concentrate around $(nq_n)^d$. Let's review why this occurs, so we can understand what happens at large depths. Suppose that $N_k$ are given for all $k<d$ and it is known that $N_k\approx(nq)^k$. Then from \eqref{eqn:nd_process}, $N_d$ would concentrate around
\begin{equation}\label{eqn:nd-conc}
    N_d\approx\left(n-\sum_{i=0}^{d-1}N_i\right)\left(1-(1-q_n)^{N_{d-1}}\right)\,.
\end{equation}
If $d$ is small enough, then we should have 
\begin{equation}\label{eqn:n-apx}
    n-\sum_{i=0}^{d-1}N_i\approx n-\sum_{i=0}^{d-1}(nq_n)^{i}\approx n
\end{equation}
and
\begin{equation}\label{eqn:q-apx}
    1-(1-q_n)^{N_{d-1}}\approx 1-e^{-q_n(nq_n)^{d-1}}\approx q_n(nq_n)^{d-1}
\end{equation}
so $N_d\approx(nq_n)^d$. 

However, this starts to fail when we run low on unexplored vertices, i.e. when $(nq_n)^d$ is close to $n$, equivalently when $d$ is close to $\ell_n^*=\frac{\log n}{\log(nq_n)}$. Suppose now that we have reached $d=d_n^*$. By the definition \eqref{eqn:dstar} of $d_n^*$, we still have $N_{d_n^*-1}\approx(nq_n)^{d_n^*-1}<\frac{n}{(\log\log n)^2}=o(n)$, so \eqref{eqn:n-apx} still holds true. 
On the other hand, the last approximation in \eqref{eqn:q-apx} fails if $q_n(nq_n)^{d_n^*-1}=\Omega(1)$. 
This is the case where, using the notation of \eqref{eqn:lambda}, we have that 
$\lambda_n=e^{-q_n(nq_n)^{d_n^*-1}} = \Omega(1)$,
and $N_{d_n^*}$ occupies a nontrivial portion of the vertices. This causes \eqref{eqn:n-apx} to now fail, and instead we
observe that if $N_{d_n^*}$ does not exhaust most of the remaining vertices, then $N_{d_n^*+1}$ will.
The formal version of this analysis is detailed in Proposition~\ref{prop:conc}; it requires some care since the number of layers $d_n^*$ is diverging with $n$, and we need fairly precise bounds to enable the analysis in what follows.



\begin{figure}
    \centering
    \includegraphics[width=0.5\linewidth]{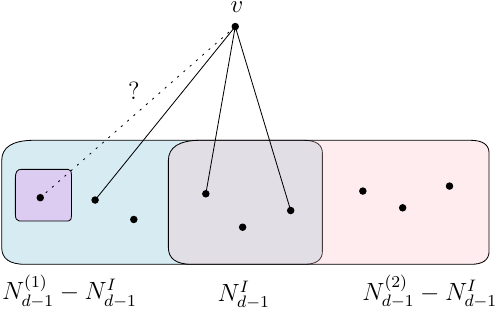}
    \caption{The ``adjacency probabilities'' problem in correlated graphs. The blue region in the left and the pink region in the right are those explored at step $d-1$ for $G_n^{(1)}$ and $G_n^{(2)}$, respectively. If the vertices in the purple box in the blue region is previously explored in $G_n^{(2)}$, then it is already known that they are not connected to the new vertex $v$.}
    \label{fig:bad}
\end{figure}

\paragraph{Correlated graphs analysis.} We now build upon the ideas above to analyze the more difficult setting of two $\rho_n$-correlated graphs $G_n^{(1)}$ and $G_n^{(2)}$.
We are now interested in not only $N_d^{(1)}$ and $N_d^{(2)}$, which denote $N_d$ for $G_n^{(1)}$ and $G_n^{(2)}$ respectively, but also $N_d^I$ --- the set of vertices that have distance $d$ in both. 
Suppose that we are running a BFS on both graphs simultaneously up to distance $d-1$, and we are currently about to explore the vertices at distance $d\leq d_n^*$. By \eqref{eqn:n-apx}, most of the vertices are still unexplored (in both graphs) and among those we are searching for vertices at distance $d$ in both $G_n^{(1)}$ and $G_n^{(2)}$. We can categorize each vertex $u$ explored in the previous step $d-1$ into the three categories:
\begin{enumerate}[label=(\alph*)]
    \item\label{item:only-1} $u$ has depth $d-1$ in $G_n^{(1)}$ but not in $G_n^{(2)}$, counted in $N_{d-1}^{(1)}-N_{d-1}^I$;
    \item\label{item:only-2} $u$ has depth $d-1$ in $G_n^{(2)}$ but not in $G_n^{(1)}$, counted in $N_{d-1}^{(2)}-N_{d-1}^I$;
    \item\label{item:both} $u$ has depth $d-1$ in both $G_n^{(1)}$ and $G_n^{(2)}$, counted in $N_{d-1}^I$.
\end{enumerate}
For each unexplored vertex $v$, we would like to express the probability of $v$ being counted towards $N_d^I$ in terms of $(N_{d-1}^{(1)}, N_{d-1}^{(2)}, N_{d-1}^I)$. Marginally, any two vertices are adjacent only in $G_n^{(1)}$ with probability $q_n(1-\rho_n)$, only in $G_n^{(2)}$ with the same probability, and both in $G_n^{(1)}$ and $G_n^{(2)}$ with probability $q_n\rho_n$. It is tempting to apply these ``adjacency probabilities'' to $u$ and $v$, but 
these adjacency probabilities might be \emph{incorrect} if $u$ is in \ref{item:only-1} or \ref{item:only-2}! For instance, if $u$ in \ref{item:only-1} has depth $d-1$ in $G_n^{(1)}$ but $d-2$ in $G_n^{(2)}$, then the fact that $v$ is unexplored already implies that $u$ and $v$ are not adjacent in $G_n^{(2)}$. This in turn implies that, due to correlation, $u$ and $v$ are unlikely to be adjacent in $G_n^{(1)}$ either. See Figure~\ref{fig:bad} for a visualization of this issue.

To handle this issue, we observe that those satisfying these adjacency probabilities are in the following two subsets of \ref{item:only-1} and \ref{item:only-2}, along with \ref{item:both}:
\begin{enumerate}[label=(\Alph*)]
    \item\label{item:only-1u} $u$ has depth $d-1$ in $G_n^{(1)}$ and is \emph{unexplored} in $G_n^{(2)}$;
    \item\label{item:only-2u} $u$ has depth $d-1$ in $G_n^{(2)}$ and is \emph{unexplored} in $G_n^{(1)}$.
\end{enumerate}
A key claim is that most of the vertices in \ref{item:only-1} and \ref{item:only-2} are in fact \ref{item:only-1u} and \ref{item:only-2u}, i.e., it is very unlikely that a vertex explored in $G_n^{(1)}$ at step $d-1$ happens to have already been explored in $G_n^{(2)}$ earlier. Since those ``bad'' vertices not conforming to the adjacency probabilities always negatively impact $N_d^I$, this claim allows us to stochastically squeeze $N_d^I$ between two random variables which are easier to establish concentration bounds. This argument is rigorously given in Proposition~\ref{prop:rhoconc}.

There are two contributions to $N_d^I$ coming from the previous step. One is \emph{contribution by correlation}, stemming from common edges of $G_n^{(1)}$ and $G_n^{(2)}$ connected to $N_{d-1}^I$, where each common edge occurs with probability $\approx q_n\rho_n$. The other is \emph{contribution by randomness}, which appears when $v$ happens to be adjacent to depth $d-1$ vertices in $G_n^{(1)}$ and in $G_n^{(2)}$ through separate edges; such pair of edges occurs with probability $\approx q_n^2$. When $\rho_n$ is large enough and $N_{d-1}^I$ constitute a non-negligible portion of $N_{d-1}^{(1)}$ and $N_{d-1}^{(2)}$, the contribution by correlation is non-negligible and yields $N_d^I\approx(nq_n\rho_n)^d$ as long as $d<d_n^*$ and $\rho_n^d=\Theta(1)$. However, if $\rho_n^d=o(1)$, then the contribution by correlation diminishes, resulting in \emph{correlation decay}. Hence, the asymptotics of $\gamma_n=\rho_n^{d_n^*}$ defined in \eqref{eqn:defgamma} plays an important role in describing what happens at step $d=d_n^*$. Indeed, Theorem~\ref{thm:correlated-intro} states that
\[
    \frac{N_{d_n^*}^I}{n}\pto1-2\lambda+\lambda^{2-\gamma}\,.
\]
When $\gamma=1$, the RHS is $1-\lambda$ which equals the single graph asymptotics in Theorem~\ref{thm:single-intro}, demonstrating the dominance of the contribution by correlation. On the other hand, when $\gamma=0$, the RHS is precisely $(1-\lambda)^2$, showing that the correlation decays completely and the contribution by randomness dominates.

\subsubsection{Overlap of shortest path trees for correlated graphs}\label{sec:to-overlap}



\paragraph{From distance asymptotics to shortest path trees.} Once we know the distance $\mathsf{d}_G(1,v)$ for all $v\in V$, then every shortest path tree can be constructed by picking for each $v\in V$ its \emph{parent} $u$ in the neighborhood of $v$ satisfying $\mathsf{d}_G(1,u)+1=\mathsf{d}_G(1,v)$. We denote the set of parent candidates of $v$ by
\[
    \mathsf{par}_G(v):=\{u\in V:\text{$u$ and $v$ are adjacent in $G$ and $\mathsf{d}_G(1,u)+1=\mathsf{d}_G(1,v)$}\}\,.
\]
Then the set of shortest path trees can be identified with the Cartesian product
\begin{equation}\label{eqn:uniform-support}
    \prod_{v\in V\setminus\{1\}}\mathsf{par}_G(v)\,.
\end{equation}
As a consequence, the uniform measure over the shortest path trees is simply a vertex-wise product measure. When we have two graphs $G^{(1)}$ and $G^{(2)}$ sharing the same vertex set, this means that the optimal coupling (under the normalized Hamming metric) of the uniform measures is the product of vertex-wise couplings, so the overlap is given by
\[
    \tilde{R}_n=\frac{1}{n}\sum_{v\in V\setminus\{1\}}\frac{|\mathsf{par}_{G^{(1)}}(v)\cap\mathsf{par}_{G^{(2)}}(v)|}{\max(|\mathsf{par}_{G^{(1)}}(v)|,|\mathsf{par}_{G^{(2)}}(v)|)}\,.
\]
By our distance asymptotics (Theorem~\ref{thm:single-intro} and the first part of Theorem~\ref{thm:correlated-intro}), we only need to include in the sum the vertices at distance $d_n^*$ or $d_n^*+1$ in both graphs. These vertices fall under one of the four subsets depending on the precise pair of distances (e.g., distance $d_n^*$ in both graphs, $d_n^*$ in $G^{(1)}$ and $d_n^*+1$ in $G^{(2)}$, and so on), and we know the asymptotic proportion of these subsets from our distance asymptotics. Then for each subset we apply concentration inequalities to control the sum. The details can be found in Section~\ref{sec:overlap-opt}.

\paragraph{From overlap concentration to the absence of OGP.} A result of this analysis, the second part of Theorem~\ref{thm:correlated-intro}, implies that the overlap concentrates for any limit values of the proxies $\lambda_n\to\lambda$, $\rho_n\to\rho$, and $\gamma_n\to\gamma$. Moreover, the probability limit of the overlap is a continuous function of $\rho$ or $\gamma$ onto $[0,1]$ for any fixed $\lambda$. This means that for any sequence of $\alpha_n$ and any open interval $\mathcal{I}\subseteq[0,1]$, by choosing a suitable sequence of $\rho_n$, $\mathcal{I}$ contains the overlap infinitely often with probability 1. As noted in Corollary~\ref{cor:no-eogp}, this violates the ensemble OGP condition typically needed to prove the nonexistence of stable algorithms.

\begin{remark}[Extension to non-convergent case]
Using a compactness argument, we can also formally rule out the possibility of OGP without actually requiring the $\lambda_n$ to converge. 
To elaborate,
we have a triangular array\footnote{We use the term ``triangular array'' to remind the reader that it is a sequence for each fixed $n$ and we will eventually send $n\to\infty$.} of $\mathcal{G}(n,q_n)$ graphs $G=G_{n}^{(0)},G_{n}^{(1)},\cdots,G_{n}^{(T)}$, where each pair of vertices is labeled from $1$ to $T=\binom{n}{2}$, and $G_{n}^{(t)}$ is obtained from $G_{n}^{(t-1)}$ by resampling a pair labeled $t$. The ensemble OGP, if it held, would tell us that there is an interval $\mathcal{I}\subseteq(0,1)$ such that the overlap $R_n^{(t)}$ of uniformly random shortest path trees of $G$ and $G_n^{(t)}$ satisfies $R_n^{(t)}\notin \mathcal{I}$ for all $1\leq t\leq T$ with high probability. Our result implies that this never happens. Indeed, suppose for contradiction that we do have the ensemble OGP. Then there exists a subsequence of indices where the OGP holds almost surely, and from that we can find a further subsequence such that $\lambda_{k_n}\to\lambda$ converges as $n\to\infty$. Assume for instance that $\lambda\in\{0,1\}$. Then for an arbitrary $r\in\mathcal{I}$, by setting
\[
    t_n=\left\lfloor \binom{n}{2}\cdot\frac{\log(1/r)}{d_n^*}\right\rfloor
\]
we have $R_{k_n}^{(t_{k_n})}\pto r$, so an overlap contained in $\mathcal{I}$ exists infinitely often almost surely, a contradiction.
\end{remark}
\subsubsection{Gibbs measures over spanning trees}\label{sec:to-gibbs}


\paragraph{Computation of the log partition function.} To analyze the Gibbs measures, we need to obtain a very precise understanding of the log partition function / free energy of the model. 
A basic principle of statistical physics is that the Gibbs measure is always obtained by maximizing entropy while keeping the average energy fixed. This can be made precise via the Gibbs variational principle (cf. Bogolyubov inequality) \cite{ellis2007entropy}, which tells us that
\begin{equation}\label{eqn:gibbs-var-intro}
    \log Z_{G,\beta} = \sup_{\nu} \left[ H_{\nu}(T)  - \beta \log\log(n) \E_{\nu}\left[ \sum_v \mathsf{d}_T(1,v) \right]\right]
\end{equation}
where the supremum ranges over all probability measures; in this formulation, we see that the inverse temperature $\beta \ge 0$ exactly functions as a Lagrange multiplier. We simplify our system to a mean-field model that is thermodynamically equivalent in the sense that its log-partition function approximates \eqref{eqn:gibbs-var-intro}, effectively converting it to a one-dimensional optimization. We will shortly demonstrate this by sketching a procedure of solving \eqref{eqn:gibbs-var-intro} up to $o_p(n\log\log n)$ error, which is leading-order in $n$ and accurate enough for Theorem~\ref{thm:intro-gibbs-free-energy}. However, it turns out to really understand the system and prove all the other results, we actually need to build a conceptually similar but more delicate analysis in Section~\ref{sec:gibbs}, where we derive a more precise formula accurate up to $o_p(n)$ error.


\paragraph{A variational upper bound.} It is discussed in the previous section that the set of shortest path trees is the Cartesian product of $\mathsf{par}_G(v)$. In general, any spanning tree $T$ can be identified with a map $\mathsf{par}_T:V\setminus\{1\}\to V$ where $u=\mathsf{par}_T(v)$ is in the neighborhood $\mathsf{N}_G(v)$ of $v$ satisfying $\mathsf{d}_T(1,u)+1=\mathsf{d}_T(1,v)$. In other words, a spanning tree can be thought of as an element of the Cartesian product
\[
    \mathsf{par}_T\in\prod_{v\in V\setminus\{1\}}\mathsf{N}_G(v)
\]
and the Gibbs measure $\mu_{G,\beta}$ can also be defined on the corresponding product $\sigma$-algebra. It should be noted that, unlike shortest path trees, not every element of the Cartesian product is a valid spanning tree --- the parent map $\mathsf{par}_T$ must be \emph{acyclic}, in the sense that no sequence of vertices $v_1,\cdots,v_k$ satisfies $v_{i+1}=\mathsf{par}_T(v_i)$ and $v_1=\mathsf{par}_T(v_k)$. 
This means we cannot apply a vanilla mean-field approximation since arbitrary product measures might place nontrivial mass outside the support of $\mu_{G,\beta}$. Nonetheless, this perspective gives an important upper bound to \eqref{eqn:gibbs-var-intro} which is useful in understanding the Gibbs measure:
\[
    \begin{split}
        \log Z_{G,\beta}& = \sup_{\nu} \left[ H_{\nu}(T)  - \beta \log\log(n) \E_{\nu}\left[ \sum_v \mathsf{d}_T(1,v) \right]\right]\\
        &\leq \sup_{\nu} \left[ \sum_{v}H_{\nu}(\mathsf{par}_T(v))  - \beta \log\log(n) \E_{\nu}\left[ \sum_v \mathsf{d}_T(1,v) \right]\right]\\
        &=\sup_{\nu}\left[\sum_{v}\left(H_{\nu}(\mathsf{par}_T(v))-\beta\log\log(n)\E_{\nu}[\mathsf{d}_T(1,v)]\right)\right]\,.
    \end{split}
\]
Making the ansatz that this upper bound is nearly tight, we can start to understand the behavior of the model by optimizing each summand in the RHS. For each $v\in V\setminus\{1\}$, it wants to maximize the \emph{entropy} $H_\nu(\mathsf{par}_T(v))$ but at the same time minimize the \emph{energy} $\mathsf{d}_T(1,v)$, and this tradeoff is balanced by the inverse temperature $\beta$. If $v$ tries to be at the minimum energy state, namely $\mathsf{d}_T(1,v)=\mathsf{d}_G(1,v)$, then the parent choice $\mathsf{par}_T(v)$ is confined to the set $\mathsf{par}_G(v)$ (of course, this alone does not guarantee that $v$ will have minimum energy), which gives the entropy of at most $\log|\mathsf{par}_G(v)|$. If this is smaller by at least $\beta\log\log n$ than the maximum possible entropy $\log|\mathsf{N}_G(v)|$, then $v$ might jump to a higher energy state $\mathsf{d}_T(1,v)\geq\mathsf{d}_G(1,v)+1$ with a non-negligible probability seeking for a larger entropy.

\begin{remark}[Transformed system]
In our rigorous analysis, we show there is a simple way to rewrite the partition function $Z_{G,\beta}$ as the partition function of a \emph{Markov random field} or ``spin system'' related to the original system. This is explained in Section~\ref{sec:transform-mrf} and helpful for formalizing many of the ideas we discuss here.
\end{remark}


\begin{figure}
    \centering
    \includegraphics[width=0.5\linewidth]{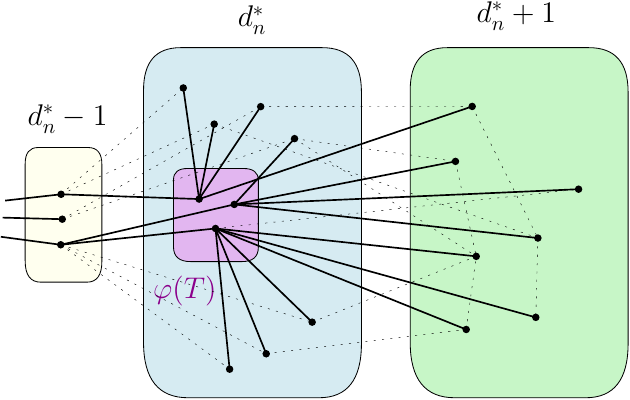}
    \caption{An illustration of the kernel $\varphi(T)$ of a spanning tree $T$. Each line (either dotted or normal) is an edge of the underlying graph $G_n$, and the normal lines are the edges of $T$. The labels $d_n^*-1$, $d_n^*$, and $d_n^*+1$ denote the distances from the root in $G_n$, not $T$.}
    \label{fig:kernel}
\end{figure}

\paragraph{Reduction to a one-dimensional system.} Note that since we are working with sparse random graphs, a typical vertex can contribute at most $H_{\nu}(\mathsf{par}_T(v))\leq\log|\mathsf{N}_G(v)|=(1+o(1))\log\log n$ to the sum, which allows us to ignore $o(n)$ summands. Here, our results on distance asymptotics come into play: since most of the vertices will be at distance either $d_n^*$ or $d_n^*+1$ with high probability, we only need to consider those vertices in $V^*=\{v:\mathsf{d}_G(1,v)\in\{d_n^*,d_n^*+1\}\}$\footnote{If we want an accuracy of at most $o(n)$ error, then $O(\frac{n}{\log\log n})$ number of vertices start to matter. This is the reason why in \eqref{eqn:dstar} we defined $d_n^*$ so that $N_{d_n^*-1}\approx(nq_n)^{d_n^*-1}<\frac{n}{(\log\log n)^2}$: this assures that it is still safe to ignore vertices at depth at most $d_n^*-1$. Those at depth at least $d_n^*+2$ can also be ignored but needs more care to handle; see Lemma~\ref{lem:gamma-a}.}. This gives
\begin{equation}\label{eqn:logz-vstar}
    \log Z_{G,\beta}\leq\sup_{\nu}\left[\sum_{v\in V^*}\left(H_{\nu}(\mathsf{par}_T(v))-\beta\log\log(n)\E_{\nu}[\mathsf{d}_T(1,v)]\right)\right]+o_p(n\log\log n)\,.
\end{equation}
Recall from the preceding discussion that each vertex decides whether to attempt to be at the minimum energy state by picking its parent from $\mathsf{par}_G(v)$, or jump to higher energy state and enjoy a possibly larger entropy. In the spirit of mean-field approximation, we make the following mean-field assumptions and observe how the system behaves.
\begin{enumerate}
    \item Every vertex $v$ with depth less than $d_n^*$ always stays at the minimum energy state $\mathsf{d}_T(1,v)=\mathsf{d}_G(1,v)$, in order to ``support'' the vertices in $V^*$.
    \item Each of the vertices at depth $d_n^*$, (approximately) independently with a common probability $p$, decides to pick its parent from $\mathsf{par}_G(v)$ and gains an energy of $d_n^*$. Let $A$ be the set of vertices that made this decision, i.e., $\mathsf{par}_T(v)\in\mathsf{par}_G(v)$.
    \item\label{item:aspt3} Every other vertex $v\in V^*\setminus A$, hoping that $A$ supplies enough entropy, always picks its parent from $A$ and gains an energy of $d_n^*+1$.
\end{enumerate}
On average, a typical vertex at depth $d_n^*$ would have $N_{d_n^*-1}q_n$ neighbors in $\mathsf{par}_G(v)$; by our results on distance asymptotics, this is approximately $(nq_n)^{d_n^*-1}q_n=(nq_n)^{\Delta_n}$. Hence, each vertex in $A$ gains a maximum entropy of $(\Delta_n+o(1))\log\log n$ on average. Similarly, a typical vertex would have $q_n|A|$ neighbors in $A$, so each vertex in $V^*\setminus A$ gains a maximum entropy of $\approx\log(|A|q_n)$. By our assumptions, the size of $A$ concentrates around $N_{d_n^*}p$, so this is  $\approx\log(N_{d_n^*}pq_n)$. Thus, the contribution of each vertex at depth $d_n^*$ to the sum in \eqref{eqn:logz-vstar} is approximately
\[
    H(p)+p\Delta_n\log\log n+(1-p)\log(N_{d_n^*}pq_n)-\beta\log\log(n)(pd_n^*+(1-p)(d_n^*+1))
\]
and for each vertex at depth $d_n^*+1$
\[
    \log(N_{d_n^*}pq_n)-\beta\log\log(n)(d_n^*+1)\,.
\]
Collecting terms depending on $p$, \eqref{eqn:logz-vstar} reduces to optimizing
\[
    \sup_{0<p\leq 1}\left[N_{d_n^*}(H(p)+p(\Delta_n+\beta)\log\log n)+((1-p)N_{d_n^*}+N_{d_n^*+1})\log(N_{d_n^*}pq_n)\right]\,.
\]
Assuming the convergence of the proxy $\Delta_n\to\Delta\in[0,1]$ and observing that $N_{d_n^*}H(p)=O(n)$ is negligible and $N_{d_n^*}+N_{d_n^*+1}=|V^*|=(1+o_p(1))n$, this can also be stated as optimizing $m=N_{d_n^*}p$:
\begin{equation}\label{eqn:1dim-opt}
    \sup_{0<m\leq N_{d^*}}\left[m(\Delta+\beta)\log\log n+(n-m)\log(mq_n)\right]\,.
\end{equation}
Note that $m$ represents an approximate size of $A$. Thus, we have reduced the problem \eqref{eqn:gibbs-var-intro} of optimizing the measure to the problem \eqref{eqn:1dim-opt} of optimizing the number of vertices at depth $d_n^*$ satisfying $\mathsf{par}_T(v)\in\mathsf{par}_G(v)$. In Definition~\ref{def:kernel}, we call this set $A$ a \emph{kernel}, depicted in Figure~\ref{fig:kernel}, and is a central object in our rigorous derivation throughout Section~\ref{sec:gibbs}. Once we have completely analyzed the variational upper bound as above, we prove its tightness by constructing measures on trees which approximately satisfy the ``mean-field assumptions'' above.

\paragraph{Phase transitions and the Franz--Parisi potential.} We now present an overview of interesting consequences of our analysis of partition function, whose formal proofs are detailed in Section~\ref{sec:gibbs2}. By solving \eqref{eqn:1dim-opt}, one can see that if $\beta>1-\Delta$ (low temperature phase) then $m=N_{d_n^*}$ is an approximate optimizer. This can be interpreted as the following: in the low temperature phase, almost all vertices at depth $d_n^*$ decides to be in the minimum energy state $\mathsf{d}_T(1,v)=\mathsf{d}_G(1,v)=d_n^*$. This indicates that a typical tree $T$ under the Gibbs measure might actually be similar to a shortest path tree, which turns out to be true in Wasserstein sense (cf. Theorem~\ref{thm:intro-phase-transition}). A crucial mathematical tool in proving this is an entropy-transport inequality: we round $\mu_{G,\beta}$ to the support \eqref{eqn:uniform-support} of the uniform measure by a negligible change in Wasserstein distance, then argue that this rounded measure has near maximum entropy; an entropy-transport inequality then guarantees that this should be close to the uniform measure in Wasserstein distance\footnote{As is evident from Corollary~\ref{cor:high-entropy-wasserstein}, this requires that we know the entropy of $\mu_{G,\beta}$ accurate up to an $o(n)$ error. See Section~\ref{sec:gibbs-low-phase} for a rigorous argument.}. In addition, in the low temperature phase, the landscape of the Franz--Parisi potential largely depends on the size of $\mathsf{par}_G(v)$ for each $v$ at depth $d_n^*$. If $\Delta>0$, then a typical size of $\mathsf{par}_G(v)$ is unbounded; this means that two independent samples from $\mu_{G,\beta}$ will have negligible overlap, leading to a strictly increasing landscape of the Franz--Parisi potential (e.g., the left plot of Figure~\ref{fig:fpp}). Otherwise, if a typical size of $\mathsf{par}_G(v)$ is bounded, then two independent samples might have nontrivial overlap with high probability (e.g., the blue curve in Figure~\ref{fig:fpp}).

Now we consider the case $\beta<1-\Delta$ (high temperature phase) where $m=\frac{n}{(1-\Delta-\beta)\log\log n}$ is an approximate optimizer, given that this is smaller than $N_{d_n^*}$. In this case, some non-negligible portion of the vertices at depth $d_n^*$ decide to jump to a higher energy state $\mathsf{d}_T(1,v)=\mathsf{d}_G(1,v)+1=d_n^*+1$. As a result, those staying in the minimum energy state tend to have more children in a typical tree under the Gibbs measure $\mu_{G,\beta}$, compared to a uniformly random shortest path tree. Thus, the ``children profile'' of the vertices at depth $d_n^*$ can be used to distinguish $\mu_{G,\beta}$ from the uniform measure; if we look at the empirical distribution of the number of children for vertices at depth $d_n^*$, trees sampled from the former would look like a mixture of two distributions, whereas those sampled from the latter would just concentrate around $N_{d_n^*+1}/N_{d_n^*}$. This children profile is in fact Lipschitz under the Hamming metric and thus plays a role as a \emph{witness} to show that $\mu_{G,\beta}$ is not close in Wasserstein metric to the uniform measure. Moreover, the fact that the optimal $m$ is only $o(n)$ implies that all but $o(n)$ vertices enjoy high entropy, i.e., have an unbounded number of parent choices. Thus, two independent samples from $\mu_{G,\beta}$ are unlikely to have a non-negligible overlap, leading to a strictly increasing landscape of the Franz--Parisi potential in any regime.

\subsection{Related work}

\paragraph{Related work on tensor PCA.} One average case problem which has been deeply investigated is the problem of \emph{tensor principle component analysis (PCA)} \cite{montanari2014statistical} for tensors of order $p \ge 3$. Notably, in this algorithm the best guarantees for natural local search algorithms with random initialization like tensor power iteration, approximate messsage passing (AMP), and gradient ascent or low-temperature Langevin dynamics on the log-likelihood are statistically suboptimal among polynomial time algorithms \cite{montanari2014statistical,arous2020algorithmic}. Instead, what are believed to be the best possible statistical guarantees for polynomial time algorithms are obtained by algorithms such as higher levels of the Sum-of-Squares (SoS) semidefinite programming hierarchy \cite{hopkins2015tensor} and more sophisticated spectral methods related to SoS \cite{hopkins2015tensor,hopkins2017power}. This example also raised questions about how the predictions from statistical physics style landscape analysis should be interpreted, since unlike in many other problems AMP does not predict the computational threshold. More generally, ``local'' algorithms with random initialization seem to be suboptimal based on existing analyses \cite{wein2019kikuchi}.

In a notable work, Wein, El Alaoui and Moore showed that ideas related to statistical physics and landscape analysis \emph{can} predict the same threshold as SoS by analyzing the Hessian of the \emph{Kikuchi free energy}. This is a higher-order analogue of the \emph{Bethe free energy} which is related to belief propagation and approximate message passing (see, e.g., \cite{yedidia2003understanding}); essentially the Bethe free energy is based on optimization over pairwise marginals, whereas the level $\ell$ of the Kikuchi free energy is an optimization problem over $\ell$-wise marginals. (The levels of the SoS hierarchy have a similar interpretation.) They showed that a linearization analysis of the Kikuchi free energy leads to a spectral algorithm on $n^{O(p)} \times n^{O(p)}$ matrices which succeeds up to the conjectured optimal threshold. 

There are several parallels between these findings for tensor PCA and our results for shortest path on random graphs. In both cases the \emph{choice of optimization landscape/geometry} has an important effect on the performance of local search algorithms like power iteration. In tensor PCA, naively ``local'' algorithms related to, e.g., optimization the Bethe free energy have suboptimal guarantees; instead, the best algorithms use other landscapes (e.g., corresponding to optimization of the Kikuchi free energy or higher levels of SoS, or the ``ironed out'' landscape of \cite{biroli2020iron}). This mirrors how direct local search on the landscape of the shortest path problem fails whereas lifting local search to the tree landscape succeeds. Relatedly, in both cases ideas related to statistical physics and spin glass theory (e.g., various variational approximations to the free energy) predict the ``correct'' computational-statistical tradeoff for polynomial time algorithms when we allow for a better choice of optimization landscape. 

\paragraph{Planted vs null problems.} 
One major difference between our problem and tensor PCA is that tensor PCA is a ``planted'' problem --- the goal for the algorithm is to estimate a planted signal buried in noise --- whereas shortest path in random graphs is a ``null'' or ``unplanted'' problem because there is no underlying signal or concept of signal-to-noise ratio. This is why the OGP is relevant to our setting (it is generally used for ``null''/pure random optimization problems) whereas AMP and related methods for making predictions in planted problems are used in tensor PCA. 

Recently there have also been studies of planted analogues of the OGP in order to study problems like the \emph{planted clique} problem. This had led to findings which appear to have a parallel moral to the case of tensor PCA --- naive local search can fail to achieve the optimal computational-statistical tradeoff \cite{gamarnik2024landscape,chen2024low,gheissari2025finding} and natural changes to the landscape like over-parameterization lead to closing the gap. See also \cite{deshpande2015finding} for a perspective on planted clique related to AMP.

\paragraph{Some pointers regarding shattering, OGP, etc.} 
Some examples of very well-studied problems where there has been a lot of deep work on understanding the solution geometry (through Gibbs measures, OGP, and other means) the $p$-spin model from spin glass theory, independent sets in random graphs, and different versions of random $k$-SAT, see \cite[e.g.]{krzakala2007gibbs,alaoui2024near,kizildaug2023sharp,gamarnik2017performance,bresler2022algorithmic,ding2015proof,coja2014asymptotic,coja2012catching,gamarnik2021overlap}. For instance, the works \cite{achlioptas2008algorithmic,mezard2005clustering,achlioptas2006solution,coja2015independent} rigorously established shattering results for the solution space of several optimization problems such as k-SAT and observed that this implies lower bounds for MCMC algorithms. Here shattering means that the solution space is largely made up of many small clusters which are pairwise distant from each other; this is one example of the type of free energy barrier which the Franz--Parisi potential can pick up on (see, e.g., \cite{el2025shattering} for more discussion). These shattering properties also motivated the introduction of the OGP \cite{gamarnik2014limits}.

As previously mentioned, the OGP framework and other analyses of the solution landscape have been applied to a variety of different random optimization problems. We will not attempt to survey the OGP literature, the related literature on spin glass theory etc (see e.g. \cite{mezard1987spin}), but just mention a couple of relevant applications where the stability of algorithms have been explicitly studied. 
One such example is the extensive literature on the symmetric and asymmetric binary perceptron \cite[e.g.]{abbe2022proof,kim1998covering,perkins2021frozen,baldassi2020clustering,abbe2022binary,gamarnik2022algorithms}. Here it was found that  although stable algorithms exist in certain regimes, they succeed by finding solutions only in a relatively small region of the solution space; around typical solutions free energy barriers actually exist. This illustrates one potential subtlety involved in making predictions based on landscape geometry. See also \cite{chou2021limitations} for another work proving the stability of a natural (quantum) algorithm, shown in order to conclude that it cannot overcome the OGP.

\paragraph{Some related work using the Franz--Parisi potential.} There is a massive physics literature around the Franz--Parisi potential and its application to predicting metastability and hardness which we will again not attempt to survey. See, e.g., \cite{sun2012following,franz1995recipes,zdeborova2010generalization} for relevant work including the ``state following'' perspective. Some recent applications in mathematics include its use in the solution of the spiked Wigner model \cite{el2018estimation} and in studying shattering in spherical spin glasses \cite{el2025shattering}. Another recent work \cite{bandeira2022franz} has shown that in planted models with Gaussian noise, a version of the \emph{annealed} approximation to the Franz--Parisi potential is (1) related/equivalent to predictions based on low-degree polynomials, and (2) yields rigorous lower bounds against local Markov chains. The latter point is nonobvious and its proof relies heavily on the noise being Gaussian, unlike the quenched Franz--Parisi potential which always implies hardness for MCMC via bottleneck theory (see, e.g., Appendix~\ref{apdx:bottleneck}). The shortest path problem example of \cite{LS2024} gives an example where the low-degree heuristic and direct application of the Franz--Parisi potential make essentially different conclusions. 

\paragraph{Overlap gap and low-degree polynomials.} Another interesting aspect of the OGP for shortest path \cite{LS2024} is the separation it leads to between the predictions of overlap gap and the low-degree polynomial method for predicting average-case hardness. One reason this is interesting is that some recent works have also established connections between these two in the positive direction, in the sense that OGP-based techniques can sometimes rule out the success of low-degree polynomials \cite{wein2022optimal, gamarnik2024hardness,huang2025strong} --- see \cite{LS2024} for a much more extensive discussion of this topic, and \cite{huang2025strong} for connections between the low degree heuristic and the probability of the OGP occurring. Intuitively, this is because low-degree polynomials of independent variables are in some sense average-case Lipschitz due to concentration of measure. 
See, e.g., \cite{brennan2020statistical, holmgren2021counterexamples,koehler2022reconstruction,huang2024low,huang2025optimal,buhai2025quasi} for some other works discussing connections between low-degree polynomials and the Statistical Query model as well as limitations of the low-degree polynomial method for predicting computational tractability. 

\paragraph{The structure of low-complexity Gibbs measures.} The rigorous analysis of our model fits into a general paradigm known as \emph{mean-field approximation} in statistical physics and probability \cite{parisi1988statistical}. In particular, rigorous work on the na\"ive mean-field approximation (e.g. \cite{austin2019structure,eldan2018decomposition,eldan2018gaussian,eldan2020taming,jain2019mean,augeri2021transportation}) shows that for some Gibbs measures of the form $p(x) \propto e^{f(x)}$ for $x \in \{\pm 1\}^n$, the log partition function $\log Z$ can be well-approximated by an optimization over the space of product measures, and in Wasserstein space the Gibbs measure can be approximated by a ``small'' mixture of product measures (pure states). Our model does not satisfy the assumptions needed to apply these general results (enforcing the tree structure requires hard constraints, and it is also unclear obvious whether our model satisfies something like the ``low gradient complexity'' conditions used in \cite{eldan2018gaussian}).
The results we obtain do informally correspond to an approximation by (mixtures of) ``almost'' product measures, and we do build on related information-theoretic techniques like entropy-transportation inequalities. Related to the literature on mean-field approximation, it is also interesting that we need very precise asymptotics (beyond leading-order) for $\log Z$ in order to understand the structure of the Gibbs measure --- a major focus in the mean-field approximation literature is determining the actual size of the error term for na\"ive mean-field bounds (see \cite{eldan2020taming,augeri2021transportation} for the state of the art in the context of Ising models).

\subsection{Organization of the paper}
In Section~\ref{sec:preliminaries}, we go through some preliminaries such as graph-theoretic notation which is used throughout the paper. 
Section~\ref{sec:single} is the beginning of our analysis where we establish some key facts about the distribution of distances in a random graph. Here as in the rest of the paper, we reference mathematical tools such as Chernoff bounds, Poisson approximation, large deviations theory, etc. which are discussed in Appendix~\ref{apdx:tools}.

Next is Section~\ref{sec:correlated}, where we analyze the overlap between shortest path trees (and other related structures) under partial rerandomization of the graph. These are the results which directly imply that the OGP does \emph{not} hold in the landscape over trees. Note that Section~\ref{sec:correlated} is the \emph{only} section where two correlated random graphs are considered --- in the Gibbs measure analysis, a single random graph is sampled once (``quenched'') and the randomness comes only from sampling the distribution.

In Section~\ref{sec:gibbs} we compute asymptotics for the log partition function of the finite temperatures Gibbs measure over random graphs. This is at the heart of the results in the following Section~\ref{sec:gibbs2}, where we prove our main results about phase transitions, the Franz--Parisi potential, and the approximate structure of the Gibbs measures/energy landscape. 
The parts of the analysis which are directly used in both the OGP and Franz-Parisi analysis are all contained in Section~\ref{sec:single}, so Section~\ref{sec:correlated} can be read before or after Sections~\ref{sec:gibbs} and \ref{sec:gibbs2}. The main text of the paper ends with a brief conclusion (Section~\ref{sec:conc}). 

The paper concludes with a few accompanying appendices. Appendix~\ref{apdx:gibbs-critical} discusses the behavior of the Gibbs measures near the critical temperature, complementing Section~\ref{sec:gibbs} and Section~\ref{sec:gibbs2}. Appendix~\ref{apdx:discussion} contains some discussion and mathematical background which is helpful for understanding and interpreting the main results of the paper. Appendix~\ref{sec:dca} is a self-contained appendix discussing the shortest path problem from the perspective of matroid theory and discrete convexity. In particular, in Appendix~\ref{sec:simplex} we elaborate on the connection between the landscape over trees and the \emph{network simplex} algorithm. This discussion explains why as a basic fact we know that local search in the space of trees must always be globally convergent, as well as generalizations of this observation which exist in the literature. 
In Appendix~\ref{apdx:path-barrier} we elaborate on the interpretation of the main structural result of \cite{LS2024} in terms of the Franz--Parisi potential on the space of paths. 
In Appendix~\ref{sec:example-lovasz}, we elaborate on the analogy between our result and local search on different extensions of a submodular function. This includes discussion of the corresponding Franz--Parisi potentials.  

\section{Preliminaries}\label{sec:preliminaries}

\subsection{Additional notations}

In Section~\ref{sec:models-notations} we introduced our models and notations which are sufficient to present our results. Here, we complement that with additional graph theoretic notations that will be handy throughout the proofs and intermediate results in the rest of the paper. Besides these, throughout the rest of the paper, especially for the proofs, we will usually omit the subscripts for $q=q_n$, $G=G_n$, $d^*=d_n^*$, and $\ell^*=\ell_n^*$. On the other hand, we never drop the subscripts for proxy sequences such as $\lambda_n$, $\Delta_n$, and $\rho_n$, since they can be confused with their limiting values.

\paragraph{Notations for general graphs.} Here we introduce several notations for common objects in graph theory. For $u,v\in V(G)$, we write $u\sim v$ to mean that $u$ and $v$ are adjacent in $G$. We denote the neighborhood of $v$ by
\[
    \mathsf{N}_G(v):=\{u\in V(G):u\sim v\}\,.
\]
More generally, for any $A,B\subseteq V(G)$, we denote the neighborhood of $A$ in $B$ by
\[
    \mathsf{N}_B(A):=\bigcup_{v\in A}\mathsf{N}_G(v)\cap B\setminus A\,.
\]
We define the \emph{distance function} $\mathsf{d}_G(u,v)$ for $u,v\in V(G)$ as the length of a shortest path from $u$ to $v$ in $G$. When there is no such path, the distance is set to be $\infty$ by convention. More generally, for a set $A\subseteq V(G)$, $\mathsf{d}_G(u,A)$ is the length of a shortest path from $u$ to any vertex in $A$. The function $\mathsf{d}_G(u,-)$ is the \emph{distance vector} of $G$ starting from $u$. Define the distance $d$ neighborhood of $v$ by
\[
    \Gamma_d^{G}(v):=\{u\in V(G):\mathsf{d}_G(u,v)=d\}
\]
and its size
\[
    N_d^G(v):=|\Gamma_d^G(v)|\,.
\]
In addition, we sometimes use
\[
    \Gamma_{\leq d}^{G}(v):=\bigcup_{k=1}^d\Gamma_{k}^G(v)
\]
to denote the set of vertices with distance at most $d$. We may sometimes drop the subscript/superscript $G$ when it is obvious from the context.

\paragraph{Notations for graphs with fixed source.} As a quick reminder, we mostly consider paths from a fixed starting vertex $1$. Thus, we may simply write $\mathsf{d}_G=\mathsf{d}_G(1,-)$, $\Gamma_d^G=\Gamma_d^G(1)$, and $N_d^G=N_d^G(1)$. As we study paths from the fixed source $1$, we naturally focus on the connected component of $G$ containing $1$. We denote the set of vertices in this component by
\[
    \overline{V}(G):=\{v\in V(G):\mathsf{d}_G(1,v)<\infty\}
\]
and
\[
    \overline{G}=G[\overline{V}]
\]
the induced subgraph of $G$ by $\overline{V}$. On top of the distance structure induced by $1$, we define the \emph{parent set} of $v\in\overline{V}\setminus\{1\}$ by
\[
    \mathsf{par}_G(v):=\{u\in\mathsf{N}_G(v):\mathsf{d}_G(1,v)=\mathsf{d}_G(1,u)+1\}\,.
\]
For any spanning tree $T$ of $\overline{G}$, the parent set of $v$ in $T$ is always a singleton and by $\mathsf{par}_T(v)$ we refer to this unique parent. In particular, $\mathsf{par}_T(v)\in\mathsf{par}_G(v)$. As an inverse function of $\mathsf{par}_G$, we define the \emph{children set} of $v\in\overline{V}$ by
\[
    \mathsf{ch}_G(v):=\{u\in \mathsf{N}_G(v):v\in\mathsf{par}_G(v)\}\,.
\]

\subsection{Wasserstein distance}
\begin{definition}[see, e.g., \cite{villani2021topics}]
The 1-Wasserstein distance between two probability measures $\mu,\nu$ over a metric space $(d,\mathcal X)$ is given by
\[ W_1(\mu,\nu) = \inf_{\Pi} \mathbb E_{(X,Y) \sim \Pi}[d(X,Y)]\]
where $\Pi$ ranges over all coupling of $\mu$ and $\nu$.
\end{definition}
\begin{proposition}\label{prop:prod-coupling}
Suppose that $\mu = \otimes_{i = 1}^m \mu_i$ and $\nu = \otimes_{i = 1}^m \nu_i$ are two product measures on $\mathcal X^n$, let $d$ be a metric on $\mathcal X$, and let $c(x,y) = \sum_{i = 1}^n d(x_i,y_i)$ be a metric on $\mathcal X^n$. Then if we consider the 1-Wasserstein distance induced by these metrics, we have that
\[ W_1(\mu,\nu) = \sum_{i = 1}^n W_1(\mu_i,\nu_i), \]
and also the infimum in the definition of $W_1(\mu,\nu)$ can be obtained by a product measure.
\end{proposition}
\begin{proof}
We have that
\[ W_1(\mu,\nu) = \inf_{\Pi} \mathbb E_{(X,Y) \sim \Pi}[c(X,Y)] = \inf_{\Pi} \sum_{i = 1}^n\mathbb E_{(X,Y) \sim \Pi}[d(X_i,Y_i)] = \sum_{i = 1}^n \inf_{\Pi_i} \mathbb E_{(X_i,Y_i) \sim \Pi_i}[d(X_i,Y_i)],  \]
where we used that $\sum_{i = 1}^n\mathbb E_{(X,Y) \sim \Pi}[d(X_i,Y_i)] \le \sum_{i = 1}^n \inf_{\Pi_i} \mathbb E_{(X_i,Y_i) \sim \Pi_i}[d(X_i,Y_i)]$ holds for all $\Pi$. Also, equality is achieved if we consider a product coupling of the optimal $\Pi_i$.
\end{proof}

\subsection{Shortest path structures}\label{sec:dagtree}
\paragraph{The shortest path DAG.} The \emph{shortest path DAG} from $v$ is the union of all shortest paths starting from $v$. This is also called the \emph{shortest path subgraph} (see, e.g., \cite{khuller1995balancing,bauer2009batch}); because our original graph is undirected, we use the term ``shortest path DAG'' to emphasize that the shortest path structure is naturally directed and acyclic.

\paragraph{Uniformly random shortest path tree.} A shortest path tree from vertex $v$ is any tree which is formed as the union of shortest paths from $v$ to each other vertex. Interestingly, given the shortest path DAG, we can easily count and sample a uniformly random element from the set of \emph{shortest path trees} of the graph. All we need to do is, independently for each vertex $u$ in the shortest path DAG, pick a uniformly incoming edge to $u$ and drop all of the other incoming edges. In other words, the uniform distribution over shortest path trees naturally factorizes as a \emph{product measure} and the shortest path DAG encodes the information about this product measure. 
This is closely connected with what happens in the finite-temperature setting, which we study later in Section~\ref{sec:gibbs}. 


\begin{proposition}
Fix a graph $G$, let $\mu_{G,v}$ be the uniform distribution over shortest path trees from vertex $v$ in $G$, and let $F$ be the shortest path DAG of $G$ also from $v$. Then $\mu_v$ factorizes as a product measure where each vertex $u \ne v$ picks its parent uniformly at random from its parents in $F$.
\end{proposition}
\begin{remark}[Stability of sampling algorithms]\label{rmk:reasonable-stability}
As mentioned before, the above proposition implies that we can sample from the uniform distribution once we compute the shortest path DAG (e.g., using Dijkstra's algorithm). It is worth comparing this to the lower bound against stable algorithms implied by Wasserstein disorder chaos (Corollary~\ref{cor:wasserstein-disorder}) --- essentially, while Dijkstra's algorithm is fairly stable, it is still $o(1)$-Lipschitz while the lower bound applies to $\Omega(1)$-Lipschitz algorithms (see \cite{el2022sampling}), 
so there is no contradiction. Informally, while message passing algorithms run for $O(1)$ rounds are stable, this consideration does not apply here because Dijkstra's algorithm corresponds to running belief propagation/message passing for $\Omega(\log(n)/\log\log(n))$ rounds (discussed later in Preliminaries). This distinction is also discussed and relevant to the low-degree polynomial lower bounds established in \cite{koehler2022reconstruction,huang2024low,huang2025optimal}. 

We also note that a recent independent work of Ma and Schramm \cite{ma2025polynomial} has, in a similar spirit, highlighted a situation where sampling is algorithmically easy (in their case, via Glauber dynamics) even though Wasserstein disorder chaos occurs. 
\end{remark}

The following propositions shows that the process of sampling a spanning tree, given as input the shortest path DAG, is \emph{stable in the worst-case} (i.e. Lipschitz). Here, $\mathsf{par}_T(v)$ and $\mathsf{par}_F(v)$ denote the set of parents of $v$ in the tree $T$ (in which case is a singleton for non-root vertices) and the DAG $F$, respectively.
\begin{proposition}\label{prop:sptree-conc}
Given two graphs $G_1$ and $G_2$ on the same vertex set $V$ of size $n$ and a choice of vertex $v \in V$, let $\mu_1,\mu_2$ be the corresponding uniform measures over shortest path trees from $v$. Let $\Pi$ be the product measure coupling of $\mu_1$ and $\mu_2$ achieving $W_1(\mu_1,\mu_2)$ under the directed\footnote{The edges of the tree inherent directions from the DAG, and are only considered equal if they point in the same direction. Ignoring direction would make negligible difference in our results since opposite-direction overlaps are not very common in our setting.} Hamming metric (from Proposition~\ref{prop:prod-coupling}) and let $(T_1,T_2) \sim \Pi$ be random shortest path trees. Then with probability at least $1 - \delta$,
\[ \left|\#\{ e: e \in T_1 \cap T_2 \} - \sum_{u \ne v} \Pr(\mathsf{par}_{T_1}(u) = \mathsf{par}_{T_2}(u))\right| = O(\sqrt{n \log(2/\delta)}). \]
\end{proposition}
\begin{proof}
This follows from Hoeffding's inequality. 
\end{proof}
\begin{proposition}
With the notation of the previous proposition, and letting $F_1$ and $F_2$ be the corresponding shortest path DAGs starting from vertex $v$, we have
\[ \Pr(\mathsf{par}_{T_1}(u) = \mathsf{par}_{T_2}(u)) = \frac{\#\{e : \mathsf{par}_{F_1}(u) \cap \mathsf{par}_{F_2}(u)\}}{\max(\# \mathsf{par}_{F_1}(u), \# \mathsf{par}_{F_2}(u))}. \]
\end{proposition}
\begin{proof}
This is the probability they are equal in the optimal coupling \cite[Proposition 4.7]{wilmer2009markov}. 
\end{proof}

\subsection{Gibbs measures}
\paragraph{Finite-temperature Gibbs measure.}
For the shortest path tree problem, there is a natural way to define a class of finite-temperature \emph{Gibbs measures}. Given a graph $G = (V,E)$, its connected component $\overline{G}$ containing the root vertex $1$, and for inverse temperature $\beta > 0$, we define a class of Gibbs measures over spanning trees $T$ of $\overline{G}$\footnote{In our study of sparse Erd\"os--R\'enyi graphs, $\overline{G}$ is typically the giant component of $G$.}
\begin{equation}\label{eqn:finite-temp}
\mu_{G,\beta}(T) = \frac{1}{Z_{G,\beta}} \exp\left(-\beta \log\log(n)  \sum_{v\in\overline{V}} \mathsf{d}_T(1,v) \right)
\end{equation}
where $Z_{G,\beta}$ is the \emph{partition function} (normalizing constant), and $\mathsf{d}_T$ is the graph metric in the tree $T$. The connection to the minimum-cost flow formulation \eqref{eqn:lp-flow} is that for a shortest path tree $T$, $\sum_v \mathsf{d}_T(1,v)$ is exactly the objective value for \eqref{eqn:lp-flow} of the unique unidirectional flow on $T$ sending $1$ unit of flow from the root $1$ to each of the other vertices $v$. (The entries of the corresponding flow vector $f$ on an edge $u \to v$ correspond to the number of shortest paths which contain this edge.) We also show in Appendix~\ref{apdx:dijkstra-is-bp} that with this definition of the Gibbs measure, one can naturally recover Dijkstra's algorithm as a special case of belief propagation.

As is customary in statistical physics, as the inverse temperature $\beta$ increases, the Gibbs measure starts to concentrate more and more on near-minimizers of
\[ H(T) = \log\log(n) \sum_v \mathsf{d}_T(1,v). \]
In other words, as $\beta$ increases the Gibbs measure reweighs itself to prefer trees where the sum of the distances of vertices to the root becomes smaller.
The scaling factor $\log \log(n)$ is included so that the regime where $\beta = \Theta(1)$ is where the model exhibits interesting changes to the Gibbs measure, as will soon become clear from the results and analysis. Note that if we take $\beta = 0$, this measure is simply the uniform measure over spanning trees, which itself has been extensively studied over general graphs (see, e.g., \cite{kirchhoff1958solution,anari2021log}). 

\section{Asymptotics for a single graph}\label{sec:single}
In this section, we study the asymptotics of $G=G_n$, centered on the set of fixed distance neighborhoods $\Gamma_d=\Gamma_d^{G_n}(1)$. These results, especially those developed in Section~\ref{sec:dv}, are fundamental and will be frequently referred back throughout the paper.


\subsection{The distribution of distances in a single graph}\label{sec:dv}




We start with the asymptotic behavior of the distance vector. As mentioned in the introduction, the distance between $d^*=d_n^*$ and $\ell^*=\ell_n^*$ will be critically important in classifying the behavior of the random graph. Because of this, even if we consider an arbitrary convergent sequence of relative densities $\alpha_n \to \alpha$, this is not enough to determine the behavior of the graph for large values of $n$. (This will be clear from the conclusion of the following theorem.) The key information which is needed to determine the behavior is actually the relative location of $d^*$ and $\ell^*$, which is encoded by the parameter $\lambda_n$ in the definition below.

\begin{theorem}\label{thm:main1}
    Suppose that $\lambda_n\to\lambda\in[0,1]$. Then we have 
    \[ N_{d_n^*}/n\pto 1-\lambda, \qquad \text{and} \qquad N_{d_n^*+1}/n\pto \lambda \]
    as $n \to \infty$.
    In other words, with probability $1 - o(1)$ all vertices except for at most $o(n)$ of them have distance either $d_n^*$ or $d_n^*+1$, and the proportion of vertices with distance $d^*$ converges to $1-\lambda$ in probability.
\end{theorem}

Recall that $d^*$ is defined to be the smallest integer such that $(nq)^{d^*}\geq n/(\log\log n)^2$. This implies
\begin{equation}\label{eqn:lambda-sandwich}
    \frac{1}{(\log\log n)^2}\leq (nq)^{d^*-\ell^*}=\log(1/\lambda_n)<\frac{nq}{(\log\log n)^2}=\frac{\alpha_n\log n}{(\log\log n)^2}
\end{equation}
which gives
\begin{equation}\label{eqn:delta-sandwich}
    -\frac{2\log\log\log n}{\log\alpha_n+\log\log n}\leq d^*-\ell^*=\Delta_n<1-\frac{2\log\log\log n}{\log\alpha_n + \log\log n}< 1
\end{equation}
for large enough $n$.

The idea of our proof of Theorem~\ref{thm:main1} is to obtain tight concentration bounds for $N_d$. Since the distribution of $N_d$ depends on $N_0,\cdots,N_{d-1}$, we need an inductive argument. The following result is much stronger than we actually need here, but this form turns out to be useful in proving Theorem~\ref{thm:main3} in Section \ref{sec:spdag}.

\begin{proposition}\label{prop:conc}
    Let $\frac{C}{\log\log n}<\delta<\frac{1}{2}$ for some constant $C>0$. Then with probability $1-o(1)$, we have
    \[
        (1-\delta)(nq)^d\leq N_d\leq (1+\delta)(nq)^d
    \]
    for all $0\leq d<d^*$. In particular, we have
    \[
        \frac{N_d}{(nq)^d}\pto1
    \]
    for all $0\leq d<d^*$.
\end{proposition}
\begin{proof}
    Let $L=(1-\delta)^{1/3}$ and $R=(1+\delta)^{1/3}$. Note that for large enough $n$,
    \begin{equation}\label{eqn:lbound}
        L^{1/k}=(1-\delta)^{1/(3k)}\leq\left(1-\frac{C}{\log\log n}\right)^{1/(3k)}\leq1-\frac{C}{3k\log\log n}
    \end{equation}
    and
    \begin{equation}\label{eqn:rbound}
        R^{1/k}=(1+\delta)^{1/(3k)}\geq\left(1+\frac{C}{\log\log n}\right)^{1/(3k)}\geq1+\frac{C}{6k\log\log n}
    \end{equation}
    where \eqref{eqn:lbound} is from Lemma~\ref{lem:alg1} and \eqref{eqn:rbound} is from the fact that $1+x\leq e^x\leq 1+2x$ for small enough $x>0$. Based on \eqref{eqn:nd_process}, our proof strategy is to apply concentration of measure and
    the union bound to establish that the collection of events $E_d$ (parameterized by depth $d$) hold simultaneously with high probability, where 
    \[
        E_d=\begin{cases}
            L^{1/2}nq\leq N_1\leq R^{1/2}nq & \text{if } d=1,\\
            L(nq)^2\leq N_2\leq R(nq)^2 & \text{if } d=2,\\
            L^{1+(d-2)/\ell^*}(nq)^d\leq N_d\leq R^{1+(d-2)/\ell^*}(nq)^d& \text{if }3\leq d\leq d^*-2,\\
            L^3(nq)^d\leq N_d\leq R^3(nq)^d&\text{if } d=d^*-1.
        \end{cases}
    \]
    The proof consists of three steps; from $d=0$ to $2$, from $d=3$ to $d^*-2$, and $d=d^*-1$. 
    For convenience, we define
    \[
        n_d:= n-\sum_{i=0}^{d-1}N_i
    \]
    and
    \[
        q_d:=1-(1-q)^{N_{d-1}}\,.
    \]
    
    \paragraph{Step 1.} We prove $\Pr(E_1)\geq1-o(1)$ and $\Pr(E_2\mid E_1)\geq 1-o(1)$. Since $N_1\sim\Binom(n-1,q)$, Lemma~\ref{lem:chernoff} together with \eqref{eqn:lbound} and \eqref{eqn:rbound} gives that
    \[
        L^{1/2}nq\leq N_1\leq R^{1/2}nq
    \]
    with probability $1-o(1)$.

    Under this event, we have
    \[
        n_2=n-N_0-N_1\geq n-1-R^{1/2}nq\geq n-2\alpha_n\log n\geq L^{1/6}n
    \]
    and
    \[
        qN_1\leq 2R^{1/2}nq^2\leq2(1-L^{1/6})
    \]
    for large $n$, which gives by Lemma~\ref{lem:alg1} and Lemma~\ref{lem:alg2}
    \[
        L^{2/3}nq^2\leq L^{1/6}qN_1\leq q_2=1-(1-q)^{N_1}\leq qN_1\leq R^{1/2}nq^2\,.
    \]
    By Lemma~\ref{lem:chernoff}, we obtain
    \[
        \Pr(N_2\leq R^{1/2}n_2q_2\leq R(nq)^2\mid E_1)\geq 1-o(1)
    \]
    and
    \[
        \Pr(N_2\geq L^{1/6}n_2q_2\geq L(nq)^2\mid E_1)\geq1-o(1)\,.
    \]
    
    \paragraph{Step 2.} We prove $\Pr(E_d\mid E_1,\cdots, E_{d-1})\geq 1-o(\frac{1}{\log n})$ for $3\leq d\leq d^*-2$ by generalizing Step 1. Suppose that under $E_1,\cdots,E_{d-1}$, we condition on $N_0,\cdots,N_{d-1}$. We first obtain a lower bound on $n_d$ and $q_d$.
    \begin{claim}
        Under $E_{d-1}$, we have $n_d\geq L^{1/(3\ell^*)}n$ and $q_d\geq L^{1/(3\ell^*)}qN_{d-1}$.
    \end{claim}
    To prove the claim, we need an upper bound of $L^{1/(3\ell^*)}$ to make it easier to work with. By \eqref{eqn:lbound} we get
    \begin{equation}\label{eqn:ubl}
        L^{1/(3\ell^*)}\leq1-\frac{C}{9\ell^*\log\log n}=1-\frac{C\log(nq)}{9\log n\log\log n}\leq1-\frac{C}{18\log n}
    \end{equation}
    for large enough $n$, since $nq=\Omega(\log n)$.
    
    Now we bound $n_d$. Under $E_{d-1}$, we have that
    \[
        \sum_{k=0}^{d-1}N_k\leq 1+Rnq+\sum_{k=2}^{d-1}R(R^{1/d^*}nq)^k\leq R^2\left(\sum_{k=0}^{d-1}(nq)^k\right)\leq 2R^2(nq)^{d-1}\,.
    \]
    From \eqref{eqn:delta-sandwich} we have $d\leq d^*-2<\ell^*-1$, so we can further bound this as
    \[
        2R^2(nq)^{d-1}\leq2R^2(nq)^{\ell^*-2}=\frac{2R^2}{(nq)^{2}}\cdot n\,.
    \]
    
    We want to show $\frac{2R^2}{(nq)^{2}}\leq 1-L^{1/(3\ell^*)}$, and by \eqref{eqn:ubl} it is enough to show
    \[
        \frac{2R^2}{(nq)^{2}}\leq\frac{C}{18\log n}
    \]
    for sufficiently large $n$, which easily follows from $nq=\Omega(\log n)$.
    
    To lower bound $q_d$, we use \eqref{eqn:lambda-sandwich} to observe that
    \[
        qN_{d-1}\leq qR^2(nq)^{d-1}\leq qR^2(nq)^{d^*-3}=R^2(nq)^{d^*-\ell^*-2}\leq\frac{R^2}{nq(\log\log n)^2}\,.
    \]
    Similar to the preceding paragraph, we have for sufficiently large $n$ that
    \[
        \frac{R^2}{nq(\log\log n)^2}\leq\frac{C}{9\log n}\leq2(1-L^{1/(3\ell^*)})\,.
    \]
    Hence, by Lemma~\ref{lem:alg1} and Lemma~\ref{lem:alg2}, we obtain
    \[
        q_d=1-(1-q)^{N_{d-1}}\geq 1-e^{-qN_{d-1}}\geq L^{1/(3\ell^*)}qN_{d-1}
    \]
    which proves the claim.

    Combining the two lower bounds, we have under $E_{d-1}$ that
    \begin{equation}\label{eqn:nqbound}
        \E[N_d\mid N_0,\cdots,N_{d-1}]=n_dq_d\geq L^{2/(3\ell^*)}nqN_{d-1}\geq L^{2/(3\ell^*)}\cdot L^{1+(d-3)/\ell^*}(nq)^d\,.
    \end{equation}
    Now we are ready to achieve bounds for the probability $\Pr(E_d\mid E_1,\cdots, E_{d-1})$. Since $n_d\leq n$ and $q_d\leq qN_{d-1}$ by Lemma~\ref{lem:alg1}, Lemma~\ref{lem:chernoff} gives
    \[
        \Pr(N_d\leq R^{1/\ell^*}nqN_{d-1}\leq R^{1+(d-2)/\ell^*}(nq)^d\mid E_1,\cdots, E_{d-1})\geq 1-\exp\left(-\frac{(R^{1/\ell^*}-1)^2n_dq_d}{3}\right)\,.
    \]
    Note that by \eqref{eqn:rbound}
    \[
        R^{1/\ell^*}-1\geq\frac{C}{6\ell^*\log\log n}\geq\frac{C}{12\log n}
    \]
    and by \eqref{eqn:nqbound}
    \[
        n_qd_q\geq L^{2/(3\ell^*)}\cdot L^{1+(d-3)/\ell^*}(nq)^d\geq L^2(nq)^3=\Omega\left((\log n)^3\right)\,.
    \]
    Thus, we arrive at
    \[
        \Pr(N_d\leq R^{1+(d-2)/\ell^*}(nq)^d\mid E_1,\cdots,E_{d-1})\geq 1-e^{-\Omega(\log n)}\geq1-o\left(\frac{1}{\log n}\right)\,.
    \]
    Similarly, from \eqref{eqn:nqbound},
    \[
        \Pr(N_d\geq L^{1/(3\ell^*)}n_dq_d\geq L^{1+(d-2)/\ell^*}(nq)^d\mid E_1,\cdots,E_{d-1})\geq1-\exp\left(-\frac{(1-L^{1/(3\ell^*)})^2n_dq_d}{2}\right)\,.
    \]
    Using \eqref{eqn:ubl}, we get a similar bound
    \[
        \Pr(N_d\geq L^{1+(d-2)/\ell^*}(nq)^d\mid E_1,\cdots, E_{d-1})\geq1-e^{\Omega(\log n)}\geq1-o\left(\frac{1}{\log n}\right)\,.
    \]

    \paragraph{Step 3.} For $d=d^*-1$ we prove $\Pr(E_d\mid E_1,\cdots, E_{d-1})\geq 1-o(1)$, and to this end we need $n_d>L^{1/3}n$ and $q_d>L^{1/3}qN_{d-1}$. The proof is almost identical, but since it is important to see what could possibly break down if $d=d^*$, we include a proof.

    Bounding $n_d$ is nearly identical; $\sum_{k=0}^{d-1}N_k\leq 2R^2(nq)^{d-1}$ still holds. By \eqref{eqn:delta-sandwich} we have
    \[
        2R^2(nq)^{d-1}\leq2R^2(nq)^{\ell^*-1}=\frac{2R^2}{nq}\cdot n\,.
    \]
    By \eqref{eqn:lbound} it follows that $\frac{2R^2}{nq}\leq 1-L^{1/3}$ for large $n$. This gives $n_d\geq L^{1/3}n$.

    For $q_d$, note that by \eqref{eqn:lambda-sandwich}
    \[
        qN_{d-1}\leq qR^2(nq)^{d^*-2}=R^2(nq)^{d^*-\ell^*-1}\leq\frac{R^2}{(\log\log n)^2}\,.
    \]
    We have $\frac{R^2}{(\log\log n)^2}\leq 2(1-L^{1/3})$, which implies
    \[
        q_d\geq L^{1/3}qN_{d-1}\,.
    \]
    Now by the same argument with Step 2 it is easy to see that
    \[
        \Pr(E_{d}\mid E_1,\cdots, E_{d-1})\geq 1-o(1)
    \]
    for $d=d^*-1$.
    
    Combining the three steps gives the desired result.
\end{proof} 

We are now ready to prove Theorem~\ref{thm:main1}. In fact, we prove a slightly stronger result which turns out to be useful in Section~\ref{sec:gibbs}.

\begin{theorem}\label{thm:main1-tight}
    For any constant $C>0$, we have
    \[
        \left|\frac{N_{d^*}}{(1-\lambda_n)n}-1\right|<\frac{C}{\log\log n}
    \]
    and
    \[
        \left|\frac{N_{d^*+1}}{n}-\lambda_n\right|<\frac{C}{\log\log n}
    \]
    with probability at least $1-o(1)$.
\end{theorem}
\begin{proof}
    For $0<\delta_1<1$ and $0<\delta_2<1$ that may depend on $n$, consider an event
    \[
        E=\{|\Gamma_{\leq d^*-1}|\leq\delta_1n\}\cap\left\{ (1-\delta_2)(nq)^{d^*-1}\leq N_{d^*-1}\leq(1+\delta_2)(nq)^{d^*-1} \right\}\,.
    \]
    Now we condition on $N_{d^*-1}$, $|\Gamma_{\leq d^*-1}|$, and the event $E$. Recall that
    \[
        N_{d^*}\sim\Binom(n_{d^*}, q_{d^*})
    \]
    where $n_{d^*}=n-|\Gamma_{\leq d^*-1}|$ and $q_{d^*}=1-(1-q)^{N_{d^*-1}}$. First, note that by Lemma~\ref{lem:alg1}
    \[
        q_{d^*}\geq1-e^{-qN_{d^*-1}}\geq 1-e^{-(1-\delta_2)q(nq)^{d^*-1}}=1-e^{-(1-\delta_2)(nq)^{d^*-\ell^*}}=1-\lambda_n^{1-\delta_2}\,.
    \]
    By Lemma~\ref{lem:alg1} we have
    \begin{equation}
        q_{d^*}\geq1-\lambda_n^{1-\delta_2}=1-(1-(1-\lambda_n))^{1-\delta_2}\geq1-(1-(1-\delta_2)(1-\lambda_n))=(1-\lambda_n)(1-\delta_2)\,.
    \end{equation}
    For an upper bound on $q_{d^*}$, we use the fact $1-q\geq e^{-q}(1-q^2)$ and again apply Lemma~\ref{lem:alg1} to get
    \[
        q_{d^*}\leq 1-e^{-qN_{d^*-1}}(1-q^2)^{N_{d^*-1}}\leq 1-e^{-qN_{d^*-1}}(1-q^2N_{d^*-1})\leq1-\lambda_n^{1+\delta_2}+nq^2\,.
    \]
    Similar to the lower bound, Lemma~\ref{lem:alg1} gives
    \[
        q_{d^*}\leq1-\lambda_n^{1+\delta_2}+nq^2=1-(1-(1-\lambda_n))^{1+\delta_2}+nq^2\leq(1+\delta_2)(1-\lambda_n)+nq^2
    \]
    Thus, we have
    \[
        n_{d^*}q_{d^*}\geq n(1-\lambda_n)(1-\delta_1)(1-\delta_2)\geq n(1-\lambda_n)(1-(\delta_1+\delta_2))
    \]
    and
    \[
        n_{d^*}q_{d^*}\leq n(1-\lambda_n)(1+\delta_2)+n^2q^2=n(1-\lambda_n)\left(1+\delta_2+\frac{nq^2}{1-\lambda_n}\right)\,.
    \]
    Note that for large enough $n$,
    \[
        1-\lambda_n\geq1-e^{-1/(\log\log n)^2}\geq\frac{1}{2(\log\log n)^2}
    \]
    by, e.g., Lemma~\ref{lem:alg2}. Thus,
    \[
        \frac{nq^2}{1-\lambda_n}\leq\frac{\alpha_n^2(\log n)^2(\log\log n)^2}{n}\,.
    \]
    We may now invoke Proposition~\ref{prop:conc} with $\delta_1=C_1/\log\log n$ and $\delta_2=C_2/\log\log n$ which gives that for any constant $C_3>0$
    \[
        \left|n_{d^*}q_{d^*}-(1-\lambda_n)n\right|<\frac{C_3(1-\lambda_n)n}{\log\log n}
    \]
    with probability $1-o(1)$. Using Lemma~\ref{lem:chernoff} we can argue that the concentration of $\Binom(n_{d^*},q_{d^*})$ is much stronger than the RHS, leading to
    \[
        \left|\frac{N_{d^*}}{(1-\lambda_n)n}-1\right|<\frac{C}{\log\log n}
    \]
    with probability $1-o(1)$.

    To prove the second part of the theorem, we note that with high probability
    \[
        N_{d^*}\geq\frac{1-\lambda_n}{2}n\geq\frac{1}{2}(1-e^{-1/(\log\log n)^2})n\geq\frac{n}{4(\log\log n)^2}\,.
    \]
    Conditioned on $\Gamma_{\leq d^*}$, the probability that a vertex in $V\setminus\Gamma_{\leq d^*}$ is not connected to any of $\Gamma_{d^*}$ is at most
    \begin{equation}\label{eqn:dp2small}
        (1-q)^{N_{d^*}}\leq e^{-qN_{d^*}}\leq e^{-\frac{nq}{4(\log\log n)^2}}
    \end{equation}
    which is far less than $o(1/\log\log n)$. Thus, the second part follows from the first part.
    
\end{proof}

We note that the asymptotics for $N_{d^*}$ is somewhat different from $N_d$ for $d\leq d^*-1$. Nevertheless, if we only consider a subsequence where $\lambda_n\to\lambda=1$, then $N_{d^*}$ follows a similar asymptotics, though slightly weaker than Proposition \ref{prop:conc}. We state this as the following lemma, which turns out to be useful in a number of proofs.

\begin{lemma}\label{lem:dstar}
    Consider a subsequence of $n\in\mathbb{Z}^+$ such that $\lambda_n\to\lambda=1$. Then we have
    \[
        \frac{N_{d^*}}{(nq)^{d^*}}\pto1\,.
    \]
\end{lemma}
\begin{proof}
    We have
    \[
        \frac{N_{d^*}}{(nq)^{d^*}}=\frac{N_{d^*}}{n\log(1/\lambda_n)}=\frac{N_{d^*}}{n(1-\lambda_n)}\cdot\frac{1-\lambda_n}{\log(1/\lambda_n)}\,.
    \]
    Since $\frac{1-\lambda_n}{\log(1/\lambda_n)}\to1$ if $\lambda_n\to1$, the conclusion follows from Theorem~\ref{thm:main1-tight}.
\end{proof}

\subsection{Shortest paths between two fixed vertices}

We develop some basic properties of sparse random graphs, which are more or less standard. These are useful in Section~\ref{sec:unstable} where we prove that the overlap of uniformly random shortest paths does not concentrate.

\begin{proposition}\label{prop:localtree}
    Let $d\leq\frac{d^*-1}{2}$. Then with probability $1-o(1)$, the shortest path DAG is locally a tree up to distance $d$. In other words, with probability $1-o(1)$, there is a unique shortest path from $1$ to any vertex with distance at most $d$.
\end{proposition}
\begin{proof}
    We prove that for all $1\leq k\leq d$ every vertex in $\Gamma_{k}$ is connected to at most one vertex in $\Gamma_{k-1}$, simultaneously with high probability. Conditioned on $\Gamma_1,\cdots,\Gamma_{k-1}$, the number of edges between $v\in\Gamma_{\geq k}$ and $\Gamma_{k-1}$ follows $\Binom(N_{k-1},q)$, so the probability that this number is greater than $1$ is
    \[
        1-(1-q)^{N_{k-1}}-qN_{k-1}(1-q)^{N_{k-1}-1}\,.
    \]
    Using Lemma~\ref{lem:alg2}, this can be bounded by
    \[
        \begin{split}
            1-(1-q)^{N_{k-1}}-qN_{k-1}(1-q)^{N_{k-1}-1}&\leq qN_{k-1}-qN_{k-1}(1-q)^{N_{k-1}-1}\\
            &\leq q^2N_{k-1}^2\,.
        \end{split}
    \]
    Hence, the probability that none of the vertices in $\Gamma_{\geq k}$ have more than one connection to $\Gamma_{k-1}$ is at least
    \[
        (1-q^2N_{k-1}^2)^{|\Gamma_{\geq k}|}\geq (1-q^2N_{k-1}^2)^n\geq1-nq^2N_{k-1}^2\,.
    \]
    Thus, the probability that the shortest path DAG is locally a tree up to distance $d$ is at least
    \[
        \prod_{k=1}^d(1-nq^2N_{k-1}^2)\geq1-nq^2\sum_{k=1}^dN_{k-1}^2\,.
    \]
    By Proposition~\ref{prop:conc}, with overwhelming probability this is further bounded by
    \[
        1-nq^2\sum_{k=1}^dN_{k-1}^2\geq1-nq^2\sum_{k=1}^d2(nq)^{2(k-1)}\geq1-4nq^2(nq)^{2(d-1)}=1-4q(nq)^{2d-1}\,.
    \]
    By \eqref{eqn:lambda-sandwich} we have
    \[
        q(nq)^{2d-1}\leq q(nq)^{d^*-2}=(nq)^{d^*-\ell^*-1}\leq\frac{1}{(\log\log n)^2}=o(1)
    \]
    which proves the result.
\end{proof}

\begin{corollary}\label{cor:no-short-cycle}
    With probability $1-o(1)$, there is no cycle of length smaller than $d^*$ that contains a fixed vertex $1$.
\end{corollary}
\begin{proof}
    First, suppose that $d^*$ is odd. Then by Proposition~\ref{prop:localtree} the shortest path DAG is locally a tree up to distance $d=\frac{d^*-1}{2}$. In addition to this, we prove that the whole graph $G$ is locally a tree up to distance $d-1$ with high probability, which would follow by showing that
    \[
        \sum_{k=1}^{d-1}\binom{N_k}{2}q\pto0\,.
    \]
    Indeed, by \eqref{eqn:lambda-sandwich} and Proposition~\ref{prop:conc}, we get
    \[
        \sum_{k=1}^{d-1}\binom{N_k}{2}q\leq\binom{2N_{d-1}}{2}q\leq 2qN_{d-1}^2\leq 8q(nq)^{2d-2}\leq8(nq)^{d^*-\ell^*-2}\leq\frac{8}{nq(\log\log n)^2}
    \]
    with probability $1-o(1)$, as desired. Under these events, the shortest possible cycle containing $1$ has length at least $2d+1=d^*$.

    A similar argument works when $d^*$ is even. In this case, we set $d=\frac{d^*}{2}-1$ and deduce that $G$ is locally a tree up to distance $d$, so under the good event the shortest possible cycle has length at least $2(d+1)=d^*$. The details are omitted.
\end{proof}

Proposition~\ref{prop:localtree} gives a nearly complete picture of the union of the shortest paths from fixed vertices $1$ to $2$. Note that by Proposition~\ref{prop:conc} the distance between $1$ and $2$ is either $d^*$ or $d^*+1$ with overwhelming probability. When $d^*$ is odd, $d=\frac{d^*-1}{2}$ is an integer, so if we run Dijkstra's algorithm (or simply BFS) then the shortest path DAGs from $1$ and from $2$ will be disjoint trees up to distance $d$ with high probability, so what happens between $\Gamma_{d}(1)$ and $\Gamma_{d}(2)$ will determine the distance between $1$ and $2$ and the set of shortest paths between them. If there exists an edge between them, the distance will be $d^*$ and otherwise $d^*+1$.

Now the most interesting case is when $0<\lambda<1$, i.e. the distance can be either $d^*$ or $d^*+1$, both with non-negligible probability. Then the number of edges between $\Gamma_d(1)$ and $\Gamma_d(2)$ will follow the distribution
\begin{equation}\label{eqn:nd12}
    \Binom(N_{d}(1)\cdot N_d(2),q)\,.
\end{equation}
By Proposition~\ref{prop:conc}, note that
\[
    \frac{N_d(1)\cdot N_d(2)\cdot q}{(nq)^{2d}q}\pto1\,.
\]
Since $(nq)^{2d}q=(nq)^{d^*-1}q=\log(1/\lambda_n)\to\log(1/\lambda)$, it follows that the limiting distribution of (\ref{eqn:nd12}) is in fact $\Pois(\log(1/\lambda))$. Hence, the probability that there exists a unique edge between $\Gamma_d(1)$ and $\Gamma_d(2)$ converges in probability to $\lambda\log(1/\lambda)$. Since the shortest path DAGs up to distance $d$ were trees, this actually implies there is a unique path between $1$ and $2$.

We summarize our discussion to the following theorem.
\begin{theorem}\label{thm:uniquepath}
    Suppose $\{\alpha_n\}_{n\in\mathbb{Z}^+}$ is chosen so that $d^*$ is odd and $\lambda_n\to\lambda\in(0,1)$. Then the probability that there is a unique shortest path of length $d^*$ from $1$ to $2$ converges to $\lambda\log(1/\lambda)$.
\end{theorem}

\section{Asymptotics for correlated Erd\"os--R\'enyi graphs}\label{sec:correlated}

In this section, we analyze the setting of \emph{correlated} random graphs relevant to OGP theory. The main result in this section is the computation of the overlap of random shortest path trees in Section~\ref{sec:overlap-trees}. Before we get there, it is helpful to first analyze the behavior of the correlated \emph{distance vectors} and \emph{shortest path DAGs} which we do first.  

\paragraph{Additional notation.} As usual, we drop the subscripts $G^{(1)}=G_n^{(1)}$, $G^{(2)}=G_n^{(2)}$. We write $\Gamma_d^{(1)}=\Gamma_d^{G^{(1)}}$ and $\Gamma_d^{(2)}=\Gamma_d^{G^{(2)}}$ to denote the set of vertices with distance $d$ within the graphs $G^{(1)}$ and $G^{(2)}$, respectively. Similarly, we define $N_d^{(1)}$ and $N_d^{(2)}$. Also, we denote by $\Gamma_d^I$ the intersection $\Gamma_d^{(1)}\cap\Gamma_d^{(2)}$, and $N_d^I=|\Gamma_d^I|$ the cardinality of the intersection.

\subsection{The overlap of the distance vectors}\label{sec:odv}

Recall that as in Theorem~\ref{thm:main1}, the shortest distance distribution heavily depends on the limit of the difference between $d_n^*$ and $\ell_n^*$. For cases 1 and 3, i.e. $\lambda_n=o(1)$ or $\lambda_n=1-o(1)$, then the overlap is quite trivial, because there is one fixed distance $d_n^*$ or $d_n^*+1$ which dominates almost the entire graph. Hence, the most interesting case to study the overlap is case 2, where $\lambda_n\to\lambda\in(0,1)$.

As a reminder, we restate the definition \eqref{eqn:defgamma} of $\gamma_n$:
\[
    \gamma_n:=\rho_n^{d_n^*}\in(0,1)\,.
\]

\begin{theorem}\label{thm:main2}
    Suppose that $\lambda_n\to\lambda$ and $\gamma_n\to\gamma$ with $\lambda,\gamma\in[0,1]$. Then as $n \to \infty$, we have that 
    \[\frac{N_{d_n^*}^I}{n} \pto 1-2\lambda+\lambda^{2-\gamma}, \]
    i.e. the proportion of the vertices that have distance $d^*$ for both $G_n^{(1)}$ and $G_n^{(2)}$ converges to $1-2\lambda+\lambda^{2-\gamma}$ in probability.
\end{theorem}

It is helpful to consider the two extreme cases. If $\gamma=1$, then $N_{d^*}^I/n\to1-\lambda$ which is the same with the single graph case; $G^{(1)}$ and $G^{(2)}$ are nearly identical. On the other hand, $\gamma=0$, we get $N_{d^*}^I/n\to(1-\lambda)^2$ which can be interpreted as $(N_{d^*}^{(1)}/n)\cdot(N_{d^*}^{(2)}/n)$; due to correlation decay, the events of a vertex contained in $\Gamma_{d^*}^{(1)}$ and in $\Gamma_{d^*}^{(2)}$ are almost independent.

\begin{corollary}\label{cor:overlap1}
    Under the same setting with Theorem~\ref{thm:main2}, the proportion of the vertices that have the same distance for $G_n^{(1)}$ and $G_n^{(2)}$ converges to $1-2\lambda+2\lambda^{2-\gamma}$ in probability.
\end{corollary}
\begin{proof}
    The remaining vertices not contained in $\bigcup_{d=0}^{d^*}\Gamma_d^{(1)}$ and $\bigcup_{d=0}^{d^*}\Gamma_d^{(2)}$ have proportion $\lambda^{2-\gamma}$ in probability limit and both $\Gamma_{d^*+1}^{(1)}$ and $\Gamma_{d^*+1}^{(2)}$ occupy them by Theorem~\ref{thm:main1}.
\end{proof}



The proof of Theorem~\ref{thm:main2} is similar to the proof of Theorem~\ref{thm:main1}, except that we need to obtain concentration bounds for $N_d^I$ as well. Thus, we will prove a result similar to Proposition~\ref{prop:conc}. Recall the $\gamma_n$ defined in \eqref{eqn:defgamma}.

\begin{proposition}\label{prop:rhoconc}
    Let $\frac{C}{\log\log n}<\delta<\frac{1}{2}$ for some constant $C$. Suppose that $\gamma_n\to\gamma>0$. Then with probability $1-o(1)$, we have
    \[
        (1-\delta)(nq\rho_n)^d\leq N_d^I\leq (1+\delta)(nq\rho_n)^d
    \]
    for all $0\leq d<d^*$. Moreover, if $\gamma=0$, then
    \[
        \frac{N_{d^*-1}^I}{(nq)^{d^*-1}}\pto0\,.
    \]
\end{proposition}
\begin{proof}
    The cardinality of the intersection $N_d^I=|\Gamma_d^I|=|\Gamma_d^{(1)}\cap\Gamma_d^{(2)}|$ unfortunately does not have a nice exact hierarchical structure like \eqref{eqn:nd_process}. Nevertheless, we will obtain bounds using a similar strategy.
    
    Suppose that we are already given (conditioned on) $\Gamma_i^{(1)}$ and $\Gamma_i^{(2)}$ for all $0\leq i\leq d-1$, and we are at the step of constructing $\Gamma_{d}^{(i)}$ and $\Gamma_{d}^I$. We define the set of ``fresh'' vertices $V_{d}:=V\setminus\Gamma_{\leq d-1}^{(1)}\setminus\Gamma_{\leq d-1}^{(2)}$ and its cardinality $n_d:=|V_d|$. Each vertex $v\in V_{d}$ is either:
    \begin{enumerate}[label=(\alph*)]
        \item\label{item:casea} connected to both $\Gamma_{d-1}^{(1)}$ and $\Gamma_{d-1}^{(2)}$,
        \item\label{item:caseb} connected to $\Gamma_{d-1}^{(1)}$ only or $\Gamma_{d-1}^{(2)}$ only, or
        \item\label{item:casec} connected to neither $\Gamma_{d-1}^{(1)}$ nor $\Gamma_{d-1}^{(2)}$.
    \end{enumerate}
    Then for \ref{item:casea} $v$ will be included in $\Gamma_{d}^I$ and for \ref{item:casec} $v$ will be included in $V_{d+1}$. Let $\sim_i$ and $\nsim_i$ denote the connectivity and disconnectivity in graph $G^{(i)}$, and when there is no subscript we mean for both. We are interested in the probability of \ref{item:casea}, which is calculated as
    \[
        \begin{split}
            q_{d}&:=\Pr(v\sim_1\Gamma_{d-1}^{(1)}\wedge v\sim_2\Gamma_{d-1}^{(2)})\\
            &=1-\Pr(\text{$v\nsim_1\Gamma_{d-1}^{(1)}$})-\Pr(\text{$v\nsim_2\Gamma_{d-1}^{(2)}$})+\Pr(\text{$v\nsim_1\Gamma_{d-1}^{(1)}\wedge v\nsim_2\Gamma_{d-1}^{(2)}$})\,.
        \end{split}
    \]
    Here, $\Pr(\text{$v\nsim_1\Gamma_{d-1}^{(1)}\wedge v\nsim_2\Gamma_{d-1}^{(2)}$})$ equals
    \begin{equation}\label{eqn:hard}
        \Pr(\text{$v\nsim_1\Gamma_{d-1}^{(1)}\wedge v\nsim_2\Gamma_{d-1}^{(2)}$})=\Pr(\text{$v\nsim_1\Gamma_{d-1}^{(1)}\setminus\Gamma_{d-1}^I$})\cdot\Pr(\text{$v\nsim_2\Gamma_{d-1}^{(2)}\setminus\Gamma_{d-1}^I$})\cdot\Pr(\text{$v\nsim\Gamma_{d-1}^I$})\,.
    \end{equation}
    
    The probability $\Pr(v\nsim\Gamma_{d-1}^I)$ is easy to compute. The probability that a pair of vertices is not connected by an edge in both $G^{(1)}$ and $G^{(2)}$ is $(1-q)(\rho_n+(1-\rho_n)(1-q))=(1-q)(1-q(1-\rho_n))$, so
    \[
        \Pr(v\nsim\Gamma_{d-1}^I)=(1-q)^{N_{d-1}^I}(1-q(1-\rho_n))^{N_{d-1}^I}\,.
    \]
    
    However, the calculation of $\Pr(v\nsim_1\Gamma_{d-1}^{(1)}\setminus\Gamma_{d-1}^{I})$ is more complicated: the edge probability in $G^{(1)}$ between $v$ and a vertex in $\Gamma_{d-1}^{(1)}\setminus\Gamma_{d-1}^I$ is not exactly $q$, because some of the vertices in $\Gamma_{d-1}^{(1)}\setminus\Gamma_{d-1}^I$ might be contained in $\Gamma_{\leq d-2}^{(2)}$. For these vertices, it is already conditioned that they are not connected to $v$ in $G^{(2)}$. Fortunately, this means that for these vertices the edge probability in $G^{(1)}$ is \emph{less} than $q$, which allows us to obtain bounds using ``truely fresh'' pair of vertices.

    Put it formally, let $\tilde{\Gamma}_d^{(1)}:=\Gamma_d^{(1)}\setminus\Gamma_{\leq d-1}^{(2)}$ and $\tilde{\Gamma}_d^{(2)}:=\Gamma_d^{(2)}\setminus\Gamma_{\leq d-1}^{(1)}$ (also set $\tilde{N}_d^{(1)}:=|\tilde{\Gamma}_d^{(1)}|$ and $\tilde{N}_d^{(2)}:=|\tilde{\Gamma}_d^{(2)}|$). Then when we set all edge probabilities to $q$, $\Pr(v\sim_1\Gamma_{d-1}^{(1)}\wedge v\sim_2\Gamma_{d-1}^{(2)})$ only gets bigger, and when we replace $\Gamma_{d-1}^{(1)}$ and $\Gamma_{d-1}^{(2)}$ with $\tilde{\Gamma}_{d-1}^{(1)}$ and $\tilde{\Gamma}_{d-1}^{(2)}$, respectively, the probability becomes smaller. This gives bounds
    \begin{equation}\label{eqn:qrbound}
        q_d \leq 1-(1-q)^{N_{d-1}^{(1)}}-(1-q)^{N_{d-1}^{(2)}}+(1-q)^{N_{d-1}^{(1)}+N_{d-1}^{(2)}-N_{d-1}^I}(1-q(1-\rho_n))^{N_{d-1}^I}
    \end{equation}
    and
    \begin{equation}\label{eqn:qlbound}
        q_d \geq 1-(1-q)^{\tilde{N}_{d-1}^{(1)}}-(1-q)^{\tilde{N}_{d-1}^{(2)}}+(1-q)^{\tilde{N}_{d-1}^{(1)}+\tilde{N}_{d-1}^{(2)}-N_{d-1}^I}(1-q(1-\rho_n))^{N_{d-1}^I}\,.
    \end{equation}

    The main part of the proof is similar to Proposition~\ref{prop:conc}. First, we consider the case $\gamma>0$. Let $L=(1-\delta)^{1/3}$ and $R=(1+\delta)^{1/3}$. We define the collection of events
    \[
        E_d=\begin{cases}
            L^{1/2}nq\rho_n\leq N_1^I\leq R^{1/2}nq\rho_n & \text{if } d=1,\\
            L(nq\rho_n)^2\leq N_2^I\leq R(nq\rho_n)^2 & \text{if } d=2,\\
            L^{1+(d-2)/\ell^*}(nq\rho_n)^d\leq N_d^I\leq R^{1+(d-2)/\ell^*}(nq\rho_n)^d& \text{if }3\leq d\leq d^*-2,\\
            L^3(nq\rho_n)^d\leq N_d^I\leq R^3(nq\rho_n)^d&\text{if } d=d^*-1.
        \end{cases}
    \]
    Again, the proof consists of three steps, namely
    \[
        \Pr(E_d\mid E_1,\cdots, E_{d-1})\geq\begin{cases}
            1-o(1) & \text{if $1\leq d\leq 2$}\,,\\
            1-o\left(\frac{1}{\log n}\right) & \text{if $3\leq d\leq d^*-2$}\,,\\
            1-o(1) & \text{if $d= d^*-1$}\,.
        \end{cases}
    \]
    We only prove for $3\leq d\leq d^*-2$ since the other cases are similar (see the proof of Proposition~\ref{prop:conc}). In order to apply Lemma~\ref{lem:chernoff}, we need to derive bounds for $n_d$ and $q_d$.
    
    We first handle $n_d$. Note that
    \[
        n_d\geq n-\sum_{k=0}^{d-1}N_k^{(1)}-\sum_{k=0}^{d-1}N_d^{(2)}\,.
    \]
    Thus, by applying the same argument as in the proof of Proposition~\ref{prop:conc}, it is easy to see that
    \begin{equation}\label{eqn:ndlb}
        n_d\geq L^{1/(3\ell^*)}n\,.
    \end{equation}

    For $q_d$, we establish a lower bound first. By Proposition~\ref{prop:conc}, we may safely assume $N_{k}^{(i)}\leq 2(nq)^{k}$ for all $k\leq d^*-1$. For some $\epsilon>0$ (not necessarily constant) to be chosen later, assume
    \begin{equation}\label{eqn:epsbound}
        q(nq)^{d-1}\leq \epsilon\,.
    \end{equation}
    Then
    \[
        qN_{d-1}^{(i)}\leq 2q(nq)^{d-1}\leq 2\epsilon\,.
    \]
    Hence, by Lemma~\ref{lem:alg1} and Lemma~\ref{lem:alg2}, we have
    \[
        (1-q)^{\tilde{N}_d^{(i)}}\leq1-(1-\epsilon)q\tilde{N}_{d-1}^{(i)}\,.
    \]
    Also by Lemma~\ref{lem:alg1}, we get
    \[
        (1-q)^{\tilde{N}_{d-1}^{(1)}+\tilde{N}_{d-1}^{(2)}-N_{d-1}^I}\geq 1-q(\tilde{N}_{d-1}^{(1)}+\tilde{N}_{d-1}^{(2)}-N_{d-1}^I)
    \]
    and
    \[
        (1-q(1-\rho_n))^{N_d^I}\geq 1-q(1-\rho_n)N_d^I\,.
    \]
    By plugging in these bounds to \eqref{eqn:qlbound}, we have
    \[
        \begin{split}
            q_d &\geq q\rho N_{d-1}^I-\epsilon q(\tilde{N}_{d-1}^{(1)}+\tilde{N}_{d-1}^{(2)})+q^2(1-\rho)(\tilde{N}_{d-1}^{(1)}+\tilde{N}_{d-1}^{(2)}-N_{d-1}^I)N_{d-1}^I\\
            &\geq q\rho N_{d-1}^I-\epsilon q(N_{d-1}^{(1)}+N_{d-1}^{(2)})\\
            &\geq q\rho N_{d-1}^I-4\epsilon q(nq)^{d-1}\,.
        \end{split}
    \]
    Under $E_{d-1}$, we have $N_{d-1}^I\geq \frac{1}{2}(nq\rho_n)^{d-1}$. Also note $\rho_n^d\geq\rho_n^{d^*}\geq\gamma/2$ for large enough $n$. This gives
    \begin{equation}\label{eqn:qrhobound}
        q\rho_n N_{d-1}^I\geq\frac{1}{2}q\rho_n(nq\rho_n)^{d-1}\geq\frac{1}{4}\gamma q(nq)^{d-1}\,.
    \end{equation}
    Hence,
    \[
        q_d\geq q\rho_n N_{d-1}^I-4\epsilon q(nq)^{d-1}\geq\left(1-\frac{16\epsilon}{\gamma}\right)q\rho_n N_{d-1}^I\,.
    \]
    Recall from \eqref{eqn:ubl} that we have for large enough $n$
    \[
        L^{1/(3\ell^*)}\leq 1-\frac{C}{18\log n}\,.
    \]
    Thus, we set (note that the assumption $\gamma>0$ is used here)
    \[
        \epsilon=\frac{C\gamma}{288\log n}=\Theta\left(\frac{1}{\log n}\right)
    \]
    which satisfies \eqref{eqn:epsbound} since
    \[
        q(nq)^{d-1}\leq(nq)^{d^*-\ell^*-2}\leq\frac{1}{nq(\log\log n)^2}\leq O\left(\frac{1}{(\log n) (\log\log n)^2}\right)\,.
    \]
    This gives a lower bound
    \[
        q_d\geq L^{1/(3\ell^*)}q\rho_n N_{d-1}^I\,.
    \]
    Combined with \eqref{eqn:ndlb} and \eqref{eqn:qrhobound}, we have
    \begin{equation}\label{eqn:nqlb}
        n_dq_d\geq L^{2/(3\ell^*)}nq\rho_n N_{d-1}^I\geq \frac{1}{8}\gamma(nq)^{d}=\Omega((\log n)^3)\,.
    \end{equation}
    Therefore, applying Lemma~\ref{lem:chernoff} and \eqref{eqn:ubl} gives
    \[
        \begin{split}
            \Pr(N_d\geq L^{1/(3\ell^*)}n_dq_d\geq L^{1+(d-2)/\ell^*}(nq)^d\mid E_1,\cdots, E_{d-1})&\geq1-\exp\left(-\frac{(1-L^{1/(3\ell^*)})^2n_dq_d}{2}\right)\\
            &\geq1-e^{-\Omega(\log n)}\\
            &\geq1-o\left(\frac{1}{\log n}\right)\,.
        \end{split}
    \]

    Deriving an upper bound is similar. We abuse notation and reuse $\epsilon$ but possibly with a different value. Since
    \[
        q(N_{d-1}^{(1)}+N_{d-1}^{(2)})\leq 4q(nq)^{d-1}\leq 4\epsilon
    \]
    by Lemma~\ref{lem:alg1} and Lemma~\ref{lem:alg2} we have
    \[
        (1-q)^{N_{d-1}^{(1)}+N_{d-1}^{(2)}-N_{d-1}^I}\leq 1-(1-2\epsilon)q(N_{d-1}^{(1)}+N_{d-1}^{(2)}-N_{d-1}^I)
    \]
    and
    \[
        (1-q(1-\rho))^{N_d^I}\geq 1-(1-2\epsilon)q(1-\rho_n)N_d^I\,.
    \]
    Also, note that $(1-q)^{N_{d-1}^{(i)}}\leq 1-qN_{d-1}^{(i)}$. Applying these inequalities to \eqref{eqn:qrbound}, one can see that
    \[
        \begin{split}
            q_d &\leq 1-(1-q)^{N_{d-1}^{(1)}}-(1-q)^{N_{d-1}^{(2)}}+(1-q)^{N_{d-1}^{(1)}+N_{d-1}^{(2)}-N_{d-1}^I}(1-q(1-\rho))^{N_{d-1}^I}\\
            &\leq q\rho N_{d-1}^I+2\epsilon q(N_{d-1}^{(1)}+N_{d-1}^{(2)}-N_{d-1}^I)+(1-2\epsilon)^2q^2(1-\rho)N_{d-1}^I(N_{d-1}^{(1)}+N_{d-1}^{(2)}-N_{d-1}^I)\\
            &\leq q\rho N_{d-1}^I+2\epsilon q(N_{d-1}^{(1)}+N_{d-1}^{(2)})+q^2(N_{d-1}^{(1)}+N_{d-1}^{(2)})^2\\
            &\leq q\rho N_{d-1}^I+6\epsilon q(N_{d-1}^{(1)}+N_{d-1}^{(2)})\\
            &\leq q\rho N_{d-1}^I+24\epsilon q(nq)^{d-1}\\
            &\leq\left(1+\frac{96\epsilon}{\gamma}\right)q\rho N_{d-1}^I\,.
        \end{split}
    \]
    Note that by \eqref{eqn:rbound} we have
    \begin{equation}\label{eqn:lbr}
       R^{1/(2\ell^*)}\geq1+\frac{C}{12\ell^*\log\log n}=1+\frac{C\log(nq)}{12\log n\log\log n}\geq1+\frac{C}{24\log n}\,.
    \end{equation}
    for large $n$. Thus, this time we set
    \[
        \epsilon=\frac{C\gamma}{2304\log n}=\Theta\left(\frac{1}{\log n}\right)\,.
    \]
    which satisfies \eqref{eqn:epsbound} as well. This gives an upper bound
    \[
        q_d\leq R^{1/(2\ell^*)}q\rho_n N_{d-1}^I
    \]
    and thus by \eqref{eqn:nqlb} and \eqref{eqn:lbr}
    \[
        \begin{split}
            \Pr(N_d\leq R^{1/(2\ell^*)}n_dq_d\leq R^{1+(d-2)/\ell^*}(nq)^d\mid E_1,\cdots, E_{d-1})&\geq1-\exp\left(-\frac{(R^{1/(2\ell^*)}-1)^2n_dq_d}{2}\right)\\
            &\geq1-e^{-\Omega(\log n)}\\
            &\geq1-o\left(\frac{1}{\log n}\right)
        \end{split}
    \]
    which concludes the proof.

    The proof for the case $\gamma=0$ is mostly similar. Note that the above proof step works as long as $\rho^d$ is bounded away from zero. Intuitively, since $\rho^{d^*}\to0$, $\rho^d$ will eventually be $o(1)$ as $d$ increases. After this step, what the inequality
    \[
        q_d\leq q\rho_n N_{d-1}^I+24\epsilon q(nq)^{d-1}
    \]
    tells us is that $\rho_n N_{d-1}^I=q(nq)^{d-1}\cdot\rho_n^d=o(q(nq)^{d-1})$, so $n_dq_d\leq o((nq)^d)$ as desired. More formally, $\gamma=0$ implies that for any $\delta'>0$ we can find $d<d^*$ such that $\rho_n^d<\delta'$. This implies that
    \[
        q_d\leq q\rho_n N_{d-1}^I+24\epsilon q(nq)^{d-1}\leq2\delta' q(nq)^{d-1}+24\epsilon q(nq)^{d-1}=(2\delta'+24\epsilon)q(nq)^{d-1}\,.
    \]
    Thus, it easily follows that $N_d^I/(nq)^d\dto0$. It is now straightforward to prove that this continues to hold for larger $d$.
\end{proof}

The following lemma complements Proposition~\ref{prop:rhoconc}, similar to Lemma~\ref{lem:dstar}.

\begin{lemma}\label{lem:rhodstar}
    Suppose that $\lambda_n\to\lambda=1$ and $\gamma_n\to\gamma$. Then
    \[
        \frac{N_{d^*}^I}{(nq)^{d^*}}\pto\gamma\,.
    \]
\end{lemma}
\begin{proof}
    The proof is nearly identical to Proposition~\ref{prop:rhoconc}, thus omitted.
\end{proof}

\begin{proof}[Proof of Theorem~\ref{thm:main2}]
    Using Propositions \ref{prop:conc} and \ref{prop:rhoconc}, we prove that each of the bounds from \eqref{eqn:qrbound} and \eqref{eqn:qlbound} converges in probability to the desired bound. The result is trivial if $\lambda=0$ or $\lambda=1$ by Theorem~\ref{thm:main1}, so we assume $0<\lambda<1$. Note that for any $1\leq d\leq d^*$ and a constant $\epsilon>0$ we have
    \begin{equation}\label{eqn:tildebound}
        \tilde{N}_{d}^{(1)}=|\Gamma_d^{(1)}\setminus\Gamma_{\leq d-1}^{(2)}|\geq|\Gamma_d^{(1)}|-|\Gamma_{\leq d-1}^{(2)}|\geq (1-\epsilon)(nq)^d-2(nq)^{d-1}\geq (1-2\epsilon)(nq)^d
    \end{equation}
    with high probability. Hence, as in the proof of Theorem~\ref{thm:main1}, the terms $(1-q)^{N_{d^*-1}^{(i)}}$ and $(1-q)^{\tilde{N}_{d^*-1}^{(i)}}$ all converge in probability to $\lambda$.
    
    First, consider the case $\gamma>0$. Then $\rho_n=1-o(1)$ so $q(1-\rho_n)N_{d^*-1}^I\pto0$, which implies
    \[
        (1-q(1-\rho_n))^{N_{d^*-1}^I}\pto0\,.
    \]
    Also, since $\rho_n^{d^*}\to \gamma$, Proposition~\ref{prop:rhoconc} gives $qN_{d^*-1}^I\pto \gamma\log(1/\lambda)$, so that
    \[
        (1-q)^{N_{d^*-1}^{(1)}+N_{d^*-1}^{(2)}-N_{d-1}^I}\pto \lambda^{2-\gamma}\,.
    \]
    Due to \eqref{eqn:tildebound} the same convergence holds when $\tilde{N}_{d^*-1}^{(i)}$ are used in place of $N_{d^*-1}^{(i)}$. Combining these facts gives the desired result.

    If $\gamma=0$, then by the second statement of Proposition~\ref{prop:rhoconc} we have $qN_{d^*-1}^I\to0$, so the same result holds.
    
\end{proof}

\subsection{The overlap of the shortest path DAGs}\label{sec:spdag}

We define the \emph{shortest path DAG} of a graph $G$, denoted by $\mathsf{SPD}_G(v)$, as the union of the shortest paths starting from $v\in V(G)$ where the orientation agrees with the direction of the outgoing paths from $v$. We also study a natural measure of overlap for shortest path DAGs which we now formally present. This result does not have direct implications for shortest path trees, but shortest path DAGs themselves are natural objects to consider from the perspective of, e.g., linear programming duality --- see the discussion in Preliminaries. When $\lambda = 0$ and $\rho = 1$, the overlap of shortest path DAGs exhibits a subtle behavior not seen for shortest path trees.

Since $1$ is the source, we simply write $\mathsf{SPD}_G=\mathsf{SPD}_G(1)$, and the overlap is defined by
\begin{equation}\label{eqn:overlap2def}
    S_n:=\frac{|\mathsf{SPD}_{G_n^{(1)}}\cap \mathsf{SPD}_{G_n^{(2)}}|}{\sqrt{|\mathsf{SPD}_{G_n^{(1)}}|\cdot|\mathsf{SPD}_{G_n^{(2)}}|}}\,.
\end{equation}
As with trees, we require that the edges in the overlap have the same orientation for $G^{(1)}$ and $G^{(2)}$.

For the distance vectors discussed in Section \ref{sec:odv}, the cases $\lambda=0$ and $\lambda=1$ were trivial and the interesting regime was $\lambda\in(0,1)$. For the shortest path DAGs, the situation is the opposite: the case $\lambda\in(0,1)$ is easy and the case $\lambda=0$ will be the trickiest. This is because when $\lambda\in(0,1)$, it is obvious that most of the edges in the shortest path DAG will connect $\Gamma_{d^*}$ and $\Gamma_{d^*+1}$, but when $\lambda=0$, it is not clear which of $\Gamma_{d^*-1}$ and $\Gamma_{d^*+1}$ is larger\footnote{At first glance, one might expect $\lambda=1$ would suffer from the same issue, but fortunately it turns out that we always have $|\Gamma_{d^*}|\gg|\Gamma_{d^*+2}|$, so it can be handled in the same way as the case $\lambda\in(0,1)$.}.

We first discuss intuitively what is happening here. By Proposition~\ref{prop:conc}, $N_{d^*-1}=|\Gamma_{d^*-1}|$ concentrates around $(nq)^{d^*-1}=\frac{\log(1/\lambda_n)}{q}$. On the other hand, by Theorem~\ref{thm:main1-tight}, $N_{d^*+1}=|\Gamma_{d^*+1}|$ concentrates around $\lambda_nn$. Thus, we can compare the two sets by comparing $\frac{\log(1/\lambda_n)}{\lambda_n}$ with $nq=\alpha_n\log n$. But since $\lambda_n$ does not converge to a particular value as discussed in Section \ref{sec:dv}, the result of the comparison is of course not consistent. Hence, just as we fixed a subsequence so that we could say $\lambda_n\to\lambda$, we will consider a finer subsequence such that
\[
    \eta_n:=\frac{\log(1/\lambda_n)}{\lambda_nnq}=\frac{\log(1/\lambda_n)}{\alpha_n\lambda_n\log n}\to\eta\in[0,\infty]\,.
\]
For instance, $\eta=0$ would mean $N_{d^*-1}\ll N_{d^*+1}$ so most of the edges in $\mathsf{SPD}_G$ connect $\Gamma_{d^*}$ and $\Gamma_{d^*+1}$.

To compute the overlap, we need to compute the proportions of $N_{d^*-1}^I$ and $N_{d^*+1}^I$ relative to $N_{d^*-1}^{(i)}$ and $N_{d^*+1}^{(i)}$, respectively. The former is well understood through Proposition~\ref{prop:rhoconc}. For the latter, we go back to the proof of Proposition~\ref{prop:rhoconc} where we were constructing $\Gamma_d^{(i)}$ and $\Gamma_d^I$. At step $d=d^*$, one can realize those vertices $v\in V_{d^*}$ that are not connected to $\Gamma_{d^*-1}^{(1)}$ or $\Gamma_{d^*-1}^{(2)}$ comprise $\Gamma_{d^*+1}^I$. Thus, our probability of interest will be
\[
    \Pr(v\nsim_1\Gamma_{d^*-1}^{(1)}\wedge v\nsim_2\Gamma_{d^*-1}^{(2)})
\]
which can be calculated through \eqref{eqn:hard}. From the proof of Proposition~\ref{prop:rhoconc}, one may expect this to be around $\lambda^{2-\rho_n^{d^*}}$. Since $|N_{d^*+1}^{(i)}|\approx\lambda n$ the proportion is roughly $\lambda^{1-\gamma}$. If $\gamma=1$ and $\lambda=0$, this is not well defined. Hence, we need to consider another proxy sequence
\[
    \xi_n:=\lambda_n^{1-\gamma_n}\to\xi\in[0,1]
\]
converges. The overlap of the shortest path DAGs will be described in terms of these quantities.

We formalize our discussion to the following theorem.
\begin{theorem}\label{thm:main3}
    Suppose that:
    \begin{itemize}
        \item $\lambda_n\to\lambda\in[0,1]$,
        \item $\gamma_n\to\gamma\in[0,1]$,
        \item $\frac{\log(1/\lambda_n)}{\alpha_n\lambda_n\log n}\to\eta\in[0,\infty]$,
        \item $\lambda_n^{1-\gamma_n}\to\xi\in[0,1]$, and
        \item $\rho_n\to\rho\in[0,1]$.
    \end{itemize}
    Then we have
    \[
        S_n\pto\begin{dcases*}
            \gamma & \text{if $\lambda=1$ or $\eta=\infty$,}\\
            \rho\cdot\frac{(1-(2-\xi)\lambda)\xi+\eta\gamma}{1-\lambda+\eta} & \text{otherwise.}
        \end{dcases*}
    \]
\end{theorem}

\begin{proof}
    We split the proof into several cases, each of which will require different concentration results. The fully rigorous proofs for these concentrations are a bit tedious, so we only detail the proof for the first case and afterwards only present informally. Let $q_I=q(q(1-\rho_n)+\rho_n)$ be the edge probability of $G^{(1)}\cap G^{(2)}$.

    \paragraph{Case 1: $\lambda\in(0,1)$.} This is the easiest case, which intuitively follows from our earlier results. Combining Lemma~\ref{lem:chernoff} and Theorem~\ref{thm:main1}, it is easy to see that
    \begin{equation}\label{eqn:spdconc}
        \frac{|\mathsf{SPD}_{G^{(i)}}|}{\lambda(1-\lambda)n^2q}\pto1\,.
    \end{equation}
    Now we need to obtain a concentration result for $\mathsf{SPD}_{G^{(1)}}\cap\mathsf{SPD}_{G^{(2)}}$. It is intuitively clear that most of the edges are between $N_{d^*}^{I}$ and $N_{d^*+1}^{I}$. To prove it, we follow the procedure of the proof of Proposition~\ref{prop:rhoconc}, and after each step $d$, obtain a concentration bound on the number of edges between $N_{d-1}^I$ and $N_d^I$. Conditioned on $\Gamma_1^{(i)},\cdots,\Gamma_{d-1}^{(i)}$, the number of edges in $\mathsf{SPD}_{G^{(1)}}\cap\mathsf{SPD}_{G^{(2)}}$ that connects $N_{d-1}^I$ and $N_d^I$ follows
    \[
        \Binom\left(N_{d-1}^I\cdot\left|V\setminus\Gamma_{\leq d-1}^{(1)}\setminus\Gamma_{\leq d-1}^{(2)}\right| ,q_I\right)\,.
    \]
    We may assume the conclusion of Proposition~\ref{prop:rhoconc} holds. Then for $d\leq d^*$, we have that $N_{d-1}^I\leq 2(nq)^{d-1}$ and $|V\setminus\Gamma_{\leq d-1}^{(1)}\setminus\Gamma_{\leq d-1}^{(2)}|\leq n$. Thus, it easily follows from the Chernoff bound that the number of edges connecting $N_{d-1}^I$ and $N_d^I$ is bounded by
    \[
        4(nq)^{d-1}\cdot n\cdot q_I\,.
    \]
    with probability $1-o(1)$, simultaneously for all $d\leq d^*$. Hence, the number of edges connecting $N_{d-1}^I$ and $N_d^I$ for all $d\leq d^*$ is bounded by
    \[
        \sum_{d=1}^{d^*}4(nq)^{d-1}\cdot n\cdot q_I\leq 8(nq)^{d^*-1}\cdot n\cdot q_I\leq\frac{8\log(1/\lambda_n)n^2q_I}{nq}\,.
    \]
    This is negligible compared to $|\mathsf{SPD}_{G^{(i)}}|$ as written in \eqref{eqn:spdconc}. For $d\geq d^*+2$, we observe that the number of edges conditioned on the previous steps is stochastically dominated by
    \[
        \Binom\left(N_{d^*+1}^I\left|V\setminus\Gamma_{\leq d^*+1}^{(1)}\setminus\Gamma_{\leq d^*+1}^{(2)}\right|+\binom{\left|V\setminus\Gamma_{\leq d^*+1}^{(1)}\setminus\Gamma_{\leq d^*+1}^{(2)}\right|}{2}, q_I\right)\,.
    \]
    Since $|V\setminus\Gamma_{\leq d^*+1}^{(1)}\setminus\Gamma_{\leq d^*+1}^{(2)}|=o(n)$, it is clear that this is also negligible compared to $|\mathsf{SPD}_{G^{(i)}}|$. Thus, we only need to consider the edges between $N_{d^*}^{I}$ and $N_{d^*+1}^{I}$, whose cardinality follows
    \[
        \Binom(N_{d^*}^I|V\setminus\Gamma_{\leq d^*}^{(1)}\setminus\Gamma_{\leq d^*}^{(2)}|,q_I)
    \]
    conditioned on $\Gamma_1^{(i)},\cdots,\Gamma_{d^*}^{(i)}$. Since we know that $N_{d^*}^I$ concentrates around
    \[
        (1-2\lambda+\lambda^{2-\gamma})n=(1-(2-\xi)\lambda)n
    \]
    and $|V\setminus\Gamma_{\leq d^*}^{(1)}\setminus\Gamma_{\leq d^*}^{(2)}|$ concentrates around
    \[
        \lambda^{2-\gamma}n=\lambda\xi n
    \]
    we see that
    \[
        \frac{|\mathsf{SPD}_{G^{(1)}}\cap\mathsf{SPD}_{G^{(2)}}|}{\lambda(1-(2-\xi)\lambda)\xi n^2q_I}\pto1\,.
    \]
    This completes the proof when $\lambda\in(0,1)$.

    \paragraph{Case 2: $\lambda=1$.} From Lemma~\ref{lem:dstar}, we know that $N_{d^*}^{(i)}$ concentrates around $(nq)^{d^*}=\log(1/\lambda_n) n$. Conditioned on $N_1^{(i)},\cdots,N_{d^*+1}^{(i)}$, the number of vertices with distance at least $d^*+2$ follows
    \[
        \Binom\left(n-\sum_{d=0}^{d^*}N_d, (1-q)^{N_{d^*}}\right)\,.
    \]
    Note that $(1-q)^{N_{d^*}}\leq e^{-qN_{d^*}}\leq e^{-\log(1/\lambda_n) nq/2}=e^{-\alpha_n\log(1/\lambda_n)\log n/2}$ with high probability. From \eqref{eqn:lambda-sandwich} we have $\log(1/\lambda_n)\geq1/(\log\log n)^2$, so for large enough $n$
    \[
        (1-q)^{N_{d^*}}\leq e^{-\alpha_n\log(1/\lambda_n)\log n/2}\leq e^{-\log n/(\log\log n)^2}\ll \frac{1}{(\log\log n)^2}\leq\log(1/\lambda_n)
    \]
    with high probability. Hence, it follows that $N_{d^*}^{(i)}\gg N_{d^*+2}^{(i)}$ in probability limit, which means we only need to consider the edges connecting $N_{d^*}^{(i)}$ and $N_{d^*+1}^{(i)}$. Applying Lemma~\ref{lem:rhodstar} gives $N_{d^*}^I/N_{d^*}^{(i)}\pto \gamma$, which implies the desired result.

    \paragraph{Case 3: $\lambda=0$ and $\eta=\infty$.} This is similar to Case 2 since we can prove that $N_{d^*-1}^{(i)}\gg N_{d^*+1}^{(i)}$. Indeed, $N_{d^*-1}^{(i)}$ concentrates around $(nq)^{d^*-1}=\log(1/\lambda_n)/(nq)$ and $N_{d^*+1}$ is bounded by
    \[
        n(1-q)^{N_{d^*-1}}\leq ne^{-qN_{d^*-1}}=\lambda_nn
    \]
    with high probability. Thus, $\eta=\infty$ implies $N_{d^*-1}^{(i)}\gg N_{d^*+1}^{(i)}$ and the concentration follows.

    \paragraph{Case 4: $\lambda=0$ and $\eta<\infty$.} This case is the trickiest, since $N_{d^*+1}^{(i)}$ is not negligible anymore. Note that $\eta<\infty$ implies
    \begin{equation}\label{eqn:betabound}
        \lambda_n\geq\frac{1}{\log n}
    \end{equation}
    for large $n$. This case is the trickiest, since $N_{d^*+1}^{(i)}$ is not negligible anymore. Our goal is obtain concentration bounds for the quantities
    \begin{align}
        N_{\geq d^*+1}^{(i)}&:=\left|V\setminus\Gamma_{\leq d^*}^{(i)}\right|\,,\label{align:g1}\\
        N_{\geq d^*+2}^{(i)}&:=\left|V\setminus\Gamma_{\leq d^*+1}^{(i)}\right|\,,\label{align:g2}\\
        N_{d^*+1}^I&:=\left|\Gamma_{d^*+1}^{(1)}\cap\Gamma_{d^*+1}^{(2)}\right|\,.\label{align:g3}
    \end{align}
    We note that (the second part of) Theorem~\ref{thm:main1-tight} is not strong enough to handle our needs here. We must leverage \eqref{eqn:betabound} and extend the proof of Theorem~\ref{thm:main1-tight} to obtain a tighter result.

    Let $\epsilon>0$ be an arbitrary constant. Conditioned on $\Gamma_1^{(i)},\cdots,\Gamma_{d^*-1}^{(i)}$, $N_{\geq d^*+1}^{(i)}$ follows
    \[
        \Binom\left(n-\sum_{d=0}^{d^*-1}N_d^{(i)}, (1-q)^{N_{d^*-1}^{(i)}}\right)\,.
    \]
    We handle $(1-q)^{N_{d^*-1}^{(i)}}$, then the concentration for $N_{\geq d^*+1}^{(i)}$ will readily follow. For the upper bound, we apply Proposition~\ref{prop:conc} with $\delta_1=\frac{\log(1+\epsilon)}{\log\log n}$, which gives
    \[
        \begin{split}
            (1-q)^{N_{d^*-1}^{(i)}}&\leq e^{-qN_{d^*-1}^{(i)}}\\
            &\leq e^{-(1-\delta_1)q(nq)^{d^*-1}}\\
            &=\lambda_n\cdot e^{\delta_1\log(1/\lambda_n)}\\
            &\leq \lambda_n\cdot e^{\delta_1\log\log n}\\
            &\leq (1+\epsilon)\lambda_n
        \end{split}
    \]
    with probability $1-o(1)$. To get a lower bound, we use
    \[
        q=\frac{\alpha_n\log n}{n}\leq \frac{2\log\left(\frac{2}{2-\epsilon}\right)}{\alpha_n\log n+\log\left(\frac{2}{2-\epsilon}\right)}
    \]
    to apply Lemma~\ref{lem:alg1} and Lemma~\ref{lem:alg2}, which implies
    \begin{equation}\label{eqn:qepsbound}
        \begin{split}
            (1-q)^{N_{d^*-1}^{(i)}}&\geq\exp\left(-\left(1+\frac{\log\left(\frac{2}{2-\epsilon}\right)}{\alpha_n\log n}\right)qN_{d^*-1}^{(i)}\right)\\
            &\geq e^{-qN_{d*-1}^{(i)}}\cdot\exp\left(-\left(\frac{\log\left(\frac{2}{2-\epsilon}\right)}{\alpha_n\log n}\right)nq\right)\\
            &\geq\left(1-\frac{\epsilon}{2}\right)e^{-qN_{d^*-1}^{(i)}}
        \end{split}
    \end{equation}
    Now in the same manner as the derivation of the upper bound, we have
    \[
        e^{-qN_{d^*-1}^{(i)}}\geq\left(1-\frac{\epsilon}{2}\right)\lambda_n
    \]
    with probability $1-o(1)$. Thus, we obtain
    \[
        (1-q)^{N_{d^*-1}^{(i)}}\geq(1-\epsilon)\lambda_n\,.
    \]
    Therefore, we have proved that
    \[
        \frac{(1-q)^{N_{d^*-1}^{(i)}}}{\lambda_n}\pto1
    \]
    which in turn implies
    \[
        \frac{N_{\geq d^*+1}^{(i)}}{\lambda_nn}\pto1\,.
    \]

    Next, we bound (\ref{align:g2}). In fact, we only need a crude upper bound, which can be achieved by
    \[
        (1-q)^{N_{d^*}^{(i)}}\leq e^{-qN_{d^*}^{(i)}}\leq e^{-\frac{1}{2}nq}\leq n^{-\frac{\alpha_n}{2}}\,.
    \]
    Since $\lambda\geq\frac{1}{\log n}$, this shows that $N_{\geq d^*+1}^{(i)}\gg N_{\geq d^*+2}^{(i)}$, so we can safely neglect $N_{\geq d^*+2}^{(i)}$. In particular, we have proved that
    \begin{equation}\label{eqn:ndconc1}
        \frac{N_{d^*+1}^{(i)}}{\lambda_nn}\pto1\,.
    \end{equation}

    Finally, we handle (\ref{align:g3}). 
    As noted in the discussion preceding the statement of Theorem~\ref{thm:main3}, we actually want to prove concentration bounds for \eqref{eqn:hard} when $d=d^*$, which we temporarily denote by $\bar{q}$. Following the argument in the proof of Proposition~\ref{prop:rhoconc} that led to the bounds \eqref{eqn:qrbound} and \eqref{eqn:qlbound}, we have
    \begin{equation}\label{eqn:qbu}
        \bar{q}\leq(1-q)^{\tilde{N}_{d^*-1}^{(1)}+\tilde{N}_{d^*-1}^{(2)}-N_{d^*-1}^I}(1-q(1-\rho_n))^{N_{d^*-1}^I}
    \end{equation}
    and
    \begin{equation}\label{eqn:qbl}
        \bar{q}\geq(1-q)^{N_{d^*-1}^{(1)}+N_{d^*-1}^{(2)}-N_{d^*-1}^I}(1-q(1-\rho_n))^{N_{d^*-1}^I}\,.
    \end{equation}
    Starting from \eqref{eqn:qbu},
    \[
        \begin{split}
            \bar{q}&\leq(1-q)^{\tilde{N}_{d^*-1}^{(1)}+\tilde{N}_{d^*-1}^{(2)}-N_{d^*-1}^I}(1-q(1-\rho_n))^{N_{d^*-1}^I}\\
            &\leq e^{-q(\tilde{N}_{d^*-1}^{(1)}+\tilde{N}_{d^*-1}^{(2)}-\rho_n N_{d^*-1}^I)}\,.
        \end{split}
    \]
    Since Proposition~\ref{prop:conc} yields
    \[
        \tilde{N}_{d^*-1}^{(i)}\geq N_{d^*-1}^{(i)}-N_{\leq d^*-2}^{(i)}\geq N_{d^*-1}^{(i)}-2(nq)^{d^*-2}
    \]
    with high probability, we obtain
    \[
        \begin{split}
            \bar{q}&\leq e^{-q(\tilde{N}_{d^*-1}^{(1)}+\tilde{N}_{d^*-1}^{(2)}-\rho_n N_{d^*-1}^I)}\\
            &\leq e^{-q(N_{d^*-1}^{(1)}+N_{d^*-1}^{(2)}-\rho_n N_{d^*-1}^I)}\cdot e^{4q(nq)^{d^*-2}}\\
            &\leq (1+\epsilon)e^{-q(N_{d^*-1}^{(1)}+N_{d^*-1}^{(2)}-\rho_n N_{d^*-1}^I)}
        \end{split}
    \]
    for large $n$. On the other hand, applying \eqref{eqn:qepsbound} to \eqref{eqn:qbl} gives
    \[
        \bar{q}\geq(1-\epsilon)e^{-q(N_{d^*-1}^{(1)}+N_{d^*-1}^{(2)}-\rho_n N_{d^*-1}^I)}\,.
    \]
    Hence, it suffices to derive concentration bounds for $e^{-q(N_{d^*-1}^{(1)}+N_{d^*-1}^{(2)}-\rho_n N_{d^*-1}^I)}$. This can be done similarly as bounding $(1-q)^{N_{d^*-1}^{(i)}}$, by combining together Proposition~\ref{prop:rhoconc}, which we omit the details. Namely, one can show that
    \[
        (1-2\epsilon)\lambda_n^{2-\rho_n^{d^*}}\leq\bar{q}\leq (1+2\epsilon)\lambda_n^{2-\rho_n^{d^*}}\,.
    \]
    Recall that $\rho_n^{d^*}=\gamma_n$. Due to the fact that $N_{\geq d^*+2}^{(i)}$ are negligible, we arrive at
    \begin{equation}\label{eqn:ndconc2}
        \frac{N_{d^*+1}^I}{\lambda_n^{2-\gamma_n}\cdot n}\pto1\,.
    \end{equation}

    With these results, we can prove concentrations for the shortest path DAGs. Combining (\ref{eqn:ndconc1}) and (\ref{eqn:ndconc2}) yields $N_{d^*+1}^{I}/N_{d^*+1}^{(i)}\pto\xi$, the same argument with Case 2 leads to $N_{d^*}^I/N_{d^*}^{(i)}\pto\gamma$, and $N_{d^*}^{(i)}/N_{d^*+1}^{(i)}\to\eta$ by the assumption; altogether completes the proof.
    
\end{proof}

\subsection{The overlap of random shortest path trees}
\label{sec:overlap-trees}
As noted in Section \ref{sec:dagtree}, the shortest path DAG of a graph naturally encodes the information about the distribution of a uniformly random shortest path tree. The goal of this section is to analyze
\[
    R_n:=\frac{1}{n}|T_n^{(1)}\cap T_n^{(1)}|\,.
\]
As briefly mentioned in the introduction, $R_n$ is a still random variable after conditioned on the graphs $G^{(1)}$ and $G^{(2)}$, and the distribution depends on how we sample a pair of shortest path trees.

We will focus on the setting of Proposition~\ref{prop:sptree-conc}, where we sample a pair of shortest path trees $(T_1,T_2)$ from the optimal coupling with respect to the Hamming metric. Another choice is to sample each tree independently --- the analysis is much simpler, and will result in the usual notion of disorder chaos.

\subsubsection{The overlap under optimal coupling}\label{sec:overlap-opt}

By Proposition~\ref{prop:sptree-conc}, it suffices to prove that
\begin{equation*}
    \tilde{R}_n:=\mathbb{E}[R_n\mid G_n^{(1)},G_n^{(2)}]=\frac{1}{n}\sum_{v\in \tilde{V}\setminus\{1\}}\frac{|\mathsf{par}_{G_n^{(1)}}(v)\cap\mathsf{par}_{G_n^{(2)}}(v)|}{\max(|\mathsf{par}_{G_n^{(1)}}(v)|,|\mathsf{par}_{G_n^{(2)}}(v)|)}
\end{equation*}
concentrates, where $\tilde{V}$ is the set of vertices reachable from $1$. An equivalent way of formulating \eqref{eqn:parconc} is to consider a weighted version of shortest path DAG, where to each incoming edge $e$ to $v$ we give weight $1/\operatorname{indeg}(v)$, and compute the $\ell_1$ distance between two weighted shortest path DAGs.

\begin{theorem}\label{thm:main4}
    Suppose that $\lambda_n\to\lambda$ and $\gamma_n\to\gamma$ with $\lambda,\gamma\in[0,1]$, and $\rho_n\to\rho$. Then as $n\to\infty$, we have that
    \[
        \tilde{R}_n\pto\begin{dcases*}
            \frac{1-2\lambda+\lambda^{2-\gamma}}{1-\lambda}\cdot\lambda^{2-\gamma}\rho+g((1-\gamma)\log(1/\lambda),\gamma\log(1/\lambda))(1-\lambda^{\gamma}) & \text{if $0<\lambda<1$,}\\
            \gamma & \text{otherwise}\\
        \end{dcases*}
    \]
    where $g(a,b)$ is defined by
    \[
        g(a,b)=\begin{dcases*}
            \mathbb{E}\left[\frac{Z}{\max(X_1,X_2)+Z}\,\middle|\,Z>0\right]&\text{if $a,b>0$,}\\
            1&\text{if $a=0$,}\\
            0&\text{if $b=0$}
        \end{dcases*}
    \]
    for independent $X_1,X_2\sim\Pois(a)$ and $Z\sim\Pois(b)$.
\end{theorem}



Now we prove Theorem~\ref{thm:main4}. In contrast to Theorem~\ref{thm:main3} where we had to deal with minority vertices, $\tilde{R}_n$ is an average over all \emph{vertices}, so by Theorem~\ref{thm:main1} we can safely disregard all vertices in $\Gamma_{\leq d^*-1}$ and $\Gamma_{\geq d^*+1}$. The proof of Theorem~\ref{thm:main4} is slightly more sophisticated since we have two graphs, so we first state two crucial observations as a separate lemma.

\begin{lemma}\label{lem:decomp}
    For each $d\geq1$ and $v\in V$, we define random variables
    \[
        Y_{d,v}=\begin{dcases*}
            \frac{|\mathsf{par}_{G^{(1)}}(v)\cap\mathsf{par}_{G^{(2)}}(v)|}{\max(|\mathsf{par}_{G^{(1)}}(v)|, |\mathsf{par}_{G^{(2)}}(v)|)}&\text{if $v\in\Gamma_d^I$ and $v\sim\Gamma_{d-1}^{I}$,}\\
            0&\text{otherwise}
        \end{dcases*}
    \]
    where $\sim$ means connectivity in both graphs. Then the following hold.
    \begin{enumerate}[label=(\alph*)]
        \item Conditioned on $\Gamma_{1}^{(i)},\cdots,\Gamma_{d-1}^{(i)}$, random variables in $\{Y_{d,v}\}_{v\in V\setminus\Gamma_{\leq d-1}^{(1)}\setminus\Gamma_{\leq d-1}^{(2)}}$ are independent and identically distributed.
        \item Let $\bar{Y}=\frac{1}{n}\sum_{v\in V}(Y_{d^*,v}+Y_{d^*+1,v})$. Then for any $\epsilon>0$ we have $\Pr(|\tilde{R}_n-\bar{Y}|>\epsilon)\to0$ as $n\to\infty$.
    \end{enumerate}
\end{lemma}
\begin{proof}
    The first part follows by observing that $Y_{d,v}$ are nonzero if and only if $v\sim\Gamma_{d-1}^I$, so $|\mathsf{par}_{G^{(i)}}(v)|$ are functions of the edge connectivity between $v$ and $\Gamma_{d-1}^{(i)}$.
    
    To prove the second part, we first note that $\bar{Y}\leq\tilde{R}_n$; intuitively, $(n-1)\bar{Y}$ is the summation as in \eqref{eqn:parconc}, but only over vertices contained in $\Gamma_{d^*}^{I}$ and $\Gamma_{d^*+1}^I$. It remains to prove that the other vertices are either contained in a small set, or the overlap of its parents is negligible. Indeed, by Proposition~\ref{prop:conc}, the sets $\Gamma_{d}^{(i)}$ where $d<d^*$ or $d>d^*+1$ have cardinality $o(n)$, which does not contribute to the sum. Thus, we only need to deal with the case $\Gamma_{d^*}^{(i)}\setminus\Gamma_{d^*}^I$, and more specifically $0<\lambda<1$.
    
    Without loss of generality, we consider $\Gamma_{d^*}^{(1)}\cap\Gamma_{d^*+1}^{(2)}$, which occupies most of $\Gamma_{d^*}^{(1)}\setminus\Gamma_{d^*}^{(2)}$ by Theorem~\ref{thm:main1}. We show that most vertices here are connected to many vertices in $\Gamma_{d^*}^{(2)}$ but few vertices in $\Gamma_{d^*-1}^{(1)}$. Conditioned on $\Gamma_{\leq d^*}^{(2)}$, each vertex $v\in V\setminus\Gamma_{\leq d^*}^{(2)}$ satisfies
    \[
        |\{w\in\Gamma_{d^*}^{(2)}:v\sim_2 w\}|\sim\Binom(N_{d^*}^{(2)},q)\,.
    \]
    Let $E_v$ be the event such that $v$ is connected to at most $\frac{1}{2}(1-\lambda)nq$ vertices in $\Gamma_{\leq d^*-1}^{(2)}$. Since $N_{d^*}$ concentrates around $(1-\lambda)n$ by Proposition~\ref{prop:conc}, Lemma~\ref{lem:chernoff} gives
    \[
        \Pr(E_v\mid\Gamma_{\leq d^*-1}^{(2)})\leq o(1)\,.
    \]
    Again, by Lemma~\ref{lem:chernoff}, this implies that at most $o(n)$ vertices in $V\setminus\Gamma_{\leq d^*-1}^{(2)}$ are connected to at most $\frac{1}{2}(1-\lambda)nq$ vertices in $\Gamma_{\leq d^*}^{(2)}$ with high probability.

    Next, conditioned on $\Gamma_{\leq d^*-1}^{(1)}$, each vertex $v\in V\setminus\Gamma_{\leq d^*-1}^{(1)}$ satisfies
    \[
        |\{w\in\Gamma_{d^*-1}^{(1)}:v\sim_2 w\}|\sim\Binom(N_{d^*-1}^{(1)},q)\,.
    \]
    Since $qN_{d^*-1}^{(1)}$ concentrates around $\log(1/\lambda)$ by Proposition~\ref{prop:conc}, we have
    \[
        \Binom(N_{d^*-1}^{(1)},q)\dto\Pois(\log(1/\lambda))
    \]
    so for any $\epsilon>0$ there exists $K$ such that at most $\epsilon n$ vertices in $V\setminus\Gamma_{\leq d^*-1}^{(1)}$ are connected to at least $K$ vertices in $\Gamma_{\leq d^*-1}^{(1)}$ with probability $1-o(1)$.

    Therefore, all but at most $o(n)$ vertices in $\Gamma_{d^*}^{(1)}\cap\Gamma_{d^*+1}^{(2)}$ have at least $\frac{1}{2}(1-\lambda)nq=\Omega(\log n)$ connections to $\Gamma_{d^*}^{(2)}$ and at most $K$ connections to $\Gamma_{d^*}^{(1)}$, so their contributions to \eqref{eqn:parconc} are $o(1)$.    
\end{proof}

\begin{lemma}\label{lem:decomp2}
    Conditioned on $\Gamma_1^{(i)},\cdots,\Gamma_{d-1}^{(i)}$ and $v\in V\setminus\Gamma_{\leq d-1}^{(1)}\setminus\Gamma_{\leq d-1}^{(2)}$, define independent random variables $X_d^{(i)}$, $Z_d$, and $W_d^{(i)}$ such that
    \[
        X_d^{(i)}\sim\Binom(N_{d-1}^{(i)}-N_{d-1}^I,q)
    \]
    and
    \[
        (Z_d, W_d^{(1)}, W_d^{(2)}, -)\sim\Multinom(N_{d-1}^I; q_I, q-q_I, q-q_I, -)
    \]
    where $q_I=q(q(1-\rho_n)+\rho_n)$ is the edge probability of $G^{(1)}\cap G^{(2)}$. Then we have that
    \[
        \begin{bmatrix}
            |\mathsf{par}_{G^{(1)}}(v)|\\
            |\mathsf{par}_{G^{(2)}}(v)|\\
            |\mathsf{par}_{G^{(1)}}(v)\cap \mathsf{par}_{G^{(2)}}(v)|
        \end{bmatrix}\mathbf{1}_{v\sim\Gamma_{d-1}^{I}}\deq\begin{bmatrix}
            X_d^{(1)}+Z_d+W_d^{(1)}\\
            X_d^{(2)}+Z_d+W_d^{(2)}\\
            Z_d
        \end{bmatrix}\mathbf{1}_{Z_d>0}\,.
    \]
    In particular, conditioned on $\Gamma_1^{(i)},\cdots,\Gamma_{d-1}^{(i)}$, we have
    \[
        Y_{d,v}\deq\begin{dcases*}
            \frac{Z_d}{\max(X_d^{(1)}+W_d^{(1)}, X_d^{(2)}+W_d^{(2)})+Z_d}\cdot\mathbf{1}_{Z_d>0}&\text{if $v\in V\setminus\Gamma_{\leq d-1}^{(1)}\setminus\Gamma_{\leq d-1}^{(2)}$,}\\
            0&\text{otherwise.}
        \end{dcases*}
    \]
\end{lemma}
\begin{proof}
    This is clear once we realize that $X_d^{(i)}$ represents edges in $G^{(i)}$ connected to $\Gamma_{d-1}^{(i)}\setminus\Gamma_{d-1}^I$, $Z_d$ represents edges in $G^{(1)}\cap G^{(2)}$ connected to $\Gamma_{d-1}^I$, and $W_d^{(i)}$ represents edges only in $G^{(i)}$ connected to $\Gamma_{d-1}^I$. These are precisely the parents of $v$ in the corresponding shortest path DAGs if $v\sim\Gamma_{d-1}^I$, i.e. $Z_d>0$.
\end{proof}

Combining the above lemmata, we obtain the following lemma which is extremely useful in analyzing $\tilde{R}_n$.

\begin{lemma}\label{lem:decomp3}
    Consider a random variable $\Pi_d$ which is a function of $\Gamma_1^{(i)},\cdots,\Gamma_{d-1}^{(i)}$ defined by
    \[
        \Pi_d:=\mathbb{E}\left[\frac{Z_{d}}{\max(X_{d}^{(1)}+W_{d}^{(1)}, X_{d}^{(2)}+W_{d}^{(2)})+Z_d}\cdot\mathbf{1}_{Z_{d}>0}\,\middle|\,\Gamma_1^{(i)},\cdots,\Gamma_{d-1}^{(i)}\right]\,.
    \]
    Then the random variable
    \[
        \bar{\Pi}:=\Pi_{d^*}+\lambda^{2-\gamma}\Pi_{d^*+1}
    \]
    satisfies $|\tilde{R}_n-\bar{\Pi}|\pto0$.
\end{lemma}
\begin{proof}
    We build on Lemma~\ref{lem:decomp2}. For convenience, we denote $V_d:=V\setminus\Gamma_{\leq d-1}^{(1)}\setminus\Gamma_{\leq d-1}^{(2)}$ and define
    \[
        \bar{Y}_{d}:=\frac{1}{|V_{d}|}\sum_{v\in V_{d}}Y_{d,v}\,.
    \]
    Note that $\mathbb{E}[\bar{Y}_{d}\mid\Gamma_1^{(i)},\cdots,\Gamma_{d-1}^{(i)}]=\Pi_{d}$. By Lemma~\ref{lem:decomp}, $\bar{Y}_{d}$ conditioned on $\Gamma_1^{(i)},\cdots,\Gamma_{d-1}^{(i)}$ is a sample mean of i.i.d. random variables. Since each $Y_{d,v}$ can have value in $[0,1]$, we have
    \[
        \Var(\bar{Y}_{d}\mid\Gamma_1^{(i)},\cdots,\Gamma_{d-1}^{(i)})=\frac{1}{|V_{d}|}\Var(Y_{d,v}\mid\Gamma_1^{(i)},\cdots,\Gamma_{d-1}^{(i)})\leq\frac{1}{4|V_{d}|}\,.
    \]
    Hence, Chebyshev's inequality gives
    \begin{equation}\label{eqn:cheby}
        \Pr(|\bar{Y}_{d}-\Pi_{d}|\geq\delta\mid\Gamma_1^{(i)},\cdots,\Gamma_{d-1}^{(i)})\leq\frac{1}{4\delta^2|V_{d}|}
    \end{equation}

    Now we write
    \[
        \begin{split}
            \bar{Y}&=\frac{1}{n}\sum_{v\in V}(Y_{d^*,v}+Y_{d^*+1,v})\\
            &=\frac{1}{n}\sum_{v\in V}Y_{d^*,v} + \frac{1}{n}\sum_{v\in V}Y_{d^*+1,v}\\
            &=\frac{|V_{d^*}|}{n}\bar{Y}_{d^*}+\frac{|V_{d^*+1}|}{n}\bar{Y}_{d^*+1}\,.
        \end{split}
    \]
    By Theorems~\ref{thm:main1} and \ref{thm:main2}, we have $|V_{d^*}|/n\pto1$ and $|V_{d^*+1}|/n\pto\lambda^{2-\gamma}$. For $\bar{Y}_{d^*}$, \eqref{eqn:cheby} easily gives $|\bar{Y}_d-\Pi_d|\pto0$. For $\bar{Y}_{d^*+1}$, \eqref{eqn:cheby} gives a similar result if $\lambda>0$; otherwise, $|V_{d^*+1}|/n\pto0$ so the result trivially holds.
    
\end{proof}

\begin{proof}[Proof of Theorem~\ref{thm:main4}]
    By Lemma~\ref{lem:decomp3}, it is enough to prove that $\bar{\Pi}$ has the given limit in probability.

    First, assume $\lambda=1$. Then by Theorem~\ref{thm:main1} (or in the language of Lemma~\ref{lem:decomp3}, $\Pr(Z_{d^*}>0)\to0$) it suffices to prove that $\Pi_{d^*+1}$ concentrates around the given limit. Since $q(nq)^{d^*}=\log(1/\lambda_n)nq=\Omega(\log n/(\log\log n)^2)$, Lemma~\ref{lem:chernoff} and Proposition~\ref{prop:conc} give
    \[
        \frac{X_{d^*+1}^{(i)}+W_{d^*+1}^{(i)}+Z_{d^*+1}}{q(nq)^{d^*}}\pto1\,.
    \]
    If $\gamma=0$, then by Proposition~\ref{prop:rhoconc} we have $N_{d^*}^{(i)}/(nq)^{d^*}\to0$ so $Z_{d^*+1}/X_{d^*+1}^{(i)}\to0$, thus the expectation tends to zero and the result is trivial. Assuming $\gamma>0$, we get
    \[
        \frac{Z_{d^*+1}}{q(nq)^{d^*}\gamma}\pto1\,.
    \]
    Also, $\gamma>0$ implies $\rho_n\to1$, so $W_{d^*+1}^{(i)}/Z_{d^*+1}\to0$. Hence, the expectation converges in probability to $\gamma$ as desired. The proof for $\lambda=0$ is almost the same by replacing $\Pi_{d^*+1}$ with $\Pi_{d^*}$.

    Now assume $0<\lambda<1$. Again, we appeal to Lemma~\ref{lem:decomp3} and establish probability limits of $\Pi_{d^*}$ and $\Pi_{d^*+1}$. We first prove that
    \begin{equation}\label{eqn:pidstar}
        \Pi_{d^*}\pto g((1-\gamma)\log(1/\lambda),\gamma\log(1/\lambda))(1-\lambda^{\gamma})\,.
    \end{equation}
    We may assume $\gamma>0$, otherwise $Z_{d^*}\to0$ by Proposition~\ref{prop:rhoconc} so the result is trivial. Since $q(nq)^{d^*-1}\pto\log(1/\lambda)$, Propositions \ref{prop:conc} and \ref{prop:rhoconc} yield $qN_{d^*-1}\pto\log(1/\lambda)$ and $qN_{d^*-1}\pto\gamma\log(1/\lambda)$ so
    \[
        X_{d^*}^{(i)}\dto\Pois((1-\gamma)\log(1/\lambda))\,.
    \]
    and
    \begin{equation}\label{eqn:zpois}
        Z_{d^*}\dto\Pois(\gamma\log(1/\lambda))\,.
    \end{equation}
    Since $\gamma>0$ implies $\rho_n=1-o(1)$, we have $(q-q_I)/q=o(1)$ so $W_{d^*}^{(i)}\pto0$. Note that we can write $\Pi_{d^*}$ as
    \[
        \mathbb{E}\left[\frac{Z_{d^*}}{\max(X_{d^*}^{(1)}+W_{d^*}^{(1)}, X_{d^*}^{(2)}+W_{d^*}^{(2)})+Z_{d^*}}\,\middle|\,Z_{d^*}>0,\Gamma_1^{(i)},\cdots,\Gamma_{d-1}^{(i)}\right]\cdot\Pr(Z_{d^*}>0\mid \Gamma_1^{(i)},\cdots,\Gamma_{d-1}^{(i)})\,.
    \]
    Since $\Pr(Z_{d^*}>0\mid \Gamma_1^{(i)},\cdots,\Gamma_{d-1}^{(i)})\to 1-\lambda^\gamma$ by \eqref{eqn:zpois}, we have proved \eqref{eqn:pidstar}.

    Next, we show
    \begin{equation}\label{eqn:pidstar1}
        \Pi_{d^*+1}\pto \frac{1-2\lambda+\lambda^{2-\gamma}}{1-\lambda}\cdot\rho\,.
    \end{equation}
    By Theorems \ref{thm:main1} and \ref{thm:main2}, we have $N_{d^*}^{(i)}/n\pto1-\lambda$ and $N_d^I/n\pto1-2\lambda+\lambda^{2-\gamma}$. This gives
    \[
        \frac{X_{d^*+1}^{(i)}+W_{d^*+1}^{(i)}+Z_{d^*+1}}{nq}\pto1-\lambda
    \]
    and
    \[
        \frac{Z_{d^*+1}}{nq_I}\pto1-2\lambda+\lambda^{2-\gamma}
    \]
    where $q_I=q(q(1-\rho_n)+\rho_n)$. This proves (\ref{eqn:pidstar1}). Combining \eqref{eqn:pidstar} and (\ref{eqn:pidstar1}) proves the desired limit in probability.
\end{proof}

\subsubsection{Aside: the overlap under independent coupling}

When shortest path trees are independently sampled, the conditional expectation of $R_n$ is given by
\begin{equation}\label{eqn:indepcoupling}
    Q_n:=\frac{1}{n-1}\sum_{v\in \tilde{V}\setminus\{1\}}\frac{|\mathsf{par}_{G^{(1)}}(v)\cap\mathsf{par}_{G^{(2)}}(v)|}{|\mathsf{par}_{G^{(1)}}(v)|\cdot|\mathsf{par}_{G^{(2)}}(v)|}\,.
\end{equation}
From this, it is obvious to see that when a vertex has an asymptotically unbounded number of parent choices, then regardless of $|\mathsf{par}_{G^{(1)}}(v)\cap\mathsf{par}_{G^{(2)}}(v)|$ the probability of choosing the same parent is $o(1)$.

Therefore, the only nontrivial contribution to \eqref{eqn:indepcoupling} comes from $\Gamma_{d^*}$ and only when $0<\lambda<1$ and $\gamma>0$ (and thus $\rho=1$), since this is the only case where $\Theta(n)$ vertices asymptotically have only a finite number of parents and a nonzero intersection. In this case, the expectation of the probability of choosing the same parent conditioned on $\Gamma_1^{(i)},\cdots,\Gamma_{d^*-1}^{(i)}$ can be described as
\[
    \mathbb{E}\left[\frac{Z_{d^*}}{(X_{d^*}^{(1)}+Z_{d^*})(X_{d^*}^{(1)}+Z_{d^*})}\,\middle|\,Z_{d^*}>0,\Gamma_1^{(i)},\cdots,\Gamma_{d^*-1}^{(i)}\right]\cdot\Pr(Z_{d^*}>0\mid\Gamma_1^{(i)},\cdots,\Gamma_{d-1}^{(i)})\,.
\]
The reader may want to compare this with the proof of \eqref{eqn:pidstar}.

This sketches the proof of the following theorem. We omit the rigorous proof since it is essentially the same with the proof of Theorem~\ref{thm:main4}.
\begin{theorem}\label{thm:indepcoupling}
    Suppose that $\lambda_n\to\lambda$, $\gamma_n\to\gamma$, and $\rho_n\to\rho$ with $\lambda,\gamma,\rho\in[0,1]$. Then as $n\to\infty$, we have that
    \[
        Q_n\pto\begin{cases}
            h((1-\gamma)\log(1/\lambda),\gamma\log(1/\lambda))(1-\lambda^\gamma) & \text{if $0<\lambda<1$ and $\gamma>0$,}\\
            0&\text{otherwise}
        \end{cases}
    \]
    where $h(a,b)$ is defined by
    \[
        h(a,b)=\begin{dcases*}
            \mathbb{E}\left[\frac{Z}{(X_1+Z)(X_2+Z)}\,\middle|\,Z>0\right]&\text{if $a,b>0$,}\\
            \mathbb{E}\left[Z^{-1}\,\middle|\,Z>0\right]&\text{if $a=0$}
        \end{dcases*}
    \]
    for independent $X_1,X_2\sim\Pois(a)$ and $Z\sim\Pois(b)$.
\end{theorem}

In particular, when $\rho_n=\rho<1$ is a fixed constant then the overlap of the shortest path trees will be asymptotically zero, implying the usual notion of disorder chaos (see, e.g., \cite{chatterjee2009disorder}).

\begin{remark}[Different notions of disorder chaos]
    This notion of disorder chaos is incomparable with Wasserstein disorder chaos, see discussion in \cite{el2022sampling}.
    In \cite{el2022sampling}, they define Wasserstein disorder chaos using $W_2$ distance instead of $W_1$, but this is equivalent for the purpose of determining if Wasserstein disorder chaos is present or not. Explicitly, this is because their underlying distance metric (normalized Hamming distance) is bounded, and for bounded random variables we can compare $L_1$ and $L_2$ norms via Jensen and H\"older's inequality.
\end{remark}

\subsection{No overlap concentration for random shortest paths}\label{sec:unstable}

In this section, we show that for fixed vertices $1$ and $2$, there is a parameter regime where the overlap of uniformly random shortest paths provably fails to concentrate. This lies in start contrast to the case of shortest path trees we studied above.
For correlated instances, we sample a pair of shortest paths from the optimal coupling, similar to shortest path trees in the previous section.

This is conceptually related to the Overlap Gap Property established in \cite{LS2024}, which proves that shortest path cannot be computed by a sufficiently stable algorithm, but not a formal consequence\footnote{As a reminder, the OGP implies that the shortest path is not ``stable'' under edge-by-edge resampling of the underlying graph. }.

\begin{theorem}\label{thm:unstable}
    There exist sequences $\alpha_n$ and $\rho_n$ such that the overlap of uniformly random shortest paths between fixed vertices $1$ and $2$ has variance bounded away from zero.
\end{theorem}

The idea to construct such parameters is to find a regime where there is non-negligible probability of having a unique path between $1$ and $2$. Then we argue that the path may either be destroyed or be preserved with non-negligible probability if we choose the right $\rho_n$. Lastly, by the absence of short cycles in a typical sparse random graph, another shortest path can't have large overlap should the original path gets destroyed.

\begin{proof}[Proof of Theorem~\ref{thm:unstable}]
    We choose $\alpha_n$ that meets the assumption of Theorem~\ref{thm:uniquepath} and $\rho_n=1-\mu/\ell_n^*$ for any constant $\mu>0$. Then conditioned on the event of Theorem~\ref{thm:uniquepath}, the number of edges along the shortest path $P$ that gets resampled converges in distribution to $\Pois(\mu)$. Thus, the probability that the path is being preserved converges to $e^{-\mu}$. From the discussion above, it is easy to see that the probability that this path is still unique (i.e. no additional edge generated between $\Gamma_d^{(2)}(1)$ and $\Gamma_d^{(2)}(2)$) in $G^{(2)}$ converges to $\lambda$, in which case the overlap is $1$.

    On the other hand, with probability converging to $1-e^{-\mu}$, the path gets destroyed. In this case, let $P'$ be any shortest path from the resampled graph $G^{(2)}$. By Theorem~\ref{thm:main1}, the length of $P'$ is either $d^*$ or $d^*+1$ with high probability, so $P\cap P'\neq\emptyset$ guarantees the existence of a cycle of length smaller than $d^*$ in the graph $G^{(1)}\cup G^{(2)}$. Since $G^{(1)}\cup G^{(2)}$ is stochastically dominated by a graph $G'\sim\mathcal{G}(n,q(2-\rho_n))$, we define a parameter
    \[
        \ell^{\#}:=\frac{\log n}{\log (nq(2-\rho_n))}=\frac{\log n}{\log(nq)+\log(2-\rho_n)}
    \]
    for $G'$. Note that
    \[
        \begin{split}
            \ell^*-\ell^{\#}&=\frac{\log n}{\log (nq)} - \frac{\log n}{\log(nq)+\log(2-\rho_n)}\\
            &=\frac{(\log n)\log(2-\rho_n)}{\log(nq)(\log (nq)+\log(2-\rho_n))}\\
            &\leq\frac{(1-\rho_n)\ell^*}{\log (nq)+\log(2-\rho_n)}\\
            &\leq\frac{\mu}{\log (nq)}
        \end{split}
    \]
    which implies $|\ell^*-\ell^\#|=o(1)$. Since $0<\lambda<1$, we actually have $|\ell^*-d^*|=o(1)$, so it follows that the closest integer to $\ell^\#$ is also $d^*$. By applying Corollary~\ref{cor:no-short-cycle} to $G'$, we have that there is no cycle of length smaller than $d^*$ with high probability in $G'$ and thus in $G^{(1)}\cup G^{(2)}$ as well. Hence, it follows that $P\cap P'=\emptyset$ for any $P'$ with probability $1-o(1)$.

    In conclusion, we have proved that, for any $\epsilon>0$, the probability that the overlap is $1$ is at least $e^{-\mu}\lambda^2\log(1/\lambda)-\epsilon$ and the probability that the overlap is zero is at least $(1-e^{-\mu})\lambda\log(1/\lambda)-\epsilon$ for sufficiently large $n$. The theorem is proved.
\end{proof}

\section{The partition function of Gibbs measures over trees}\label{sec:gibbs}
In this section, we compute the limiting asymptotics for the log partition function $\log Z$ of the Gibbs measure over trees. This is the technical heart of the analysis. We leave statements and proofs of most structural properties to Section~\ref{sec:gibbs2} afterwards. 

One interesting observation in our analysis is that we can compute the partition function of our model in terms of a ``spin system'' or \emph{Markov random field} (MRF), by relating the distribution over trees to the corresponding distribution over distance vectors. This is explained in Section~\ref{sec:transform-mrf} and is helpful for conceptualizing understanding our analysis, so we make use of it throughout the rest of the section.
\subsection{Preliminaries} 
Recall the definition from \eqref{eqn:finite-temp} of the Gibbs measure over spanning trees $T$ at inverse temperature $\beta$:
\begin{equation}\label{eqn:gibbs}
\mu_{G,\beta}(T) = \frac{1}{Z_{G,\beta}} \exp\left(-\beta \log\log(n)  \sum_{v\in\overline{V}} \mathsf{d}_T(1,v) \right)
\end{equation}
where $\overline{V}$ is the set of vertices in the connected component $\overline{G}$ of $G$ containing the source vertex $1$.

\begin{definition}[Entropy and Relative Entropy]
Given a discrete probability measure $\nu$, we define 
\[ H_{\nu}(X) = -\E_{X \sim \nu}[\log(\nu(X))] \]
to be the Shannon entropy.
Similarly, for any probability measure $P$ which is absolute continuous with respect to $Q$, 
\[ \KL(P,Q) = \E_P\left[\log \frac{dP}{dQ}\right] = \E_Q\left[\frac{dP}{dQ} \log \frac{dP}{dQ}\right] \]
is the relative entropy or Kullback--Leibler divergence (see, e.g., \cite{dembo1991information,van2014probability}). When $\mu$ is the uniform measure on a finite set $S$, then for any probability measure $\nu$ on $S$, we have the relation
\[ \KL(\nu,\mu) = \E_{X \sim \nu}[\log \nu(X) + \log |S|] = \log |S| - H_{\nu}(X). \]
\end{definition}
\begin{proposition}[Gibbs variational principle, see, e.g., \cite{ellis2007entropy,van2014probability,wainwright2008graphical}]\label{prop:gibbs-variational}
\[ \log Z_{G,\beta} = \sup_{\nu} \left[ H_{\nu}(T)  - \beta \log\log(n) \E_{\nu}\left[ \sum_v \mathsf{d}_T(1,v) \right]\right]\]
where the supremum ranges over all probability measures on the set of spanning trees of the giant component $\overline{G}$ of $G$. Furthermore, the supremum is uniquely attained at the Gibbs measure $\mu_{G,\beta}$. 
\end{proposition}
\begin{proof}
This follows by rearranging
\begin{equation}\label{eqn:kl-variational}
0 \le \KL(\nu, \mu_{G,\beta}) = \E_{\nu}\left[\log \frac{\nu(T)}{\mu_{G,\beta}(T)}\right] = \log Z_{G,\beta} - H_{\nu}(T) + \beta \log\log(n) \E_{\nu}\left[ \sum_v \mathsf{d}_T(1,v) \right]
\end{equation}
where we used the fact that the KL is zero iff $\nu = \mu_{G,\beta}$. 
\end{proof}
\begin{definition}\label{def:naive-mf}
Define the \emph{na\"ive mean-field approximation} to $\log Z$ by the right hand side of the following inequality:
\[ \log Z_{G,\beta} \ge \sup_{\nu \in \mathcal M}\left[ H_{\nu}(T)  - \beta \log\log(n) \E_{\nu}\left[ \sum_v \mathsf{d}_T(1,v) \right]\right] \]
where $\nu$ ranges over the set $\mathcal M$ of vertex-wise product measures over spanning trees in the following sense: under $\nu$, the parent $\mathsf{par}_T(v)$ of every vertex $v$ is chosen independently. The inequality is just a special case of the Gibbs variational principle.  
\end{definition}
Informally, we will say that the na\"ive mean-field approximation is accurate if it gets the value of $\log Z$ correct to leading order. Although we do not apply these results, deterministic conditions for guaranteeing the accuracy of the approximation have been extensively studied \cite[e.g.]{eldan2020taming,chatterjee2016nonlinear,jain2019mean,austin2019structure,eldan2018gaussian,basak2017universality,augeri2021transportation}. We note that accuracy of the mean-field approximation should not always be interpreted as implying that the Gibbs measure is close to a product measure --- this is true in a sense for one direction of the KL divergence, but false for the Wasserstein metric in some simple examples. Instead, it generally means that the measure is well-approximated by a small mixture of product measures (see \cite{eldan2018gaussian,austin2019structure}). We will explain and confirm this part of the picture later in Section~\ref{sec:gibbs2} with more detail.


\paragraph{Diameter bound.}
We use a fundamental result about the diameter of random graphs, which tells us the diameter of the giant component is not too much bigger than $d^*$. The references prove tighter concentration bounds than what is stated here, but the below statement is sufficient for our purposes. 
\begin{theorem}[\cite{chung2001diameter,riordan2010diameter}]\label{thm:diameter}
Let $p_n$ be a sequence of elements of $(0,1)$.
Asymptotically almost surely, the diameter of the giant component of $G(n, p_n)$ is
\[ (1 + o(1)) \frac{\log n}{\log(n p_n)} \]
provided that $n p_n \to \infty$.
\end{theorem}

\subsection{Transformation into a Markov random field}\label{sec:transform-mrf}
We next describe a transformation of our model to a Markov random field (MRF).
This transformation works on \emph{arbitrary graphs} and is not specific to the random graph setting.
Recall that we consider the Gibbs measure on spanning trees of a weighted undirected graph $G$ given by
\begin{equation*}
\mu_{G,\overline \beta}(T) = \frac{1}{Z_{G,\overline \beta}} \exp\left(-\overline \beta  \sum_v \mathsf{d}_T(1,v) \right)
\end{equation*}
where\footnote{We use the notation $\overline \beta$, because as mentioned above on random graphs we will ultimately let $\overline \beta = \beta \log \log n$ and view $\beta$ as the actual inverse temperature. This scaling is needed because the degree of vertices in our graph diverges as $n \to \infty$.} $\overline \beta \ge 0$. Note that the solution to the single-source shortest path problem itself is the mode of the distribution for any $\overline \beta > 0$.

If we think of $\mu_{G,\overline \beta}$ as inducing a probability measure over distance vectors $D_v = \mathsf{d}_T(1,v)$ and let $\ell_{uv}$ denote the length of the edge between each pair of neighbors $u$ and $v$,
then we find $D$ is a Markov random field\footnote{I.e., a random vector which satisfies the Markov property that $D_v$ is conditionally independent of the rest of $D$ given $(D_w)_{w \sim v}$, see e.g., \cite{mezard2009information,lauritzen1996graphical,wainwright2008graphical}. This always holds if the PMF is a product over factors on cliques of the graph $G$, c.f. the Hammersley-Clifford theorem \cite{lauritzen1996graphical}} on the underlying graph $G$. More precisely, we find:
\begin{proposition}
The marginal law of $D$ given $T \sim \mu$ is
\begin{equation}\label{eqn:mrf}
\Pr[D = d] = \frac{\mathsf{1}(d_1 = 0)}{Z_{G,\overline \beta}} \prod_{v \ne 1} \exp\left[H(\mathsf{par}_T(v) \mid D = d) - \overline \beta d_v\right] \mathsf{1}(\exists w \sim v, d_v = d_w + \ell_{vw}).
\end{equation}
Furthermore, the conditional entropy respects the graph structure, more precisely
\[ H(\mathsf{par}_T(v) \mid D = d) = H(\mathsf{par}_T(v) \mid D_v = d_v, \forall w \sim v, D_w = d_w), \]
so $D$ is a Markov random field on $G$.
\end{proposition}
\begin{proof}
Observe that its marginal law under $\mu = \mu_{G,\overline \beta}$ is
\[ \Pr_{\mu}[D = d] = \Pr_{\mu}[\forall v: d_T(1,v) = d_v] = \frac{1}{Z_{G,\overline \beta}} \sum_T  \exp\left(-\overline \beta  \sum_v d_v \right) \mathsf{1}(\forall v: d_T(1,v) = d_v).  \]
Observe
\[  \exp\left(-\overline \beta  \sum_v d_v \right) \mathsf{1}(\forall v: d_T(1,v) = d_v) = \mathsf{1}(d_1 = 0) \prod_{v \ne 1} \exp\left(- \overline \beta d_v\right) \mathsf{1}(\exists w \sim v, d_v = d_w + \ell_{vw}). \]
Furthermore, given that $\Pr(D = d) > 0$ we know that $T$ is a node-wise product measure in the sense that 
\[ \Pr[\forall v \ne 1, \mathsf{par}(v)_v = p_v \mid D] = \prod_{v \ne 1} \Pr[\mathsf{par}_T(v)_v = p_v \mid D] \]
since, once the distances to the root are known, the choice of parent of each vertex is an independent process which picks for each vertex $v$ uniformly at random among its neighbors which are at the correct distance, so we have that
\[ \log \#\{T : \forall v, d_T(1,v) = d_v\} = H(T \mid D = d) = \sum_v H(\mathsf{par}_T(v) \mid D = d). \]
Combining the above equations gives the stated conclusion. 
\end{proof}

\subsection{Analysis of the partition function}\label{sec:partition-function}

We analyze the partition function of the Gibbs measure defined in \eqref{eqn:gibbs}. For notational convenience, we write
\begin{equation}\label{eqn:gibbs-single}
    \mu_{G,\beta}(T)=\frac{1}{Z_{G,\beta}}\exp\left(-\overline{\beta}\sum_v\mathsf{d}_T(1,v)\right)
\end{equation}
where $\overline{\beta}=\beta\log\log n$. From the Gibbs variational principle (Proposition \ref{prop:gibbs-variational}), the log partition function can be written as
\begin{equation}\label{eqn:free-single}
    \Phi_{G,\beta}:=\log Z_{G,\beta}=\sup_{\nu}\left[H_{\nu}(T)  - \overline{\beta} \E_{\nu}\left[ \sum_v \mathsf{d}_T(1,v) \right]\right]\,.
\end{equation}
Here, $\nu$ is a measure on the set of spanning trees $T$ of the connected component $\overline{G}$ containing $1$ of $G$. Our goal of this section is to find a formula for \eqref{eqn:free-single} that is accurate up to an $o(n)$ (absolute) error asymptotically almost surely. More formally, we find a quantity $\tilde{\Phi}_{G,\beta}^*$ such that
\[
    \Phi_{G,\beta}=\tilde{\Phi}_{G,\beta}^*+o_p(n)\,.
\]

\paragraph{Organization of the analysis.} The intuition behind our analysis is simple, but the mathematical rigor to justify the idea and formalize the intuition might seem rather long and sophisticated. To present the formal results early but also to guide the readers through the proof idea and an intuitive picture, we organize this section as follows. In Section~\ref{sec:gibbs-formal}, we formally state core results including the formula for $\log Z_{G,\beta}$ with brief explanations and remarks, but we do not yet attempt to explain in detail what the terms in the formula mean. Section~\ref{sec:gibbs-regularity} serves as a quick review of the single graph asymptotics developed in Section~\ref{sec:single} which the rest of the sections rely on. Section~\ref{sec:gibbs-informal} then gives a non-rigorous analysis based on the Markov random field representation. Section~\ref{sec:kernels-cond-gibbs} through Section~\ref{sec:gibbs-opt-sol} utilize a handful of mathematical tools to give a formal proof to the main results.

\subsubsection{Main results}\label{sec:gibbs-formal}

\paragraph{Reduction to a one-dimensional optimization.} Our first result shows that we can solve the variational optimization \eqref{eqn:free-single} by instead solving a one-dimensional optimization. Namely, given the graph $G=G_n$, we construct a set of measures
\begin{equation}\label{eqn:pre-hatmu}
    \hat{\mu}_{G,1},\hat{\mu}_{G,2},\cdots,\hat{\mu}_{G,N_{d^*}}
\end{equation}
indexed by the set of integers at most $N_{d^*}$, the number of vertices at distance $d^*$ in $G$. The index $m$ of each measure $\hat{\mu}_{G,m}$ represents the number of vertices at depth $d^*$ satisfying $\mathsf{par}_T(v)\in\mathsf{par}_G(v)$ (i.e., the size of the \emph{kernel} of $T$ which we study in Section~\ref{sec:kernels-cond-gibbs}), and it maps to a value $\tilde{\Phi}_{G,\beta}(m)$ which we will define shortly. A measure $\hat{\mu}_{G,m}$ with the largest value of $\tilde{\Phi}_{G,\beta}(m)$ is supposed be a good approximation of $\mu_{G,\beta}$. An explicit construction of these measures is delayed to \eqref{eqn:hatmu} in Section~\ref{sec:kernels-cond-gibbs}.

To define $\tilde{\Phi}_{G,\beta}(m)$, we let $\Psi(m;n,\lambda)$ be a function of integers $1\leq m\leq n$ and a real number $0<\lambda<1$
\begin{equation}\label{eqn:psi}
    \Psi(m;n,\lambda):=\E\left[\log\left(\sum_{I\in\binom{[n]}{m}}\prod_{i\in I}X_i\right)\right]\,,\qquad X_1,\cdots,X_n\overset{\text{i.i.d.}}\sim\ZTP(\log(1/\lambda))
\end{equation}
where $\ZTP$ is the \emph{zero-truncated Poisson} distribution defined in Definition~\ref{def:ztp}. Now $\tilde{\Phi}_{G,\beta}(m)$ is defined by
\begin{equation}\label{eqn:gibbs-pts}
    \tilde{\Phi}_{G,\beta}(m):=\Psi(m;N_{d^*},\lambda_n)+(n-m)\log(mq)+\overline{\beta}m-\overline{\beta}(d^*+1)n\,.
\end{equation}
Note that \eqref{eqn:gibbs-pts} is formulated only in terms of $n$, $\alpha_n$, and $N_{d^*}$. The following theorem shows that we can approximate $\Phi_{G,\beta}=\log Z_{G,\beta}$ by solving a one-dimensional optimization
\begin{equation}\label{eqn:gibbs-opt}
    \tilde{\Phi}_{G,\beta}^*:=\max_{1\leq m\leq N_{d^*}}\tilde{\Phi}_{G,\beta}(m)\,.
\end{equation}



\begin{theorem}\label{thm:main-gibbs}
    There exists a collection of measures \eqref{eqn:pre-hatmu} with the following property. For any fixed constant $\beta>0$, let $m^*$ be any function of $n$, $\alpha_n$, and $N_{d_n^*}$ (and $\beta$) such that $m^*\in\{1,\cdots,N_{d_n^*}\}$ almost surely and it approximately solves \eqref{eqn:gibbs-opt} in a sense that
    \begin{equation}\label{eqn:gibbs-apx-sol}
        \tilde{\Phi}_{G_n,\beta}^*=\tilde{\Phi}_{G_n,\beta}(m^*)+o_p(n)\,.
    \end{equation}
    Then we have the following.
    \begin{itemize}
        \item $\hat{\mu}_{G_n,m^*}$ approximates $\mu_{G_n,\beta}$ in KL divergence:
        \[
            \KL(\hat{\mu}_{G_n,m^*},\mu_{G_n,\beta})=o_p(n)\,.
        \]

        \item The log partition function $\Phi_{G_n,\beta}$ is well approximated by $\tilde{\Phi}_{G_n,\beta}(m^*)$:
        \[
            \Phi_{G_n,\beta}=\tilde{\Phi}_{G_n,\beta}(m^*)+o_p(n)\,.
        \]
    \end{itemize}
\end{theorem}
\begin{proof}
    See Section~\ref{sec:gibbs-1d-opt}.
\end{proof}

An equivalent but slightly different viewpoint to this reduction is to consider the following mean-field spin system with a random external field. Suppose that the nature samples $A_1,\cdots,A_n\overset{\text{i.i.d.}}\sim\Pois(\log(1/\lambda_n))$. We consider a (random) measure defined on $\mathbf{x}=(x_1,\cdots,x_n)\in\{0,1\}^n$
\begin{equation}\label{eqn:mf-cw}
    \frac{\mathbf{1}_{A_i=0\implies x_i=0}}{Z_{\mathbf{A}}}\exp\left(\sum_{i:A_i>0}(\log A_i-\log(nq)+\overline{\beta})x_i+n\cdot g(\bar{\mathbf{x}})\right)
\end{equation}
where $\bar{\mathbf{x}}=\frac{1}{n}\sum_{i=1}^nx_i$ and $g(x)=(1-x)\log x$. Then the second part of Theorem~\ref{thm:main-gibbs} can be stated in the following way:
\[
    \log Z_{G,\beta}=\log Z_{\mathbf{A}}+n\log(nq)-\overline{\beta}(d^*+1)n+o_p(n)\,.
\]
This fact is not needed nor is explicitly proved in our analysis. However, this formulation may provide further insights to our model and is especially useful when we discuss the behavior near the critical temperature. See Appendix~\ref{apdx:gibbs-critical} for details.

\paragraph{Solving the optimization.} The next step is to find a suitable $m^*$ that satisfies \eqref{eqn:gibbs-apx-sol}. As we will see shortly, depending on the constant $\beta>0$ and the asymptotics of the proxy sequences, there are two regimes where the solutions are far apart. To state the result, we define a new proxy
\[
    \kappa_n:=(1-\lambda_n)\log\log n\,.
\]
Since $N_{d^*}\approx(1-\lambda_n)n$ by Theorem~\ref{thm:main1-tight}, $\kappa_n$ expresses the size of $N_{d^*}$ at the order of $n/\log\log n$.

Now we precisely define the two regimes where we state our result. The first is the low temperature phase, which is the union of two regimes:
\begin{equation}\label{eqn:low-temp-phase}
    \begin{dcases}
        \beta>\limsup_{n\to\infty}(1-\Delta_n-\kappa_n^{-1})\\
        \beta=\limsup_{n\to\infty}(1-\Delta_n-\kappa_n^{-1})\quad\text{and}\quad\limsup_{n\to\infty}\kappa_n<\infty\,.
    \end{dcases}\tag{LT}
\end{equation}
Next is the high temperature phase
\begin{equation}\label{eqn:high-temp-phase}
    \beta<\liminf_{n\to\infty}(1-\Delta_n-\kappa_n^{-1})\,.\tag{HT}
\end{equation}
Note that these do not cover all possible regimes, especially when $\beta$ is near the limits of $1-\Delta_n-\kappa_n^{-1}$. Nevertheless, we present a formula of $\log Z_{G,\beta}$ for the two regimes.


\begin{theorem}[Approximate log partition function]\label{thm:logz-formula}
    In each phase, we have the following formula for the log partition function $\Phi_{G_n,\beta}=\log Z_{G_n,\beta}$.
    \begin{itemize}
        \item In the low temperature phase \eqref{eqn:low-temp-phase}, $m^*=N_{d_n^*}$ satisfies \eqref{eqn:gibbs-apx-sol} and we have
        \[
            \Phi_{G_n,\beta}=n\Psi_0(\lambda_n)+\lambda_nn\log(\alpha_n(1-\lambda_n)\log n)-\overline{\beta}(d_n^*+\lambda_n)n+ o_p(n)
        \]
        where
        \[
            \Psi_0(\lambda):=\E[\log(X\vee 1)]\,,\quad X\sim\Pois(\log(1/\lambda))\,.
        \]

        \item In the high temperature phase \eqref{eqn:high-temp-phase}, $m^*=\lfloor\frac{N_{d_n^*}}{(1-\Delta_n-\beta)\kappa_n}\rfloor$ satisfies \eqref{eqn:gibbs-apx-sol} and we have
        \[
            \Phi_{G_n,\beta}=n\log\left(\frac{\alpha_n\log n}{(1-\Delta_n-\beta)\log\log n}\right)-n-\overline{\beta}(d_n^*+1)n+ o_p(n)\,.
        \]
    \end{itemize}
\end{theorem}
\begin{proof}
    See Section~\ref{sec:gibbs-opt-sol}.
\end{proof}

\begin{remark}
    In Section~\ref{sec:kernels-cond-gibbs}, we will construct the measures $\hat{\mu}_{G,m}$ in a way that the last measure $\hat{\mu}_{G,N_{d^*}}$ corresponding to $m=N_{d^*}$ is precisely the uniform measure $\mu_{G,\infty}$ over the shortest path trees. Then Theorem~\ref{thm:logz-formula} indicates that the Gibbs measure $\mu_{G,\beta}$ is close in KL divergence to the uniform measure $\mu_{G,\infty}$ in the low temperature phase \eqref{eqn:low-temp-phase}.
\end{remark}

\subsubsection{Regularity properties of the underlying graph}\label{sec:gibbs-regularity}

Before we proceed on studying the Gibbs measures $\mu_{G,\beta}$, we remind the reader that $\mu_{G,\beta}$ themselves are \emph{random} measures, since the graph $G$ is drawn from the Erd\"os--R\'enyi model $\mathcal{G}(n,q)$. Thus, when we say we sample a random spanning tree $T$ from $\mu_{G,\beta}$, there are in fact two sources of randomness: one from $G$, and one from $\mu_{G,\beta}$ conditioned on $G$. Thus, for most of the results in this section, we actually prove that $\mu_{G,\beta}$ enjoys certain properties for ``nice'' graphs $G$, where $G\sim\mathcal{G}(n,q)$ is ``nice'' with probability tending to $1$ as $n\to\infty$. Many of such ``nice'' features of a random graph come from what we have developed in Section~\ref{sec:single}. In this section, we collect these ``regularity properties'' on $G$ we will frequently refer back to throughout the proofs in Section~\ref{sec:gibbs}. Some of the proofs need additional regularity properties specific to them, and these will usually be stated as a separate lemma.

As for notations, we try to distinguish the randomness from $G$ and from $\mu_{G,\beta}$ to minimize confusion. Whenever we study the probability of a certain event $E$ of spanning trees under the Gibbs measure $\mu_{G,\beta}$ conditioned on $G$, we write
\[
    \mu_{G,\beta}(E)
\]
instead of the usual probability symbol, e.g., $\Pr_{T\sim\mu_{G,\beta}\mid G}(E)$. For the expectation, we always include a subscript to indicate under what measure we are taking the expectation over, such as
\[
    \E_{\mu_{G,\beta}}[f(T)]\,.
\]

Before we proceed, we briefly recap some relevant facts regarding the shortest path structure of $G$ we developed in Section~\ref{sec:single}. The purpose of this section is to provide an intuitive picture, so the results will be stated rather informally.
\begin{itemize}
    \item The single most important groundwork for our analysis is perhaps the stochastic process \eqref{eqn:nd_process}. A useful corollary of this is that conditioned on $\Gamma_{\leq d-1}$, we have
    \[
        |\mathsf{par}_G(v)|\sim\ZTB(N_{d-1},q)
    \]
    for all $v\in\Gamma_d$, where $\ZTB(n, p)$ is the \emph{zero-truncated binomial} distribution, the binomial. Essential properties of zero-truncated distributions can be found in Appendix~\ref{apdx:log-binomial}.
    \item Proposition~\ref{prop:conc} tells us that
    \begin{equation}\label{eqn:dsm1-neg}
        |\Gamma_{\leq d^*-1}|=O((nq)^{d^*-1})=O_p\left(\frac{n}{(\log\log n)^2}\right)\,.
    \end{equation}
    Since a typical vertex has degree $O(\log n)$, the entropy gained from picking its parent $\mathsf{par}_T(v)$ is at most $O(\log\log n)$. This implies that the entropy coming from vertices of distance at most $d^*-1$ is only $o_p(n)$ and thus can be safely ignored.
    \item By Theorem~\ref{thm:main1-tight}, most of the vertices are in $\Gamma_{d^*}\cup\Gamma_{d^*+1}$, i.e., have distance either $d^*$ or $d^*+1$ to the root $1$ in $G$. In addition, $N_{d^*}$ is well approximated by $(1-\lambda_n)n$:
    \begin{equation}\label{eqn:dstar-apx}
        N_{d^*}=\left(1+o_p\left(\frac{1}{\log\log n}\right)\right)(1-\lambda_n)n\,.
    \end{equation}
    \item Conditioned on $\Gamma_{\leq d^*-1}$, each $v\in\Gamma_{d^*}$ has $\ZTB(N_{d^*-1},q)$ parents. Since $qN_{d^*-1}$ concentrates around $q(nq)^{d^*-1}=\log(1/\lambda_n)$ by Proposition~\ref{prop:conc}, this is approximately $\ZTP(\log(1/\lambda_n))$, the \emph{zero-truncated Poisson} distribution which is briefly discussed in Appendix~\ref{apdx:pois-apx}. If, for instance, $\Delta_n\to\Delta\in(0,1]$ (and thus $\lambda_n\to\lambda=0$), this concentrates around $(nq)^{\Delta}$.
    \item Conditioned on $\Gamma_{\leq d^*}$, $G[\Gamma_{d^*}]$ is marginally an Erd\"os--R\'enyi graph; consequently, a typical vertex $v\in\Gamma_{d^*}$ has approximately $qN_{d^*}$ neighbors in $\Gamma_{d^*}$. By Lemma~\ref{lem:dstar} this concentrates around $q(nq)^{d^*}$, which is $\Omega(\log n/(\log\log n)^2)$.
    \item Conditioned on $\Gamma_{\leq d^*}$, each $v\in\Gamma_{d^*+1}$ has $\ZTB(N_{d^*},q)$ parents, which concentrates around $qN_{d^*}$. Again, this is at least $\Omega(\log n/(\log\log n)^2)$ asymptotically almost surely.
\end{itemize}

\subsubsection{An informal description of the Gibbs measures}\label{sec:gibbs-informal}


To get a sense of the behavior of $\mu_{G,\beta}$, we revisit the informal analysis in Section~\ref{sec:to-gibbs} from a slightly different viewpoint --- a Markov Random Field perspective from Section~\ref{sec:transform-mrf}. Namely, we examine the (marginal) law of the \emph{distance vector} $\mathsf{d}_T$ under $\mu_{G,\beta}$. Imagine that we are sampling a distance vector using a Gibbs sampler. Namely, we sample a value of $\mathsf{d}_T(1,v)$ for each non-root vertex $1\neq v\in V$, conditioned on the distance values of all the other vertices. The (unnormalized) log probability of a particular distance $d$ is the sum of the \emph{parent selection entropy}
\[
    H(\mathsf{par}_T(v)):=\log\left|\left\{u\in\mathsf{N}_G(v):\mathsf{d}_T(1,u)=d-1\right\}\right|=\log|\mathsf{N}_G(v)\cap\Gamma_{d-1}^T|
\]
and the negative \emph{potential energy} times the inverse temperature
\[
    -\overline{\beta}d=-\beta d\log\log n\,.
\]
Now we run an ``informal'' Gibbs sampling, starting with the initial vector $\mathsf{d}_T=\mathsf{d}_G$. In terms of the distribution over spanning trees, we are starting with the uniform measure over the shortest path trees. For each vertex $v$, we consider candidate distances that $v$ can choose, and roughly estimate how much parent selection entropy and potential energy $v$ might gain.
\begin{itemize}
    \item Suppose that $v\in\Gamma_{d^*}$. If we choose to leave $v$ at the same level, i.e., $\mathsf{d}_T(1,v)=d^*$, then the parent selection entropy is approximately the log of $\ZTP(\log(1/\lambda_n))$. For $\Delta>0$, this would concentrate around $\Delta\log(nq)\approx\Delta\log\log n$. On the other hand, if we choose $d^*+1$, the parent selection entropy is at least $\log(\log n/(\log\log n)^2)\approx\log\log n$, but $v$ also gains $\beta\log\log n$ potential energy. Since $v$ gains almost the maximum parent selection entropy $\log\log n$ at level $d^*+1$, it will almost never move to level $d^*+2$ or greater. From this observation, we make the following guess:
    \begin{itemize}
        \item if $\beta>1-\Delta$, then $v$ will most likely stay at distance $d^*$ and
        \item if $\beta<1-\Delta$, then $v$ will most likely jump to distance $d^*+1$.
    \end{itemize}
    This is an example of a \emph{phase transition}. See Section~\ref{sec:phase-transition} for more details including formal statements.

    \item Suppose that $v\in\Gamma_{d^*+1}$. Then the parent selection entropy is at least $\log(\log/(\log\log n)^2)\approx\log\log n$, which is already nearly the maximum. Thus, as long as $\beta$ is a positive constant, $v$ will almost never move.

    \item The behavior of the vertices at level $d^*-1$ or below might be a little more complicated. For now, we will just accept that they will almost never move. Roughly speaking, this is because a typical vertex in $\Gamma_{\leq d^*-1}$ has a lot of descendants at level $d^*$ and $d^*+1$; if too many of $\Gamma_{\leq d^*-1}$ leave their level, then a lot of vertices might gain a lot of potential energy which wouldn't justify the entropy gain.
\end{itemize}
This suggests that the law of a spanning tree $T\sim\mu_{G,\beta}$ of $\overline{G}$ largely depends on the dynamics of the vertices in $\Gamma_{d^*}$, which is captured by the notion of \emph{kernels} defined in the following section.

\subsubsection{Kernels and conditional Gibbs measures}\label{sec:kernels-cond-gibbs}

As mentioned in the preceding section, a key notion that lies in the heart of our analysis is the \emph{kernel} of a spanning tree, which we define as follows.

\begin{definition}\label{def:kernel}
    Let $T$ be a spanning tree of $\overline{G}$. The \emph{kernel} of $T$, denoted $\varphi(T)$, is defined by
    \[
        \varphi(T):=\Gamma_{d^*}\cap\mathsf{ch}_T(\Gamma_{d^*-1})\,.
    \]
\end{definition}

In other words, $\varphi(T)$ is the set of vertices which are children of $\Gamma_{d^*-1}$, both for $G$ and $T$. Intuitively, from the view of our imaginary Gibbs sampler from Section~\ref{sec:gibbs-informal}, $\varphi(T)$ is related to the set of vertices in $\Gamma_{d^*}$ that ``stay'', though we don't require vertices in the kernel to have distance $d^*$ in $T$ (in fact, they do most of the time; see Section~\ref{sec:gibbs-w1}). The following proposition characterizes all possible kernels.

\begin{proposition}\label{prop:kernel-char}
    Any kernel intersects every connected component of $\overline{G}\setminus\Gamma_{\leq d^*-1}$. Conversely, any subset of $\Gamma_{d^*}$ intersecting every component of $\overline{G}\setminus\Gamma_{\leq d^*-1}$ is the kernel of a spanning tree of $\overline{G}$.
\end{proposition}
\begin{proof}
    Assume $A=\varphi(T)$. For any vertex $v\in\overline{V}\setminus\Gamma_{\leq d^*-1}$, a (unique) path in $T$ from $v$ to the root must contain an edge $u_1\to u_2$ where $u_1\in\Gamma_{d^*}$ and $u_2\in\Gamma_{d^*-1}$. This implies $u_1\in\mathsf{ch}_T(\Gamma_{d^*-1})$ and thus $u_1\in A$. Hence, there is a path in $T$ from $v$ to $A$.

    Conversely, let $A\subseteq\Gamma_{d^*}$ be a set that intersects every component of $\overline{G}\setminus\Gamma_{\leq d^*-1}$. We construct a spanning tree $T$ as follows:
    \begin{enumerate}
        \item choose a spanning tree $T_0$ of $G[\Gamma_{\leq d^*-1}]$;
        \item choose edges $e_1,\cdots,e_{|A|}$ between the vertices in $A$ and $\Gamma_{d^*-1}$, one for each vertex;
        \item choose spanning trees $T_1,\cdots,T_k$ of the all $k$ components of $\overline{G}\setminus\Gamma_{\leq d^*-1}$;
        \item set $T=T_0\cup T_1\cup\cdots\cup T_k\cup\{e_1,\cdots,e_{|A|}\}$.
    \end{enumerate}
    It is straightforward to verify that $T$ is indeed a spanning tree of $\overline{G}$ and that $A=\varphi(T)$.
\end{proof}

We denote by
\[
    \mathcal{K}(G):=\{\varphi(T):\text{$T$ is a spanning tree of $\overline{G}$}\}
\]
the set of all possible kernels as characterized in the proposition.

\paragraph{Conditioning on the kernel.} One of the reasons why kernels are useful is because the Gibbs measures $\mu_{G,\beta}$ are easy to describe when conditioned on the kernel. Consider a decomposition of the partition function $Z_{G,\beta}$ by grouping together the spanning trees with the same kernel, i.e.,
\begin{equation}\label{eqn:gibbs-partition-kernel}
    Z_{G,\beta}=\sum_{A\in\mathcal{K}(G)}Z_{G,\beta}|_A
\end{equation}
where
\[
    Z_{G,\beta}|_A:=\sum_{T:\varphi(T)=A}\exp\left(-\overline{\beta}\sum_{v}\mathsf{d}_T(1,v)\right)\,.
\]
This corresponds to decomposing a measure $\nu$ on the spanning trees of $\overline{G}$ by
\[
    \nu=\sum_{A\in\mathcal{K}(G)}\varphi_*(\nu)(A)\cdot\nu|_A
\]
where
\[
    \varphi_*(\nu)(A):=\nu(\varphi(T)=A)
\]
are the mixture weights and
\[
    \nu|_A(T)=\nu(T\mid\varphi(T)=A)
\]
are the conditional measures. Here, $\varphi_*(\nu)$ can equivalently be understood as the pushforward of $\nu$ by $\varphi$.
Our first result is that whenever $|A|$ is reasonably large, regardless of the constant $\beta>0$, $\mu_{G,\beta}|_A$ behaves like its zero-temperature limit $\mu_{G,\infty}|_A$ which can be described by the following sampling procedure:
\begin{enumerate}
    \item for each $v\in\Gamma_{\leq d^*-1}$, sample $\mathsf{par}_T(v)$ from $\mathsf{par}_G(v)$ uniformly at random;
    \item for each $v\in A$, sample $\mathsf{par}_T(v)$ from $\mathsf{par}_G(v)$ uniformly at random;
    \item for each $v\in\overline{V}\setminus\Gamma_{\leq d^*-1}\setminus A$, sample $\mathsf{par}_T(v)$ uniformly at random from the set of vertices $u\in\mathsf{N}_G(v)$ such that $\mathsf{d}_G(v, A)=\mathsf{d}_G(u,A)+1$.
\end{enumerate}
In other words, $\mu_{G,\infty}|_A$ is a product measure that always samples a minimum energy configuration, conditioned on $\varphi(T)=A$. To state the formal result, we define
\begin{equation}\label{eqn:ml-def}
    m_{\ell}:=m_0\vee\left\lceil\frac{n}{3(\log\log n)^2}\right\rceil
\end{equation}
where
\begin{equation}\label{eqn:m0-def}
    m_0:=\min_{A\in\mathcal{K}(G)}|A|=\#\{\text{connected components of $\overline{G}\setminus\Gamma_{\leq d^*-1}$}\}
\end{equation}
is the minimum size of a kernel.

\begin{proposition}\label{prop:fixed-kernel}
    Let $\beta>0$ be a constant. Then with probability at least $1-o(1)$, we have
    \[
        \KL(\mu_{G,\infty}|_A,\mu_{G,\beta}|_A)=o(n)
    \]
    and
    \begin{equation}\label{eqn:gibbs-cond-zeta}
        \log (Z_{G,\beta}|_A)=\sum_{v\in A}\log|\mathsf{par}_G(v)|+(n-|A|)\log(|A|q)+\overline{\beta}|A|-\overline{\beta}(d^*+1)n+o(n)
    \end{equation}
    simultaneously for all $A\in\mathcal{K}(G)$ with size at least
    \begin{equation}\label{eqn:gibbs-ker-lb}
        |A|\geq m_\ell\,.
    \end{equation}
\end{proposition}

We postpone the proof of Proposition~\ref{prop:fixed-kernel} to Section~\ref{sec:gibbs-kernel-proof}. Once we understand $\log(Z_{G,\beta}|_A)$, \eqref{eqn:gibbs-partition-kernel} gives
\[
    \begin{split}
        \Phi_{G,\beta} &=\log Z_{G,\beta}\\
        &\leq\max_{A\in\mathcal{K}(G)}\log (Z_{G,\beta}|_A)+\log|\mathcal{K}(G)|\\
        &\leq\max_{A\in\mathcal{K}(G)}\log (Z_{G,\beta}|_A)+O(n)\,.
    \end{split}
\]
We clearly have $\max_{A\in\mathcal{K}(G)}\log(Z_{G,\beta}|_A)\leq\Phi_{G,\beta}$, so this in fact implies
\begin{equation}\label{eqn:gibbs-opt-bigo}
    \Phi_{G,\beta}=\max_{A\in\mathcal{K}(G)}\log (Z_{G,\beta}|_A)+O(n)\,.
\end{equation}
Thus, if we know that the maximum is attained where \eqref{eqn:gibbs-ker-lb} holds, then we obtain an approximate value of $\Phi_{G,\beta}$ by optimizing \eqref{eqn:gibbs-cond-zeta} over $A$, and the corresponding measure $\mu_{G,\infty}|_A$ will be close to $\mu_{G,\beta}$ in KL distance. However, this approximation has error up to $O(n)$, due to $\mathcal{K}(G)$ being too large.

\paragraph{Conditioning on the kernel size.} If we aim at a more delicate analysis, we should use a coarser decomposition of $Z_{G,\beta}$. Namely, instead of grouping by the kernel, we group by the kernel size. We define
\[
    \mathcal{K}_m(G):=\{A\in\mathcal{K}(G):|A|=m\}
\]
as the collection of kernels of size $m$. Consider a decomposition of $Z_{G,\beta}$ by
\[
    Z_{G,\beta}=\sum_{m=m_0}^{N_{d^*}}Z_{G,\beta}\|_m
\]
where
\begin{equation}\label{eqn:gibbs-partition-kernel-size}
    Z_{G,\beta}\|_m:=\sum_{A\in\mathcal{K}_m(G)}Z_{G,\beta}|_A\,.
\end{equation}
The corresponding decomposition of a measure $\nu$ can be written as
\[
    \nu=\sum_{m=m_0}^{N_{d^*}}|\varphi|_*(\nu)(m)\cdot\nu\|_m
\]
where
\[
    |\varphi|_*(\nu)(m):=\nu(|\varphi(T)|=m)
\]
and
\[
    \nu\|_m(T):=\nu(T\mid|\varphi(T)|=m)\,.
\]
Here, in contrast to \eqref{eqn:gibbs-opt-bigo}, we have
\[
    \begin{split}
        \Phi_{G,\beta}&\leq\max_{m_0\leq m\leq N_{d^*}}\log(Z_{G,\beta}\|_m)+\log(N_{d^*}-m_0+1)\\
        &\leq\max_{m_0\leq m\leq N_{d^*}}\log(Z_{G,\beta}\|_m)+O(\log n)
    \end{split}
\]
so
\begin{equation}\label{eqn:gibbs-via-ks}
    \Phi_{G,\beta}=\max_{m_0\leq m\leq N_{d^*}}\log(Z_{G,\beta}\|_m)+O(\log n)
\end{equation}
which is tight enough. Our goal now is to analyze \eqref{eqn:gibbs-partition-kernel-size}. Plugging in \eqref{eqn:gibbs-cond-zeta} from Proposition~\ref{prop:fixed-kernel}, we get
\[
    \begin{split}
        \log(Z_{G,\beta}\|_m) &= \log\left(\sum_{A\in\mathcal{K}_m(G)}Z_{G,\beta}|_A\right)\\
        &= \log\left(\sum_{A\in\mathcal{K}_m(G)}\prod_{v\in A}|\mathsf{par}_G(v)|\right)+(n-m)\log(mq)+\overline{\beta}m-\overline{\beta}(d^*+1)n+o_p(n)\,.
    \end{split}
\]
Here, we define
\begin{equation}\label{eqn:hg-def}
    \mathcal{H}_G(m):=\log\left(\sum_{A\in\mathcal{K}_m(G)}\prod_{v\in A}|\mathsf{par}_G(v)|\right)
\end{equation}
which is the log partition function of the measure $\tau_{G,m}$ on $\mathcal{K}_m(G)$ defined by
\[
    \tau_{G,m}(A)\propto\prod_{v\in A}|\mathsf{par}_G(A)|\,.
\]
Thus, to approximate the measures $\mu_{G,\beta}\|_m$, we build measures $\hat{\mu}_{G,m}$ by mixing $\mu_{G,\infty}|_A$ with mixture weight $\tau_{G,m}$:
\begin{equation}\label{eqn:hatmu}
    \hat{\mu}_{G,m}:=\sum_{A\in\mathcal{K}_m(G)}\tau_{G,m}(A)\cdot\mu_{G,\infty}|_A\,.
\end{equation}
Then we arrive at the following result.

\begin{proposition}\label{prop:fixed-kernel-size}
    Let $\beta>0$ be a constant. Then with probability at least $1-o(1)$, we have
    \[
        \KL(\hat{\mu}_{G,m},\mu_{G,\beta}\|_m)=o(n)
    \]
    and
    \[
        \log(Z_{G,\beta}\|_m)=\mathcal{H}_G(m)+(n-m)\log(mq)+\overline{\beta}m-\overline{\beta}(d^*+1)n+o(n)
    \]
    simultaneously for all $m_\ell\leq m\leq N_{d^*}$.
\end{proposition}
\begin{proof}
    We already proved the second part. For the first part, for the sake of convenience, we define
    \[
        f(\nu):=H_{\nu}(T)-\overline{\beta}\E_{\nu}\left[\sum_v\mathsf{d}_T(1,v)\right]\,.
    \]
    We observe that by Proposition~\ref{prop:fixed-kernel} and the Gibbs variational principle (Proposition~\ref{prop:gibbs-variational})
    \[
        \begin{split}
            \KL(\hat{\mu}_{G,m},\mu_{G,\beta}\|_m) &= \log(Z_{G,\beta}\|_m)-f(\hat{\mu}_{G,m})\\
            &=\log(Z_{G,\beta}\|_m)-\left(H_{\tau_{G,m}}(\mathbf{A})+\E_{\tau_{G,m}}\left[f(\mu_{G,\infty}|_{\mathbf{A}})\right]\right)\\
            &=\log(Z_{G,\beta}\|_m)-\left(H_{\tau_{G,m}}(\mathbf{A})+\E_{\tau_{G,m}}\left[\log(Z_{G,\beta}|_{\mathbf{A}})\right]\right)+o_p(n)\\
            &=\mathcal{H}_G(m)-\left(H_{\tau_{G,m}}(\mathbf{A})+\E_{\tau_{G,m}}\left[\sum_{v\in\mathbf{A}}\log|\mathsf{par}_G(v)|\right]\right)+o_p(n)\\
            &=o_p(n)
        \end{split}
    \]
    as desired.
\end{proof}

\begin{remark}
    The measures $\hat{\mu}_{G,m}$ defined in \eqref{eqn:hatmu} will be our desired set of measures \eqref{eqn:pre-hatmu} that satisfies Theorem~\ref{thm:main-gibbs}. Of course we only have $\hat{\mu}_{G,m}$ for $m\geq m_0$ because otherwise there is no available kernel. For $m<m_0$ we can just set $\hat{\mu}_{G,m}$ arbitrarily since these values of $m$ can never approximately solve \eqref{eqn:gibbs-opt}; this will be proved in Section~\ref{sec:gibbs-1d-opt}. Note that when $m=N_{d^*}$, $\tau_{G,m}$ simply puts all mass to $A=\Gamma_{d^*}$ so $\hat{\mu}_{G,m}=\mu_{G,\infty}$ as we discussed in the remark below Theorem~\ref{thm:logz-formula}.
\end{remark}

\subsubsection{Proof of Proposition~\ref{prop:fixed-kernel}}\label{sec:gibbs-kernel-proof}
\paragraph{An upper bound.} We start by deriving an upper bound of $\log(Z_{G,\beta}|_A)$ which leverages the Gibbs variational principle and a nice property of kernels.

\begin{lemma}\label{lem:free-core}
    Let $M$ be a constant greater than the maximum degree of $G$. Then with probability at least $1-o(1)$, we have
    \[
        \log(Z_{G,\beta}|_A)\leq\sum_{v\in A}\log|\mathsf{par}_G(v)|+\sum_{v\in \overline{V}\setminus\Gamma_{\leq d^*-1}\setminus A}\log(|\mathsf{N}_A(v)|+Me^{-\overline{\beta}})+\overline{\beta}|A|-\overline{\beta}(d^*+1)n+o(n)
    \]
    simultaneously for all $A\in\mathcal{K}(G)$.
\end{lemma}
\begin{proof}
    First, let $A\in\mathcal{K}(G)$ be fixed. By the Gibbs variational principle (Proposition~\ref{prop:gibbs-variational}), we have
    \[
        \log(Z_{G,\beta}|_A)=\sup_{\nu}\left[H_\nu(T)-\overline{\beta}\E_{\nu}\left[\sum_{v\in\overline{V}}\mathsf{d}_T(1,v)\right]\right]
    \]
    where $\nu$ ranges over all measures over the spanning trees $T$ of $\overline{G}$ such that $\varphi(T)=A$. Using the sub-additivity of entropy (see, e.g., \cite{dembo1991information}), we write
    \[
        \begin{split}
            H_{\nu}(T)&\leq\sum_{v\in\Gamma_{\leq d^*-1}\setminus\{1\}}H_\nu(\mathsf{par}_T(v))+\sum_{v\in A}H_\nu(\mathsf{par}_T(v))+\sum_{v\in\overline{V}\setminus\Gamma_{\leq d^*-1}\setminus A}H_\nu(\mathsf{par}_T(v))\\
            &\leq |\Gamma_{\leq d^*-1}|\log M+\sum_{v\in A}\log|\mathsf{par}_G(v)|+\sum_{v\in\overline{V}\setminus\Gamma_{\leq d^*-1}\setminus A}H_{\nu}(\mathsf{par}_T(v))
        \end{split}
    \]
    where we have used the fact that since $A$ is the kernel of $T$, the choice of parents for each $v\in A$ is limited to its parents in $G$. For the energy part, we have
    \[
        \begin{split}
            \sum_{v\in\overline{V}}\mathsf{d}_T(1,v) &=\sum_{v\in A}\mathsf{d}_T(1,v)+\sum_{v\in\overline{V}\setminus A}\mathsf{d}_T(1,v)\\
            &\geq\sum_{v\in A}\mathsf{d}_T(1,v)+\sum_{v\in\Gamma_{d^*-1}}\mathsf{d}_T(1,v)+\sum_{v\in\overline{V}\setminus\Gamma_{\leq d^*-1}\setminus A}\mathsf{d}_T(1,v)\\
            &\geq N_{d^*-1}(d^*-1)+|A|d^*+\sum_{v\in\overline{V}\setminus\Gamma_{\leq d^*-1}\setminus A}\mathsf{d}_T(1,v)\,.
        \end{split}
    \]
    A vertex $v\in\overline{V}\setminus\Gamma_{\leq d^*-1}\setminus A$ is not in the kernel, so any path to the root must pass through $A$. This implies that $\mathsf{d}_T(1,v)\geq d^*+\mathsf{d}_T(v,A)$. Thus,
    \[
        \begin{split}
            \sum_{v\in\overline{V}}\mathsf{d}_T(1,v) &\geq N_{d^*-1}(d^*-1)+|A|d^*+\sum_{v\in\overline{V}\setminus\Gamma_{\leq d^*-1}\setminus A}(d^*+\mathsf{d}_T(v, A))\\
            &=|\overline{V}\setminus\Gamma_{\leq d^*-2}|d^*+\sum_{v\in\overline{V}\setminus\Gamma_{\leq d^*-1}\setminus A}\mathsf{d}_T(v,A)-N_{d^*-1}\,.
        \end{split}
    \]
    This gives
    \begin{equation}\label{eqn:gibbs-cond-kernel-ub}
        \begin{split}
            \log(Z_{G,\beta}|_A)&\leq |\Gamma_{\leq d^*-1}|\log M+\sum_{v\in A}\log|\mathsf{par}_G(v)|-\overline{\beta}|\overline{V}\setminus\Gamma_{\leq d^*-2}|d^*+\overline{\beta}N_{d^*-1}\\
            &\quad+\sup_{\nu}\left[\sum_{v\in\overline{V}\setminus\Gamma_{\leq d^*-1}\setminus A}\left(H_{\nu}(\mathsf{par}_T(v))-\overline{\beta}\E_{\nu}\left[\mathsf{d}_T(v,A)\right]\right)\right]\,.
        \end{split}
    \end{equation}
    Thus supremum in this expression is bounded by
    \[
        \sum_{v\in\overline{V}\setminus\Gamma_{\leq d^*-1}\setminus A}\sup_{\nu}\left[H_{\nu}(\mathsf{par}_T(v))-\overline{\beta}\E_{\nu}\left[\mathsf{d}_T(v,A)\right]\right]\,.
    \]
    A trick here is to use a simple bound $\mathsf{d}_T(v,A)\geq1+\mathbf{1}_{\mathsf{par}_T(v)\notin A}$. Again using the Gibbs variational principle, we have
    \[
        \begin{split}
            \sup_{\nu}\left[H_{\nu}(\mathsf{par}_T(v))-\overline{\beta}\E_{\nu}\left[\mathbf{1}_{\mathsf{par}_T(v)\notin A}\right]\right]&\leq\log\left(\sum_{u\in\mathsf{N}_A(v)}e^0+\sum_{u\in\mathsf{N}_G(v)\setminus\mathsf{N}_A(v)}e^{-\overline{\beta}}\right)\\
            &\leq\log(|\mathsf{N}_A(v)|+Me^{-\overline{\beta}})\,.
        \end{split}
    \]
    Plugging into \eqref{eqn:gibbs-cond-kernel-ub}, we get
    \[
        \begin{split}
            \log(Z_{G,\beta}|_A)&\leq |\Gamma_{\leq d^*-1}|\log M+\sum_{v\in A}\log|\mathsf{par}_G(v)|+\sum_{v\in\overline{V}\setminus\Gamma_{\leq d^*-1}\setminus A}\log(|\mathsf{N}_A(v)|+Me^{-\overline{\beta}})+\overline{\beta}|A|\\
            &\quad-\overline{\beta}|\overline{V}\setminus\Gamma_{\leq d^*-2}|(d^*+1)+\overline{\beta}N_{d^*-1}\,.
        \end{split}
    \]
    Now by Proposition~\ref{prop:conc} and the definition on $d^*$, we have
    \[
        |\Gamma_{\leq d^*-1}|=O_p((nq)^{d^*-1})=O_p\left(\frac{n}{(\log\log n)^2}\right)
    \]
    which implies $|\Gamma_{\leq d^*-1}|\log M=o_p(n)$. Similarly, $\overline{\beta}N_{d^*-1}=o_p(n)$. Also, note that
    \begin{equation}\label{eqn:gibbs-res-bound}
        |\overline{V}\setminus\Gamma_{\leq d^*-2}|=n-O_p((nq)^{d^*-2})=n-O_p\left(\frac{n}{\log n(\log\log n)^2}\right)
    \end{equation}
    so that
    \[
        \overline{\beta}|\overline{V}\setminus\Gamma_{\leq d^*-1}|(d^*+1)=\overline{\beta}n(d^*+1)-o_p(n)\,.
    \]
    This gives the desired bound.
\end{proof}

Our goal here is to control the sum
\[
    \sum_{v\in}\log(\mathsf{N}_A(v)+Me^{-\overline{\beta}})\,.
\]
A key observation is that conditioned on $\Gamma_{\leq d^*-1}$ and $A$, we have that $|\mathsf{N}_A(v)|$ for $v\in V\setminus\Gamma_{\leq d^*-1}\setminus A$ are i.i.d. with distribution $\Binom(|A|, q)$. Thus, we hope to have an approximation
\[
    \sum_{v\in \overline{V}\setminus\Gamma_{\leq d^*-1}\setminus A}\log(Me^{-\overline{\beta}}+|\mathsf{N}_A(v)|)\approx |\overline{V}\setminus\Gamma_{\leq d^*-1}\setminus A|\E[\log(Me^{-\overline{\beta}}+\Binom(|A|, q))]
\]
which are accurate up to, hopefully, $o(n)$. We would like to apply Lemma~\ref{lem:logbin-subg} to $\log(Me^{-\overline{\beta}}+|\mathsf{N}_A(v)|)$. An issue is that if $\overline{\beta}$ is very large then $Me^{-\overline{\beta}}$ gets close to zero, which may cause its subgaussian norm to blow up. Fortunately, we know intuitively that the Gibbs measure prefers $A$ to be large enough to support $V\setminus\Gamma_{\leq d^*-1}\setminus A$, i.e., $\log|\mathsf{N}_A(v)|$ is likely to be of order $\Omega(\log\log n)$. Thus, our trick is to replace $Me^{-\overline{\beta}}$ with larger value when it is too small, large enough to achieve strong concentration but strictly smaller in order than $\log n$.

\begin{lemma}\label{lem:logsum-ub}
    Let $L=Me^{-\overline{\beta}}\vee \sqrt{\log n}$. Then with probability at least $1-o(1)$,
    \[
        \sum_{v\in \overline{V}\setminus\Gamma_{\leq d^*-1}\setminus A}\log(Me^{-\overline{\beta}}+|\mathsf{N}_A(v)|)\leq(n-|A|)\log(L+|A|q)+o(n)
    \]
    simultaneously for all $A\in\mathcal{K}(G)$.
\end{lemma}
\begin{proof}
    We first condition on $\Gamma_{\leq d^*-1}$ and $A$. We first slightly extend the sum to
    \[
        \sum_{v\in \overline{V}\setminus\Gamma_{\leq d^*-1}\setminus A}\log(Me^{-\overline{\beta}}+|\mathsf{N}_A(v)|)\leq\sum_{v\in V\setminus\Gamma_{\leq d^*-1}\setminus A}\log(L+|\mathsf{N}_A(v)|)
    \]
    so that we can say $|\mathsf{N}_A(v)|$ are (conditionally) i.i.d. binomial random variables. Letting $Y=L+|\mathsf{N}_A(v)|$, by Lemma~\ref{lem:logbin-subg} we have
    \[
        \lVert\log(L+|\mathsf{N}_A(v)|)-\log(L+|A|q)\rVert_{\psi_2}^2\leq\frac{6}{L}\leq\frac{6}{\sqrt{\log n}}\,.
    \]
    Hence, by Corollary~\ref{cor:tail-sum}, the probability that the sum $\sum\log(L+|\mathsf{N}_A(v)|)$ is away from its mean by at least $t$ is bounded by
    \[
        2\exp\left(-\frac{t^2\sqrt{\log n}}{432|V\setminus\Gamma_{\leq d^*-1}\setminus A|}\right)\leq2\exp\left(-\frac{t^2\sqrt{\log n}}{432n}\right)\,.
    \]
    By setting $t=n/\log\log n$, we get
    \[
        2\exp\left(-\frac{n\sqrt{\log n}}{432(\log\log n)^2}\right)\,.
    \]
    
    Now we take the union bound over all possible $A$, which gives
    \[
        \begin{split}
            \sum_{v\in V\setminus\Gamma_{\leq d^*-1}\setminus A}\log(L+|\mathsf{N}_A(v)|)&\leq|V\setminus\Gamma_{\leq d^*-1}\setminus A|\E[\log(L+|\mathsf{N}_A(v)|)\mid \Gamma_{\leq d^*-1},A]+\frac{n}{\log\log n}\\
            &\leq(n-|A|)\E[\log(L+|\mathsf{N}_A(v)|)\mid \Gamma_{\leq d^*-1},A]+\frac{n}{\log\log n}
        \end{split}
    \]
    for all $A\subseteq\mathcal{K}(G)$ with probability at least
    \[
        1-2\exp\left(-\frac{n\sqrt{\log n}}{432(\log\log n)^2}+n\right)=1-o(1)\,.
    \]
    Finally, Jensen's inequality gives
    \[
        \E[\log(L+|\mathsf{N}_A(v)|)\mid \Gamma_{\leq d^*-1},A]\leq\log(L+\E[|\mathsf{N}_A(v)|\mid\Gamma_{\leq d^*-1},A])=\log(L+|A|q)
    \]
    completing the proof.
\end{proof}

Combining Lemma~\ref{lem:free-core} with Lemma~\ref{lem:logsum-ub} gives the following result.

\begin{lemma}\label{lem:gibbs-cond-ub}
    Let $L=Me^{-\overline{\beta}}\vee\sqrt{\log n}$. Then with probability at least $1-o(1)$,
    \[
        \log(Z_{G,\beta}|_A)\leq\sum_{v\in A}\log|\mathsf{par}_G(v)|+(n-|A|)\log(L+|A|q)+\overline{\beta}|A|-\overline{\beta}(d^*+1)n+o(n)
    \]
    simultaneously for all $A\in\mathcal{K}(G)$.
\end{lemma}

This also immediately yields a bound for $Z_{G,\beta}\|_m$ by following an argument similar to the proof of Proposition~\ref{prop:fixed-kernel-size}. This is not needed for the proof of Proposition~\ref{prop:fixed-kernel} but is useful for our analysis in Section~\ref{sec:gibbs-w1}.

\begin{lemma}\label{lem:gibbs-cond-size-ub}
    Let $L=Me^{-\overline{\beta}}\vee\sqrt{\log n}$. Then with probability at least $1-o(1)$,
    \[
        \log(Z_{G,\beta}\|_m)\leq\mathcal{H}_G(m)+(n-m)\log(L+mq)+\overline{\beta}m-\overline{\beta}(d^*+1)n+o(n)
    \]
    simultaneously for all $m_0\leq m\leq N_{d^*}$.
\end{lemma}

\paragraph{A lower bound.} By the Gibbs variational principle, $\log(Z_{G,\beta}|_A)$ is lower bounded by
\begin{equation}\label{eqn:gibbs-pre-lb}
    \log(Z_{G,\beta}|_A)\geq H_{\mu_{G,\infty}|_A}(T)-\overline{\beta}\E_{\mu_{G,\infty}|_A}\left[\sum_{v}\mathsf{d}_T(1,v)\right]\,.
\end{equation}
Since $\mu_{G,\infty}|_A$ is a product measure, we have
\begin{equation}\label{eqn:gibbs-pre-entropy-lb}
    \begin{split}
        H_{\mu_{G,\infty}|_A}(T) &= \sum_{v\in\overline{V}\setminus\{1\}}H(\mathsf{par}_T(v))\\
        &\geq\sum_{v\in A}\log|\mathsf{par}_G(v)|+\sum_{v\in V\setminus\Gamma_{\leq d^*-1}\setminus A}\log(|\mathsf{N}_A(v)|\vee1)
    \end{split}
\end{equation}
The first thing we need is the concentration of
\begin{equation}\label{eqn:lbctrl1}
    \sum_{v\in V\setminus\Gamma_{d^*-1}\setminus A}\log(|\mathsf{N}_A(v)|\vee1)\,.
\end{equation}
simultaneously for all $A\subseteq\Gamma_{d^*}$.

\begin{lemma}\label{lem:logsum-lb}
    There is a universal constant $C>0$ such that with probability at least $1-o(1)$,
    \[
        \sum_{v\in V\setminus\Gamma_{\leq d^*-1}\setminus A}\log(|\mathsf{N}_A(v)|\vee1)\geq|V\setminus\Gamma_{\leq d^*-1}\setminus A|\log(|A|q)-\frac{Cn}{\sqrt{|A|q}}
    \]
    simultaneously for all $A\subseteq\Gamma_{d^*}$.
\end{lemma}
\begin{proof}
    Condition on $\Gamma_{d^*-1}$ and $A$. By Lemma~\ref{lem:logbinv-subg} we have
    \[
        \lVert\log(|\mathsf{N}_A(v)|\vee1)-\log(|A|q\vee1)\rVert_{\psi_2}\leq\frac{C_1}{\sqrt{|A|q}}
    \]
    for some constant $C_1>0$. By Corollary~\ref{cor:tail-sum} the probability that the sum $\sum\log(|\mathsf{N}_A(v)\vee1)$ is away from its mean by at least $t>0$ is bounded by
    \[
        2\exp\left(-\frac{2t^2|A|q}{C_2^2|V\setminus\Gamma_{\leq d^*-1}\setminus A|}\right)\leq2\exp\left(-\frac{2t^2|A|q}{C_2^2n}\right)
    \]
    for some constant $C_2$. By setting $t=C_2n/\sqrt{|A|q}$ the probability above is bounded by $2e^{-2n}$. Taking the union bound over all $A\subseteq\Gamma_{d^*}$ we get
    \[
        \sum_{v\in V\setminus\Gamma_{d^*-1}\setminus A}\log(|\mathsf{N}_A(v)|\vee1)\geq|V\setminus\Gamma_{d^*-1}\setminus A|\E[\log(|\mathsf{par}_A(v)|\vee1)\mid\Gamma_{\leq d^*-1},A]-\frac{C_3n}{\sqrt{|A|q}}
    \]
    simultaneously for all $A$, with probability at least $1-o(1)$. Now using Lemma~\ref{lem:subg-mean} we have
    \[
        \begin{split}
            \sum_{v\in V\setminus\Gamma_{d^*-1}\setminus A}\log(|\mathsf{N}_A(v)|\vee1)&\geq|V\setminus\Gamma_{d^*-1}\setminus A|\left(\log(|A|q\vee1)-\frac{C_1}{\sqrt{|A|q}}\right)-\frac{C_3n}{\sqrt{|A|q}}\\
            &\geq|V\setminus\Gamma_{d^*-1}\setminus A|\log(|A|q)-\frac{Cn}{\sqrt{|A|q}}
        \end{split}
    \]
    for some constant $C>0$, which is as desired.
\end{proof}

Applying Lemma~\ref{lem:logsum-lb} to \eqref{eqn:gibbs-pre-entropy-lb}, we get
\[
    H_{\mu_{G,\infty}|_A}(T)\geq\sum_{v\in A}|\mathsf{par}_G(v)|+|V\setminus\Gamma_{\leq d^*-1}\setminus A|\log(|A|q)-\frac{Cn}{\sqrt{|A|q}}\,.
\]
By Proposition~\ref{prop:conc} we have $|\Gamma_{\leq d^*-1}|=O_p((nq)^{d^*-1})=O_p(n/(\log\log n)^2)$ and clearly $\log(|A|q)=O(\log\log n)$, so
\begin{equation}\label{eqn:gibbs-entropy-lb}
    H_{\mu_{G,\infty}|_A}(T)\geq\sum_{v\in A}|\mathsf{par}_G(v)|+(n-|A|)\log(|A|q)-\frac{Cn}{\sqrt{|A|q}}-O_p\left(\frac{n}{\log\log n}\right)
\end{equation}
simultaneously for all $A\in\mathcal{K}(G)$.

Now we bound the energy term of \eqref{eqn:gibbs-pre-lb}. We have
\begin{equation}\label{eqn:energy-ub}
    \begin{split}
        \sum_v \mathsf{d}_T(1,v) &\leq d^*|A|+(d^*+1)(n-|A|)+\sum_{v\in\overline{V}\setminus\Gamma_{\leq d^*+1}^T}\mathsf{d}_G(v,A)\\
        &\leq d^*|A|+(d^*+1)(n-|A|)+|V\setminus\Gamma_{\leq d^*+1}^T|+|V\setminus\Gamma_{\leq d^*+2}^T|\operatorname{diam}(\overline{G}\setminus\Gamma_{\leq d^*-1})
    \end{split}
\end{equation}
where $\operatorname{diam}(\overline{G}\setminus\Gamma_{\leq d^*-1})$ is the graph diameter of $\overline{G}\setminus\Gamma_{\leq d^*-1}$. Since $G\setminus \Gamma_{\leq d^*-1}$ is an Erd\"os--R\'enyi graph conditioned on $\Gamma_{\leq d^*-1}$, Theorem~\ref{thm:diameter} implies that
\begin{equation}\label{eqn:dbound}
    \operatorname{diam}(\overline{G}\setminus\Gamma_{\leq d^*-1})=O\left(\frac{\log n}{\log\log n}\right)
\end{equation}
asymptotically almost surely. Now we control $|V\setminus\Gamma_{\leq d^*+1}^T|$ and $|V\setminus\Gamma_{\leq d^*+2}^T|$ given that $|A|$ is reasonably large. Note that by the definition of $\mu_{G,\infty}|_A$, $\Gamma_{\leq d^*+1}^T$ and $\Gamma_{\leq d^*+2}^T$ are fixed for any $T\in\operatorname{supp}(\mu_{G,\infty}|_A)$ as long as $A$ is fixed. As such, we abuse notation and write $\Gamma_{\leq d^*+1}^A$ and $\Gamma_{\leq d^*+2}^A$ to emphasize that they are functions of $A$ (and $G$). Also, recall the definition of $m_\ell$ in \eqref{eqn:ml-def}.

\begin{lemma}\label{lem:gamma-a}
    With probability at least $1-o(1)$, we have
    \[
        |V\setminus\Gamma_{\leq d^*+1}^A|= o\left(\frac{n}{\log\log n}\right)
    \]
    and 
    \[
        |V\setminus\Gamma_{\leq d^*+2}^A|=o\left(\frac{n}{\log n}\right)
    \]
    simultaneously for all $A\in\mathcal{K}(G)$ with $|A|\geq m_\ell$.
\end{lemma}
\begin{proof}
    It suffices to prove that with high probability the bound simultaneously holds for all $A$ of the same size $|A|=\lceil\frac{n}{3(\log\log n)^2}\rceil$. Conditioned on $\Gamma_{d^*-1}$ and $A$, we have
    \[
        |V\setminus\Gamma_{\leq d^*+1}^A|=\sum_{v\in V\setminus\Gamma_{\leq d^*-1}\setminus A}\mathbf{1}\{\mathsf{N}_A(v)=\emptyset\}\sim\Binom(|V\setminus\Gamma_{d^*-1}\setminus A|,(1-q)^{|A|})\,.
    \]
    This is stochastically dominated by $\Binom(n, e^{-q|A|})$. For simplicity write $\mu=ne^{-q|A|}$. Using Lemma~\ref{lem:raw-chernoff}, we have
    \[
        \Pr\left(|V\setminus\Gamma_{\leq d^*+1}^A|\geq \mu+t\right) \leq\exp\left(t-(\mu+t)\log\left(1+\frac{t}{\mu}\right)\right)\,.
    \]
    Setting $\mu+t=\frac{n}{(\log\log n)^2}$ which is greater than $\mu$ for large enough $n$,
    \[
        \begin{split}
            (\mu+t)\log\left(1+\frac{t}{\mu}\right) &= \frac{n}{(\log\log n)^2}\log\left(\frac{e^{q|A|}}{(\log\log n)^2}\right)\\
            &=\frac{\alpha_n|A|\log n}{(\log\log n)^2}-\frac{2n\log\log\log n}{(\log\log n)^2}\\
            &\geq\frac{\alpha_nn\log n}{3(\log\log n)^{4}}-o(n)\,.
        \end{split}
    \]
    This gives
    \begin{equation}\label{eqn:d1tail}
        \Pr\left(|V\setminus\Gamma_{\leq d^*+1}^A|\geq \mu+t\right)\leq\exp\left(-\frac{\alpha_nn\log n}{3(\log\log n)^{4}}+o(n)\right)\,.
    \end{equation}
    For $|V\setminus\Gamma_{\leq d^*+2}^A|$, we proceed similarly. Under the good event $|V\setminus\Gamma_{\leq d^*+1}^A|<\mu+t$, $\Gamma_{d^*+1}^A$ has size $n-o(n)$, so we may assume $N_{d^*+1}^A\geq n/2$. Thus, $|V\setminus\Gamma_{\leq d^*+2}^A|$ is stochastically dominated by $\Binom(n, e^{-nq/2})$. Applying the same Chernoff bound, this time with $\mu+t=\frac{n}{(\log n)(\log\log n)}$, it is easy to see
    \begin{equation}\label{eqn:d2tail}
        \Pr\left(|V\setminus\Gamma_{\leq d^*+2}^A|\geq\frac{n}{(\log n)(\log\log n)}\right)\leq\exp\left(-\Omega\left(\frac{n}{\log\log n}\right)\right)\,.
    \end{equation}
    
    Now we take the union bound over all $A$. Note that
    \begin{equation}\label{eqn:nchooseml}
        \begin{split}
            \binom{N_{d^*}}{\lceil\frac{n}{3(\log\log n)^2}\rceil}&\leq\binom{n}{\lceil\frac{n}{3(\log\log n)^2}\rceil}\\
            &\leq\exp\left(nH\left(\frac{1}{(\log\log n)^2}\right)\right)\\
            &\leq\exp\left(\frac{4n\log\log\log n}{(\log\log n)^2}\right)
        \end{split}
    \end{equation}
    where for the last inequality we used the fact that $H(p)\leq-2p\log p$ for all $0<p<1/2$. Hence, taking the union bound gives $o(1)$ bounds for the tail probabilities \eqref{eqn:d1tail} and \eqref{eqn:d2tail}.
\end{proof}

Now we have a bound
\begin{equation}\label{eqn:gibbs-energy-ub}
    \sum_v\mathsf{d}_T(1,v)\leq(d^*+1)n-|A|+o_p\left(\frac{n}{\log\log n}\right)
\end{equation}
for the energy term \eqref{eqn:energy-ub}, simultaneously for all $A$ with $|A|\geq m_\ell$. Combining with \eqref{eqn:gibbs-entropy-lb} yields the following result.

\begin{lemma}\label{lem:gibbs-cond-lb}
    With probability at least $1-o(1)$, we have
    \[
        \log(Z_{G,\beta}|_A)\geq\sum_{v\in A}\log|\mathsf{par}_G(v)|+(n-|A|)\log(|A|q)+\overline{\beta}|A|-\overline{\beta}(d^*+1)n-o(n)
    \]
    simultaneously for all $A\in\mathcal{K}(G)$ with $|A|\geq m_{\ell}$.
\end{lemma}

\paragraph{Putting the bounds together.} Lemma~\ref{lem:gibbs-cond-ub} and Lemma~\ref{lem:gibbs-cond-lb} are more than sufficient to prove Proposition~\ref{prop:fixed-kernel}. For $L=Me^{-\overline{\beta}}\vee\sqrt{\log n}$ in Lemma~\ref{lem:gibbs-cond-ub}, we have
\[
    \log(L+|A|q)=\log(|A|q)+\log\left(1+\frac{L}{|A|q}\right)\leq\log(|A|q)+\frac{L}{|A|q}\,.
\]
Since $|A|\geq m_\ell$, we have
\[
    |A|q\geq\frac{nq}{3(\log\log n)^2}=\Omega\left(\frac{\log n}{(\log\log n)^2}\right)
\]
while, since $M=O_p(\log n)$,
\[
    L=O_p((\log n)^{(1-\beta)\vee\frac{1}{2}})\,.
\]
Hence, we have
\[
    \frac{L}{|A|q}=o_p(n)\,.
\]
Now applying \eqref{eqn:kl-variational}, Lemma~\ref{lem:gibbs-cond-ub}, and Lemma~\ref{lem:gibbs-cond-lb} proves Proposition~\ref{prop:fixed-kernel}.

\subsubsection{Concentration of the log partition function}\label{sec:hg-conc}

The formulae in Proposition~\ref{prop:fixed-kernel} and Proposition~\ref{prop:fixed-kernel-size} involve $\mathcal{H}_G(m)$ which heavily depends on $G$. We prove that this concentrates uniformly over all $m_0\leq m\leq N_{d^*}$, which allows us to replace it with $\Psi(m;N_{d^*},\lambda_n)$ which depends on $G$ only through $N_{d^*}$. The goal of this section is to prove the following lemma.

\begin{lemma}\label{lem:hg-conc}
    With probability at least $1-o(1)$, we have
    \[
        \mathcal{H}_G(m)=\Psi(m;N_{d^*},\lambda_n)+o(n)
    \]
    simultaneously for all $m_0\leq m\leq N_{d^*}$.
\end{lemma}

Revisit \eqref{eqn:hg-def} and \eqref{eqn:psi} for definitions. We repeat the definition of $\mathcal{H}_G(m)$:
\[
    \mathcal{H}_G(m):=\log\left(\sum_{A\in\mathcal{K}_m(G)}\prod_{v\in A}|\mathsf{par}_G(v)|\right)\,.
\]
Here, $A$ ranges over all kernels in $\mathcal{K}_m(G)$, but the collection $\mathcal{K}_m(G)$ is a bit unwieldy. Our first task is to replace this with the collection of all $m$-elements subsets which we temporarily write
\[
    \tilde{\mathcal{H}}_{G}(m):=\log\left(\sum_{A\in\binom{\Gamma_{d^*}}{m}}\prod_{v\in A}|\mathsf{par}_G(v)|\right)\,.
\]

\begin{lemma}\label{lem:kl-inter-low}
    With probability at least $1-o(1)$, we have
    \[
        \tilde{\mathcal{H}}_{G}(m) = \mathcal{H}_G(m)+o(n)
    \]
    simultaneously for all $m_0\leq m\leq N_{d^*}$.
\end{lemma}
\begin{proof}
    We let $J\subseteq\overline{V}\setminus\Gamma_{\leq d^*-1}$ be the set of vertices outside the giant component of $\overline{G}\setminus\Gamma_{\leq d^*-1}$, which has cardinality $|J|=n^{1-\Theta(1)}$ with high probability. Also, let $M$ be the maximum degree of $G$.

    If $m\leq|J|=n^{1-\Theta(1)}$, then $\tilde{\mathcal{H}}_G(m)=o_p(n)$ so the result is trivial. Otherwise, if $m>|J|$, then any subset of $\Gamma_{d^*}$ containing $J$ is a kernel, so we have
    \[
        \exp(\tilde{\mathcal{H}}_G(m)-\mathcal{H}_G(m))\leq\frac{\sum_A\prod_{v\in A}|\mathsf{par}_G(v)|}{\sum_{A\supseteq J}\prod_{v\in A}|\mathsf{par}_G(v)|}\,.
    \]
    For the denominator, we use a simple bound $|\mathsf{par}_G(v)|\geq1$ to obtain
    \[
        \sum_{\substack{J\subseteq A\subseteq\Gamma_{d^*}\\|A|=m}}\prod_{v\in A}|\mathsf{par}_G(v)|\geq\sum_{\substack{B\subseteq\Gamma_{d^*}\setminus J\\|B|=m-|J|}}\prod_{v\in B}|\mathsf{par}_G(v)|\,.
    \]
    Now we handle the numerator. The idea is to factor out $(m-|J|)$ sized subsets of $A$ to cancel the numerator, and bound the remaining terms by $M$. For this, we bound each product as
    \[
        \begin{split}
            \prod_{v\in A}|\mathsf{par}_G(v)| &\leq M^{|A\cap J|}\prod_{v\in A\setminus J}|\mathsf{par}_G(v)|\\
            &\leq M^{|A\cap J|}\left(\frac{1}{\binom{|A\setminus J|}{m-|J|}}\sum_{\substack{B\subseteq A\setminus J\\|B|=m-|J|}}\prod_{v\in B}|\mathsf{par}_G(v)|\cdot M^{|J|-|A\cap J|}\right)\\
            &=\frac{M^{|J|}}{\binom{|A\setminus J|}{m-|J|}}\sum_{\substack{B\subseteq A\setminus J\\|B|=m-|J|}}\prod_{v\in B}|\mathsf{par}_G(v)|\,.
        \end{split}
    \]
    Now we take the sum over all $A\subseteq\Gamma_{d^*}$. By the double counting argument the multiplicative factor of each $B\subseteq\Gamma_{d^*}\setminus J$ with size $|B|=m-|J|$ is
    \[
        M^{|J|}\sum_{k=0}^{|J|}\frac{\binom{|J|}{k}\binom{n-m}{|J|-k}}{\binom{m-k}{m-|J|}}
    \]
    where $k$ encodes $|A\cap J|$. Combining with the numerator bound, we have
    \[
        \exp(\tilde{\mathcal{H}}_G(m)-\mathcal{H}_G(m))\leq M^{|J|}\sum_{k=0}^{|J|}\frac{\binom{|J|}{k}\binom{n-m}{|J|-k}}{\binom{m-k}{m-|J|}}\leq M^{|J|}|J|n^{|J|}\,.
    \]
    Since $|J|=n^{1-\Theta(1)}$ the result follows.
\end{proof}

Next, we show that $\tilde{\mathcal{H}}_G(m)$ concentrates around $\Psi(m;N_{d^*},\lambda_n)$. Throughout this part, we treat $\Gamma_{\leq d^*-1}$ as fixed unless otherwise stated. We first construct random variables that have the same distribution as $|\mathsf{par}_G(v)|$ and set up notations to better describe $\tilde{\mathcal{H}}_G(m)$ in terms of these random variables. Let $n^*=|V\setminus\Gamma_{\leq d^*-1}|$ and consider i.i.d. random variables
\[
    X_1,\cdots,X_{n^*}\sim\Binom(N_{d^*-1},q)\,.
\]
We also define an index set $I_{\mathbf{X}>0}$ as a function of $\mathbf{X}=(X_1,\cdots,X_{n^*})$ by
\[
    I_{\mathbf{X}>0}:=\{i\in[n^*]:X_i>0\}
\]
and its size
\[
    N_{\mathbf{X}>0}:=|I_{\mathbf{X}>0}|\,.
\]
Note that
\[
    (|\mathsf{par}_G(v)|)_{v\in\Gamma_{d^*}}\overset{d}=(X_i)_{i\in I_{\mathbf{X}>0}}\,.
\]
We also define $Y_i=\log(X_i\vee1)$. For fixed $N_{d^*}=N_{\mathbf{X}>0}$, we have
\begin{equation}\label{eqn:hg-ztb}
    \begin{split}
        \tilde{\mathcal{H}}_G(m)&\overset{d}=\log\left(\sum_{J\in\binom{I_{\mathbf{X}>0}}{m}}\prod_{i\in J}e^{Y_i}\right) \\
        &=\log\left(\sum_{J\in\binom{I_{\mathbf{X}>0}}{m}}\exp\left(\sum_{i\in J}Y_i\right)\right)\\
        &=\mathsf{LSE}_m(\mathbf{Y}_{\mathbf{X}>0})
    \end{split}
\end{equation}
where $\mathsf{LSE}_m$ is the LogSumExp function applied over subset sums of $m$ elements, and $\mathbf{Y}_{\mathbf{X}>0}:=(Y_i){i\in I_{\mathbf{X}>0}}$. Note that \eqref{eqn:hg-ztb} holds jointly, i.e.,
\begin{equation}\label{eqn:hg-joint}
    \left(\tilde{\mathcal{H}}_G(m)\right)_{m=1}^{N_{d^*}}\overset{d}=\left(\mathsf{LSE}_m\left(\mathbf{Y}_{\mathbf{X}>0}\right)\right)_{m=1}^{N_{\mathbf{X}>0}}
\end{equation}
where $\mathbf{Y}_{\mathbf{X}>0}:=(Y_i)_{i\in I_{\mathbf{X}>0}}$. Necessary properties of the $m$-wise LogSumExp function can be found in Appendix~\ref{apdx:lse}.

For the sake of convenience, we introduce several shorthand notations. First, we define
\[
    \phi_{n,p}:=\operatorname{law}(\log(X\vee1))\,,\quad X\sim\Binom(n,p)
\]
which will be used to describe the law of $Y_i$. Also, we define
\[
    \pi_{\lambda}:=\operatorname{law}(\log(X\vee1))\,,\quad X\sim\Pois(\lambda)
\]
and
\[
    \tilde{\pi}_{p,\lambda}:=\operatorname{law}(\log(X\vee1))\,,\quad X\sim p\delta_{0}+(1-p)\ZTP(\lambda)
\]
which are useful for $\Psi(m;N_{d^*},\lambda_n)$. Here, $\tilde{\pi}_{p,\lambda}$ can be understood as a ``patched'' version of $\pi_{\lambda}$, where the mass at zero has been adjusted. In particular, $\pi_\lambda=\tilde{\pi}_{p,\lambda}$ if $p=e^{-\lambda}$.


By Lemma~\ref{lem:subg-ztp}, $\tilde{\pi}_{0,\lambda}$ has a bounded subgaussian norm regardless of $\lambda$, so Lemma~\ref{lem:lse2} implies that there is a universal constant $C$ such that
\begin{equation}\label{eqn:lse-goal}
    \begin{split}
        |\mathsf{LSE}_m(\mathbf{Y}_{\mathbf{X}>0})-\Psi(m;N_{\mathbf{X}>0},\lambda_n)|&\leq N_{\mathbf{X}>0}W_1(\tilde{\pi}_{0,\log(1/\lambda_n)},\mu_{\mathbf{Y}_{\mathbf{X}>0}})+C\sqrt{N_{\mathbf{X}>0}}\\
        &=n^*W_1(\tilde{\pi}_{p_{\mathbf{X}=0},\log(1/\lambda_n)},\mu_{\mathbf{Y}})+C\sqrt{N_{\mathbf{X}>0}}
    \end{split}
\end{equation}
where
\[
    p_{\mathbf{X}=0}:=\frac{N_{\mathbf{X}=0}}{n^*}=1-\frac{N_{\mathbf{X}>0}}{n^*}\,.
\]
We can decompose the RHS as
\begin{equation}\label{eqn:lse-interm}
    n^*W_1(\tilde{\pi}_{p_{\mathbf{X}=0},\log(1/\lambda_n)},\mu_{\mathbf{Y}})=\underbrace{n^*W_1(\tilde{\pi}_{p_{\mathbf{X}=0},\log(1/\lambda_n)},\pi_{\log(1/\lambda_n)})}_{\text{(A)}}+n^*W_1(\pi_{\log(1/\lambda_n)},\mu_{\mathbf{Y}})\,.
\end{equation}
The second term is further decomposed into the following.
\[
    \begin{split}
        n^*W_1(\pi_{\log(1/\lambda_n)},\mu_{\mathbf{Y}}) &\leq \underbrace{n^*W_1(\pi_{\log(1/\lambda_n)},\pi_{qN_{d^*-1}})}_{\text{(B)}}+\underbrace{n^*W_1(\pi_{qN_{d^*-1}},\phi_{N_{d^*-1},q})}_{\text{(C)}}+\underbrace{n^*W_1(\phi_{N_{d^*-1},q},\mu_{\mathbf{Y}})}_{\text{(D)}}\,.
    \end{split}
\]
There are four terms we need to bound. We proceed in reverse order. For (D), Lemma~\ref{lem:logbinv-subg} and Lemma~\ref{lem:empirical-conv} suffice:
\[
    \text{(D)}\leq O_p(\sqrt{n^*})=O_p(\sqrt{n})\,.
\]
For (C) we use Theorem~\ref{thm:w1-pois} and the fact that $x\mapsto\log(x\vee1)$ is $1$-Lipschitz:
\[
    \begin{split}
        \text{(C)} &\leq n^*W_1(\Binom(N_{d^*-1},q),\Pois(qN_{d^*-1}))\\
        &=O\left(n^*\sqrt{N_{d^*-1}}q^{3/2}\right)\\
        &=O\left((nq)^{3/2}\right)\\
        &=O((\log n)^{3/2})\,.
    \end{split}
\]
To deal with (B) we use the following lemma.

\begin{lemma}
    For $0<\lambda_1<\lambda_2$ we have
    \[
        W_1(\pi_{\lambda_1},\pi_{\lambda_2})\leq\frac{2(\lambda_2-\lambda_1)}{\lambda_1}\cdot(1-e^{-\lambda_1})\,.
    \]
\end{lemma}
\begin{proof}
    For $X\sim\Pois(\lambda_1)$ and $Y\sim\Pois(\lambda_2)$, we consider $Z\sim\Pois(\lambda_2-\lambda_1)$ independent of $X$ and construct a coupling of $X$ and $Y$ by $Y=X+Z$. Then
    \[
        \begin{split}
        W_1(\pi_{\lambda_1},\pi_{\lambda_2}) &\leq \E[\log(Y\vee1)-\log(X\vee1)]\\
        &=\E\left[\log\left(\frac{(X+Z)\vee1}{X\vee1}\right)\right]\\
        &\leq\E\left[\log\left(1+\frac{Z}{X\vee1}\right)\right]\\
        &\leq\E\left[\frac{Z}{X\vee1}\right]\\
        &\leq2(\lambda_2-\lambda_1)\E\left[\frac{1}{X+1}\right]\\
        &=\frac{2(\lambda_2-\lambda_1)}{\lambda_1}\cdot(1-e^{-\lambda_1})\,.
    \end{split}
    \]
\end{proof}

Using this we have
\[
    \begin{split}
       \text{(B)} &\leq n^*\cdot\frac{2|qN_{d^*-1}-\log(1/\lambda_n)|}{qN_{d^*-1}\wedge\log(1/\lambda_n)}\\
        &=n^*\cdot\frac{2|N_{d^*-1}-(nq)^{d^*-1}|}{N_{d^*-1}\wedge(nq)^{d^*-1}}\\
        &=O_p\left(\frac{n}{\log\log n}\right)
    \end{split}
\]
where the last line is due to Proposition~\ref{prop:conc}. Now we are left with (A). We first establish the following lemma.

\begin{lemma}\label{lem:patched-ztp}
    We have
    \[
        W_1(p\delta_0+(1-p)\ZTP(\lambda),\Pois(\lambda))=\frac{\lambda|p-e^{-\lambda}|}{1-e^{-\lambda}}\,.
    \]
\end{lemma}
\begin{proof}
    Let $X\sim p\delta_0+(1-p)\ZTP(\lambda)$ and $Y\sim\Pois(\lambda)$. For an integer $x\geq0$, we have
    \[
        \begin{split}
            F_X(x):=\Pr(X\leq x) &= p+\sum_{i=1}^x\Pr(X=x)\\
            &=p+\frac{1-p}{1-e^{-\lambda}}\sum_{i=1}^x\Pr(Y=x)\\
            &=p+\frac{1-p}{1-e^{-\lambda}}(F_Y(x)-e^{-\lambda})\\
            &=\frac{1-p}{1-e^{-\lambda}}F_Y(x)+\frac{p-e^{-\lambda}}{1-e^{-\lambda}}
        \end{split}
    \]
    This gives
    \[
        \begin{split}
            W_1(p\delta_0+(1-p)\ZTP(\lambda),\Pois(\lambda)) &= \sum_{x=0}^\infty|F_X(x)-F_Y(x)|\\
            &=\sum_{x=0}^\infty\frac{|p-e^{-\lambda}|}{1-e^{-\lambda}}(1-F_Y(x))\\
            &=\frac{|p-e^{-\lambda}|}{1-e^{-\lambda}}\E[Y]\\
            &=\frac{\lambda|p-e^{-\lambda}|}{1-e^{-\lambda}}\,.
        \end{split}
    \]
\end{proof}

Since $x\mapsto\log(x\vee1)$ is $1$-Lipschitz, Lemma~\ref{lem:patched-ztp} implies
\[
    \begin{split}
        \text{(A)} &\leq n^*\cdot\frac{\log(1/\lambda_n)}{1-\lambda_n}|p_{\mathbf{X}=0}-\lambda_n|\\
        &\leq n^*\cdot\frac{\log(1/\lambda_n)}{1-\lambda_n}\left(\left|p_{\mathbf{X}=0}-(1-q)^{N_{d^*-1}}\right|+\left|(1-q)^{N_{d^*-1}}-e^{-qN_{d^*-1}}\right|+\left|e^{-qN_{d^*-1}}-\lambda_n\right|\right)\\
        &=\underbrace{\frac{\log(1/\lambda_n)}{1-\lambda_n}\left|N_{\mathbf{X}=0}-n^*(1-q)^{N_{d^*-1}}\right|}_{\text{(A1)}}+\underbrace{\frac{n^*\log(1/\lambda_n)}{1-\lambda_n}\left|(1-q)^{N_{d^*-1}}-e^{-qN_{d^*-1}}\right|}_{\text{(A2)}}\\
        &\quad+\underbrace{\frac{n^*\log(1/\lambda_n)}{1-\lambda_n}\left|e^{-qN_{d^*-1}}-\lambda_n\right|}_{\text{(A3)}}\,.
    \end{split}
\]
This again consists of three terms which we bound separately. Before that, we note that
\[
    \frac{\log(1/\lambda_n)}{1-\lambda_n}=O\left(\frac{\log n}{(\log\log n)^2}\right)
\]
which can be derived from \eqref{eqn:lambda-sandwich}. For (A1), recall that
\[
    N_{\mathbf{X}=0}\sim\Binom(n^*,(1-q)^{N_{d^*-1}})
\]
which gives
\[
    \text{(A1)}=\frac{\log(1/\lambda_n)}{1-\lambda_n}\cdot O(\sqrt{n})=O\left(\frac{\sqrt{n}\log n}{(\log\log n)^2}\right)\,.
\]
For (A2), we utilize an inequality
\[
    e^{-np}-(1-p)^n\leq\frac{1}{2}np^2
\]
which holds for all $n\geq0$ and $0\leq p\leq 1$. This can be verified by taking derivative of $e^{-np}-(1-p)^n-\frac{1}{2}np^2$ by $p$ and showing that it is decreasing. We have
\[
    \begin{split}
        \text{(A2)} &\leq\frac{n^*\log(1/\lambda_n)}{1-\lambda_n}\cdot \frac{1}{2}N_{d^*-1}q^2\\
        &\leq\frac{1}{2}\cdot\frac{\log(1/\lambda_n)}{1-\lambda_n}\cdot (nq)^2\\
        &=O\left(\frac{(\log n)^3}{(\log\log n)^2}\right)\,.
    \end{split}
\]
Finally, for (A3), we use
\[
    \frac{e^x-1}{x}\geq e^{x/2}
\]
which holds for all $x>0$. Thus,
\[
    \frac{\log(1/\lambda_n)}{1-\lambda_n}\geq\lambda_n^{-1/2}\,.
\]
This gives
\[
    \begin{split}
        \text{(A3)} &\leq n^*\left|\lambda_n^{-1/2}e^{-qN_{d^*-1}}-\lambda_n^{1/2}\right|\\
        &=n^*\left|\lambda_n^{\frac{qN_{d^*-1}}{\log(1/\lambda_n)}-\frac{1}{2}}-\lambda_n^{1/2}\right|\,.
    \end{split}
\]
To further upper bound this, we note that for $t>0$ and $0\leq x\leq 1$ the function $|x^{t-1/2}-x^{1/2}|$ attains its maximum at $x^*=(2t-1)^{-1/(t-1)}$ (other than the trivial edge case $t=1$) which implies
\[
    |x^{t-1/2}-x^{1/2}|\leq\frac{2|t-1|}{2t-1}(x^*)^{1/2}\leq\frac{2|t-1|}{2t-1}\,.
\]
In our case,
\[
    \frac{qN_{d^*-1}}{\log(1/\lambda_n)}=\frac{N_{d^*-1}}{(nq)^{d^*-1}}=1\pm O_p\left(\frac{1}{\log\log n}\right)
\]
by Proposition~\ref{prop:conc}, which gives
\[
    \text{(A3)}=O_p\left(\frac{n}{\log\log n}\right)\,.
\]
As a result, we have proved that \eqref{eqn:lse-goal} is $o_p(n)$ simultaneously for all $m$. Going back to \eqref{eqn:hg-joint}, this is equivalent to the following result.

\begin{lemma}\label{lem:hg-tilde-conc}
    With probability at least $1-o(1)$, we have
    \[
        \tilde{\mathcal{H}}_G(m)=\Psi(m;N_{d^*},\lambda_n)+o(n)
    \]
    simultaneously for all $1\leq m\leq N_{d^*}$.
\end{lemma}

Combining this with Lemma~\ref{lem:kl-inter-low} completes the proof of Lemma~\ref{lem:hg-conc}. Note that Lemma~\ref{lem:hg-conc} also yields the following version of Proposition~\ref{prop:fixed-kernel-size}.

\begin{corollary}\label{cor:fixed-kernel-size}
    Let $\beta>0$ be a constant. Then with probability at least $1-o(1)$, we have
    \[
        \log(Z_{G,\beta}\|_m)=\tilde{\Phi}_{G,\beta}(m)+o(n)
    \]
    simultaneously for all $m_\ell\leq m\leq N_{d^*}$.
\end{corollary}

\subsubsection{Proof of Theorem~\ref{thm:main-gibbs}}\label{sec:gibbs-1d-opt}

We show that the measures \eqref{eqn:hatmu} (for $m<m_0$, set them arbitrary) satisfy the desired properties. Recall that by Lemma~\ref{lem:gibbs-cond-size-ub} we have an upper bound of $\log(Z_{G,\beta}\|_m)$ for $m_0\leq m\leq N_{d^*}$:
\begin{equation}\label{eqn:gibbs-ub-tmp}
    \Psi(m;N_{d^*},\lambda_n)+(n-m)\log(L+mq)+\overline{\beta}m\,.
\end{equation}
We claim that it suffices to prove that:
\begin{enumerate}[label=(\alph*)]
    \item\label{item:claim1} $\Psi(m;N_{d^*},\lambda_n)$ is an increasing function of $m$ for $m\leq m_\ell$ with high probability,
    \item\label{item:claim2} $(n-m)\log(L+mq)$ is also an increasing function of $m$ for $m\leq m_\ell$  with high probability, and
    \item\label{item:claim3} $f(m)=(n-m)\log(mq)$ is also an increasing function of $m$ for $m\leq m_\ell$ with high probability and there exists an integer $m^\#\in[m_\ell,N_{d^*}]$ as a function of $n$ such that $f(m^{\#})-f(m_\ell)=\Omega_p(n)$.
\end{enumerate}
If we have these, then we get the following facts:
\begin{itemize}
    \item By \ref{item:claim1} and \ref{item:claim2}, the maximum value of \eqref{eqn:gibbs-ub-tmp} is attained at $m\geq m_\ell$ with high probability, at which the bound is tight by Corollary~\ref{cor:fixed-kernel-size}. This implies that
    \begin{equation}\label{eqn:gibbs-max-ml}
        \log Z_{G,\beta}=\max_{m_\ell\leq m\leq N_{d^*}}\log(Z_{G,\beta}\|_m)+o_p(n)\,.
    \end{equation}
    \item By \ref{item:claim1} and \ref{item:claim3}, any $m^*$ satisfying the assumption of the theorem has $m^*\geq m_\ell$ with high probability, in which case the objective \eqref{eqn:gibbs-pts} agrees with \eqref{eqn:gibbs-ub-tmp} ignoring $o_p(n)$ terms. Thus, $m^*$ is an approximate solution to the maximization in \eqref{eqn:gibbs-max-ml}, yielding
    \[
        \log Z_{G,\beta}=\log (Z_{G,\beta}\|_{m^*})
    \]
    which is enough to prove the theorem, similar to Proposition~\ref{prop:fixed-kernel-size}.
\end{itemize}

Note that we may assume
\[
    m_\ell=\left\lceil\frac{n}{3(\log\log n)^2}\right\rceil
\]
since $m_0$ is by definition the number of connected components of $\overline{G}\setminus\Gamma_{\leq d^*-1}$ which is $n^{1-\Theta(1)}$ with high probability. We first prove \ref{item:claim1} that $\Psi(m;N_{d^*},\lambda_n)$ is non-decreasing if $m\leq m_\ell$. By Lemma~\ref{lem:dstar} and the definition of $d^*$, $N_{d^*}$ concentrates around $(nq)^{d^*}\geq\frac{n}{(\log\log n)^2}$ when it is the smallest. Thus, we have that
\[
    m_{\ell}=\left\lceil\frac{n}{3(\log\log n)^2}\right\rceil<\frac{N_{d^*}}{2}
\]
asymptotically almost surely. As a consequence, for any $m<m_{\ell}$ and $x_1,\cdots,x_{N_{d^*}}\geq1$, we have
\[
    \begin{split}
        \sum_{I\in\binom{[N_{d^*}]}{m}}\prod_{i\in I}x_i &= \frac{1}{N_{d^*}-m}\sum_{I\in\binom{[N_{d^*}]}{m+1}}\sum_{J\in\binom{I}{m}}\prod_{i\in J}x_i\\
        &\leq\frac{m+1}{N_{d^*}-m}\sum_{I\in\binom{[N_{d^*}]}{m+1}}\prod_{i\in I}x_i
    \end{split}
\]
which implies that $\Psi(m;N_{d^*},\lambda_n)\leq\Psi(m+1;N_{d^*},\lambda_n)$ simultaneously for all $m<m_{\ell}$, asymptotically almost surely.

Now we prove \ref{item:claim2}. To see that $g(x)=(n-x)\log(L+qx)$ is non-decreasing, observe that
\[
    g'(x)=-\log(L+qx)+\frac{q(n-x)}{L+qx}\,.
\]
Using that $L+qx\leq L+m_\ell q=O_p(\log n/(\log\log n)^2)$, we get
\[
    \log(L+qx)=O_p(\log\log n)
\]
and
\[
    \frac{q(n-x)}{L+qx}=\Omega_p((\log\log n)^2)
\]
which imply $g'(x)>0$ uniformly over all $x<m_\ell$ for sufficiently large $n$, with high probability.

The first part of \ref{item:claim3} is similar to \ref{item:claim2}. For the second part, we may take
\[
    m^\#=\left\lceil\frac{n}{2(\log\log n)^2}\right\rceil\wedge N_{d^*}\,.
\]
By the same argument with the previous paragraph, we may assume $m^\#=\lceil\frac{n}{2(\log\log n)^2}\rceil$. Then it is straightforward to see that
\[
    f(m^\#)-f(m_\ell)\geq\log(3/2)n-o_p(n)
\]
as desired. This concludes the proof of Theorem~\ref{thm:main-gibbs}.

\subsubsection{Proof of Theorem~\ref{thm:logz-formula}}\label{sec:gibbs-opt-sol}

In this section, we prove Theorem~\ref{thm:logz-formula} by solving the optimization \eqref{eqn:gibbs-opt}, which we re-state here for the reader's convenience:
\begin{equation}\label{eqn:gibbs-opt-restate}
    \tilde{\Phi}_{G,\beta}^*:=\max_{m_0\leq m\leq N_{d^*}}\left\{\Psi(m;N_{d^*},\lambda_n)+(n-m)\log(mq)+\overline{\beta} m\right\}-\overline{\beta}(d^*+1)n\,.
\end{equation}
The first thing we need to handle is the mysterious function $\Psi(m;N_{d^*},\lambda_n)$. Conditioned on $N_{d^*}$, consider i.i.d. random variables $X_1,\cdots,X_{N_{d^*}}\sim\ZTP(\log(1/\lambda_n))$. Recall that by definition
\[
    \begin{split}
        \Psi(m;N_{d^*},\lambda_n) &= \E\left[\mathsf{LSE}_m(\log X_1,\cdots,\log X_{N_{d^*}})\,\middle|\,N_{d^*}\right]\\
        &=\E\left[\log\left(\sum_{I\in\binom{[N_{d^*}]}{m}}\exp\left(\sum_{i\in I}\log X_i\right)\right)\,\middle|\,N_{d^*}\right]\,.
    \end{split}
\]
If the sum $\sum\log X_i$ enjoys strong concentration, we might be able to replace each of these with $m\Psi_1(\lambda_n)$ where we define
\[
    \Psi_1(\lambda_n):=\E[\log X_1]\,,\quad X_1\sim\ZTP(\log(1/\lambda_n))\,.
\]
This might introduce a large error (up to $O(n)$) if $\log(1/\lambda_n)=\Theta(1)$ and $m=\Theta(n)$. In other cases, namely if $m$ is either small or close to $N_{d^*}$, it is possible to obtain more explicit results by concentration of measure.


\begin{lemma}\label{lem:parconc}
    There is a universal constant $C>0$ such that with probability at least $1-o(1)$,
    \[
        \left|\sum_{i\in A}\log X_i-|A|\Psi_1(\lambda_n)\right|\leq C\sqrt{(|A|\wedge(N_{d^*}-|A|+\sqrt{N_{d^*}})N_{d^*}}
    \]
    simultaneously for all $A\subseteq[N_{d^*}]$.
\end{lemma}
\begin{proof}
    Define
    \[
        K=\lVert\log X_1-\log(\log(1/\lambda_n)\vee 1)\rVert_{\psi_2}\,.
    \]
    By Corollary~\ref{cor:tail-sum} we have for any $t>0$ that
    \[
        \Pr\left(\left|\sum_{i=1}^{N_{d^*}}\log X_i-N_{d^*}\Psi_1(\lambda_n)\right|\geq t\right)\leq 2\exp\left(\frac{C_1t^2}{N_{d^*}K^2}\right)
    \]
    for some constant $C_1>0$. By setting $t=KN_{d^*}^{3/4}$, we have
    \begin{equation}\label{eqn:total-conc}
        \left|\sum_{i=1}^{N_{d^*}}\log X_i-N_{d^*}\Psi_1(\lambda_n)\right|\leq KN_{d^*}^{3/4}
    \end{equation}
    with probability at least $1-2e^{-C_1\sqrt{N_{d^*}}}$. Similarly, for any $t>0$,
    \[
        \Pr\left(\left|\sum_{i\in A}\log X_i-|A|\Psi_1(\lambda_n)\right|\geq t\right)\leq 2\exp\left(\frac{C_1t^2}{|A|K^2}\right)\,.
    \]
    By setting $t=2K\sqrt{|A|N_{d^*}/C_1}$, we have
    \[
        \Pr\left(\left|\sum_{i\in A}\log X_i-|A|\Psi_1(\lambda_n)\right|\geq C_2K\sqrt{|A|N_{d^*}}\right)\leq 2e^{-4N_{d^*}}
    \]
    for some constant $C_2>0$. Taking the union bound over all $A\subseteq[N_{d^*}]$, we have that
    \begin{equation}\label{eqn:partial-conc}
        \left|\sum_{i\in A}\log X_i-|A|\Psi_1(\lambda_n)\right|\leq C_2K\sqrt{|A|N_{d^*}}
    \end{equation}
    for all $A\subseteq[N_{d^*}]$, with probability at least $1-2e^{-3N_{d^*}}$, conditioned on $\Gamma_{\leq d^*-1}$. Now we write
    \[
        \sum_{i\in A}\log X_i=\sum_{i=1}^{N_{d^*}}\log X_i-\sum_{i\notin A}\log X_i\,.
    \]
    Under the events \eqref{eqn:total-conc} and \eqref{eqn:partial-conc}, we get
    \[
        \left|\sum_{i\in A}\log X_i-|A|\Psi_1(\lambda_n)\right|\leq KN_{d^*}^{3/4}+C_2K\sqrt{(N_{d^*}-|A|)N_{d^*}}\,.
    \]
    Thus, we conclude that
    \[
        \left|\sum_{i\in A}\log X_i-|A|\Psi_1(\lambda_n)\right|\leq C_3K\sqrt{(|A|\wedge (N_{d^*}-|A|+\sqrt{N_{d^*}}))N_{d^*}}
    \]
    for some constant $C_3>0$, with probability at least $1-2e^{-C_1\sqrt{N_{d^*}}}-2e^{-3N_{d^*}}=1-o(1)$. The conclusion now follows since $K$ is bounded by a universal constant by Lemma~\ref{lem:subg-ztp}.
    
\end{proof}

\begin{lemma}\label{lem:psi-bounds}
    There is a universal constant $C>0$ such that with probability at least $1-o(1)$,
    \[
        \Psi(m;N_{d^*},\lambda_n)\leq N_{d^*}H(m/N_{d^*})+m\Psi_1(\lambda_n)+C\sqrt{(m\wedge(N_{d^*}-m+\sqrt{N_{d^*}})N_{d^*}}+o(n)
    \]
    and
    \[
        \Psi(m;N_{d^*},\lambda_n)\geq m\Psi_1(\lambda_n)-o(n)
    \]
    simultaneously for all $1\leq m\leq N_{d^*}$.
\end{lemma}
\begin{proof}
    Lemma~\ref{lem:parconc} implies that
    \[
        \begin{split}
            \mathsf{LSE}_m(\log X_1,\cdots,\log X_{N_{d^*}}) &\leq \log\binom{N_{d^*}}{m}+m\Psi_1(\lambda_n)+C\sqrt{(m\wedge(N_{d^*}-m+\sqrt{N_{d^*}})N_{d^*}}\\
            &\leq N_{d^*}H(m/N_{d^*})+m\Psi_1(\lambda_n)+C\sqrt{(m\wedge(N_{d^*}-m+\sqrt{N_{d^*}})N_{d^*}}\,.
        \end{split}
    \]
    simultaneously for all $m$, with probability $1-o(1)$. By Lemma~\ref{lem:lse2} $\mathsf{LSE}_m$ concentrates around its mean given that the distribution is subgaussian, so this in turn implies
    \[
        \Psi(m;N_{d^*},\lambda_n)\leq N_{d^*}H(m/N_{d^*})+m\Psi_1(\lambda_n)+C\sqrt{(m\wedge(N_{d^*}-m+\sqrt{N_{d^*}})N_{d^*}}+o_p(n)
    \]
    simultaneously for all $m$. On the other hand, Jensen's inequality gives
    \[
        \begin{split}
            \mathsf{LSE}_m(\log X_1,\cdots,\log X_{N_{d^*}}) &= \log\left(\frac{1}{\binom{N_{d^*}}{m}}\sum_{I\in\binom{[N_{d^*}]}{m}}\exp\left(\sum_{i\in I}\log X_i\right)\right)+\log\binom{N_{d^*}}{m}\\
            &\geq\frac{\binom{N_{d^*}-1}{m-1}}{\binom{N_{d^*}}{m}}\sum_{i=1}^{N_{d^*}}\log X_i\\
            &=\frac{m}{N_{d^*}}\sum_{i=1}^{N_{d^*}}\log X_i
        \end{split}
    \]
    which implies
    \[
        \Psi(m;N_{d^*},\lambda_n)\geq m\Psi_1(\lambda_n)-o_p(n)\,.
    \]
\end{proof}

Now we prove Theorem~\ref{thm:logz-formula}. Our objective function is
\begin{equation}\label{eqn:gibbs-opt-obj}
    f(x)=\Psi(x;N_{d^*},\lambda_n)+(n-x)\log(qx)+\overline{\beta}x-\overline{\beta}(d^*+1)n
\end{equation}
with domain $\mathcal{X}=\{1,\cdots,N_{d^*}\}$. We also define
\begin{equation}\label{eqn:gibbs-opt-obj-ub}
    f_u(x):=g(x)+h(x)
\end{equation}
where
\[
    g(x)=\Psi_1(\lambda_n)x+(n-x)\log(qx)+\overline{\beta}x-\overline{\beta}(d^*+1)n
\]
and
\[
    h(x)=N_{d^*}H(x/N_{d^*})+C\sqrt{\left(x\wedge\left(N_{d^*}-x+\sqrt{N_{d^*}}\right)\right)N_{d^*}}\,.
\]
Note that $f_u(x)$ is a concave function. By Lemma~\ref{lem:psi-bounds} we have with high probability
\begin{equation}\label{eqn:f-sandwich}
    g(x)-o(n)\leq f(x)\leq f_u(x)+o(n)
\end{equation}
for all $x\in\mathcal{X}$. Our proof strategy is as follows: instead of directly optimizing $f(x)$, we find a set $I\subseteq\mathcal{X}$ where $f_u(x)$ attains its maximum $f_u(x_u^*)$, and prove that $h(x_u^*)=o(n)$. This implies $|f(x_u^*)-g(x_u^*)|=o(n)$ so $g(x_u^*)$ can be a good approximation for the maximum of $f(x)$ asymptotically almost surely.

Before we start the main proof, we note that Lemma~\ref{lem:subg-ztp} gives
\begin{equation}\label{eqn:reg2}
    \left|\Psi_1(\lambda_n)-(\Delta_n\vee0)\log(nq)\right|=O(1)
\end{equation}
with high probability. This further implies
\begin{equation}\label{eqn:psi-delta}
    \left|\Psi_1(\lambda_n)-\Delta_n\log(nq)\right|=O(\log\log\log n)\,.
\end{equation}

\paragraph{Case 1: low temperature phase.} To deal with the phase \eqref{eqn:low-temp-phase}, we consider $x=(1-\epsilon\kappa_n^{-1})N_{d^*}$ where $0\leq\epsilon\leq(\kappa_n\wedge1)/2$. Note that
\begin{equation}\label{eqn:ltp-nds-pert}
    \begin{split}
        (1-\epsilon\kappa_n^{-1})N_{d^*}&=N_{d^*}-\frac{\epsilon N_{d^*}}{(1-\lambda_n)\log\log n}\\
        &=N_{d^*}-\frac{\epsilon n}{\log\log n}-o\left(\frac{n}{\log\log n}\right)\\
        &\geq N_{d^*}-\frac{n}{\log\log n}-o\left(\frac{n}{\log\log n}\right)\,.
    \end{split}
\end{equation}
This readily implies $h((1-\epsilon\kappa_n^{-1})N_{d^*})=o(n)$. For $g((1-\epsilon\kappa_n^{-1})N_{d^*})$ we have
\begin{equation}\label{eqn:ltp-gdiff}
    \begin{split}
        &g(N_{d^*})-g((1-\epsilon\kappa_n^{-1})N_{d^*})\\
        &= \epsilon\kappa_n^{-1}N_{d^*}(\Psi_1(\lambda_n)-\log(qN_{d^*})+\overline{\beta})-(n-(1-\epsilon\kappa_n^{-1})N_{d^*})\log(1-\epsilon\kappa_n^{-1})\\
        &\geq \epsilon\kappa_n^{-1}N_{d^*}(\Psi_1(\lambda_n)-\log(qN_{d^*})+\overline{\beta})-n\log(1-\epsilon\kappa_n^{-1})\\
        &=\epsilon(\beta-1+\Delta_n+\kappa_n^{-1})n-(\epsilon\kappa_n^{-1}+\log(1-\epsilon\kappa_n^{-1}))n-o(n)
    \end{split}
\end{equation}
where the $o(n)$ term does not depend on $\epsilon$. Since $\beta\geq\limsup_{n\to\infty}(1-\Delta_n-\kappa_n^{-1})$ implies $\beta-1+\Delta_n+\kappa_n^{-1}\geq-o(1)$, this gives
\[
    f_u(N_{d^*})\geq f_u((1-\epsilon\kappa_n^{-1})N_{d^*})-o(n)
\]
for all $\epsilon$. Now we argue that the maximum of $f_u(x)$ is attained at some $\epsilon$. Since $f_u(x)$ is concave, it suffices to show that there is some $0<\epsilon_0<\kappa_n\wedge 1$ such that $f_u(N_{d^*})>f_u((1-\epsilon_0\kappa_n^{-1})N_{d^*})$. We prove that $\epsilon_0=(\kappa_n\wedge1)/2$ satisfies our desired property by further splitting the low temperature phase into the following two regimes.

\begin{itemize}
    \item Suppose that $\beta>\limsup_{n\to\infty}(1-\Delta_n-\kappa_n^{-1})$. This implies that there exists a $\delta>0$ such that $\beta-1+\Delta_n+\kappa_n^{-1}>\delta$ for all large $n$. If $\kappa_n\geq1$, then $\epsilon_0=1/2$ so
    \[
        \epsilon(\beta-1+\Delta_n+\kappa_n^{-1})n\geq \frac{\delta}{2}n\,.
    \]
    Otherwise, if $\kappa_n<1$, then $\epsilon_0=\kappa_n/2$ so
    \[
        \begin{split}
            -(\epsilon_0\kappa_n^{-1}+\log(1-\epsilon_0\kappa_n^{-1}))n &= \left(-\frac{1}{2}+\log 2\right)n\\
            &\geq\frac{1}{6}n\,.
        \end{split}
    \]
    This implies that
    \[
        g(N_{d^*})-g((1-\epsilon_0\kappa_n^{-1})N_{d^*})\geq\left(\frac{\delta}{2}\wedge\frac{1}{6}\right)n-o(n)\,.
    \]
    In particular, $f_u(N_{d^*})> f_u((1-\epsilon_0\kappa_n^{-1})N_{d^*})$ for all sufficiently large $n$.

    \item Otherwise, we have $\beta=\limsup_{n\to\infty}(1-\Delta_n-\kappa_n^{-1})$ and $\limsup_{n\to\infty}\kappa_n=K<\infty$ for a constant $K$. Then $\liminf_{n\to\infty}\kappa_n^{-1}=1/K$, so by substituting $\epsilon$ with $\epsilon_0=(\kappa_n\wedge1)/2$ we have $\epsilon_0\kappa_n^{-1}\geq\frac{1}{2}\wedge\frac{1}{4K}$ which implies that $-\epsilon_0\kappa_n^{-1}-\log(1-\epsilon_0\kappa_n^{-1})$ is bounded away from zero. Thus, similar to the preceding regime, we get
    \[
        g(N_{d^*})-g((1-\epsilon_0\kappa_n^{-1})N_{d^*})=\Omega(n)
    \]
    and $f_u(N_{d^*})> f_u((1-\epsilon_0\kappa_n^{-1})N_{d^*})$ for all large $n$.
\end{itemize}

Therefore, we have proved that $\max_{x\in\mathcal{X}}f(x)=f_u(N_{d^*})+o_p(n)$. Since $N_{d^*}\Psi_1(\lambda_n)=n\Psi_0(\lambda_n)+o_p(n)$, this gives the formula
\[
    \Phi_{G,\beta}=n\Psi_0(\lambda_n)+\lambda_nn\log(\alpha_n(1-\lambda_n)\log n)-\overline{\beta}(d^*+\lambda_n)n\pm o_p(n)\,.
\]

\paragraph{Case 2: high temperature phase.} In the high temperature phase \eqref{eqn:high-temp-phase}, we can pick a constant $\delta>0$ such that $\beta<1-\Delta_n-\kappa_n^{-1}-\delta$ for all sufficiently large $n$. This also implies $\kappa_n>1$. Here, we consider $x=(1+\epsilon)\frac{\kappa_n^{-1}N_{d^*}}{1-\Delta_n-\beta}$ for $-1/2\leq \epsilon\leq\delta$. Note that
\[
    (1+\delta)\frac{\kappa_n^{-1}N_{d^*}}{1-\Delta_n-\beta}<\frac{(\kappa_n^{-1}+\delta\kappa_n^{-1})N_{d^*}}{\kappa_n^{-1}+\delta}<N_{d^*}\,.
\]
Also, we have
\begin{equation}\label{eqn:htp-nds-pert}
    \begin{split}
        \frac{\kappa_n^{-1}N_{d^*}}{1-\Delta_n-\beta} &=\frac{N_{d^*}}{(1-\Delta_n-\beta)(1-\lambda_n)\log\log n}\\
        &=\frac{n}{(1-\Delta_n-\beta)\log\log n}+o\left(\frac{n}{\log\log n}\right)\,.
    \end{split}
\end{equation}
This gives $h\left((1+\epsilon)\frac{\kappa_n^{-1}N_{d^*}}{1-\Delta_n-\beta}\right)=o(n)$. Moreover,
\begin{equation}\label{eqn:htp-gdiff}
    \begin{split}
        &g\left(\frac{\kappa_n^{-1}N_{d^*}}{1-\Delta_n-\beta}\right)-g\left(\frac{(1+\epsilon)\kappa_n^{-1}N_{d^*}}{1-\Delta_n-\beta}\right)\\
        &= \frac{\epsilon\kappa_n^{-1}N_{d^*}}{1-\Delta_n-\beta}\left(\Psi_1(\lambda_n)-\log\left(\frac{\kappa_n^{-1}N_{d^*}q}{1-\Delta_n-\beta}\right)+\overline{\beta}\right)-\left(n-\frac{(1+\epsilon)\kappa_n^{-1}N_{d^*}}{1-\Delta_n-\beta}\right)\log(1+\epsilon)\\
        &\geq\frac{\epsilon\kappa_n^{-1}N_{d^*}}{1-\Delta_n-\beta}\left(\Psi_1(\lambda_n)-\log\left(\frac{\kappa_n^{-1}N_{d^*}q}{1-\Delta_n-\beta}\right)+\overline{\beta}\right)-n\log(1+\epsilon)\\
        &=(\epsilon-\log(1+\epsilon))n-o(n)\,.
    \end{split}
\end{equation}
This implies
\[
    f_u\left(\frac{\kappa_n^{-1}N_{d^*}}{1-\Delta_n-\beta}\right)\geq f_u\left((1+\epsilon)\frac{\kappa_n^{-1}N_{d^*}}{1-\Delta_n-\beta}\right)-o(n)
\]
for all $\epsilon$. Also, by plugging in $\epsilon=\epsilon_1=-1/2$ or $\epsilon=\epsilon_2=\delta$, we have for sufficiently large $n$
\[
    \begin{split}
        f_u\left(\frac{\kappa_n^{-1}N_{d^*}}{1-\Delta_n-\beta}\right)&\geq f_u\left((1+\epsilon_i)\frac{\kappa_n^{-1}N_{d^*}}{1-\Delta_n-\beta}\right)+\Omega(n) \\
        &>f_u\left((1+\epsilon_i)\frac{\kappa_n^{-1}N_{d^*}}{1-\Delta_n-\beta}\right)
    \end{split}
\]
for either $\epsilon_i$, so by the concavity of $f_u$ it follows that the maximum is attained in $(\epsilon_1,\epsilon_2)$. Therefore, an approximate maximum can be obtained by $\epsilon=0$ leading to
\[
    \Phi_{G,\beta}=n\log\left(\frac{\alpha_n\log n}{(1-\Delta_n-\beta)\log\log n}\right)-n-\overline{\beta}(d^*+1)n\pm o_p(n)\,.
\]

\section{Analysis of Gibbs measures over trees}\label{sec:gibbs2}

Building on top of our analysis of the log partition function $\log Z_{G,\beta}$, we continue on analyzing the Gibbs measures $\mu_{G,\beta}$. We show the existence of different forms of phase transitions in the system by analyzing the trajectory of $\mu_{G,\beta}$ in the space of probability measures equipped with $1$-Wasserstein metric. We also compute several thermodynamic quantities such as free energy density and Franz--Parisi potential, which let us give a picture of the optimization landscape from a statistical physics perspective.

\paragraph{Entropy and transportation.} 
We will make use of \emph{entropy-transport} inequalities, specifically:
\begin{theorem}[Corollary 4.16, Example 4.17 of \cite{van2014probability}]\label{thm:entropy-transport}
Suppose that $\mu_i$ is a probability measure on metric space $(\mathcal X_i,d_i)$ for $i = 1$ to $n$ and that for some $\sigma > 0$, each satisfies the $W_1$-entropy transport inequality $W_1(\nu,\mu_i) \le \sqrt{2 \sigma^2 \KL(\nu, \mu_i)}$ for all $\nu$. Then for all probability measures $\nu$ on $\bigotimes_i \mathcal X_i$ and letting $\mu = \bigotimes_i \mu_i$, we have\footnote{We use a different normalization convention for the Hamming metric compared to \cite{van2014probability}.}
\[ W_1(\nu, \mu) \le \sqrt{2 \sigma^2 n \KL(\nu, \mu)} \]
where we equip $\bigotimes_i \mathcal X_i$ with the metric $d(x,y) = \sum_i d_i(x_i,y_i)$. 

In particular, by combining the above with Pinsker's inequality, this holds with $\sigma^2 = 1$ if $d_i(x,y) = 1(x \ne y)$, in which case $d$ is the Hamming metric. 
\end{theorem}
\begin{corollary}[High entropy $\;\Longrightarrow\;$ Wasserstein closeness to the uniform measure]
\label{cor:high-entropy-wasserstein}
Let $\mathcal X_1,\dots,\mathcal X_n$ be finite sets and write
\[
S \;=\;\prod_{i=1}^n \mathcal X_i,
\qquad\qquad
\mu \;=\;\bigotimes_{i=1}^n \mu_i,
\]
where each $\mu_i$ is the uniform measure on~$\mathcal X_i$.  
Equip $S$ with the \emph{(unnormalised) Hamming metric}
\[
d(x,y) \;=\; \sum_{i=1}^n \mathbf 1\!\bigl(x_i \neq y_i\bigr),
\qquad x=(x_1,\dots,x_n),\;y=(y_1,\dots,y_n)\in S.
\]
For any probability measure $\nu$ on $S$,
such that $\nu$ has $\varepsilon$-near-maximal entropy, i.e., 
\[
H_\nu(X)\;\ge\;\log |S| - \varepsilon
\quad(\varepsilon>0),
\]
we have
\[
W_1(\nu,\mu)
\;\le\;
\sqrt{2n\,\varepsilon}.
\]
\end{corollary}
\begin{proof}
Because each $\mu_i$ is uniform on the finite set $\mathcal X_i$, it satisfies the $W_1$–entropy transport inequality with parameter $\sigma^2=1$ under the $0$–$1$ metric $d_i(x,y)=\mathbf 1(x\neq y)$ (cf.\ the remark following Theorem~\ref{thm:entropy-transport}).  
Applying Theorem~\ref{thm:entropy-transport} to the product space $(S,d)$ therefore gives
\[
W_1(\nu,\mu) \;\le\; \sqrt{2n\,\KL(\nu,\mu)}.
\]
For the uniform measure $\mu$ on $S$, we recall the standard identity
\(
\KL(\nu,\mu)=\E_{\nu}\left[\log\frac{\nu}{\mu}\right]=\log|S|-H_\nu(X).
\)
Substituting this into the previous inequality yields the result.
\end{proof}

\subsection{Structure of the Gibbs measures}\label{sec:gibbs-w1}

In Section~\ref{sec:partition-function}, we have constructed measures $\hat{\mu}_{G,m}$ that approximates $\mu_{G,\beta}$ in KL divergence. Built on these results, we analyze the geometric structure of a typical spanning tree under the Gibbs measures, with respect to the $1$-Wasserstein metric. Our main goal is to prove Theorem~\ref{thm:gibbs-low-phase} and Theorem~\ref{thm:gibbs-high-phase}, which state that the Gibbs measures are Wasserstein close to the uniform measure over the set of shortest path trees in the low temperature phase, and far from the uniform measure in the high temperature phase.

\subsubsection{Uniform measure as ground state}\label{sec:ground-state}
As a warm-up, we analyze the uniform measure $\mu_{G,\infty}$ over the shortest path trees of $\overline{G}$. In statistical physics terminology, this is referred to as the \emph{ground state}. In this section, we compute the ground state energy and entropy of our system, which is quite similar to the computation in the proof of Proposition~\ref{prop:fixed-kernel} in Section~\ref{sec:gibbs-kernel-proof}.

The ground state energy of the system is equivalent to the sum of $\mathsf{d}(1,v)$ over all $v$ in the giant component of $G$.

\begin{theorem}[Ground state energy]\label{thm:ground-state-energy}
    We have
    \[
        \sum_{v\in\overline{V}}\mathsf{d}_{G_n}(1,v)=(d_n^*+\lambda_n)n+o_p\left(\frac{n}{\log\log n}\right)
    \]
\end{theorem}
\begin{proof}
    An upper bound is already proved in \eqref{eqn:gibbs-energy-ub}:
    \[
        \sum_v\mathsf{d}_G(1,v)\leq (d^*+1)n-N_{d^*}+o_p\left(\frac{n}{\log\log n}\right)\,.
    \]
    For the lower bound, we have
    \[
        \begin{split}
            \sum_v\mathsf{d}_G(1,v) &\geq N_{d^*-1}(d^*-1)+N_{d^*}d^*+|\overline{V}\setminus\Gamma_{\leq d^*}|(d^*+1)\\
            &=|\overline{V}\setminus\Gamma_{\leq d^*-2}|(d^*+1)-2N_{d^*-1}-N_{d^*}\,.
        \end{split}
    \]
    By Proposition~\ref{prop:conc}, $N_{d^*-1}$ concentrates around $(nq)^{d^*-1}=O_p(n/(\log\log n)^2)$. Also applying \eqref{eqn:gibbs-res-bound}, we get the desired lower bound
    \[
        \sum_v\mathsf{d}_G(1,v)\geq (d^*+1)n-N_{d^*}-o_p\left(\frac{n}{\log\log n}\right)\,.
    \]
    Now applying Theorem~\ref{thm:main1-tight}, we obtain
    \[
        \begin{split}
            \sum_v\mathsf{d}_G(1,v) &= (d^*+1)n-(1-\lambda_n)n+o_p\left(\frac{n}{\log\log n}\right)\\
            &=(d^*+\lambda_n)n+o_p\left(\frac{n}{\log\log n}\right)\,.
        \end{split}
    \]
\end{proof}

The ground state entropy is just the log of the number of shortest path trees. In terms of the conditional Gibbs measure \eqref{eqn:gibbs-partition-kernel-size}, we have
\[
    Z_{G,\beta}\|_{N_{d^*}}=\#\{\text{shortest path trees of $\overline{G}$}\}\cdot\exp\left(-\overline{\beta}\sum_v\mathsf{d}_G(1,v)\right)
\]
so
\[
    \log(Z_{G,\beta}\|_{N_{d^*}})=\log\#\{\text{shortest path trees of $\overline{G}$}\}-\beta\log\log (n)\sum_v\mathsf{d}_G(1,v)\,.
\]
Since we know the ground state energy, the entropy can be easily obtained from $\log(Z_{G,\beta}\|_{N_{d^*}})$. By Corollary~\ref{cor:fixed-kernel-size} this is approximately $\tilde{\Phi}_{G,\beta}(N_{d^*})$ which is calculated in the low temperature part of Theorem~\ref{thm:logz-formula}. This immediately yields the following theorem.

\begin{theorem}[Entropy of the ground state]\label{thm:log-shortest-path-trees-acc}
    We have
    \[
        \log\#\{\text{shortest path trees of $\overline{G}_n$}\}=n\Psi_0(\lambda_n)+\lambda_nn\log(\alpha_n(1-\lambda_n)\log n)+o_p(n)\,.
    \]
    In particular, if $\lambda_n\to\lambda\in[0, 1]$ and $\Delta_n\to\Delta\in[0, 1]$, then
    \[
        \frac{\log\#\{\text{shortest path trees of $\overline{G}$}\}}{n\log\log n}\pto\lambda+\Delta=\begin{cases}
            \Delta&\text{if $0<\Delta\leq 1$,}\\
            \lambda&\text{if $\Delta=0$.}
        \end{cases}
    \]
\end{theorem}

\subsubsection{Distribution of distance vectors}

In Section~\ref{sec:kernels-cond-gibbs}, we studied the conditional Gibbs measures $\mu_{G,\beta}|_A$ for a fixed kernel $A$, which are building blocks for the Gibbs measures and are easier to analyze. Similarly, we start by analyzing $\mu_{G,\beta}|_A$, extending the results from the previous section where $A=N_{d^*}$. The first step is to prove that a typical spanning tree under $\mu_{G,\beta}|_A$ has near-minimum energy.

\begin{proposition}\label{prop:cond-gibbs-dv}
    For any constant $C>0$, consider a set of spanning trees of $\overline{G}$ defined by
    \begin{equation}\label{eqn:energy-opt}
        E_{G,m}=\left\{T:\sum_{v}\mathsf{d}_T(1,v)\leq (d^*+1)n-m+\frac{Cn}{\log\log n}\right\}\,.
    \end{equation}
    Let $\beta>0$ be a fixed constant. Then with probability at least $1-o(1)$, we have
    \[
        \mu_{G,\beta}|_A(E_{G,|A|})=1-\exp\left(-(C\beta-o(1))n\right)
    \]
    simultaneously for all $A\in\mathcal{K}(G)$ of size $|A|\geq m_\ell$.
\end{proposition}

To prove Proposition~\ref{prop:cond-gibbs-dv}, we observe that the log partition function $\log(Z_{G,\beta}|_A)$ established in Proposition~\ref{prop:fixed-kernel} can be decomposed into the entropy part
\[
    \sum_{v\in A}\log|\mathsf{par}_G(v)|+(n-|A|)\log(|A|q)
\]
and the energy part
\[
    \overline{\beta}|A|-\overline{\beta}(d^*+1)n\,.
\]
The entropy part is in fact nearly the largest possible. This is because by the definition of kernel, any vertex $v\in A$ has at most $\log|\mathsf{par}_G(v)|$ parent selection entropy, and for all other vertices $\log(|A|q)\approx\log\log n$ is the maximum possible for sparse graphs. Hence, a typical tree under $\mu_{G,\beta}|_A$ cannot have energy too greater than $\overline{\beta}|A|-\overline{\beta}(d^*+1)n$.

Before we prove the proposition, we first establish the following lemma, which allows us to replace $\log(|A|q)$ with a more convenient form.

\begin{lemma}\label{lem:multi-ztb-conc}
    There is a universal constant $C>0$ such that with probability at least $1-o(1)$,
    \[
        \sum_{v\in B}\log|\mathsf{N}_A(v)|\leq|B|\log(|A|q\vee1)+\frac{Cn}{\sqrt{|A|q}}
    \]
    simultaneously for all $A\subseteq\Gamma_{d^*}$ and $B\subseteq\mathsf{N}_{G\setminus\Gamma_{\leq d^*-1}}(A)$.
\end{lemma}
\begin{proof}
    The proof is similar to that of Lemma~\ref{lem:logsum-lb}. Treating $\Gamma_{\leq d^*-1}$ and $A$ as fixed, by Lemma~\ref{lem:subg-ztb}
    \[
        \lVert\log|\mathsf{N}_A(v)|-\log(|A|q\vee1)\rVert_{\psi_2}\leq\frac{C_1}{\sqrt{|A|q}}
    \]
    for some constant $C_1>0$. By Corollary~\ref{cor:tail-sum} the probability that $\sum\log|\mathsf{N}_A(v)|$ is away from its mean by at least $t>0$ is bounded by
    \[
        2\exp\left(-\frac{2t^2|A|q}{C_2^2|B|}\right)\leq 2\exp\left(-\frac{2t^2|A|q}{C_2^2n}\right)
    \]
    for some constant $C_2$. By setting $t=C_2n/\sqrt{|A|q}$ this is bounded by $2e^{-2n}$, so taking the union bound over all $A$ and $B$ gives
    \[
        \sum_{v\in B}\log|\mathsf{N}_A(v)|\leq|B|\E[\log|\mathsf{N}_A(v)|\mid\Gamma_{\leq d^*-1},A]+\frac{C_3n}{\sqrt{|A|q}}
    \]
    simultaneously for all $A$ and $B$ asymptotically almost surely. Finally, Lemma~\ref{lem:subg-mean} gives
    \[
        \begin{split}
            \sum_{v\in B}\log|\mathsf{N}_A(v)|&\leq|B|\left(\log(|A|q\vee1)+\frac{C_1}{\sqrt{|A|q}}\right)+\frac{C_3n}{\sqrt{|A|q}}\\
            &\leq|B|\log(|A|q\vee1)+\frac{Cn}{\sqrt{|A|q}}
        \end{split}
    \]
    for some constant $C>0$.
\end{proof}

\begin{proof}[Proof of Proposition~\ref{prop:cond-gibbs-dv}]
    By Proposition~\ref{prop:fixed-kernel} we may assume
    \[
        \begin{split}
            \log (Z_{G,\beta}|_A) &= \log\left(\sum_{T\in\varphi^{-1}(A)}\exp\left(-\overline{\beta}\sum_{v}\mathsf{d}_T(1,v)\right)\right)\\
            &=\sum_{v\in A}\log|\mathsf{par}_G(v)|+(n-|A|)\log(|A|q)+\overline{\beta}|A|-\overline{\beta}(d^*+1)n+o(n)\,.
        \end{split}
    \]
    Define
    \begin{equation}\label{eqn:energy-opt-bar}
        \overline{E}_{G,A}:=\varphi^{-1}(A)\cap E_{G,|A|}
    \end{equation}
    and
    \[
        \overline{E}_{G,A}^{\mathsf{c}}:=\varphi^{-1}(A)\setminus E_{G,|A|}\,.
    \]
    Note that we have
    \[
        \log\left(\mu_{G,\beta}|_A(E_{G,|A|}^{\mathsf{c}})\right)=\log\left(\sum_{T\in\overline{E}_{G,A}^{\mathsf{c}}}\exp\left(-\overline{\beta}\sum_{v}\mathsf{d}_T(1,v)\right)\right)-\log(Z_{G,\beta}|_A)
    \]
    so our goal is to analyze
    \[
        \sum_{T\in\overline{E}_{G,A}^\mathsf{c}}\exp\left(-\overline{\beta}\sum_v\mathsf{d}_T(1,v)\right)\,.
    \]
    We decompose
    \[
        \overline{E}_{G,A}^{\mathsf{c}}=\bigsqcup_{B}\overline{E}_{G,A,B}^{\mathsf{c}}
    \]
    where
    \[
        \overline{E}_{G,A,B}^{\mathsf{c}}:=\left\{T\in\overline{E}_{G,A}^{\mathsf{c}}:\Gamma_{d^*+1}^T=B\right\}\,.
    \]
    If $T\in\overline{E}_{G,A,B}^{\mathsf{c}}$ we have
    \[
        \begin{split}
            \sum_v\mathsf{d}_T(1,v)&= d^*|A|+(d^*+1)|B|+(d^*+2)(n-|A|-|B|)\pm o_p\left(\frac{\log\log n}{n}\right)\\
            &=(d^*+1)n-|A|+(n-|A|-|B|)\pm o_p\left(\frac{\log\log n}{n}\right)\,.
        \end{split}
    \]
    Thus, $\overline{E}_{G,A,B}^{\mathsf{c}}\neq\emptyset$ implies that
    \begin{equation}\label{eqn:nab-nonempty}
        n-|A|-|B|\geq\frac{(C-o_p(1))n}{\log\log n}\,.
    \end{equation}
    To bound the size of $\overline{E}_{G,A,B}^{\mathsf{c}}$ given that it is nonempty, we observe that $v\in A$ has $|\mathsf{par}_G(v)|$ parent choices, $v\in B$ has $|\mathsf{N}_A(v)|$ parent choices, and the other vertices have $O_p(\log n)$ parent choices which is simply the bound for the maximum degree of a typical sparse graph. Using Lemma~\ref{lem:multi-ztb-conc} this gives
    \[
        \begin{split}
            \log\left|\overline{E}_{G,A,B}^{\mathsf{c}}\right|&\leq\sum_{v\in A}\log|\mathsf{par}_G(v)|+\sum_{v\in B}\log|\mathsf{N}_A(v)|+(n-|A|-|B|)\log\log n+o_p(n)\\
            &\leq\sum_{v\in A}\log|\mathsf{par}_G(v)|+|B|\log(|A|q)+(n-|A|-|B|)\log\log n+o_p(n)\\
            &\leq\sum_{v\in A}\log|\mathsf{par}_G(v)|+(n-|A|)\log(|A|q)+(n-|A|-|B|)\log(3(\log\log n)^2)+o_p(n)
        \end{split}
    \]
    which holds simultaneously for all $A$ and $B$ with $|A|\geq m_\ell$.
    Hence, it follows that
    \[
        \begin{split}
            &\log\left(\sum_{T\in\overline{E}_{G,A,B}^\mathsf{c}}\exp\left(-\overline{\beta}\sum_v\mathsf{d}_T(1,v)\right)\right)\\
            &\leq\log(Z_{G,\beta}|_A)+(n-|A|-|B|)(\log(3(\log\log n)^2)-\overline{\beta})+o_p(n)\\
            &=\log(Z_{G,\beta}|_A)-(\beta-o(1))(n-|A|-|B|)\log\log n+o_p(n)\,.
        \end{split}
    \]
    This gives
    \[
        \begin{split}
            &\log\left(\sum_{|B|=n-|A|-k}\sum_{T\in\overline{E}_{G,A,B}^\mathsf{c}}\exp\left(-\overline{\beta}\sum_v\mathsf{d}_T(1,v)\right)\right)\\
            &\leq\log\binom{n}{k}+\log(Z_{G,\beta}|_A)-(\beta-o(1))k\log\log n+o_p(n)\\
            &\leq k\log\left(\frac{en}{k}\right)+\log(Z_{G,\beta}|_A)-(\beta-o(1))k\log\log n+o_p(n)\\
            &\leq k\log\left(\frac{e\log\log n}{C-o_p(1)}\right)+\log(Z_{G,\beta}|_A)-(\beta-o(1))k\log\log n+o_p(n)\\
            &\leq\log(Z_{G,\beta}|_A)-(\beta-o_p(1))k\log\log n+o_p(n)\\
            &\leq\log(Z_{G,\beta}|_A)-(C\beta-o_p(1))n+o_p(n)
        \end{split}
    \]
    where we have used \eqref{eqn:nab-nonempty}. This holds simultaneously for all $k$ such that the sum is nonzero. Hence, we have
    \[
        \begin{split}
            &\log\left(\sum_{T\in\overline{E}_{G,A}^{\mathsf{c}}}\exp\left(-\overline{\beta}\sum_v\mathsf{d}_T(1,v)\right)\right)\\
            &\leq\max_k\left\{\log\left(\sum_{|B|=n-|A|-k}\sum_{T\in\overline{E}_{G,A,B}^\mathsf{c}}\exp\left(-\overline{\beta}\sum_v\mathsf{d}_T(1,v)\right)\right)\right\}+O(\log n)\\
            &\leq\log(Z_{G,\beta}|_A)-(C\beta-o_p(1))n+o_p(n)\,.
        \end{split}
    \]
    The proposition is proved.
\end{proof}

\subsubsection{Conditional Gibbs measures are almost uniform}\label{sec:cond-gibbs-uniform}

Now we show that $\mu_{G,\beta}|_A$ is close in Wasserstein distance to $\mu_{G,\infty}|_A$, the uniform measure over the spanning trees with kernel $A$ and minimum energy. We again use the fact that the entropy is nearly maximum and apply the entropy-transport inequality through Corollary~\ref{cor:high-entropy-wasserstein}. Since some vertices might have parents not in the support of $\mu_{G,\infty}|_A$, we need to round them, for which we appeal to Proposition~\ref{prop:cond-gibbs-dv}.

\begin{proposition}\label{prop:cond-gibbs-w1}
    Let $\beta>0$ be a fixed constant. Then with probability at least $1-o(1)$, we have
    \[
        W_1(\mu_{G,\infty}|_A,\mu_{G,\beta}|_A)=o(n)
    \]
    simultaneously for all $A\in\mathcal{K}(G)$ of size $|A|\geq m_\ell$.
\end{proposition}
\begin{proof}
    For an arbitrary constant $C>0$, let $E_{G,m}$ and $\overline{E}_{G,A}$ be as defined in \eqref{eqn:energy-opt} and \eqref{eqn:energy-opt-bar}. Let $\mu_{G,\beta}|_{\overline{E}_{G,A}}$ be the measure $\mu_{G,\beta}|_A$ conditioned on $E_{G,|A|}$ (or $\overline{E}_{G,A}$). From Proposition~\ref{prop:cond-gibbs-dv}, it is clear that
    \begin{equation}\label{eqn:etw-1}
        W_1(\mu_{G,\beta}|_{\overline{E}_{G,A}},\mu_{G,\beta}|_A)=o_p(n)
    \end{equation}
    and
    \[
        \log(Z_{G,\beta}|_{\overline{E}_{G,A}}):=\log\left(\sum_{T\in\overline{E}_{G,A}}\exp\left(-\overline{\beta}\sum_{v}\mathsf{d}_T(1,v)\right)\right)=\log(Z_{G,\beta}|_A)-o_p(n)
    \]
    simultaneously for all $|A|\geq m_{\ell}$. By Proposition~\ref{prop:fixed-kernel} and the Gibbs variational principle (Proposition~\ref{prop:gibbs-variational}), this implies that
    \[
        \begin{split}
            H_{\mu_{G,\beta}|_{\overline{E}_{G,A}}}(T)&=\log(Z_{G,\beta}|_{\overline{E}_{G,A}})+\overline{\beta}\E_{\mu_{G,\beta}|_{\overline{E}_{G,A}}}\left[\sum_{v}\mathsf{d}_T(1,v)\right]\\
            &\geq\log(Z_{G,\beta}|_A)+\overline{\beta}(d^*+1)n-\overline{\beta}|A|-o_p(n)\\
            &=\sum_{v\in A}\log|\mathsf{par}_G(v)|+(n-|A|)\log(|A|q)-o_p(n)
        \end{split}
    \]
    and
    \begin{equation}\label{eqn:mu-e-energy}
        \E_{\mu_{G,\beta}|_{\overline{E}_{G,A}}}\left[\sum_{v}\mathsf{d}_T(1,v)\right]\leq (d^*+1)n-|A|+\frac{Cn}{\log\log n}\,.
    \end{equation}
    Now we consider the distance vector $D=\mathsf{d}_T:v\mapsto\mathsf{d}_T(1,v)$. Since $D$ is a function of $T$ we have
    \[
        H_{\mu_{G,\beta}|_{\overline{E}_{G,A}}}(T)=H_{\mu_{G,\beta}|_{\overline{E}_{G,A}}}(D)+H_{\mu_{G,\beta}|_{\overline{E}_{G,A}}}(T\mid D)\,.
    \]
    By the stars and bars method, we see that the number of possible distance vectors of a spanning tree in $E_{G,m}$ is bounded by
    \[
        \binom{n+\frac{Cn}{\log\log n}+o(n)-1}{\frac{Cn}{\log\log n}+o(n)}=\binom{n+o(n)}{o(n)}
    \]
    which implies that
    \[
        H_{\mu_{G,\beta}|_{\overline{E}_{G,A}}}(D)\leq\log\binom{n+o(n)}{o(n)}=o(n)\,.
    \]
    As a consequence, we get
    \begin{equation}\label{eqn:h-t-d}
        \begin{split}
            H_{\mu_{G,\beta}|_{\overline{E}_{G,A}}}(T\mid D) &= H_{\mu_{G,\beta}|_{\overline{E}_{G,A}}}(T)-o(n)\\
            &\geq \sum_{v\in A}\log|\mathsf{par}_G(v)|+(n-|A|)\log(|A|q)-o_p(n)\,.
        \end{split}
    \end{equation}
    
    To apply the entropy-transport inequality (Corollary~\ref{cor:high-entropy-wasserstein}), we round $\mu_{G,\beta}|_{\overline{E}_{G,A}}$ to the support of $\mu_{G,\infty}|_A$. This is equivalent to saying that we are forcing every vertex to have minimum distance to the root, provided that the kernel is $A$. We denote by $\mathsf{d}_A$ the unique distance vector in the support of $\mu_{G,\infty}|_A$, i.e.,
    \[
        \mathsf{d}_A(v):=\begin{cases}
            \mathsf{d}_G(1,v) & \text{if $\mathsf{d}_G(1,v)\leq d^*-1$,}\\
            d^* + \mathsf{d}_G(v,A) & \text{if $v\in\overline{V}\setminus\Gamma_{\leq d^*-1}$.}
        \end{cases}
    \]
    Now we define the rounding procedure from $\mu_{G,\beta}|_{\overline{E}_{G,A}}$ to $\nu$ as follows. For a spanning tree $T\sim\mu_{G,\beta}|_{\overline{E}_{G,A}}$, for vertices $v$ such that $\mathsf{d}_A(1,v)\neq1+\mathsf{d}_A(1, \mathsf{par}_T(v))$ (which implies $\mathsf{d}_T(1,v)\neq\mathsf{d}_A(v)$), we replace each $\mathsf{par}_T(v)$ with an arbitrary vertex $u\in\mathsf{par}_G(v)$ satisfying $\mathsf{d}_A(1,v)=1+\mathsf{d}_A(1,u)$. It is easy to see that the resulting graph $S$ is a valid spanning tree in the support of $\mu_{G,\infty}|_A$ (i.e., $\mathsf{d}_T=\mathsf{d}_A$). Also, since we do not modify the parents of more than $\frac{Cn}{\log\log n}$ vertices, $\nu$ is not too far from the original measure in terms of Wasserstein distance:
    \begin{equation}\label{eqn:etw-2}
        W_1(\nu,\mu_{G,\beta}|_{\overline{E}_{G,A}})=o(n)\,.
    \end{equation}
    We claim that this procedure does not lose entropy by more than $o_p(n)$, by analyzing the entropy on a per distance vector basis. Since the law of $T$ conditioned on distance vector is a product measure, we have
    \[
        \begin{split}
            H_{\mu_{G,\beta}|_{\overline{E}_{G,A}}}(T\mid D=\mathsf{d}) &= \sum_{v:\mathsf{d}(v)\neq\mathsf{d}_A(v)}H_{\mu_{G,\beta}|_{\overline{E}_{G,A}}}(\mathsf{par}_T(v)\mid D=\mathsf{d})\\
            &\quad+ \sum_{v:\mathsf{d}(v)=\mathsf{d}_A(v)}H_{\mu_{G,\beta}|_{\overline{E}_{G,A}}}(\mathsf{par}_T(v)\mid D=\mathsf{d})\\
            &\leq|\{v:\mathsf{d}(v)\neq\mathsf{d}_A(v)\}|\log\log n+o_p(n)\\
            &\quad+ \sum_{v:\mathsf{d}(v)=\mathsf{d}_A(v)}H_{\mu_{G,\beta}|_{\overline{E}_{G,A}}}(\mathsf{par}_T(v)\mid D=\mathsf{d})
        \end{split}
    \]
    where we have used the fact that the maximum degree is $O_p(\log n)$. By our construction of $S\sim\nu$, we have
    \[
        \begin{split}
            H_{\nu}(S\mid D=\mathsf{d})&\geq\sum_{v:\mathsf{d}(v)=\mathsf{d}_A(v)}H_{\mu_{G,\beta}|_{\overline{E}_{G,A}}}(\mathsf{par}_T(v)\mid D=d)\\
            &\geq H_{\mu_{G,\beta}|_{\overline{E}_{G,A}}}(T\mid D=d)-|\{v:\mathsf{d}(v)\neq\mathsf{d}_A(v)\}|\log\log n-o_p(n)\,.
        \end{split}
    \]
    This gives
    \[
        \begin{split}
            H_{\nu}(S) &\ge H_{\nu}(S\mid D)\\
            &\geq H_{\nu_{G,\beta}|_{\overline{E}_{G,A}}}(T\mid D)-(\log\log n)\E_{\mu_{G,\beta}|_{\overline{E}_{G,A}}}\left[|\{v:\mathsf{d}_T(1,v)\neq\mathsf{d}_A(v)\}|\right]-o_p(n)\,.
        \end{split}
    \]
    Now \eqref{eqn:mu-e-energy} implies that
    \[
        (\log\log n)\E_{\mu_{G,\beta}|_{\overline{E}_{G,A}}}\left[|\{v:\mathsf{d}_T(1,v)\neq\mathsf{d}_A(v)\}|\right]\leq Cn
    \]
    which by \eqref{eqn:h-t-d} gives
    \[
        \begin{split}
            H_{\nu}(S) &\geq H_{\nu_{G,\beta}|_{\overline{E}_{G,A}}}(T\mid D)-Cn-o_p(n)\\
            &\geq\sum_{v\in A}\log|\mathsf{par}_G(v)|+(n-|A|)\log(|A|q)-Cn-o_p(n)\,.
        \end{split}
    \]
    This is at most $Cn+o_p(n)$ smaller than the maximum entropy attained by $\mu_{G,\infty}|_A$. Therefore, applying Corollary~\ref{cor:high-entropy-wasserstein} with $\nu$ and $\mu=\mu_{G,\infty}|_A$, we arrive at
    \[
        W_1(\nu,\mu_{G,\infty}|_A)\leq\sqrt{2C}n+o_p(n)\,.
    \]
    Since $C$ is arbitrary, this implies $W_1(\nu,\mu_{G,\infty}|_A)=o_p(n)$. Together with \eqref{eqn:etw-1} and \eqref{eqn:etw-2}, we complete the proof.
\end{proof}

Now that we better understand the structure of $\mu_{G,\beta}|_A$, we mix these measures to build $\mu_{G,\beta}$. Since our results so far only apply to the case where the kernel size is at least $m_\ell$, we prove that the kernel is indeed large enough asymptotically almost surely.

\begin{lemma}\label{lem:large-kernel}
    Let $\beta>0$ be a fixed constant. Then we have
    \[
        \mu_{G,\beta}(\{T:|\varphi(T)|\geq m_{\ell}\})\geq1-\exp((1/3-o(1))n)
    \]
    with probability at least $1-o(1)$.
\end{lemma}
\begin{proof}
    This is trivial if $m_\ell=m_0$ so we assume
    \[
        m_\ell=\frac{n}{3(\log\log n)^2}\,.
    \]
    We begin with
    \begin{equation}\label{eqn:small-prob-bound}
        \begin{split}
            \log\left(\mu_{G,\beta}(\{T:|\varphi(T)|\leq m_{\ell}\})\right) &= \log\left(\sum_{m=m_0}^{m_{\ell}}Z_{G,\beta}\|_m\right) - \log Z_{G,\beta}\\
            &\leq\max_{m_0\leq m\leq m_{\ell}}\log(Z_{G,\beta}\|_m)-\log Z_{G,\beta}+\log n\,.
        \end{split}
    \end{equation}
    By Lemma~\ref{lem:gibbs-cond-size-ub}, we have
    \[
        \log(Z_{G,\beta}\|_m)\leq\mathcal{H}_G(m)+(n-m)\log(L+mq)+\overline{\beta}m-\overline{\beta}(d^*+1)n+o_p(n)\,.
    \]
    By Lemma~\ref{lem:hg-conc} and Lemma~\ref{lem:psi-bounds}, we see that $\mathcal{H}_G(m)=o_p(n)$ for $m\leq m_{\ell}$, which gives
    \[
        \begin{split}
            \log(Z_{G,\beta}\|_m)&\leq(n-m)\log(L+mq)+\overline{\beta}m-\overline{\beta}(d^*+1)n+o_p(n)\\
            &\leq n\log(L+mq)-\overline{\beta}(d^*+1)n+o_p(n)\\
            &\leq n\log\left(\frac{nq}{(\log\log n)^2}\right)-n\log 3-\overline{\beta}(d^*+1)n+o_p(n)\,.
        \end{split}
    \]
    Now we lower bound $\log Z_{G,\beta}$. Let
    \[
        m_1:=\left\lfloor\frac{n}{2(\log\log n)^2}\right\rfloor
    \]
    which is strictly smaller than $N_{d^*}$ asymptotically almost surely (by Lemma~\ref{lem:dstar}). By Proposition~\ref{prop:fixed-kernel-size} we have
    \[
        \begin{split}
            \log Z_{G,\beta} &\geq \log (Z_{G,\beta}\|_{m_1})\\
            &=\mathcal{H}_G(m_1)+(n-m_1)\log(m_1q)+\overline{\beta}m_1-\overline{\beta}(d^*+1)n-o_p(n)\\
            &\geq (n-m_1)\log(m_1q)-\overline{\beta}(d^*+1)n-o_p(n)\\
            &\geq\left(1-\frac{1}{2(\log\log n)^2}-\frac{1}{n}\right)n\log\left(\frac{nq}{2(\log\log n)^2}\right)-\overline{\beta}(d^*+1)n-o_p(n)\\
            &=n\log\left(\frac{nq}{(\log\log n)^2}\right)-n\log 2-\overline{\beta}(d^*+1)n-o_p(n)\,.
        \end{split}
    \]
    Continuing from \eqref{eqn:small-prob-bound}, we have
    \[
        \begin{split}
            \log\left(\mu_{G,\beta}(\{T:|\varphi(T)|\leq m_{\ell}\})\right)&\leq-n\log(3/2)+o_p(n)\\
            &=-\frac{1}{3}n+o_p(n)\,.
        \end{split}
    \]
\end{proof}

Since the measures $\mu_{G,\beta}|_A$ can be approximated by $\mu_{G,\infty}$, we mix these using the mixture weights of $\mu_{G,\beta}$. Namely, we define
\[
    \tilde{\mu}_{G,\beta}(T):=\sum_{A\in\mathcal{K}(G)}\mu_{G,\beta}(\varphi(T)=A)\cdot\mu_{G,\infty}|_A(T)\,.
\]
An equivalent way of defining $\tilde{\mu}_{G,\beta}$ is through the following sampling algorithm:
\begin{enumerate}
    \item for each $v\in\Gamma_{\leq d^*-1}$, sample $\mathsf{par}_T(v)$ from $\mathsf{par}_G(v)$ uniformly at random;
    \item sample $A\subseteq\Gamma_{d^*}$ according to the marginal law of the kernel $\varphi(T)$ under $\mu_{G,\beta}$;
    \item for each $v\in A$, sample $\mathsf{par}_T(v)$ from $\mathsf{par}_G(v)$ uniformly at random;
    \item for each $v\in\overline{V}\setminus\Gamma_{d^*-1}\setminus A$, sample $\mathsf{par}_T(v)$ uniformly at random from the set of vertices $u\in\mathsf{N}_G(v)$ such that $\mathsf{d}_G(v,A)=\mathsf{d}_G(u,A)+1$.
\end{enumerate}
Then our results immediately yield the following corollary.

\begin{corollary}\label{cor:gibbs-w1}
    Let $\beta>0$ be a fixed constant. Then we have
    \[
        \KL(\tilde{\mu}_{G,\beta},\mu_{G,\beta})=o_p(n)
    \]
    and
    \[
        W_1(\tilde{\mu}_{G,\beta},\mu_{G,\beta})=o_p(n)\,.
    \]
\end{corollary}
\begin{proof}
    The first result follows by Lemma~\ref{lem:large-kernel}, Proposition~\ref{prop:fixed-kernel}, and the convexity of the KL divergence. The second result follows by Lemma~\ref{lem:large-kernel} and Proposition~\ref{prop:cond-gibbs-w1}.
\end{proof}

\subsubsection{Gibbs measures in the low temperature phase}\label{sec:gibbs-low-phase}

Now we complement Theorem~\ref{thm:logz-formula} by studying how the Gibbs measures are compared to the uniform measure $\mu_{G,\infty}$ in $1$-Wasserstein metric. We begin with the low temperature phase \eqref{eqn:low-temp-phase}, in which case the Gibbs measures are close to $\mu_{G,\infty}$.

\begin{theorem}\label{thm:gibbs-low-phase}
    In the low temperature phase \eqref{eqn:low-temp-phase} we have
    \[
        W_1(\mu_{G_n,\infty},\mu_{G_n,\beta})=o_p(n)\,.
    \]
\end{theorem}

The idea of our proof is to show that a typical tree in the low temperature phase has a nearly maximum sized kernel. Similar to the our proof method in Section~\ref{sec:cond-gibbs-uniform}, we observe that the entropy part of the formula in Theorem~\ref{thm:logz-formula} is almost the maximum possible, and then apply the entropy-transport inequality.

\begin{lemma}\label{lem:gibbs-low-apx-unif}
    For any constant $C>0$, let
    \begin{equation}\label{eqn:gibbs-m1-def}
        m_1:=m_\ell\vee\lfloor(1-C\kappa_n^{-1})N_{d^*}\rfloor\,.
    \end{equation}
    Then in the low temperature phase \eqref{eqn:low-temp-phase}, we have with probability at least $1-o(1)$ that
    \[
        \mu_{G,\beta}(\{T:|\varphi(T)|\leq m_1\})\leq \exp(-\Omega(n))\,.
    \]
\end{lemma}
\begin{proof}
    The statement follows by Lemma~\ref{lem:large-kernel} if $m_1=m_{\ell}$ so we assume $m_1=\lfloor(1-C\kappa_n^{-1})N_{d^*}\rfloor$. We have
    \[
        \begin{split}
            \log\left(\mu_{G,\beta}(\{T:|\varphi(T)|\leq m_1\})\right) &= \log\left(\sum_{m=m_0}^{m_1}Z_{G,\beta}\|_m\right)-\log Z_{G,\beta}\\
            &\leq\max_{m_0\leq m\leq m_1}\log(Z_{G,\beta}\|_m)-\log Z_{G,\beta}+\log n\,.
        \end{split}
    \]
    We will use $f_u(x)$ defined in \eqref{eqn:gibbs-opt-obj-ub} to upper bound $\tilde{\Phi}_{G,\beta}(m)$. Namely, by Corollary~\ref{cor:fixed-kernel-size}, we have
    \[
        \begin{split}
            \max_{m_0\leq m\leq m_1}\log(Z_{G,\beta}\|_m) &=\max_{m_0\leq m\leq m_1}\tilde{\Phi}_{G,\beta}(m)\\
            &\leq\max_{m_0\leq x\leq m_1}f_u(x)+o_p(n)
        \end{split}
    \]
    and
    \[
        \log Z_{G,\beta}\geq f_u(N_{d^*})\,.
    \]
    We now take a look at $f_u((1-C\kappa_n^{-1})N_{d^*})$. Similar to \eqref{eqn:ltp-nds-pert}, we have
    \begin{equation}\label{eqn:m1-asymp}
        (1-C\kappa_n^{-1})N_{d^*}\geq N_{d^*}-\frac{Cn}{\log\log n}-o_p\left(\frac{n}{\log\log n}\right)
    \end{equation}
    which implies $h((1-C\kappa_n^{-1})N_{d^*})=o_p(n)$. Also similar to \eqref{eqn:ltp-gdiff},
    \[
        g(N_{d^*})-g((1-C\kappa_n^{-1})N_{d^*})=C(\beta-1+\Delta_n+\kappa_n^{-1})n-(C\kappa_n^{-1}+\log(1-C\kappa_n^{-1}))n-o_p(n)\,.
    \]
    We consider two regimes for the low temperature phase. If $\beta>\limsup_{n\to\infty}(1-\Delta_n-\kappa_n^{-1})$, then the RHS is clearly $\Omega_p(n)$. Otherwise, if $\beta=\limsup_{n\to\infty}(1-\Delta_n-\kappa_n^{-1})$ and $\limsup_{n\to\infty}\kappa_n=K<\infty$ for some constant $K$, then $C\kappa_n^{-1}>\frac{C}{2K}$ for all sufficiently large $n$, which again implies that the RHS is $\Omega(n)$. In either case, we arrive at $g(N_{d^*})-g((1-C\kappa_n^{-1})N_{d^*})=\Omega_p(n)$ which gives
    \[
        f_u(N_{d^*})-f_u((1-C\kappa_n^{-1})N_{d^*})=\Omega_p(n)\,.
    \]
    By the concavity of $f_u$ this implies that $f_u(x)$ is increasing in $x\leq(1-C\kappa_n^{-1})N_{d^*}$. As a result, we conclude that
    \[
        \begin{split}
            \log\left(\mu_{G,\beta}(\{T:|\varphi(T)|\leq m_1\})\right) &\leq \max_{m_0\leq m\leq m_1}\log(Z_{G,\beta}\|_m)-\log Z_{G,\beta}+o_p(n)\\
            &\leq f_u((1-C\kappa_n^{-1})N_{d^*})-f_u(N_{d^*})+o_p(n)\\
            &\leq -\Omega_p(n)\,.
        \end{split}
    \]
    proving the lemma.
\end{proof}

\begin{proof}[Proof of Theorem~\ref{thm:gibbs-low-phase}]
    Let $C>0$ be an arbitrary constant and define $m_1$ as in \eqref{eqn:gibbs-m1-def}. Consider the measure $\tilde{\mu}_{G,\beta}\|_{\geq m_1}$ defined by $\tilde{\mu}_{G,\beta}$ conditioned on $|\varphi(T)|\geq m_1$:
    \[
        \tilde{\mu}_{G,\beta}\|_{\geq m_1}(T):=\sum_{A\in\mathcal{K}(G)}\mu_{G,\beta}(\varphi(T)=A\mid|\varphi(T)|\geq m_1)\cdot\mu_{G,\infty}|_A(T)\,.
    \]
    By Lemma~\ref{lem:gibbs-low-apx-unif}, Proposition~\ref{prop:fixed-kernel}, and the convexity of the KL divergence we have
    \begin{equation}\label{eqn:gibbs-lk-kl}
        \KL(\tilde{\mu}_{G,\beta}\|_{\geq m_1},\mu_{G,\beta})=o_p(n)
    \end{equation}
    and by Proposition~\ref{prop:cond-gibbs-w1} and Lemma~\ref{lem:gibbs-low-apx-unif}
    \begin{equation}\label{eqn:main-w11}
        W_1(\tilde{\mu}_{G,\beta}\|_{\geq m_1},\mu_{G,\beta})=o_p(n)\,.
    \end{equation}
    By Proposition~\ref{prop:fixed-kernel} and the Gibbs variational principle (Proposition~\ref{prop:gibbs-variational}) we have
    \begin{equation}\label{eqn:ent-bound}
        \begin{split}
            H_{\tilde{\mu}_{G,\beta}\|_{\geq m_1}}(T) &\geq \log Z_{G,\beta}+\overline{\beta}\E_{\tilde{\mu}_{G,\beta}\|_{\geq m_1}}\left[\sum_v\mathsf{d}_T(1,v)\right]-o_p(n)\\
            &\geq\log (Z_{G,\beta}|_{\Gamma_{d^*}})+\overline{\beta}\sum_v\mathsf{d}_G(1,v)-o_p(n)\\
            &=\sum_{v\in V\setminus\{1\}}\log|\mathsf{par}_G(v)|-o_p(n)
        \end{split}
    \end{equation}
    and by \eqref{eqn:m1-asymp}
    \begin{equation}\label{eqn:gibbs-m1-dbound}
        \E_{\tilde{\mu}_{G,\beta}\|_{\geq m_1}}\left[\sum_v\mathsf{d}_T(1,v)\right]\leq\sum_{v}\mathsf{d}_G(1,v)+\frac{Cn}{\log\log n}+o_p\left(\frac{n}{\log\log n}\right)\,.
    \end{equation}
    Now we round $\tilde{\mu}_{G,\beta}\|_{\geq m_1}$ to a measure $\nu$ supported on the set of shortest path trees using the following procedure. For a spanning tree $T\sim\tilde{\mu}_{G,\beta}\|_{\geq m_1}$, for each vertex $v\in\Gamma_{d^*}\setminus\varphi(T)$, replace $\mathsf{par}_T(v)$ with an arbitrary vertex in $\Gamma_{d^*-1}$. Since $\tilde{\mu}_{G,\beta}\|_{\geq m_1}$ is a mixture of measures of the form $\mu_{G,\infty}|_A$, it is not difficult to see that the resulting graph $S$ is a shortest path tree and that
    \begin{equation}\label{eqn:main-w12}
        W_1(\nu,\tilde{\mu}_{G,\beta}\|_{\geq m_1})=o_p(n)\,.
    \end{equation}
    For $\mathbf{A}=\varphi(T)$ note that
    \[
        H_{\tilde{\mu}_{G,\beta}\|_{\geq m_1}}(T)=H_{\tilde{\mu}_{G,\beta}\|_{\geq m_1}}(\mathbf{A})+H_{\tilde{\mu}_{G,\beta}\|_{\geq m_1}}(T\mid\mathbf{A})
    \]
    where
    \[
        \begin{split}
            H_{\tilde{\mu}_{G,\beta}\|_{\geq m_1}}(\mathbf{A}) &\leq \log\left(\sum_{m=m_1}^{N_{d^*}}\binom{N_{d^*}}{m}\right)\\
            &\leq\log\binom{n}{m_1}+\log n\\
            &=o_p(n)\,.
        \end{split}
    \]
    Thus,
    \begin{equation}\label{eqn:cond-ent-bound}
        H_{\tilde{\mu}_{G,\beta}\|_{\geq m_1}}(T\mid\mathbf{A})\geq H_{\tilde{\mu}_{G,\beta}\|_{\geq m_1}}(T)-o_p(n)\,.
    \end{equation}
    Also,
    \[
        \begin{split}
            H_{\tilde{\mu}_{G,\beta}\|_{\geq m_1}}(T\mid\mathbf{A}=A) &= H_{\mu_{G,\infty}|_A}(T)\\
            &=\sum_{v\in\Gamma_{d^*}\setminus A}H_{\mu_{G,\infty}|_A}(\mathsf{par}_T(v))+\sum_{v\in V\setminus\{1\}\setminus(\Gamma_{d^*}\setminus A)}H_{\mu_{G,\infty}|_A}(\mathsf{par}_T(v))\\
            &\leq|\Gamma_{d^*}\setminus A|\log\log n+\sum_{v\in V\setminus\{1\}\setminus(\Gamma_{d^*}\setminus A)}H_{\mu_{G,\infty}|_A}(\mathsf{par}_T(v))+o_p(n)
        \end{split}
    \]
    which gives
    \[
        \begin{split}
            H_\nu(S\mid\mathbf{A}=A) &\geq \sum_{v\in V\setminus\{1\}\setminus(\Gamma_{d^*}\setminus A)}H_{\mu_{G,\infty}|_A}(\mathsf{par}_T(v))\\
            &\geq H_{\tilde{\mu}_{G,\beta}\|_{\geq m_1}}(T\mid\mathbf{A}=A)-|\Gamma_{d^*}\setminus A|\log\log n-o_p(n)\,.
        \end{split}
    \]
    Thus,
    \[
        \begin{split}
            H_{\nu}(S) &\geq H_{\nu}(S\mid\mathbf{A})\\
            &\geq H_{\tilde{\mu}_{G,\beta}\|_{\geq m_1}}(T\mid\mathbf{A})-(\log\log n)\E_{\tilde{\mu}_{G,\beta}\|_{\geq m_1}}[|\Gamma_{d^*}\setminus\mathbf{A}|]-o_p(n)\,.
        \end{split}
    \]
    Here, \eqref{eqn:gibbs-m1-dbound} gives
    \[
        \begin{split}
            (\log\log n)\E_{\tilde{\mu}_{G,\beta}\|_{\geq m_1}}[|\Gamma_{d^*}\setminus\mathbf{A}|]&\leq (\log\log n)\E_{\tilde{\mu}_{G,\beta}\|_{\geq m_1}}[|v\in V:\mathsf{d}_T(1,v)\neq\mathsf{d}_G(1,v)|]\\
            &\leq Cn+o_p(n)\,.
        \end{split}
    \]
    Combined with \eqref{eqn:ent-bound} and \eqref{eqn:cond-ent-bound}, we get
    \[
        H_{\nu}(S)\geq\sum_{v\in V\setminus\{1\}}\log|\mathsf{par}_G(v)|-Cn-o_p(n)\,.
    \]
    Applying Corollary~\ref{cor:high-entropy-wasserstein} to $\nu$ and $\mu=\mu_{G,\infty}$ gives
    \begin{equation}\label{eqn:w1-small}
        W_1(\nu,\mu_{G,\infty})\leq\sqrt{2C}n+o_p(n)\,.
    \end{equation}
    Since $C$ is arbitrary, this implies $W_1(\nu,\mu_{G,\infty})=o_p(n)$. Together with \eqref{eqn:main-w11} and \eqref{eqn:main-w12} the proposition is proved.
\end{proof}

\subsubsection{Gibbs measures in the high temperature phase}\label{sec:gibbs-high-phase}

Now we move to the high temperature phase. Here, we wish to prove that $\mu_{G,\beta}$ is far from the uniform measure over the shortest path trees.

\begin{theorem}\label{thm:gibbs-high-phase}
    In the high temperature phase \eqref{eqn:high-temp-phase} we have
    \[
        W_1(\mu_{G_n,\infty},\mu_{G_n,\beta})=\Theta_p(n)\,.
    \]
\end{theorem}

To show that the Wasserstein distance is large, we use the dual formulation of Wasserstein metric by explicitly constructing a $1$-Lipschitz function of spanning trees and compare the expectations under $\mu_{G,\beta}$ and $\mu_{G,\infty}$. The key object that distinguishes the high temperature phase from the low temperature phase is again kernels. Recall from Theorem~\ref{thm:logz-formula} that the optimal kernel size that contributes the most to the log partition function is
\[
    m^*=\left\lfloor\frac{N_{d^*}}{(1-\Delta_n-\beta)\kappa_n}\right\rfloor\,.
\]
When there is an asymptotic gap between $\beta$ and $1-\Delta_n-\kappa_n^{-1}$ as in \eqref{eqn:high-temp-phase}, this optimal kernel size is strictly smaller than $N_{d^*}$. We first prove that a typical tree in the high temperature phase has kernel size very close to this $m^*$, thereby is strictly smaller than $N_{d^*}$.

\begin{lemma}\label{lem:gibbs-small-kernel}
    Suppose that $\beta<1-\Delta_n-\kappa_n^{-1}-\delta$ for a constant $\delta>0$. Let
    \begin{equation}\label{eqn:gibbs-m2-def}
        m_2:=\left\lceil\left(1+\frac{\delta}{7}\right)\cdot\frac{\kappa_n^{-1}N_{d^*}}{1-\Delta_n-\beta}\right\rceil\,.
    \end{equation}
    Then we have asymptotically almost surely that
    \[
        \mu_{G,\beta}(\{T:|\varphi(T)|\geq m_2\})\leq\exp\left(-\Omega(n)\right)\,.
    \]
\end{lemma}
\begin{proof}
    We proceed similar to the proof of Lemma~\ref{lem:gibbs-low-apx-unif}. We have
    \[
        \begin{split}
            \log\left(\mu_{G,\beta}(\{T:|\varphi(T)|\geq m_2\})\right) &= \log\left(\sum_{m=m_2}^{N_{d^*}}Z_{G,\beta}\|_m\right)-\log Z_{G,\beta}\\
            &\leq\max_{m_0\vee m_2\leq m\leq N_{d^*}}\log(Z_{G,\beta}\|_m)-\log Z_{G,\beta}+\log n\\
            &\leq \max_{m_2\leq m\leq N_{d^*}}f_u(x)-f_u\left(\frac{\kappa_n^{-1}N_{d^*}}{1-\Delta_n-\beta}\right)+o_p(n)\,.
        \end{split}
    \]
    From \eqref{eqn:htp-nds-pert}, we have
    \[
        h\left(\left(1+\frac{\delta}{7}\right)\cdot\frac{\kappa_n^{-1}N_{d^*}}{1-\Delta_n-\beta}\right)=o_p(n)
    \]
    and by \eqref{eqn:htp-gdiff}
    \[
        \begin{split}
            g\left(\frac{\kappa_n^{-1}N_{d^*}}{1-\Delta_n-\beta}\right)-g\left(\left(1+\frac{\delta}{7}\right)\cdot\frac{\kappa_n^{-1}N_{d^*}}{1-\Delta_n-\beta}\right) &= \left(\frac{\delta}{7}-\log\left(1+\frac{\delta}{7}\right)\right)n-o_p(n)\\
            &=\Omega_p(n)\,.
        \end{split}
    \]
    This in particular gives
    \[
        f_u\left(\frac{\kappa_n^{-1}N_{d^*}}{1-\Delta_n-\beta}\right)-f_u\left(\left(1+\frac{\delta}{7}\right)\cdot\frac{\kappa_n^{-1}N_{d^*}}{1-\Delta_n-\beta}\right)=\Omega_p(n)\,.
    \]
    Since $f_u$ is concave, for all large enough $n$, $f_u(x)$ is decreasing in $x\geq\left(1+\frac{\delta}{7}\right)\cdot\frac{\kappa_n^{-1}N_{d^*}}{1-\Delta_n-\beta}$. Therefore,
    \[
        \begin{split}
            \log\left(\mu_{G,\beta}(\{T:|\varphi(T)|\geq m_2\})\right) &\leq f_u\left(\left(1+\frac{\delta}{7}\right)\cdot\frac{\kappa_n^{-1}N_{d^*}}{1-\Delta_n-\beta}\right)-f_u\left(\frac{\kappa_n^{-1}N_{d^*}}{1-\Delta_n-\beta}\right)+o_p(n)\\
            &\leq-\Omega_p(n)\,.
        \end{split}
    \]

\end{proof}

One might be tempted to try proving that $W_1(\mu_{G,\beta},\mu_{G,\infty})$ is large using the kernel size $|\varphi(T)|$ as the Lipschitz function that witnesses a large shift in Wasserstein metric, since it clearly ``jumps'' when moving between the two phases. However, $|\varphi(T)|$ might not be big enough if $N_{d^*}=o_p(n)$ (i.e., $\lambda_n\to1$), in which case it is not possible to witness a $\Theta(n)$ shift.

Instead, we notice that the kernel must ``support'' most of the remaining vertices. More precisely, from Proposition~\ref{prop:cond-gibbs-dv}, we know that most of the vertices not in the kernel must have distance $d^*+1$, which means that the parents of those vertices must be in the kernel. As a consequence, vertices in the kernel must have large numbers of children, which is the fact we leverage to construct a Lipschitz function. To be specific, we make the following (informal) hypotheses.
\begin{itemize}
    \item In the low temperature phase, the kernel of a typical tree is nearly the entire $\Gamma_{d^*}$. Hence, a typical vertex in $\Gamma_{d^*}$ would have the number of children around $N_{d^*+1}/N_{d^*}$.
    \item In the high temperature phase, a nontrivial portion of $\Gamma_{d^*}$ is not in the kernel. Thus, those in the kernel $A$ would have the number of children around $(n-|A|)/|A|$, and those not in the kernel would mostly have no child.
\end{itemize}
Based on these observations, the Lipschitz function we construct is the sum of absolute deviation from $N_{d^*+1}/N_{d^*}$, given by
\begin{equation}\label{eqn:w1-witness}
    \sum_{v\in\Gamma_{d^*}}\left|\mathsf{ch}_T(v)|-\frac{N_{d^*+1}}{N_{d^*}}\right|\,.
\end{equation}

Thus, we first need to understand how the absolution deviation behaves. For the sake of convenience, we introduce a few notations here. For a fixed kernel $A\in\mathcal{K}(G)$, we define
\begin{equation}\label{eqn:def-dav}
    \mathsf{d}_A(v):=\begin{cases}
        \mathsf{d}_G(1,v) & \text{if $\mathsf{d}_G(1,v)\leq d^*-1$,}\\
        d^* + \mathsf{d}_G(v,A) & \text{if $v\in\overline{V}\setminus\Gamma_{\leq d^*-1}$.}
    \end{cases}
\end{equation}
as the distance vector of any $T$ in the support of $\mu_{G,\infty}|_A$. Also, we write
\begin{equation}\label{eqn:def-nda}
    N_{d}^A:=|\{v:\mathsf{d}_A(v)=d\}|\,.
\end{equation}

\begin{lemma}\label{lem:multinom-conc}
    Conditioned on $A\subseteq\Gamma_{d^*}$, for any $t,\epsilon>0$, we have
    \[
        \E_{\mu_{G,\infty}|_A}\left[\sum_{v\in A}\left||\mathsf{ch}_T(v)|-\frac{N_{d^*+1}^A}{|A|}\right|\right]\leq t+2n\epsilon
    \]
    with probability at least
    \[
        1-2^{|A|}\epsilon^{-1}\exp\left(-\frac{t^2}{2(n-|A|)}\right)\,.
    \]
\end{lemma}
\begin{proof}
    For notational convenience we define
    \[
        f_{G,A}(T):=\sum_{v\in A}\left||\mathsf{ch}_T(v)|-\frac{N_{d^*+1}^A}{|A|}\right|\,.
    \]
    We first note that
    \[
        \sum_{v\in A}|\mathsf{ch}_T(v)|=N_{d^*+1}^A\,.
    \]
    Hence,
    \[
        \begin{split}
            f_{G,A}(T)&=\sum_{v\in A}\left||\mathsf{ch}_T(v)|-\frac{N_{d^*+1}^A}{|A|}\right|\\
            &= 2\max_{B\subseteq A}\sum_{v\in B}\left(|\mathsf{ch}_T(v)|-\frac{N_{d^*+1}^A}{|A|}\right)\\
            &=2\max_{B\subseteq A}\left\{\sum_{v\in B}|\mathsf{ch}_T(v)|-\frac{|B|N_{d^*+1}^A}{|A|}\right\}\\
            &=2\max_{B\subseteq A}\left\{\sum_{u\in\Gamma_{d^*+1}^A}\mathbf{1}_{\mathsf{par}_T(u)\in B}-\frac{|B|N_{d^*+1}^A}{|A|}\right\}
        \end{split}
    \]
    Conditioned on $\Gamma_{d^*}$, $A\subseteq\Gamma_{d^*}$, and $\Gamma_{d^*+1}^A$, the events $\mathsf{par}_T(u)\in B$ are independent over all $u\in\Gamma_{d^*+1}^A$ and have probability $|B|/|A|$. Hence, $\sum_{u\in\Gamma_{d^*+1}^A}\mathbf{1}_{\mathsf{par}_T(u)\in B}$ is conditionally $\Binom(N_{d^*+1}^A,|B|/|A|)$. Thus, Hoeffding's inequality gives
    \[
        \E\left[\mu_{G,\infty}|_A\left(\left\{T:\sum_{u\in\Gamma_{d^*+1}^A}\mathbf{1}_{\mathsf{par}_T(u)\in B}-\frac{|B|N_{d^*+1}^A}{|A|}\geq t\right\}\right)\,\middle|\, A\subseteq\Gamma_{d^*}, N_{d^*+1}^A\right]\leq\exp\left(-\frac{2t^2}{N_{d^*+1}^A}\right)\,.
    \]
    Taking the union bound, we get
    \[
        \begin{split}
            \E\left[\mu_{G,\infty}|_A(\{T:f_{G,A}(T)\geq t\}\,\middle|\, A\subseteq\Gamma_{d^*},N_{d^*+1}^A\right]
            &\leq 2^{|A|}\exp\left(-\frac{t^2}{2N_{d^*+1}^A}\right)\\
            &\leq 2^{|A|}\exp\left(-\frac{t^2}{2(n-|A|)}\right)
        \end{split}
    \]
    By Markov's inequality, we have
    \[
        \Pr\left(\mu_{G,\infty}|_A(\{T:f_{G,A}(T)\geq t\})\geq\epsilon\,\middle|\,A\subseteq\Gamma_{d^*}\right)\leq 2^{|A|}\epsilon^{-1}\exp\left(-\frac{t^2}{2(n-|A|)}\right)\,.
    \]
    Note that
    \[
        \begin{split}
            f_{G,A}(T)&\leq\sum_{v\in A}|\mathsf{ch}_T(v)|+\sum_{v\in A}\frac{N_{d^*+1}^A}{|A|}\\
            &=2N_{d^*+1}^A\\
            &\leq 2n\,.
        \end{split}
    \]
    Thus, $\mu_{G,\infty}|_A(\{T:f_{G,A}(T)\geq t\})<\epsilon$ implies
    \[
        \E_{\mu_{G,\infty}|_A}[f_{G,A}(T)]< t+2n\epsilon\,.
    \]
    This gives
    \[
        \Pr\left(\E_{\mu_{G,\infty}|_A}[f_{G,A}(T)]\geq t+2n\epsilon\,\middle|\,A\subseteq\Gamma_{d^*}\right)\leq 2^{|A|}\epsilon^{-1}\exp\left(-\frac{t^2}{2(n-|A|)}\right)
    \]
    which is as desired.
\end{proof}

With this lemma, we prove our main result by analyzing \eqref{eqn:w1-witness}.

\begin{proof}[Proof of Theorem~\ref{thm:gibbs-high-phase}]
    Let $m_2$ be defined in \eqref{eqn:gibbs-m2-def}. By \eqref{eqn:htp-nds-pert}, we may assume $m_\ell\leq m_2$ which holds asymptotically almost surely. We consider a measure $\tilde{\mu}_{G,\beta}\|_{[m_\ell,m_2]}$ defined by $\tilde{\mu}_{G,\beta}$ conditioned on $m_\ell\leq|\varphi(T)|\leq m_2$. In other words,
    \[
        \tilde{\mu}_{G,\beta}\|_{[m_\ell,m_2]}(T):=\sum_{A\in\mathcal{K}(G)}\mu_{G,\beta}(\varphi(T)=A\mid m_\ell\leq|\varphi(T)|\leq m_2)\cdot\mu_{G,\infty}|_A(T)\,.
    \]
    Lemma~\ref{lem:large-kernel}, Lemma~\ref{lem:gibbs-small-kernel}, Proposition~\ref{prop:fixed-kernel}, and the convexity of the KL divergence give
    \[
        \KL(\tilde{\mu}_{G,\beta}\|_{[m_\ell,m_2]},\mu_{G,\beta})=o_p(n)
    \]
    and Proposition~\ref{prop:cond-gibbs-w1}, Lemma~\ref{lem:large-kernel}, and Lemma~\ref{lem:gibbs-small-kernel} give
    \[
        W_1(\tilde{\mu}_{G,\beta}\|_{[m_\ell,m_2]},\mu_{G,\beta})=o_p(n)\,.
    \]
    Now our goal is to prove that $\mu_{G,\infty}$ and $\tilde{\mu}_{G,\beta}\|_{[m_\ell,m_2]}$ are far apart in Wasserstein distance. Thus, it suffices to find a $1$-Lipschitz function $f$ such that
    \[
        \E_{\mu_{G,\infty}|_A}[f(T)]-\E_{\mu_{G,\infty}}[f(T)]=\Theta_p(n)
    \]
    simultaneously for all $A\in\mathcal{K}(G)$ with $m_\ell\leq|A|\leq m_2$.
    
    We define a function $f$ on the set of spanning trees of $\overline{G}$ by
    \[
        f(T):=\sum_{v\in\Gamma_{d^*}}\left||\mathsf{ch}_T(v)|-\frac{N_{d^*+1}}{N_{d^*}}\right|\,.
    \]
    Since removing or adding an edge in $T$ changes one of $\mathsf{ch}_T(v)$ by at most $1$, $f$ is $1$-Lipschitz. Using Lemma~\ref{lem:multinom-conc} with $t=\sqrt{2N_{d^*}(n-N_{d^*})}$ and $\epsilon=e^{-N_{d^*}/8}$, we get by Theorem~\ref{thm:main1-tight}
    \begin{equation}\label{eqn:ef-first-bound}
        \begin{split}
            \E_{\mu_{G,\infty}}[f(T)]&\leq\sqrt{2N_{d^*}(n-N_{d^*})}+o_p(n)\\
            &=\sqrt{2\cdot\frac{N_{d^*}}{n}\cdot\left(1-\frac{N_{d^*}}{n}\right)}\cdot n+o_p(n)\\
            &=\sqrt{2\lambda_n(1-\lambda_ n)}\cdot n+o_p(n)\\
            &=\sqrt{2\cdot\frac{\kappa_n}{\log\log n}\left(1-\frac{\kappa_n}{\log\log n}\right)}\cdot n+o_p(n)
        \end{split}
    \end{equation}
    conditioned on $N_{d^*}$ with probability at least
    \[
        \begin{split}
            1-2^{N_{d^*}}\epsilon^{-1}\exp\left(-\frac{t^2}{2(n-N_{d^*})}\right) &= 1-2^{N_{d^*}}e^{N_{d^*}/8}e^{-N_{d^*}}\\
            &\geq 1-e^{-N_{d^*}/8}\\
            &=1-o_p(1)\,.
        \end{split}
    \]
    
    Next, we analyze $\E_{\mu_{G,\infty}|_A}[f(T)]$. We write
    \[
        \begin{split}
            \E_{\mu_{G,\infty}|_A}[f(T)] &= \E_{\mu_{G,\infty}|_A}\left[\sum_{v\in\Gamma_{d^*}}\left||\mathsf{ch}_T(v)|-\frac{N_{d^*+1}}{N_{d^*}}\right|\right]\\
            &\geq\E_{\mu_{G,\infty}|_A}\left[\sum_{v\in A}\left||\mathsf{ch}_T(v)|-\frac{N_{d^*+1}}{N_{d^*}}\right|\right]\\
            &\geq|A|\left|\frac{N_{d^*+1}^A}{|A|}-\frac{N_{d^*+1}}{N_{d^*}}\right|-\E_{\mu_{G,\infty}|_A}\left[\sum_{v\in A}\left||\mathsf{ch}_T(v)|-\frac{N_{d^*+1}^A}{|A|}\right|\right]\\
            &=\left|N_{d^*+1}^A-|A|\cdot\frac{N_{d^*+1}}{N_{d^*}}\right|-\E_{\mu_{G,\infty}|_A}\left[\sum_{v\in A}\left||\mathsf{ch}_T(v)|-\frac{N_{d^*+1}^A}{|A|}\right|\right]\,.
        \end{split}
    \]
    For the first term in the RHS, Lemma~\ref{lem:gamma-a} gives
    \begin{equation}\label{eqn:nda-bound}
        \begin{split}
            \left|N_{d^*+1}^A-|A|\cdot\frac{N_{d^*+1}}{N_{d^*}}\right| &\geq N_{d^*+1}^A-|A|\cdot\frac{N_{d^*+1}}{N_{d^*}}\\
            &\geq n-|A|-|A|\cdot\frac{N_{d^*+1}}{N_{d^*}}-o_p\left(\frac{n}{\log\log n}\right)\\
            &\geq n\left(1-\frac{|A|}{N_{d^*}}\right)-o_p\left(\frac{n}{\log\log n}\right)\\
            &\geq\frac{6}{7}\cdot\frac{\delta\kappa_n}{1+\delta\kappa_n}\cdot n-o_p\left(\frac{n}{\log\log n}\right)
        \end{split}
    \end{equation}
    where the last inequality is due to $|A|\leq m_2$. Also, applying Lemma~\ref{lem:multinom-conc} with $t=n/(\log\log n)^{1/4}$ and $\epsilon=1/n$ gives
    \begin{equation}\label{eqn:ef-bound}
        \begin{split}
            \E_{\mu_{G,\infty}|_A}\left[\sum_{v\in A}\left||\mathsf{ch}_T(v)|-\frac{N_{d^*+1}^A}{|A|}\right|\right] &\leq \frac{n}{(\log\log n)^{1/4}}+2\\
            &=o(n)
        \end{split}
    \end{equation}
    with probability at least
    \[
        \begin{split}
            1-2^{|A|}\epsilon^{-1}\exp\left(-\frac{t^2}{2(n-|A|)}\right) &\geq 1-\exp\left(-\frac{t^2}{2n}+|A|+\log(1/\epsilon)\right)\\
            &\geq1-\exp\left(-\frac{n}{2\sqrt{\log\log n}}+m_2+\log n\right)\,.
        \end{split}
    \]
    Note that \eqref{eqn:nda-bound} and \eqref{eqn:ef-bound} imply
    \begin{equation}\label{eqn:ef-second-bound}
        \E_{\mu_{G,\infty}|_A}[f(T)]\geq\frac{\delta\kappa_n/2}{1+\delta\kappa_n/2}\cdot n-o_p(n)\,.
    \end{equation}
    Now we will apply the union bound for all $A\in\mathcal{K}(G)$ with $m_\ell\leq |A|\leq m_2$. Note that
    \[
        \begin{split}
            |\{A\in\mathcal{K}(G):m_\ell\leq |A|\leq m_2\}| &\leq \sum_{m=m_\ell}^{m_2}\binom{n}{m}\\
            &\leq\sum_{m=m_\ell}^{m_2}\left(\frac{en}{m}\right)^m\\
            &\leq n\left(\frac{en}{m_\ell}\right)^{m_2}\\
            &\leq n\left(3e(\log\log n)^2\right)^{m_2}\,.
        \end{split}
    \]
    Hence, it follows that \eqref{eqn:ef-second-bound} holds simultaneously for all $A\in\mathcal{K}(G)$ with $m_\ell\leq |A|\leq m_2$ with probability at least
    \[
        1-\exp\left(-\frac{n}{2\sqrt{\log\log n}}+m_2+2\log n+m_2\log\left(3e(\log\log n)^2\right)\right)\,.
    \]
    By \eqref{eqn:htp-nds-pert}, this is
    \[
        1-\exp\left(-\frac{n}{2\sqrt{\log\log n}}+o_p\left(\frac{n}{\sqrt{\log\log n}}\right)\right)=1-o_p(1)\,.
    \]
    Hence, asymptotically almost surely, we have both \eqref{eqn:ef-first-bound} and \eqref{eqn:ef-second-bound}. This implies
    \begin{equation}\label{eqn:w1-lip-diff}
        \E_{\mu_{G,\infty}|_A}[f(T)]-\E_{\mu_{G,\infty}}[f(T)] \geq \left(\frac{6}{7}\cdot\frac{\delta\kappa_n}{1+\delta\kappa_n}-\sqrt{2\cdot\frac{\kappa_n}{\log\log n}\left(1-\frac{\kappa_n}{\log\log n}\right)}\right)n-o_p(n)
    \end{equation}
    simultaneously for all $A\in\mathcal{K}(G)$ with $m_\ell\leq|A|\leq m_2$.
    
    Now it remains to prove that
    \[
        \frac{6}{7}\cdot\frac{\delta\kappa_n}{1+\delta\kappa_n}-\sqrt{2\cdot\frac{\kappa_n}{\log\log n}\left(1-\frac{\kappa_n}{\log\log n}\right)}
    \]
    is bounded below by a constant. First, recall our assumption $\beta<1-\Delta_n-\kappa_n^{-1}-\delta$ which implies $\kappa_n>1$. This gives
    \begin{equation}\label{eqn:kn-case1}
        \sqrt{2\cdot\frac{\kappa_n}{\log\log n}\left(1-\frac{\kappa_n}{\log\log n}\right)}\leq\sqrt{\frac{2}{\log\log n}}\cdot\kappa_n\,.
    \end{equation}
    On the other hand, the AM--GM inequality gives
    \begin{equation}\label{eqn:kn-case2}
        \sqrt{2\cdot\frac{\kappa_n}{\log\log n}\left(1-\frac{\kappa_n}{\log\log n}\right)}\leq\frac{1}{\sqrt{2}}\,.
    \end{equation}
    Now we consider two cases. If $\delta\kappa_n<7$, then for all sufficiently large $n$ such that $\sqrt{2/\log\log n}\leq\delta/28$, \eqref{eqn:kn-case1} gives
    \[
        \begin{split}
            \frac{6}{7}\cdot\frac{\delta\kappa_n}{1+\delta\kappa_n}-\sqrt{2\cdot\frac{\kappa_n}{\log\log n}\left(1-\frac{\kappa_n}{\log\log n}\right)}&\geq\frac{3\delta\kappa_n}{28}-\sqrt{\frac{2}{\log\log n}}\cdot\kappa_n\\
            &\geq\frac{\delta\kappa_n}{14}\\
            &\geq\frac{\delta}{14}\,.
        \end{split}
    \]
    Otherwise, if $\delta\kappa_n\geq7$, then \eqref{eqn:kn-case2} gives
    \begin{equation}\label{eqn:lip-diff-global}
        \frac{6}{7}\cdot\frac{\delta\kappa_n}{1+\delta\kappa_n}-\sqrt{2\cdot\frac{\kappa_n}{\log\log n}\left(1-\frac{\kappa_n}{\log\log n}\right)}\geq\frac{3}{4}-\frac{1}{\sqrt{2}}\,.
    \end{equation}
    This completes the proof.
\end{proof}

\begin{remark}
    Our lower bound on the Wasserstein distance can be written only in terms of $\delta$, which specifies the gap between $\beta$ and $1-\Delta_n-\kappa_n^{-1}$. This proves the second part of Theorem~\ref{thm:intro-phase-transition} which is slightly stronger than Theorem~\ref{thm:gibbs-high-phase} on its own.
\end{remark}

\subsubsection{Replica symmetry}\label{sec:replica-symmetry}

As a consequence of our analysis, we show that our system $\mu_{G,\beta}$ is replica symmetric in the sense of the following theorem. Namely, this states that whenever $\lambda\in\{0,1\}$ or the temperature is bounded away from the critical temperature, then the overlap of two independently sampled trees concentrates.


\begin{theorem}[Replica symmetry]\label{thm:replica-symmetry}
    Let $\beta>0$ be fixed and assume $\lambda_n\to\lambda\in[0,1]$. Suppose that for given $\{G_n\}$ we independently sample $T_n,T_n'\sim\mu_{G_n,\beta}$. Then
    \[
        \mathsf{R}(T_n,T_n')\pto\begin{dcases}
            0&\text{if $\lambda\in\{0,1\}$ or $\beta<1$,}\\
            f(\lambda)&\text{if $\lambda\in(0,1)$ and $\beta>1$}
        \end{dcases}
    \]
    where
    \[
        f(\lambda)=(1-\lambda)\E[1/X\mid X>0]\,,\quad X\sim\Pois(\log(1/\lambda))\,.
    \]
\end{theorem}
\begin{proof}
    First, consider the case $\lambda\in\{0,1\}$ or $\beta<1$. By Corollary~\ref{cor:gibbs-w1}, we may assume $T_n,T_n'\sim\tilde{\mu}_{G,\beta}$. In addition, by Lemma~\ref{lem:large-kernel}, it suffices to show that for any $\epsilon>0$
    \[
        \Pr(\mathsf{R}(T,T')>\epsilon\mid G,T',|\varphi(T)|\geq m_\ell)=o_p(1)
    \]
    for any spanning tree $T'$ of $\overline{G}$. Note that
    \begin{equation}\label{eqn:rs-pr-bound}
        \Pr(\mathsf{R}(T,T')>\epsilon\mid G,T',|\varphi(T)|\geq m_\ell) \leq\max_{A\in\mathcal{K}(G):|A|\geq m_\ell}\Pr(\mathsf{R}(T,T')>\epsilon\mid G,T',\varphi(T)=A)\,.
    \end{equation}
    Thus, our task is to bound $\Pr(\mathsf{R}(T,T')>\epsilon\mid G,T',\varphi(T)=A)$ where it is guaranteed that $|A|\geq m_\ell$. By our assumption, the law of $T$ conditioned on $\varphi(T)=A$ is $\mu_{G,\infty}|_A$ which is a product measure. Define
    \[
        \mathsf{par}_A(v):=\{u\in\overline{V}:\mathsf{d}_A(1,u)=\mathsf{d}_A(1,v)+1\}
    \]
    where $\mathsf{d}_A$ is defined in \eqref{eqn:def-dav}. Then $\mathsf{R}(T,T')$ can be written as the mean of (conditionally) independent random variables
    \begin{equation}\label{eqn:rs-overlap}
        \mathsf{R}(T,T')=\frac{1}{|\overline{V}|-1}\sum_{v\in\overline{V}\setminus\{1\}}\mathbf{1}_{\mathsf{par}_T(v)=\mathsf{par}_{T'}(v)}
    \end{equation}
    where
    \[
        \mathbf{1}_{\mathsf{par}_T(v)=\mathsf{par}_{T'}(v)}\sim\Bernoulli\left(\frac{\mathbf{1}_{\mathsf{par}_{T'}(v)\in\mathsf{par}_A(v)}}{|\mathsf{par}_A(v)|}\right)\,.
    \]
    Since we need an upper bound, we can upper bound this using the sum of $\Bernoulli(1/|\mathsf{par}_A(v)|)$. Thus, the quantity of interest here is the mean
    \[
        p_{G,A}:=\frac{1}{|\overline{V}|-1}\sum_{v\in\overline{V}\setminus\{1\}}\frac{1}{|\mathsf{par}_A(v)|}\,.
    \]
    In light of \eqref{eqn:rs-pr-bound}, it suffices to show that with high probability, $p_{G,A}=o(1)$ simultaneously for all $A\in\mathcal{K}(G)$ with $|A|\geq m_\ell$. By Theorem~\ref{thm:main1-tight}, with high probability,
    \[
        \begin{split}
            (|\overline{V}|-1)p_{G,A} &= \sum_{v\in\overline{V}\setminus\{1\}}\frac{1}{|\mathsf{par}_A(v)|}\\
            &\leq\sum_{v\in A}\frac{1}{|\mathsf{par}_G(v)|}+\sum_{v\in\overline{V}\setminus\Gamma_{\leq d^*-1}\setminus A}\frac{1}{|\mathsf{N}_A(v)|\vee1}+o(n)\,.
        \end{split}
    \]
    We bound the two sums in the RHS separately, starting with the first sum. If $\lambda\in(0,1)$ and $\beta<1$, then we are in the high temperature phase \eqref{eqn:high-temp-phase} so Lemma~\ref{lem:gibbs-small-kernel} we may assume $|A|=o(n)$ for all $A$, which immediately yields $\sum_{v\in A}1/|\mathsf{par}_G(v)|\leq|A|=o_p(n)$. otherwise, if $\lambda\in\{0,1\}$, then we use
    \[
        \sum_{v\in A}\frac{1}{|\mathsf{par}_G(v)|}\leq\sum_{v\in\Gamma_{d^*}}\frac{1}{|\mathsf{par}_G(v)|}\,.
    \]
    Recall that conditioned on $\Gamma_{\leq d^*-1}$, $|\mathsf{par}_G(v)|$ are (conditionally) i.i.d. following $\Binom(N_{d^*-1},q)$. If $\lambda=0$, then $qN_{d^*-1}\pto\infty$ so it is straightforward to see $\sum_{v\in\Gamma_{d^*}}1/|\mathsf{par}_G(v)|=o_p(n)$, and if $\lambda=1$, then $N_{d^*}=o_p(n)$ so the same bound obviously holds. Now it remains to deal with the term
    \begin{equation}\label{eqn:1onav}
        \sum_{v\in\overline{V}\setminus\Gamma_{\leq d^*-1}\setminus A}\frac{1}{|\mathsf{N}_A(v)|\vee1}\,.
    \end{equation}
    Replacing $A$ with any of its superset decreases the sum, so it suffices to show that this is $o(n)$ simultaneouesly for all $A$ of size exactly $\lceil\frac{n}{3(\log\log n)^2}\rceil$. Since $|\mathsf{N}_A(v)|$ are i.i.d. $\Binom(|A|,q)$ conditioned on $\Gamma_{\leq d^*-1}$ and $A$, we define
    \[
        g(m,p)=\E\left[\frac{1}{X\vee1}\right]\,,\quad X\sim\Binom(m,p)
    \]
    so that by Hoeffding's inequality
    \[
        \Pr\left(\sum_{v\in\overline{V}\setminus\Gamma_{\leq d^*-1}\setminus A}\frac{1}{|\mathsf{N}_A(v)|\vee1}\geq g(|A|,q)+t\,\middle|\,\Gamma_{\leq d^*-1},A\right)\leq e^{-\frac{2t^2}{n}}\,.
    \]
    Since $|A|\geq m_\ell$ we have $|A|q\to\infty$ so $g(|A|,q)=o(n)$, and by plugging in $t=n/\log\log n$, \eqref{eqn:nchooseml} gives that \eqref{eqn:1onav} is $o_p(n)$ as desired.

    Now assume $\lambda\in(0,1)$ and $\beta>1$. Then we are in the low temperature phase \eqref{eqn:low-temp-phase}, so by Theorem~\ref{thm:gibbs-low-phase} we may assume $T,T'\sim\mu_{G,\infty}$. Then \eqref{eqn:rs-overlap} is the sum of independent random variables $\mathbf{1}_{\mathsf{par}_T(v)=\mathsf{par}_{T'}(v)}\sim\Bernoulli(1/|\mathsf{par}_G(v)|)$. By the same argument, it is straightforward to see that this can be written as
    \[
        \frac{1}{|\overline{V}|-1}\sum_{v\in\Gamma_{d^*}}\frac{1}{|\mathsf{par}_G(v)|}=\frac{N_{d^*}}{|\overline{V}|-1}\cdot\frac{1}{N_{d^*}}\sum_{v\in\Gamma_{d^*}}\frac{1}{|\mathsf{par}_G(v)|}+o_p(1)\,.
    \]
    From Theorem~\ref{thm:main1-tight} we have $N_{d^*}/(|\overline{V}|-1)\pto1-\lambda$ and $\frac{1}{N_{d^*}}\sum_{v\in\Gamma_{d^*}}\frac{1}{|\mathsf{par}_G(v)|}$ concentrates around
    \[
        \E\left[\sum_{v\in\Gamma_{d^*}}\frac{1}{|\mathsf{par}_G(v)|}\,\middle|\,\Gamma_{\leq d^*-1},\Gamma_{d^*}\right]
    \]
    which again converges to $f(\lambda)$, since, $\Binom(N_{d^*-1},q)\pto\Pois(\log(1/\lambda))$.
\end{proof}

\subsection{Phase transitions}\label{sec:phase-transition}

Our results have shown that, both in KL divergence and Wasserstein distance, the Gibbs measures behave quite differently in the low temperature phase and the high temperature phase. This is fundamentally a phase transition widely studied in statistical physics. In this section, we further investigate the nature of phase transitions occurring in our model of shortest paths in sparse random graphs.

\subsubsection{The free energy density}\label{sec:free-energy-density}

The usual notion of \emph{free energy} in statistical physics is the negative log partition function multiplied by the temperature. In our context, the free energy can be expressed as
\[
    -\overline{\beta}^{-1}\Phi_{G,\beta}=-\frac{1}{\beta\log\log n}\log Z_{G,\beta}\,.
\]
Recall that by Theorem~\ref{thm:ground-state-energy} the ground state energy is approximately $(d^*+\lambda_n)n$ so it is reasonable to subtract off this term to see more precisely how the free energy changes relative to the ground state. As such, we define the relative free energy density as
\[
    \mathcal{F}_{G,\beta}:=-\frac{1}{\beta n\log\log n}\log Z_{G,\beta}-(d_n^*+\lambda_n)\,.
\]
We first establish a formula for the limiting free energy density.

\begin{corollary}\label{cor:free-energy}
    Let $\beta>0$, $\Delta\in[0, 1]$, and suppose that $\Delta_n\to\Delta$ and $\lambda_n\to\lambda$. Then the free energy density converges in probability to
    \[
        \mathcal{F}_{G_n,\beta}\pto\mathcal{F}_{\lambda,\Delta}(\beta)
    \]
    where $\mathcal{F}_{\lambda,\Delta}$ is a continuous function of $\beta$ and is one of the following.
    \begin{enumerate}[label=(\Alph*)]
        \item If $\Delta=1$, then $\lambda=0$ and
        \[
            \mathcal{F}_{\lambda,\Delta}(\beta)=-\frac{1}{\beta}\,.
        \]
        \item\label{item:supregime2} If $\Delta\in[0, 1)$ and $\lambda\in[0, 1)$, then
        \[
            \mathcal{F}_{\lambda,\Delta}(\beta)=\begin{dcases}
                -\frac{\lambda+\Delta}{\beta}&\text{if $\beta\geq1-\Delta$,}\\
                \frac{1-(
                \lambda+\Delta)}{1-\Delta}-\frac{1}{\beta}&\text{if $\beta\leq1-\Delta$.}
            \end{dcases}
        \]
        
        \item If $\lambda=1$, then $\Delta=0$ and
        \[
            \mathcal{F}_{\lambda,\Delta}(\beta)=-\frac{1}{\beta}\,.
        \]

    \end{enumerate}
\end{corollary}

If we directly use Theorem~\ref{thm:logz-formula} to prove this, we miss the case $\beta=1-\Delta$. Instead, it is easier and perhaps more illuminating to start from more primitive results, namely Lemma~\ref{lem:free-core} and Lemma~\ref{lem:gibbs-cond-lb}. Together with \eqref{eqn:reg2}, it is not difficult to deduce the following corollaries.

\begin{lemma}\label{lem:logz-ub-lglg}
    With probability at least $1-o(1)$, we have
    \[
        \frac{1}{\log\log n}\log(Z_{G,\beta}|_A)\leq n-(1-\Delta-\beta)|A|-\beta(d^*+1)n+o(n)
    \]
    simultaneously for all $A\in\mathcal{K}(G)$.
\end{lemma}

\begin{lemma}\label{lem:logz-lb-lglg}
    With probability at least $1-o(1)$, we have
    \[
        \frac{1}{\log\log n}\log(Z_{G,\beta}|_A)\geq n-(1-\Delta-\beta)|A|-\beta(d^*+1)n-o(n)
    \]
    for all $A\in\mathcal{K}(G)$ with $|A|\geq m_\ell$.
\end{lemma}

From Section~\ref{sec:gibbs-1d-opt}, we already know that we only need to consider kernels of size at least $m_\ell$ to approximate $\log Z_{G,\beta}$. Thus,
\[
    \frac{1}{\log\log n}\log Z_{G,\beta}=\max_{m_\ell\leq m\leq N_{d^*}}\{n-(1-\Delta-\beta)m-\beta(d^*+1)n\}+o_p(n)\,.
\]
The objective we are maximizing over is simply a linear function, so this is easy to solve and yields Corollary~\ref{cor:free-energy}.

For the case $\Delta\in[0,1)$ and $\lambda\in[0,1)$, the free energy density as a function of $\beta$ has a non-differentiable point $\beta=1-\Delta$, indicating the existence of \emph{first-order phase transition}. Interestingly, zooming in the free energy by $\log\log n$ reveals another form of phase transition. This is easily obtained as a direct consequence of Theorem~\ref{thm:logz-formula}, simply by splitting the formula into different parameter regimes based on the asymptotics of $\Delta_n$, $\lambda_n$, and $\kappa_n$. Note that the ground state energy is still $(d^*+\lambda_n)n$ even after zooming in, since the error term in Theorem~\ref{thm:ground-state-energy} is $o_p(n/\log\log n)$.

\begin{corollary}\label{cor:logz-phase}
    Let $\beta>0$, $\Delta\in[0,1]$, $\lambda\in[0,1]$, and $\kappa\geq0$ be constants. Suppose that $\Delta_n\to\Delta$ and $\lambda_n\to\lambda$. Then
    \[
        \mathcal{F}_{G_n,\beta}\log\log n=-\frac{1}{\beta}\zeta_{G_n,\beta}+o_p(1)
    \]
    where $\zeta_{G_n,\beta}$ can be one of the following.
    \begin{enumerate}[label=(\Alph*)]
        \item\label{item:regime1} If $\Delta=1$, then $\lambda=0$ and
        \[
            \zeta_{G_n,\beta}=(\lambda_n+\Delta_n)\log\log n+\log\alpha_n\,.
        \]
        \item\label{item:regime2} We further split the case $\Delta\in[0,1)$ and $\lambda\in[0,1)$ into the following.
            \begin{enumerate}[label=(\Alph{enumi}.\roman*), ref=(\Alph{enumi}.\roman*), align=left]
            \item\label{item:regime2-1} If $\Delta\in[0,1)$ and $\lambda=0$, then
                \[
                    \zeta_{G_n,\beta}=\begin{dcases}
                        (\lambda_n+\Delta_n)\log\log n+\Delta\log\alpha_n&\text{if $\beta>1-\Delta$,}\\
                        (1-\beta+\beta\lambda_n)\log\log n-\log\log\log n+\log\left(\frac{\alpha_n}{1-\Delta-\beta}\right) -1&\text{if $\beta<1-\Delta$.}
                    \end{dcases}
                \]
        
                \item\label{item:regime2-2} If $\Delta=0$, $\lambda\in(0,1)$, then
                \[
                    \zeta_{G_n,\beta}= \begin{dcases}
                    \lambda_n\log\log n+\lambda\log((1-\lambda)\alpha_n) + \Psi_0(\lambda)& \text{if $\beta > 1$,} \\
                    (1-\beta+\beta\lambda_n)\log\log n-\log\log\log n+\log\left(\frac{\alpha_n}{1 - \beta}\right) - 1 & \text{if $\beta < 1$.}
                    \end{dcases}
                \]
        \end{enumerate}

        
        \item\label{item:regime3} If $\lambda=1$, then $\Delta=0$ and
        \[
            \zeta_{G_n,\beta}=\begin{dcases}
                \log\log n-\log\log\log n+\log(\kappa_n\alpha_n) -\kappa_n&\text{if $\beta>1-\kappa^{-1}$,}\\
                \log\log n-\log\log\log n+\log\left(\frac{\alpha_n}{1-\beta}\right)-\kappa_n\beta-1&\text{if $\beta<1-\kappa^{-1}$}
            \end{dcases}
        \]
        where $\kappa_n\in[0,\infty]$.
        \begin{enumerate}[label=(\Alph{enumi}.\roman*), ref=(\Alph{enumi}.\roman*), align=left]
            \item\label{item:regime3-1} If $\kappa=\infty$, then $\kappa^{-1}$ is replaced with zero.
            \item\label{item:regime3-2} If $\kappa\in(1,\infty)$, then $\kappa_n$ in the formula can be replaced with $\kappa$.
            \item\label{item:regime3-3} If $\kappa\in[0,1]$, then $\zeta_{G_n,\beta}$ can be set to either formula for any value of $\beta>0$, and can equivalently be written as
            \[
                \zeta_{G_n,\beta}=(\lambda_n+\Delta_n)\log\log n+\log\alpha_n\,.
            \]
        \end{enumerate}

    \end{enumerate}
\end{corollary}

The regimes \ref{item:regime1}, \ref{item:regime2}, and \ref{item:regime3} are set to match the corresponding cases in Corollary~\ref{cor:free-energy}. The most interesting regimes are \ref{item:regime3-1} and \ref{item:regime3-2}, which are something new that is not visible only from Corollary~\ref{cor:free-energy}. An especially intriguing regime is \ref{item:regime3-2}, which turns out to exhibit a \emph{second-order phase transition} which we will further discuss in the following section.



\subsubsection{Phase transitions in Wasserstein space}

An important feature of the Gibbs measures in the regimes \ref{item:regime2}, \ref{item:regime3-1}, and \ref{item:regime3-2} is that as temperature increases (i.e., $\beta$ decreases), the distance between $\mu_{G,\beta}$ and the uniform measure $\mu_{G,\infty}$ suddenly changes. As such, we can examine the phase transition phenomena by viewing the Gibbs measures $\mu_{G,\beta}$ as a function $\beta\mapsto\mu_{G,\beta}$ to the space of measures on spanning trees.

Formally, for a fixed graph $G$, we let $\mathcal{P}_1(\mathcal{T}_G)$ to be the space of probability measures on the spanning trees of the component $\overline{G}$ of $G$ containing a fixed vertex $1$, equipped with the normalized $1$-Wasserstein metric
\[
    \mathsf{w}_1(\mu,\nu):=\frac{1}{n}W_1(\mu, \nu),\qquad\mu,\nu\in\mathcal{P}_1(\mathcal{T}_G)\,.
\]
Also, we define the \emph{critical temperature} $\beta_c$ as
\[
    \beta_c:=\lim_{n\to\infty}(1-\Delta_n-\kappa_n^{-1})
\]
whenever the limit exists and is positive. In particular,
\[
    \beta_c=\begin{dcases}
        1-\Delta&\text{if $\Delta_n\to\Delta\in[0,1)$ and $\lambda_n\to0$,}\\
        1&\text{if $\lambda_n\to\lambda\in(0,1]$ and $\kappa_n\to\infty$,}\\
        1-1/\kappa&\text{if $\kappa_n\to\kappa\in(1,\infty)$.}
    \end{dcases}
\]
Our system associated with $G$ can be completely described by the function $f_G:\mathbb{R}_{>0}\to\mathcal{P}_1(\mathcal{T}_G)$, $\beta\mapsto\mu_{G,\beta}$. Then Corollary~\ref{cor:logz-phase} implies that
\[
    \mathsf{w}_1(f_G(\beta),\mu_{G,\infty})=\frac{1}{n}W_1(\mu_{G,\beta},\mu_{G,\infty})\pto0
\]
if $\beta>\beta_c$; in other words, $f_G$ is asymptotically almost surely constant in $\beta>\beta_c$. However, the distance to $\mu_{G,\infty}$ is bounded away from zero when $\beta<\beta_c$ --- a non-analyticity in $f_G$ occurs at $\beta_c$ in the thermodynamic limit, indicating a phase transition.

In fact, the phase transition in the regime \ref{item:regime3-2} is somewhat different from that in the regimes \ref{item:regime2} and \ref{item:regime3-1}. One way to see this is to consider the \emph{optimal kernel size} as denoted by $m^*$ in Theorem~\ref{thm:logz-formula}. In other words, we look at the kernel size that minimizes the free energy. Since there is no significant difference between \ref{item:regime2} and \ref{item:regime3-1}, for simplicity, we shall focus on comparing \ref{item:regime3-1} and \ref{item:regime3-2}.
\begin{itemize}
    \item In the regime \ref{item:regime3-1}, we have $\Delta_n\to0$ and $\kappa_n\to\infty$. In the low temperature phase, the optimal kernel proportion $m^*/N_{d^*}=1$, but in the high temperature phase, the optimal kernel proportion is approximately
    \[
        \frac{m^*}{N_{d^*}}\approx\frac{1}{(1-\Delta_n-\beta)\kappa_n}\to0\,.
    \]
    Hence, as $\beta$ decreases from above $\beta_c=1$ to below $\beta_c=1$, the optimal kernel proportion suddenly drops from $1$ to $0$ \emph{discontinuously}.

    \item In the regime \ref{item:regime3-2}, we also have $\Delta_n\to0$ but $\kappa_n\to\kappa\in(1,\infty)$. The optimal kernel proportion in the low temperature phase is the same, but in the high temperature phase we have
    \[
        \frac{m^*}{N_{d^*}}\approx\frac{1}{(1-\Delta_n-\beta)\kappa_n}\to\frac{1}{(1-\beta)\kappa}\,.
    \]
    Note that $\beta_c=1-1/\kappa$ and thus its left limit $\beta\to\beta_c-0$ is $1$. Hence, the optimal kernel proportion changes \emph{continuously}.
\end{itemize}
This observation can be adapted to our Wasserstein framework in the following way. For a sequence of graphs $\{G_n\}_{n\in\mathbb{Z}^+}$ with $|V(G_n)|=n$, we get a sequence of functions $\{f_{G_n}\}_{n\in\mathbb{Z}^+}$ where $f_{G_n}:\mathbb{R}_{>0}\to\mathcal{P}(\mathcal{T}_{G_n})$ and we check if it is \emph{equicontinuous} at $\beta_c$. This is formally stated and proved in the following results.

\begin{proposition}[Discontinuous phase transition]
    In regimes \ref{item:regime2} and \ref{item:regime3-1}, there exists a universal constant $C>0$ such that for all constant $\delta>0$, we can find a constant $\beta>0$ with $|\beta-\beta_c|<\delta$ satisfying
    \[
        \frac{1}{n}W_1(\mu_{G_n,\beta},\mu_{G_n,\beta_c})\geq C
    \]
    asymptotically almost surely.
\end{proposition}
\begin{proof}
    Choose $C=1/50$. Set $\beta_1=\beta_c+\delta/2$ and $\beta_2=\beta_c-\delta/2$ which are in \eqref{eqn:low-temp-phase} and \eqref{eqn:high-temp-phase}, respectively. By Theorem~\ref{thm:gibbs-low-phase}, we have
    \[
        \frac{1}{n}W_1(\mu_{G,\beta_1},\mu_{G,\infty})=o_p(1)\,.
    \]
    For $\beta_2$, we note that in the proof of Theorem~\ref{thm:gibbs-high-phase}, \eqref{eqn:w1-lip-diff} and \eqref{eqn:lip-diff-global} imply that
    \[
        \frac{1}{n}W_1(\mu_{G,\beta_2},\mu_{G,\infty})\geq\frac{3}{4}-\frac{1}{\sqrt{2}}-o_p(1)
    \]
    if $(\delta/2)\kappa_n\geq7$. Since $\kappa_n\to\infty$ in regimes \ref{item:regime2} and \ref{item:regime3-1}, for any $\delta>0$ this holds for sufficiently large $n$. Thus, we have that
    \[
        \begin{split}
            \frac{1}{n}W_1(\mu_{G,\beta_1},\mu_{G,\beta_2})&\geq\frac{3}{4}-\frac{1}{\sqrt{2}}-o_p(1)\\
        &\geq2C
        \end{split}
    \]
    asymptotically almost surely. Hence, at least one of $\beta_1$ or $\beta_2$ satisfies the desired condition.
\end{proof}

\begin{proposition}[Continuous phase transition]
    In regime \ref{item:regime3-2}, for any constant $\delta>0$, any $\beta>0$ with $|\beta-\beta_c|<\delta$ satisfies
    \[
        \frac{1}{n}W_1(\mu_{G_n,\beta},\mu_{G_n,\beta_c})\leq\kappa\sqrt{12\delta}
    \]
    asymptotically almost surely.
\end{proposition}
\begin{proof}
    We set $\delta=\frac{\epsilon^2}{12\kappa^2}$. The case $\beta_c\leq\beta<\beta_c+\delta$ is contained in \eqref{eqn:low-temp-phase}, so by Theorem~\ref{thm:gibbs-low-phase}
    \[
        \frac{1}{n}W_1(\mu_{G,\beta_c},\mu_{G,\infty})=o_p(1)
    \]
    and the conclusion follows. Thus, it remains to handle the case $\beta_c-\delta<\beta<\beta_c$. Define
    \[
        m_3:=m_\ell\vee\left\lfloor\left(1-\frac{\epsilon^2}{3}\kappa_n^{-1}\right)N_{d^*}\right\rfloor\,.
    \]
    We claim that
    \begin{equation}\label{eqn:cpt-subclaim}
        \mu_{G,\beta}(\{T:|\varphi(T)|\leq m_3\})\leq\exp(-\Omega(n))\,.
    \end{equation}
    Before we show this, we first see how \eqref{eqn:cpt-subclaim} yields the conclusion. By following the exact same argument with the proof of Theorem~\ref{thm:gibbs-low-phase}, \eqref{eqn:w1-small} gives
    \[
        \frac{1}{n}W_1(\mu_{G,\beta},\mu_{G,\infty})\leq\sqrt{\frac{2}{3}\epsilon^2}+o_p(1)<\epsilon
    \]
    asymptotically almost surely. Thus, it suffices to prove \eqref{eqn:cpt-subclaim}. Let
    \[
        \delta':=\frac{\epsilon^2}{12\kappa}\wedge\frac{1}{2}\,.
    \]
    This time, we follow the argument in the proof of Lemma~\ref{lem:gibbs-small-kernel}. From \eqref{eqn:htp-gdiff} we have
    \[
        \begin{split}
            g\left(\frac{\kappa_n^{-1}N_{d^*}}{1-\Delta_n-\beta}\right)-g\left(\frac{(1-\delta')\kappa_n^{-1}N_{d^*}}{1-\Delta_n-\beta}\right)&\geq(-\delta'-\log(1-\delta'))n-o_p(n)\\
            &=\Omega_p(n)
        \end{split}
    \]
    which yields
    \[
        \log\left(\mu_{G,\beta}\left(\left\{T:|\varphi(T)|\geq\frac{(1-\delta')\kappa_n^{-1}N_{d^*}}{1-\Delta_n-\beta}\right\}\right)\right)\leq-\Omega_p(n)\,.
    \]
    Since $\beta>\beta_c-\delta$ and $\kappa_n\to\kappa$, we have for sufficiently large $n$ that $\beta>1-\Delta_n-\kappa_n^{-1}-\delta$ and $\kappa_n<2\kappa$. Thus,
    \[
        \begin{split}
            \frac{(1-\delta')\kappa_n^{-1}}{1-\Delta_n-\beta} &\geq\frac{1-\delta'}{1+\delta\kappa_n}\\
            &\geq(1-\delta')(1-\delta\kappa_n)\\
            &\geq1-\delta'-\delta\kappa_n\\
            &\geq1-\frac{\epsilon^2}{3\kappa_n}\,.
        \end{split}
    \]
    This, together with Lemma~\ref{lem:large-kernel}, proves \eqref{eqn:cpt-subclaim}.
\end{proof}

\subsection{The Franz--Parisi potential}\label{sec:franz-parisi}

\subsubsection{Definitions and main results}
Fix a constant $\beta>0$. Consider a sequence of graphs $G_1,G_2,\cdots$ independently sampled from $G_n\sim\mathcal{G}(n,q_n)$ and a sequence of spanning trees $\{T_n\}$. We consider the non-asymptotic Franz--Parisi potential centered at $T_n$, defined by
\[
    \mathcal{F}_{G_n,T_n,\beta}^{\text{FP}}(r, \epsilon):=-\frac{1}{\beta n\log\log n}\log Z_{r,\epsilon}-(d_n^*+\lambda_n)
\]
where
\[
    Z_{r,\epsilon}:=\sum_{T':\mathsf{R}(T_n,T')\in[r-\epsilon,r+\epsilon]}\exp\left(-\overline{\beta}\sum_{v}\mathsf{d}_{T'}(1,v)\right)
\]
and
\[
    \mathsf{R}(T,T')=\frac{1}{|\overline{V}|-1}\sum_{v\in \overline{V}\setminus\{1\}}\mathbf{1}_{\mathsf{par}_T(v)\neq\mathsf{par}_{T'}(v)}
\]
is the normalized Hamming distance between the two (rooted) spanning trees. Similar to the free energy density, we subtract off the ground state energy. Note that $Z_{r,\epsilon}$ clearly depends on $G_n$, $T_n$, and $\beta$, but we drop the subscripts in the interest of readability.

Now assume further that $\{T_n\}$ are independently drawn from the Gibbs measures $T_n\sim\mu_{G_n,\beta}$. We define the (asymptotic) two-temperature Franz--Parisi potential by
\[
    \mathcal{F}_{\beta,\beta'}^{\text{FP}}(r):=\lim_{\epsilon\to0^+}\plim_{n\to\infty}\mathcal{F}_{G_n,T_n\beta'}^{\text{FP}}(r,\epsilon)
\]
and the one-temperature version
\[
    \mathcal{F}_{\beta}^{\text{FP}}(r):=\mathcal{F}_{\beta,\beta}^{\text{FP}}(r)
\]
whenever the limits exist. To compute the Franz--Parisi potential, we derive a slightly more general formula which holds for any sequence of trees $\{T_n\}$ with mild asymptotic conditions. Recall the definition of kernels in Definition~\ref{def:kernel}.

\begin{proposition}\label{prop:fpp-general}
    Suppose that $\Delta_n\to\Delta\in[0,1]$ and $\lambda_n\to\lambda\in[0,1]$. Suppose that $\{T_n\}$ is a sequence of random spanning trees such that
    \[
        \frac{|\varphi(T_n)|}{n}\pto r_1\,,\qquad\frac{|\mathsf{ch}_{T_n}(\varphi(T_n))|}{n}\pto1-r_1
    \]
    and
    \[
        \frac{|\{v\in\varphi(T_n):|\mathsf{par}_{G_n}(v)|=1\}|}{n}\pto r_2\,.
    \]
    Then the limit
    \[
        \lim_{\epsilon\to0}\plim_{n\to\infty}\mathcal{F}_{G_n,T_n,\beta}^{\text{FP}}(r,\epsilon)=f_{\lambda,\Delta}(r)
    \]
    exists and is a continuous function of $r$. Moreover, if $\beta\geq1-\Delta$, we have
    \[
        f_{\lambda,\Delta}(r)=\begin{dcases}
            -r_1-\lambda-\frac{\Delta}{\beta}+\left(1+\frac{\Delta}{\beta}\right) r&\text{if $r_1+\lambda\leq r\leq 1$,}\\
            -\frac{1-(1-\lambda-r_1)(1-\Delta)}{\beta}+\frac{1}{\beta} r&\text{if $r_1\leq r\leq r_1+\lambda$,}\\
            -\frac{1-(1-\lambda)(1-\Delta)}{\beta}+\frac{\Delta}{\beta} r&\text{if $r_2\leq r\leq r_1$,}\\
            r_2-\frac{1-(1-\lambda-r_2)(1-\Delta)}{\beta}-\left(1-\frac{1}{\beta}\right)r&\text{if $0\leq r\leq r_2$.}
        \end{dcases}
    \]
    and if $\beta\leq1-\Delta$, we have
    \[
        f_{\lambda,\Delta}(r)=\begin{dcases}
            1-r_1-\lambda-\frac{1}{\beta}+\frac{1}{\beta}r&\text{if $r_1\leq r\leq 1$.}\\
            1-\lambda-\frac{1}{\beta}+\left(\frac{1}{\beta}-1\right)r&\text{if $0\leq r\leq r_1$,}
        \end{dcases}
    \]
\end{proposition}

We postpone the proof to Section~\ref{sec:fpp-general} and first specialize this result to the asymptotic case where $\{T_n\}$ are drawn from the Gibbs measures. Note that the cases are set to match Corollary~\ref{cor:free-energy} and Corollary~\ref{cor:logz-phase}.

\begin{corollary}\label{cor:fpp-general}
    Suppose that $\Delta_n\to\Delta\in[0,1]$ and $\lambda_n\to\lambda\in[0,1]$.

    \begin{enumerate}[label=(\Alph*)]
        \item\label{item:fpp-regime1} If $\Delta=1$, then $\lambda=0$ and
        \[
            \mathcal{F}_{\beta}^{\text{FP}}(r)=-\frac{1}{\beta}+\frac{1}{\beta}r
        \]
        for all $\beta>0$.

        \item\label{item:fpp-regime2} We further split the case $\Delta\in[0,1)$ and $\lambda\in[0,1)$ into the following.
        \begin{enumerate}[label=(\Alph{enumi}.\roman*), ref=(\Alph{enumi}.\roman*), align=left]
            \item If $\Delta\in[0,1)$ and $\lambda=0$, then one of the following holds.
            \begin{itemize}
                \item If $\beta>1-\Delta$, then
                \[
                    \mathcal{F}_{\beta}^{\text{FP}}(r)=-\frac{\Delta}{\beta}+\frac{\Delta}{\beta}r\,.
                \]
    
                \item If $\beta=1-\Delta$, then $\mathcal{F}_{\beta}(r)$ may not exist. Instead, we have
                \[
                    \mathcal{F}_{\beta+\delta,\beta}^{\text{FP}}(r)=-\frac{\Delta}{\beta}+\frac{\Delta}{\beta}r=-\frac{\Delta}{1-\Delta}+\frac{\Delta}{1-\Delta}r
                \]
                for all $\delta>0$ and
                \[
                    \mathcal{F}^{\text{FP}}_{\beta-\delta,\beta}(r)=1-\frac{1}{\beta}+\frac{1}{\beta}r=\frac{\Delta}{1-\Delta}+\frac{1}{1-\Delta}r
                \]
                for all $0<\delta<1-\Delta$.
    
                \item If $\beta<1-\Delta$, then
                \[
                    \mathcal{F}^{\text{FP}}_{\beta}(r)=1-\frac{1}{\beta}+\frac{1}{\beta}r\,.
                \]
            \end{itemize}
    
            \item If $\Delta=0$ and $\lambda\in(0,1)$, then one of the following holds.
            \begin{itemize}
                \item If $\beta>1$, then
                \[
                    \mathcal{F}^{\text{FP}}_{\beta}(r)=\begin{dcases}
                        -\frac{1}{\beta}+\frac{1}{\beta}r&\text{if $r\geq1-\lambda$,}\\
                        -\frac{\lambda}{\beta}&\text{if $\lambda\log(1/\lambda)\leq r\leq 1-\lambda$,}\\
                        \frac{-\lambda+(\beta-1)\lambda\log(1/\lambda)}{\beta}-\frac{\beta-1}{\beta}r&\text{if $r\leq \lambda\log(1/\lambda)$.}
                    \end{dcases}
                \]
    
                \item If $\beta=1$, then $\mathcal{F}_{\beta}(r)$ may not exist but we have
                \[
                    \mathcal{F}^{\text{FP}}_{\beta+\delta,\beta}(r)=\begin{dcases}
                        -\frac{1}{\beta}+\frac{1}{\beta}r=-1+r&\text{if $r\geq 1-\lambda$,}\\
                        -\frac{\lambda}{\beta}=-\lambda&\text{if $r\leq 1-\lambda$}
                    \end{dcases}
                \]
                for all $\delta>0$ and
                \[
                    \mathcal{F}^{\text{FP}}_{\beta-\delta,\beta}(r)=1-\lambda-\frac{1}{\beta}+\frac{1}{\beta}r=-\lambda+r
                \]
                for all $0<\delta<1$.
    
                \item If $\beta<1$, then
                \[
                    \mathcal{F}^{\text{FP}}_{\beta}(r)=1-\lambda-\frac{1}{\beta}+\frac{1}{\beta}r\,.
                \]
            \end{itemize}
        \end{enumerate}

        \item\label{item:fpp-regime3} If $\lambda=1$, then $\Delta=0$ and
        \[
            \mathcal{F}^{\text{FP}}_{\beta}(r)=-\frac{1}{\beta}+\frac{1}{\beta}r
        \]
        for all $\beta>0$.
    \end{enumerate}
\end{corollary}
\begin{proof}
    This is immediate from Proposition~\ref{prop:fpp-general} by observing the following.
    \begin{itemize}
        \item Recall that by Theorem~\ref{thm:main1-tight} $N_{d^*}/n\pto1-\lambda$.
        
        \item If $\{T_n\}$ are sampled in the low temperature phase, Lemma~\ref{lem:gibbs-low-apx-unif} and \eqref{eqn:m1-asymp} imply $r_1=1-\lambda$. Also conditioned on $\Gamma_{\leq d^*-1}$, for each $v\in V\setminus\Gamma_{\leq d^*-1}$, we have $|\mathsf{par}_G(v)|\sim\Binom(N_{d^*-1},q)$, so $\Pr(|\mathsf{par}_G(v)|=1\mid \Gamma_{\leq d^*-1})=qN_{d^*-1}(1-q)^{N_{d^*}-1}\pto\lambda\log(1/\lambda)$ by $N_{d^*-1}/(nq)^{d^*-1}\pto1$ from Proposition~\ref{prop:conc}. This gives $r_2=\lambda\log(1/\lambda)$ (which is zero when $\lambda=0$).

        \item If $\{T_n\}$ are sampled in the high temperature phase, Lemma~\ref{lem:gibbs-small-kernel} implies that $r_1=0$, which automatically gives $r_2=0$.
    \end{itemize}
\end{proof}

We complement the case (B.ii) by further investigating the ``flat'' region for $\lambda\log(1/\lambda)<r<1-\lambda$. As it turns out, if we zoom in this region, it is a smooth convex basin. This is proved in the following proposition. It might be helpful to revisit the formula in Corollary~\ref{cor:logz-phase}.

\begin{proposition}\label{prop:fpp-special}
    Suppose that $\Delta_n\to\Delta=0$, $\lambda_n\to\lambda\in(0,1)$, $\beta>1$, and $\lambda\log(1/\lambda)<r<1-\lambda$. Letting $\bar{r}=\frac{r-\lambda\log(1/\lambda)}{1-\lambda-\lambda\log(1/\lambda)}\in(0,1)$ be a normalized value of $r$, we have
    \[
        \lim_{\epsilon\to0}\plim_{n\to\infty}\left(\left(\mathcal{F}_{G_n,T_n\beta}(r,\epsilon)+\frac{\lambda_n}{\beta}\right)\log\log n+\frac{\lambda\log((1-\lambda)\alpha_n)+\Psi_0(\lambda)}{\beta}\right)=\frac{1}{\beta}I_\lambda(\bar{r})
    \]
    where
    \[
        I_\lambda(x):=\sup_{t\in\mathbb{R}}(tx-\Lambda(t;\lambda))\,,\qquad0\leq x\leq 1
    \]
    is a convex function defined by the Legendre transform of
    \[
        \Lambda(t;\lambda):=\E\left[\log\left(1-\frac{1}{X}+\frac{e^t}{X}\right)\,\middle|\,X\geq2\right]\,,\qquad X\sim\Pois(\log(1/\lambda))\,.
    \]
\end{proposition}

The proof of Proposition~\ref{prop:fpp-special}, outlined in Section~\ref{sec:fpp-special}, borrows techniques from large deviations theory (see, e.g., \cite{dembo2009large}). The function $I_{\lambda}$ is called a \emph{rate function}, which encodes how likely the overlap $r$ will occur if we sample two independent spanning trees from $\mu_{G,\beta}$ in the low temperature phase. This can also be understood as a relative entropy of $\mu_{G,\beta}$ conditioned on the overlap, with respect to $\mu_{G,\beta}$. For instance, for each vertex $v\in\Gamma_{d^*}$, if we independently sample two parents $\mathsf{par}_T(v)$ from $\mathsf{par}_G(v)$, then they overlap with probability $1/|\mathsf{par}_G(v)|$. Since the law of $|\mathsf{par}_G(v)|$ is approximately $\Pois(\log(1/\lambda))$, among the vertices with $|\mathsf{par}_G(v)|\geq2$, we expect that the typical overlap will be approximately
\[
    \bar{r}^*=\E[1/X\mid X\geq2]\,,\qquad X\sim\Pois(\log(1/\lambda))\,.
\]
With simple calculus, one can verify that this is where the line $y=\bar{r}^*t$ with slope $\bar{r}^*$ is tangent to the convex function $y=\Lambda(t;\lambda)$, hence $I_{\lambda}(\bar{r}^*)=0$. This corresponds to the bottom of the zoomed-in region, in which case the Franz--Parisi potential agrees with the free energy density (cf. Corollary~\ref{cor:logz-phase}).

One may notice that $\mathcal{F}_\beta(r)$ in the regime \ref{item:regime2-1} is also flat if $\Delta=0$ and $\beta>1$. This is fundamentally the same phenomenon with \ref{item:regime2-2}, and a similar analysis can be used to show that in this case the Franz--Parisi potential is increasing in $r$ in the entire regime if we ``zoom in''. Intuitively, \ref{item:regime2-1} corresponds to the case where a typical vertex has distance $d^*$ and has an unbounded number of parent choices, so any two independent samples are unlikely to have nontrivial overlap. We do not provide detailed analysis for this case.

\subsubsection{Proof of Proposition~\ref{prop:fpp-general}}\label{sec:fpp-general}
As usual, we drop the subscripts and write $T=T_n$ and $G=G_n$ for conciseness. Define
\[
    Q:=\{v\in\varphi(T):|\mathsf{par}_{G}(v)|=1\}\,.
\]
and recall that
\[
    Z_{r,\epsilon}:=\sum_{T':\mathsf{R}(T,T')\in[r-\epsilon,r+\epsilon]}\exp\left(-\overline{\beta}\sum_{v}\mathsf{d}_{T'}(1,v)\right)\,.
\]
Also, for a set $R\subseteq V\setminus\{1\}$, define
\[
    \mathcal{D}(T, R)=\{T':\mathsf{par}_T(v)=\mathsf{par}_{T'}(v)\iff v\in R\}
\]
which is the collection of spanning trees whose overlap with $T$ is exactly $R$.

Our goal is to analyze
\[
    \mathcal{F}_{G,T,\beta}(r,\epsilon)=-\frac{1}{\beta n\log\log n}\log Z_{r,\epsilon}-(d_n^*+\lambda_n)\,.
\]
We first write
\begin{equation}\label{eqn:zreps}
    Z_{r,\epsilon}=\sum_{R:|R|/(\bar{n}-1)\in[r-\epsilon,r+\epsilon]}\sum_{T'\in\mathcal{D}(T,R)}\exp\left(-\overline{\beta}\sum_{v}\mathsf{d}_{T'}(1,v)\right)
\end{equation}
where for convenience we denote $\bar{n}=|\overline{V}|$. Recall that the key notion in analyzing $\log Z_{G,\beta}$ is the kernel of a spanning tree, because the sum over a fixed kernel is easy to approximate (cf. Proposition~\ref{prop:fixed-kernel}). We apply the same idea to analyze $\log Z_{r,\epsilon}$. However, if the overlap of $T'$ and $T$ is fixed to be $R$, not all sets in $\mathcal{K}(G)$ can be the kernel of $T'$.
\begin{itemize}
    \item If $v\in R\cap\varphi(T)$, then $\mathsf{par}_{T'}(v)=\mathsf{par}_T(v)\in\Gamma_{d^*-1}$, so $v$ must be in the kernel of $T'$.
    \item If $v\in R\cap(\Gamma_{d^*}\setminus\varphi(T))$, then $\mathsf{par}_{T'}(v)=\mathsf{par}_T(v)\notin\Gamma_{d^*-1}$, so $v$ must not be in the kernel of $T'$.
    \item If $v\in Q\setminus R$, then $\mathsf{par}_{T'}(v)\neq\mathsf{par}_T(v)\in\Gamma_{d^*-1}$ but $|\mathsf{par}_G(v)|=1$ so $\mathsf{par}_T(v)$ is the only neighbor of $v$ in $\Gamma_{d^*-1}$. Hence, $v$ must not be in the kernel of $T'$.
    \item For all the other vertices in $\Gamma_{d^*}$, $\mathsf{par}_{T'}(v)\neq\mathsf{par}_T(v)$ but either $v\notin\varphi(T)$ or $|\mathsf{par}_G(v)|>1$, so $v$ may or may not be in the kernel of $T'$.
\end{itemize}
In fact, the kernels of trees in $\mathcal{D}(T,R)$ can be characterized by these properties.
\begin{lemma}
    The kernel $A=\varphi(T')$ of any $T'\in\mathcal{D}(T,R)$ satisfies
    \begin{equation}\label{eqn:tp-ker-sandwich}
        R\cap\varphi(T)\subseteq A\subseteq (R\cap\varphi(T))\cup(\Gamma_{d^*}\setminus R\setminus Q)\,.
    \end{equation}
    Conversely, for any $A\in\mathcal{K}(G)$ satisfying \eqref{eqn:tp-ker-sandwich} there is a $T'\in\mathcal{D}(T,R)$ such that $A=\varphi(T')$.
\end{lemma}
We omit the proof of this fact. Let $\mathcal{K}(T,R)$ be this set of all possible kernels. Then \eqref{eqn:zreps} can be further decomposed as
\[
    Z_{r,\epsilon}=\sum_{R:|R|/(\bar{n}-1)\in[r-\epsilon,r+\epsilon]}Z_R
\]
where
\[
    Z_R:=\sum_{A\in\mathcal{K}(T,R)}Z_{R}|_A
\]
and
\[
    Z_R|_A:=\sum_{\substack{T'\in\mathcal{D}(T,R)\\\varphi(T')=A}}\exp\left(-\overline{\beta}\sum_v\mathsf{d}_{T'}(1,v)\right)\,.
\]
Similar to what we did with $\log Z_{G,\beta}$ and $\log(Z_{G,\beta}|_A)$ (cf. \eqref{eqn:gibbs-opt-bigo}), we have that
\[
    \begin{split}
        \log Z_{r,\epsilon}&=\max_{R:|R|/(\bar{n}-1)\in[r-\epsilon,r+\epsilon]}\log Z_R+O(n)\\
        &=\max_{R:|R|/(\bar{n}-1)\in[r-\epsilon,r+\epsilon]}\max_{A\in\mathcal{K}(T,R)}\log(Z_{R}|_A)+O(n)\,.
    \end{split}
\]
Since we only need $\log Z_{r,\epsilon}$ up to $O(n\log\log n)$ error, it suffices to analyze and maximize $\log(Z_{R}|_A)$. Throughout, we assume that $A\in\mathcal{K}(T,R)$, i.e., $A\in\mathcal{K}(G)$ and satisfies \eqref{eqn:tp-ker-sandwich}. Note that letting
\[
    B:=A\setminus(R\cap\varphi(T))
\]
\eqref{eqn:tp-ker-sandwich} is equivalent to
\begin{equation}\label{eqn:tp-ker-sandwich-b}
    B\subseteq\Gamma_{d^*}\setminus R\setminus Q\,.
\end{equation}

\begin{figure}[t]
    \centering
    \includegraphics[width=0.5\textwidth]{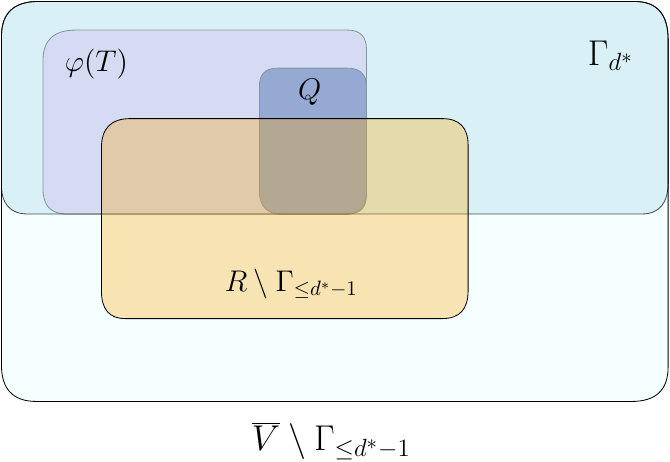}
    \caption{A Venn diagram of $\Gamma_{d^*}$, $\varphi(T)$, $Q$, and $R$ intersecting them.}
    \label{fig:venn}
\end{figure}

\paragraph{An upper bound.} Now we upper bound $\log (Z_R|_A)$ analogous to Lemma~\ref{lem:gibbs-cond-ub} and Lemma~\ref{lem:logz-ub-lglg}. The difference from these formulae is that here vertices in $R$ have a fixed parent and thus do not contribute any entropy.

\begin{lemma}\label{lem:logzra-ub}
    Let $L=Me^{-\overline{\beta}}\vee\sqrt{\log n}$. With probability at least $1-o(1)$, we have
    \[
        \log (Z_{R}|_A) \leq \sum_{v\in B}\log(|\mathsf{par}_G(v)|-\mathbf{1}_{v\in\varphi(T)})+(n-|R|-|B|)\log(L+|A|q)+\overline{\beta}|A|-\overline{\beta}(d^*+1)n+o(n)
    \]
    simultaneously for all $R\subseteq\overline{V}\setminus\{1\}$ and $A\in\mathcal{K}(T,R)$.
\end{lemma}

The proof is nearly identical to Lemma~\ref{lem:gibbs-cond-ub} so we omit it. Similar to Lemma~\ref{lem:logz-ub-lglg}, this yields
\[
    \frac{1}{\log\log n}\log(Z_{R}|_A)\leq |B|\Delta+(n-|R|-|B|)+\beta|A|-\beta(d^*+1)n+o(n)\,.
\]
Rewriting this, we have that
\begin{equation}\label{eqn:logzra-ub}
    \frac{1}{\log\log n}\log(Z_{R}|_A) \leq n-|R|-|B|(1-\Delta-\beta)+\beta|R\cap\varphi(T)|-\beta(d^*+1)n+o(n)\,.
\end{equation}
Now we consider different phases.
\begin{itemize}
    \item In the low to moderate temperature phase $\beta\geq1-\Delta$, maximizing $\log(Z_R|_A)$ for a fixed $R$ is done by maximizing $|B|$. As seen in \eqref{eqn:tp-ker-sandwich}, $|B|$ is at most $|\Gamma_{d^*}\setminus R\setminus Q|$. Thus, we get an upper bound of $\log Z_R$
    \[
        \frac{1}{\log\log n}\log Z_R\leq n-|R|+|\Gamma_{d^*}\setminus R\setminus Q|(\beta-1+\Delta)+\beta|R\cap \varphi(T)|-\beta(d^*+1)n+o(n)\,.
    \]
    Thus, our goal is to optimize
    \begin{equation}\label{eqn:r-opt}
        |\Gamma_{d^*}\setminus R\setminus Q|(\beta-1+\Delta)+\beta|R\cap \varphi(T)|
    \end{equation}
    over $R$. We partition $R\setminus\Gamma_{\leq d^*-1}$ it into four sets
    \begin{equation}\label{eqn:xi}
        \begin{split}
            U_1 &= R\cap Q\\
            U_2 &= R\cap(\varphi(T)\setminus Q)\\
            U_3 &= R\setminus\Gamma_{\leq d^*}\\
            U_4 &= R\cap (\Gamma_{d^*}\setminus\varphi(T))\,.
        \end{split}
    \end{equation}
    See Figure~\ref{fig:venn} for an illustration. Letting $u_i=|U_i|$ the sizes, \eqref{eqn:r-opt} is written as
    \[
        (N_{d^*}-|Q|-u_2-u_4)(\beta-1+\Delta)+\beta(u_1+u_2)\,.
    \]
    This is a linear combination of $u_i$, and the coefficients of $u_i$ are $\beta$, $1-\Delta$, $0$, and $1-\Delta-\beta$, respectively. Since $\beta\geq1-\Delta\geq0\geq1-\Delta-\beta$, we increase $u_1$, $u_2$, $u_3$, and $u_4$ in this order. In other words, the optimal strategy to construct $R$ is to pick vertices from $Q$, $\varphi(T)$, $\overline{V}\setminus\Gamma_{d^*}$, and then $\Gamma_{d^*}\setminus\varphi(T)$ until the target size $|R|$ is reached. Now obvious calculations yield 
    \begin{equation}\label{eqn:fpp-lb}
        \mathcal{F}_{G,T,\beta}(r,\epsilon)\geq \inf_{r'\in[r-\epsilon,r+\epsilon]}f_{\lambda,\Delta}(r')-o_p(1)\,.
    \end{equation}

    \item In the moderate to high temperature phase $\beta\leq1-\Delta$, an upper bound is obtained simply by plugging in $|B|=0$, which gives
    \[
        \max_R\{n-|R|+\beta|R\cap\varphi(T)|\}-\beta(d^*+1)n+o_p(n)\,.
    \]
    This time, an optimal way to construct $R$ is simply exhausting $\varphi(T)$ and then picking vertices outside $\varphi(T)$ arbitrarily. This also leads to \eqref{eqn:fpp-lb}.
\end{itemize}

\paragraph{A lower bound.} Similar to the upper bound, we have the following analogue of Lemma~\ref{lem:gibbs-cond-lb}.

\begin{lemma}
    With probability at least $1-o(1)$, we have
    \[
        \log (Z_{R}|_A) \geq \sum_{v\in B}\log(|\mathsf{par}_G(v)|-\mathbf{1}_{v\in\varphi(T)})+(n-|R|-|B|)\log(|A|q)+\overline{\beta}|A|-\overline{\beta}(d^*+1)n-o(n)
    \]
    simultaneously for all $R\subseteq\overline{V}\setminus\{1\}$ and $A\in\mathcal{K}(T,R)$ with $|A|\geq m_\ell$.
\end{lemma}

The proof is simply repeating the argument in the proof of Lemma~\ref{lem:gibbs-cond-lb} so we omit it. This also gives
\[
    \frac{1}{\log\log n}\log(Z_{R}|_A) \geq n-|R|-|B|(1-\Delta-\beta)+\beta|R\cap\varphi(T)|-\beta(d^*+1)n-o(n)
\]
which matches \eqref{eqn:logzra-ub}. Thus, it remains to choose $R$ and $A\in\mathcal{K}(T,R)$ such that $|A|\geq m_{\ell}$ and the optimal value of \eqref{eqn:r-opt} is achieved. This can be done by constructing $R$ (with size within $(1\pm\epsilon)n$) in the same way we obtained \eqref{eqn:fpp-lb} and, if such $A$ does not exist, add any $m_{\ell}$ vertices from $\Gamma_{d^*}$ intersecting every connected component of $\overline{G}\setminus\Gamma_{\leq d^*-1}$. Since $m_\ell/n\to0$ this has little influence on the size constraint on $R$. This leads to an upper bound of $\mathcal{F}_{G,T,\beta}(r,\epsilon)$:
\begin{equation}\label{eqn:fpp-ub}
    \mathcal{F}_{G,T,\beta}(r,\epsilon)\leq\sup_{r'\in[r-\epsilon,r+\epsilon]}f_{\lambda,\Delta}(r')+o_p(1)\,.
\end{equation}
Since $f_{\lambda,\Delta}$ is continuous on $[0,1]$, \eqref{eqn:fpp-lb} and \eqref{eqn:fpp-ub} prove the conclusion.

\subsubsection{Proof of Proposition~\ref{prop:fpp-special}}\label{sec:fpp-special}

Since $T$ is sampled from the low temperature phase, by Theorem~\ref{thm:gibbs-low-phase} we may assume that $T$ is a uniformly random shortest path tree. In other words, throughout this section, we assume that
\[
    \varphi(T)=\Gamma_{d^*}
\]
and
\[
    Q=\{v\in\Gamma_{d^*}:|\mathsf{par}_G(v)|=1\}\,.
\]
Similar to our analysis of $\log Z_{G,\beta}$, to get a more accurate estimate of $\log Z_{r,\epsilon}$, we partition $Z_R$ with respect to the kernel size $|A|$. In other words,
\[
    Z_R=\sum_{m=0}^{|\Gamma_{d^*}\setminus R\setminus Q|}\sum_{|A|=m+|R\cap\varphi(T)|}Z_{R}|_A
\]
where $m$ encodes the size of $B$ satisfying \eqref{eqn:tp-ker-sandwich-b}. Note that
\[
    \log Z_R=\max_{m=0}^{|\Gamma_{d^*}\setminus R\setminus Q|}\log\left(\sum_{|A|=m+|R\cap\varphi(T)|}Z_{R}|_A\right)+o(n)\,.
\]
Since it is assumed that $\beta>1=1-\Delta$, we are in the low temperature phase. Similar to the proof of Theorem~\ref{thm:logz-formula}, it is not difficult to see that an approximate maximum is achieved when $m$ is the largest uniformly over $R$, namely, $m=|\Gamma_{d^*}\setminus R\setminus Q|$. This is equivalent to saying that
\[
    A=(\Gamma_{d^*}\setminus Q)\cup U_1
\]
where $U_i$ are defined in \eqref{eqn:xi}. Thus,
\[
    \log Z_R\leq\log(Z_{R}|_{(\Gamma_{d^*}\setminus Q)\cup U_1})+o(n)
\]
simultaneously for all $R$. Now we apply Lemma~\ref{lem:logzra-ub} and rewrite the bound in terms of $U_i$ and $u_i=|U_i|$. Since we have assumed $\varphi(T)=\Gamma_{d^*}$, we have $U_4=\emptyset$. Thus,
\begin{equation}\label{eqn:logzr-explicit}
    \begin{split}
        \log Z_R&\leq\sum_{v\in\varphi(T)\setminus Q\setminus U_2}\log(|\mathsf{par}_G(v)|-1)+(|\overline{V}\setminus\Gamma_{\leq d^*}|+|Q|-u_1-u_3)\log(qN_{d^*})\\
        &\quad+\overline{\beta}(N_{d^*}-|Q|+u_1)-\overline{\beta}(d^*+1)n+o(n)\\
        &=\sum_{v\in\varphi(T)\setminus Q\setminus U_2}\log(|\mathsf{par}_G(v)|-1)+(\lambda_nn+|Q|-u_1-u_3)\log((1-\lambda_n)nq)\\
        &\quad+\overline{\beta}(-|Q|+u_1)-\overline{\beta}(d^*+\lambda_n)n+o(n)
    \end{split}
\end{equation}
where we also used the fact that $L+|A|q\leq qN_{d^*}$ asymptotically almost surely. Now note that
\begin{equation}\label{eqn:zreps-max}
    \log Z_{r,\epsilon} = \max_{u_1,u_3}\log\left(\sum_{\substack{R:|U_i|=u_i\\i=1,3}}Z_R\right)+o(n)\,.
\end{equation}
Recall our optimal way of constructing $R$ in the previous section, where we picked vertices from $Q$, $\varphi(T)\setminus Q$, $\overline{V}\setminus\Gamma_{d^*}$, and $\Gamma_{d^*}\setminus\varphi(T)$. In the case of Proposition~\ref{prop:fpp-special}, we have $Q/n\pto\lambda\log(1/\lambda)$ and $|\varphi(T)|/n=N_{d^*}/n\pto1-\lambda$ by our assumptions, so an optimal choice of $R$ satisfies $u_1=|Q|-o(n)$ and $u_3=o(n)$. These values are very close to extreme and thus do not contribute to any entropy of choosing $U_1$ and $U_3$ for those sizes. Thus, we may simply use $u_1=|Q|$ and $u_3=0$ as an optimal choice. In other words, an approximate maximum of \eqref{eqn:zreps-max} is achieved at these values of $u_1$ and $u_3$. Thus, we can write
\[
    \log Z_{r,\epsilon}=\log\left(\sum_{\substack{Q\subseteq R\subseteq\Gamma_{d^*}\\|R|/(\bar{n}-1)\in[r-\epsilon,r+\epsilon]}}Z_R\right)+o(n)\,.
\]
For $R$ with $Q\subseteq R\subseteq\Gamma_{d^*}$, \eqref{eqn:logzr-explicit} is written as
\[
    \begin{split}
        \log Z_R&\leq\sum_{v\in\Gamma_{d^*}\setminus Q\setminus U_2}\log(|\mathsf{par}_G(v)|-1)+\lambda_nn\log((1-\lambda_n)nq)-\overline{\beta}(d^*+\lambda_n)n+o(n)\,.
    \end{split}
\]
Using this, it is straightforward to see that it suffices to show
\[
    \lim_{\epsilon\to0}\plim_{n\to\infty}\left(\frac{1}{n}\log\left(\sum_{|U_2|\in[\bar{r}-\epsilon,\bar{r}+\epsilon]}\prod_{v\in\Gamma_{d^*}\setminus Q\setminus U_2}(|\mathsf{par}_G(v)|-1)\right)-\Psi_0(\lambda)\right)=-I_{\lambda}(\bar{r})\,.
\]
Note that the summand can be written as
\[
    \begin{split}
        \prod_{v\in\Gamma_{d^*}\setminus Q\setminus U_2}(|\mathsf{par}_G(v)|-1) &= \prod_{v\in\Gamma_{d^*}\setminus Q\setminus U_2}(|\mathsf{par}_G(v)|-1)\prod_{v\in U_2} 1\\
        &=\prod_{v\in\Gamma_{d^*}\setminus Q}|\mathsf{par}_G(v)|\prod_{v\in\Gamma_{d^*}\setminus Q\setminus U_2}\frac{|\mathsf{par}_G(v)|-1}{|\mathsf{par}_G(v)|}\prod_{v\in U_2}\frac{1}{|\mathsf{par}_G(v)|}\,.
    \end{split}
\]
Here, we have
\[
    \begin{split}
        \sum_{v\in\Gamma_{d^*}\setminus Q}\log|\mathsf{par}_G(v)| &= \sum_{v\in\overline{V}\setminus\Gamma_{\leq d^*-1}}\log(|\mathsf{N}_{\Gamma_{d^*-1}}(v)|\vee1)\\
        &=n\Psi_0(\lambda)+o_p(n)
    \end{split}
\]
where the last line is due to Corollary~\ref{cor:tail-sum} and Lemma~\ref{lem:logbinv-subg} (a similar argument was used in Section~\ref{sec:gibbs-1d-opt}). On the other hand, the product
\[
    \prod_{v\in\Gamma_{d^*}\setminus Q\setminus U_2}\frac{|\mathsf{par}_G(v)|-1}{|\mathsf{par}_G(v)|}\prod_{v\in U_2}\frac{1}{|\mathsf{par}_G(v)|}
\]
is in fact the joint probability mass function of independent $\operatorname{Bernoulli}(1/|\mathsf{par}_G(v)|)$ random variables evaluated at $\mathbf{1}_{v\in U_2}$. Thus, we seek to apply Proposition~\ref{prop:ldp} to conclude that
\[
    \frac{1}{n}\log\left(\sum_{|U_2|\in[\bar{r}-\epsilon,\bar{r}+\epsilon]}\prod_{v\in\Gamma_{d^*}\setminus Q\setminus U_2}\frac{|\mathsf{par}_G(v)|-1}{|\mathsf{par}_G(v)|}\prod_{v\in U_2}\frac{1}{|\mathsf{par}_G(v)|}\right)\pto-\inf_{x\in[\bar{r}-\epsilon,\bar{r}+\epsilon]}I_{\lambda}(x)
\]
which is enough to prove Proposition~\ref{prop:fpp-special}. It only remains to verify that the empirical measure of $1/|\mathsf{par}_G(v)|$ for $v\in\Gamma_{d^*}\setminus Q$ is close in $1$-Wasserstein metric to the distribution of $1/X$ conditioned on $X\geq2$ where $X\sim\Pois(\log(1/\lambda))$. This can easily be verified using Lemma~\ref{lem:empirical-conv} and Theorem~\ref{thm:w1-pois}; we omit the details.

\subsection{The energy landscape of the Gibbs measures}\label{sec:energy-landscape}

We conclude Section~\ref{sec:gibbs} with a simple and intuitive description of the energy landscape depicted in Figure~\ref{fig:potential-well}. Note that the figure only aims to visualize the relationship between energy and entropy, and does not reflect the actual geometry of the space of spanning trees. Specifically, the depth and width of the basin represent the energy and the entropy of a particular state, respectively. Thus, the slope represents the derivative of energy with respect to the entropy, which is the temperature needed to overcome that slope. It might be helpful to think that there is a moving particle inside the basin which moves faster as temperature increases.

From Proposition~\ref{prop:cond-gibbs-dv}, we know that for a fixed kernel size $m\geq m_\ell$, the energy of a typical tree under any Gibbs measure is $(d^*+1)n-m+o_p(n/\log\log n)$. Hence, loosely speaking, we can think of $-m$ as the energy by subtracting off the common term $(d^*+1)n$. In other words, the kernel size $m$ can be thought of as representing the \emph{depth} of a particular state in the basin. From Corollary~\ref{cor:fixed-kernel-size}, we see that the entropy of the state where $|\varphi(T)|=m$ is approximately
\begin{equation}\label{eqn:entropy-fixed-kernel}
    \Psi(m;N_{d^*},\lambda_n)+(n-m)\log(mq)+o(n)\,.
\end{equation}

\paragraph{The low temperature state (\textcolor{blue}{blue} dotted line).} This is the ground state of the system, corresponding to the shortest path trees. Unless the particle is moving fast enough to escape, it stays in the ground state. This agrees with our theoretical results that the Gibbs measure $\mu_{G,\beta}$ is close to the uniform measure $\mu_{G,\infty}$ if $\beta>\beta_c$.

\paragraph{The intermediate state (\textcolor{gray}{gray} region).} Now the particle is energetic enough to climb up the hill. This is the state that only exists in the regimes \ref{item:regime2} and \ref{item:regime3-1}, i.e., $\kappa_n\to\infty$. This implies that $N_{d^*}=\omega_p(n/\log\log n)$. Suppose that $m=CN_{d^*}$ for some constant $0<C\leq1$. Then using Lemma~\ref{lem:psi-bounds} and \eqref{eqn:psi-delta} the entropy \eqref{eqn:entropy-fixed-kernel} can be written as
\[
    \begin{split}
        \Psi(m;N_{d^*},\lambda_n)+(n-m)\log(mq) &= (\Delta +o_p(1))m\log\log n+n\log(mq)-(1+o_p(1))m\log\log n\\
        &=-(1-\Delta+o_p(1))m\log\log n+n\log(mq)\,.
    \end{split}
\]
Since $m\log\log n=\omega_p(n)$, $n\log(mq)$ part is negligible so this is simply $-(1-\Delta+o_p(1))m\log\log n$. Hence, the entropy is asymptotically a linear function of the energy $-m$, which explains the (approximately) constant slope shown in Figure~\ref{fig:potential-well}. Also, the approximate slope $(1-\Delta)\log\log n$ matches our critical temperature $\beta_c\log\log n$. Note that in the regime \ref{item:regime3-2} where $\kappa_n$ is asymptotically bounded, $N_{d^*}\log\log n=O_p(n)$ so $n\log(mq)$ is no longer negligible and this argument does not apply.

\paragraph{The high temperature state (\textcolor{red}{red} region).} This is the state where $m=O(n/\log\log n)$. Thus, we let $m=\frac{Cn}{\log\log n}$ for a constant $0<C<\liminf_{n\to\infty}\kappa_n$ (where the upper bound condition is to ensure $m<N_{d^*}$ asymptotically almost surely). In this part of the basin, the entropy \eqref{eqn:entropy-fixed-kernel} can be expressed as
\[
    \begin{split}
        \Psi(m;N_{d^*},\lambda_n)+(n-m)\log(mq) &= \Delta m\log\log n+n\log(mq)-m\log\log n+o_p(n)\\
        &=-(1-\Delta)m\log\log n+n\log(mq)+o_p(n)\,.
    \end{split}
\]
As noted above, $m\log\log n=Cn$ is a linear term, so $n\log(mq)=n\log(Cnq/\log\log n)$ is now a non-negligible term. The slope can be estimated by taking derivative with respect to $m$, which gives
\[
    (1-\Delta)\log\log n+\frac{n}{m}=(1-\Delta+C^{-1})\log\log n\,.
\]
Hence, the slope gets steeper as $C\to0$, i.e., as the depth decreases. This explains a curvy shape of the edges in this red region. In the regimes \ref{item:regime2} and \ref{item:regime3}, $C$ is unbounded in the right, and as $C\to\infty$ the slope approaches $\beta_c=1-\Delta$, connecting to the intermediate state. In the regime \ref{item:regime3-2}, $C$ is bounded by $\kappa\in(1,\infty)$ so the slope approaches $\beta_c=1-\Delta-\kappa^{-1}$. Since $C=\kappa$ corresponds to the ground state $m=N_{d^*}$, the intermediate state does not exist and thus the state changes continuously from low temperature to high temperature.

In either case, as we go beyond the high temperature state, we have $1-\Delta+C^{-1}\to\infty$, meaning that the system requires an unbounded temperature to escape the basin. Recall that $\beta=0$ gives the uniform measure over the spanning trees.

\section{Conclusions}\label{sec:conc}
From the perspective of complexity theory, we are interested in estimating the runtime and space complexity of \emph{any} algorithm which solves a problem. On the other hand, from the perspective of algorithm design, optimization, and MCMC, it is very useful to be able to compare different polynomial-time equivalent formulations of a problem in order to find the one which is most amenable to efficient algorithms. A good example, besides the ones already mentioned, is in the context of the ferromagnetic Ising model. It is an objective fact that sampling the ferromagnetic Ising model is polynomial time equivalent to sampling from the corresponding random cluster model. Nevertheless, the reformulation of the Ising model in terms of the random cluster was historically very important to design a polynomial time sampling algorithm \cite{swendsen1987nonuniversal,jerrum1993polynomial,guo2017random}. Analogously, the fact that heuristics for average case problems are sensitive to geometry may be a good thing. It also remains an interesting direction to understand the tractability of sampling for many problems which the Franz--Parisi potential suggests lack free energy barriers. 

Li and Schramm \cite{LS2024} studied a variant of the shortest path setup in terms of First Passage Percolation (FPP) on exponentially weighted random graphs, in part to test the possible hypothesis that the failure of OGP in Erd\"os--R\'enyi shortest paths is due to the inherent sparsity of the problem. This setting possesses some major qualitative differences from the Erd\"os--R\'enyi case --- in particular, the ``uniform distribution'' over shortest path trees in FPP has zero entropy since the shortest path tree is almost surely unique. So looking at the stability of the (uniform distribution over the single) shortest path tree in FPP is probably not the ``right'' object to study anymore. This is related to choosing the correct value of $\mu$ in the definition of the ensemble-OGP. (In our case, we proved there is no OGP regardless of the choice of $\mu$.) We leave a more detailed investigation of this and other models to future work.




\cleardoublepage
\phantomsection
\addcontentsline{toc}{section}{References}
\bibliographystyle{plain}
\bibliography{refs}

\appendix

\section{Mathematical tools}\label{apdx:tools}

Here we collect mathematical results used in our main proofs. 

\subsection{Concentration inequalities}

\paragraph{Subgaussian random variables.} We start with basic facts about subgaussian random variables; here we match the conventions in \cite{Vershynin_2018} where many further details can be found. We include standard proofs of several results in order to make the constants explicit.

\begin{definition}\label{def:subg}
    For a random variable $X$, the \emph{subgaussian norm} of $X$ is defined by
    \[
        \lVert X\rVert_{\psi_2}:=\inf\{K>0:\E[e^{X^2/K^2}]\leq2\}\,.
    \]
    We say $X$ is \emph{subgaussian} if it has finite subgaussian norm.
\end{definition}
Note that we do not require subgaussian random variables to have mean zero.

\begin{proposition}\label{prop:subg-tail}
    Let $X$ be a random variable. If $X$ is subgaussian then it satisfies the tail bound
    \[
        \Pr(|X|\geq t)\leq 2e^{-t^2/\lVert X\rVert_{\psi_2}^2}
    \]
    for all $t\geq0$. Conversely, if for some constant $K>0$, $X$ satisfies a tail bound
    \[
        \Pr(|X|\geq t)\leq 2e^{-t^2/K^2}
    \]
    for all $t\geq0$, then we have $\lVert X\rVert_{\psi_2}\leq\sqrt{3}K$.
\end{proposition}
\begin{proof}
    If $X$ is subgaussian then the Markov inequality gives
    \[
        \Pr(|X|\geq t)=\Pr(e^{X^2/K^2}\geq e^{t^2/K^2})\leq e^{-t^2/K^2}\E[e^{X^2/K^2}]\leq 2e^{-t^2/K^2}\,.
    \]
    For the converse, note that
    \[
        \begin{split}
            \E[e^{X^2/C^2}]&=\int_0^{\infty}\Pr(e^{X^2/C^2}\geq x)dx\\
            &=1+\int_1^{\infty}\Pr(|X|\geq C\sqrt{\log x})dx\\
            &\leq 1+2\int_1^{\infty}x^{-C^2/K^2}dx\,.
        \end{split}
    \]
    Setting $C=\sqrt{3}K$ the RHS is at most $2$, so $\lVert X\rVert_{\psi_2}\leq C=\sqrt{3}K$.
\end{proof}

\begin{proposition}\label{prop:subg-mgf}
    Let $X$ be a random variable with mean zero. If $X$ is subgaussian then its MGF satisfies
    \[
        \E[e^{\lambda X}]\leq e^{3\lambda^2\lVert X\rVert_{\psi_2}^2/2}\,.
    \]
    for all $\lambda\in\mathbb{R}$. Conversely, if for some constant $\sigma>0$ the MGF of $X$ satisfies
    \[
        \E[e^{\lambda X}]\leq e^{\lambda^2\sigma^2/2}
    \]
    for all $\lambda\in\mathbb{R}$, then $\lVert X\rVert_{\psi_2}\leq\sqrt{6}\sigma$.
\end{proposition}
\begin{proof}
    First, assume that $X$ is subgaussian with $\lVert X\rVert_{\psi_2}=K$. We use the inequality
    \[
        e^x\leq1+x+\frac{x^2}{2}e^{|x|}
    \]
    which can be easily verified with the Taylor expansion. Then we have
    \[
        \E[e^{\lambda X}]\leq 1+\E[\lambda X]+\frac{\lambda^2}{2}\E[X^2e^{|\lambda X|}]=1+\frac{\lambda^2K^2}{2}\E\left[\frac{X^2}{K^2} e^{|\lambda X|}\right]\,.
    \]
    Again we use the inequalities $x^2\leq e^{x^2/2}$ and $|\lambda X|\leq \frac{X^2}{2K^2}+\frac{\lambda^2K^2}{2}$ which give
    \[
        \E\left[\frac{X^2}{K^2} e^{|\lambda X|}\right]\leq e^{\lambda^2K^2/2}\E[e^{X^2/K^2}]\leq 2 e^{\lambda^2K^2/2}\,.
    \]
    Hence,
    \[
        \E[e^{\lambda X}]\leq 1+\lambda^2K^2e^{\lambda^2K^2/2}\leq (1+\lambda^2K^2)e^{\lambda^2K^2/2}\leq e^{3\lambda^2K^2/2}\,.
    \]

    For the converse, we have for any $\lambda>0$ that
    \[
        \Pr(|X|\geq t)=\Pr(e^{\lambda X}\geq e^{\lambda t})+\Pr(e^{-\lambda X}\geq e^{\lambda t})\leq 2e^{-\lambda t+\lambda^2\sigma^2/2}\,.
    \]
    Setting $\lambda=t/\sigma^2$ we obtain
    \[
        \Pr(|X|\geq t)\leq2 e^{-\frac{t^2}{2\sigma^2}}\,.
    \]
    Applying Proposition~\ref{prop:subg-tail} completes the proof.
\end{proof}

\begin{proposition}\label{prop:subg-center}
    If $X$ is subgaussian, then $X-\E[X]$ is also subgaussian with
    \[
        \lVert X-\E[X]\rVert_{\psi_2}\leq2\lVert X\rVert_{\psi_2}\,.
    \]
\end{proposition}
\begin{proof}
    Let $\mu=\E[X]$ and $K=\lVert X\rVert_{\psi_2}$. Then
    \[
        \E\left[e^{\frac{(X-\mu)^2}{4K^2}}\right] \leq \E\left[e^{\frac{X^2+\mu^2}{2K^2}}\right]=e^{\frac{(\E[X])^2}{2K^2}}\E\left[e^{\frac{X^2}{2K^2}}\right]\leq\left(\E\left[e^{\frac{X^2}{2K^2}}\right]\right)^2\leq\E[e^{X^2/K^2}]
    \]
    where the last two inequalities are due to Jensen's inequality.
\end{proof}

\begin{proposition}\label{prop:subg-sum}
    Let $X_1,\cdots,X_n$ be independent subgaussian random variables with mean zero. Then $X_1+\cdots+X_n$ has subgaussian norm
    \[
        \lVert X_1+\cdots+X_n\rVert_{\psi_2}^2\leq 18\sum_{i=1}^n\lVert X_i\rVert_{\psi_2}^2\,.
    \]
\end{proposition}
\begin{proof}
    Since the MGF of $X_1+\cdots+X_n$ is the product of the MGFs of $X_i$, Proposition~\ref{prop:subg-mgf} easily gives the result.
\end{proof}

\begin{lemma}\label{lem:subg-mean}
    If $X$ is subgaussian then $|\E[X]|\leq\lVert X\rVert_{\psi_2}$.
\end{lemma}
\begin{proof}
    By Proposition~\ref{prop:subg-tail}
    \[
        \begin{split}
            \E[|X|] &= \int_0^\infty\Pr(|X|\geq t)dt\\
            &\leq\int_0^\infty 2e^{-t^2/\lVert X\rVert_{\psi_2}^2}dt\\
            &\leq\sqrt{\pi}\lVert X\rVert_{\psi_2}\,.
        \end{split}
    \]
    This also implies the existence of $\E[X]$. To get a sharper result, we can appeal to Jensen's inequality:
    \[
        e^{(\E[X])^2/\lVert X\rVert_{\psi_2}^2}\leq\E\left[e^{X^2/\lVert X\rVert_{\psi_2}^2}\right]\leq 2\,.
    \]
    Hence, $|\E[X]|\leq\sqrt{\log 2}\lVert X\rVert_{\psi_2}\leq\lVert X\rVert_{\psi_2}$.
\end{proof}
This leads to the following subgaussian version of Hoeffding's inequality.
\begin{corollary}\label{cor:tail-sum}
    Suppose that $X_1,\cdots,X_n$ are i.i.d. random variables with mean $\mu$ and subgaussian norm at most $K$. Then we have
    \[
        \Pr\left(\left|\sum_{i=1}^nX_i-n\mu\right|\geq t\right)\leq 2e^{-{\frac{t^2}{72nK^2}}}\,.
    \]
\end{corollary}
\begin{proof}
    This follows by Proposition~\ref{prop:subg-center} and Proposition~\ref{prop:subg-sum}.
\end{proof}

\paragraph{Chernoff bound.} 
We use the well-known Chernoff bound for binomials throughout this paper.
\begin{lemma}\label{lem:raw-chernoff}
    Let $Y\sim\Binom(n,p)$. Then for any $-1<\delta<1$ we have
    \[
        \Pr(Y\geq (1+\delta)np)\leq \exp(np(\delta-(1+\delta)\log(1+\delta)))\,.
    \]
\end{lemma}
This yields the following bounds which are usually more convenient.
\begin{lemma}[Chernoff bound, see, e.g. \cite{motwani1995randomized}]\label{lem:chernoff}
    Let $Y\sim\Binom(n,p)$. Then for any $\delta\geq0$ we have
    \[
        \Pr(Y\geq (1+\delta)np)\leq e^{-\frac{\delta^2np}{2+\delta}}
    \]
    and
    \[
        \Pr(Y\leq (1-\delta)np)\leq e^{-\frac{\delta^2np}{2}}\,.
    \]
    In particular, we have
    \[
        \Pr(|Y-np|\geq\delta np)\leq2e^{\frac{\delta^2np}{3}}\,.
    \]
\end{lemma}


\subsection{Logarithm of binomial random variables}\label{apdx:log-binomial}

In Section~\ref{sec:gibbs} when we study Gibbs measures and random spanning trees, we are interested in the entropy of choosing a parent for each vertex, which is naturally related to the sum of the \emph{logarithm} of binomial random variables. In this section, we study the concentration of such random variables. The binomial distribution has a nonzero mass at $0$, so its logarithm is of course not well defined. Instead, we consider transformed versions of binomial distributions which put no mass at zero. The first family we consider are \emph{shifted} binomial distributions.

\begin{lemma}\label{lem:logbin-subg}
    Let $X\sim\Binom(n, p)$. Then for $a>0$ we have
    \[
        \lVert\log(X+a)-\log(np+a)\rVert_{\psi_2}\leq\sqrt{\frac{6}{a}}\,.
    \]
\end{lemma}
\begin{proof}
    Let $Y=\log(X+a)$ and $\mu=np$. By Lemma~\ref{lem:chernoff} we have
    \[
        \begin{split}
            \Pr(\log(X+a)-\log(\mu+a)\geq t) &= \Pr(X-\mu\geq(e^t-1)(\mu+a))\\
            &\leq \exp\left(-\frac{(e^t-1)^2(\mu+a)^2}{2\mu+(e^t-1)(\mu+a)}\right)\\
            &\leq \exp\left(-\frac{(e^t-1)^2(\mu+a)^2}{(e^t+1)(\mu+a)}\right)\\
            &=\exp\left(-(\mu+a)\cdot\frac{(e^t-1)^2}{e^t+1}\right)
        \end{split}
    \]
    for $t>0$. Note that
    \begin{equation}\label{eqn:tmp1}
        \frac{(e^t-1)^2}{e^t+1}\geq\frac{(e^t-1)^2}{2e^t}=\cosh t-1\geq\frac{t^2}{2}\,.
    \end{equation}
    Thus, we have
    \[
        \Pr(\log(X+a)-\log(\mu+a)\geq t) \leq \exp\left(-\frac{\mu+a}{2}\cdot t^2\right)\,.
    \]
    For the lower bound, we begin with
    \[
        \begin{split}
            \Pr(\log(X+a)-\log(\mu+a)\leq -t) &= \Pr(X-\mu\leq-(1-e^{-t})(\mu+a))\\
            &\leq \exp\left(-\frac{(1-e^{-t})^2(\mu+a)^2}{2\mu}\right)
        \end{split}
    \]
    for $t>0$. We only need to care about the case $t\leq\log(\mu+a)-\log a$, since otherwise the LHS is zero. This implies that
    \[
        1-e^{-t}\geq\frac{1-\frac{a}{\mu+a}}{\log(\mu+a)-\log a}\cdot t=\frac{1}{\mu + a}\cdot\frac{\mu}{\log(\mu + a)-\log a}\cdot t\,.
    \]
    We use the inequality (verifiable with, e.g., Taylor expansion)
    \[
        \frac{e^x-1}{x}\geq e^{x/2}
    \]
    with $x=\log(\mu+a)-\log a$, giving
    \[
        1-e^{-t}\geq\sqrt{\frac{a}{\mu+a}}\cdot t\,.
    \]
    Hence,
    \[
        \Pr(\log(X+a)-\log(\mu+a)\leq -t)\leq\exp\left(-\frac{a(\mu+a)}{2\mu}\cdot t^2\right)
    \]
    for all $t>0$.
    
    By Proposition~\ref{prop:subg-tail}, we have
    \[
        \lVert\log(X+a)-\log(np+a)\rVert_{\psi_2}^2\leq\frac{6}{\mu+a}+\frac{6\mu}{a(\mu+a)}\leq\frac{6}{a}\,.
    \]
\end{proof}

Next, we analyze \emph{rectified} binomial distributions.

\begin{lemma}\label{lem:logbinv-subg}
    Let $X\sim\Binom(n,p)$. Then we have
    \[
        \lVert\log(X\vee1)-\log(np\vee1)\rVert_{\psi_2}\leq\frac{8\sqrt{3}}{\sqrt{np\vee1}}\,.
    \]
\end{lemma}
\begin{proof}
    Let $\mu=np$. First, assume $\mu>1$. Then by Lemma~\ref{lem:chernoff}
    \[
        \begin{split}
            \Pr(\log(X\vee1)-\log\mu\geq t) &= \Pr(X\vee1\geq e^t\mu)\\
            &=\Pr(X\geq e^t\mu)\\
            &=\Pr(X-\mu\geq(e^t-1)\mu)\\
            &\leq\exp\left(-\mu\cdot\frac{(e^t-1)^2}{e^t+1}\right)\\
            &\leq\exp\left(-\frac{\mu}{2}\cdot t^2\right)
        \end{split}
    \]
    for $t>0$ where the last line follows from \eqref{eqn:tmp1}. Also,
    \[
        \Pr(\log(X\vee1)-\log\mu\leq-t)\leq\Pr(X\leq e^{-t}\mu)\leq\exp\left(-\frac{\mu(1-e^{-t})^2}{2}\right)
    \]
    for $t>0$. The LHS is zero if $t>\log\mu$, otherwise $t\leq\log\mu$ which implies
    \[
        1-e^{-t}\geq\frac{\mu-1}{\mu\log\mu}\cdot t\,.
    \]
    Here we use an inequality (verifiable with simple calculus)
    \[
        \frac{e^x-1}{x}\geq \frac{1}{2}e^{\frac{3}{4}x}
    \]
    and take $x=\log\mu$. Then
    \[
        1-e^{-t}\geq\frac{1}{2\mu^{1/4}}t
    \]
    which gives
    \[
        \Pr(\log(X\vee1)-\log\mu\leq-t)\leq\exp\left(-\frac{\sqrt{\mu}}{8}\cdot t^2\right)
    \]
    for all $t>0$. Since
    \[
        \frac{\sqrt{\mu}}{8}\leq\frac{\mu}{2}
    \]
    for all $\mu>1$, we have by Proposition~\ref{prop:subg-tail}
    \[
        \lVert\log(X\vee1)-\log(np\vee1)\rVert_{\psi_2}\leq\frac{8\sqrt{3}}{\sqrt{\mu}}\,.
    \]

    Next, we assume $\mu\leq1$. Then
    \[
        \begin{split}
            \Pr(\log(X\vee1)\geq t)&=\Pr(X\geq e^t)\\
            &=\Pr(X-\mu\geq e^t-\mu)\\
            &\leq\exp\left(-\frac{(e^t-\mu)^2}{e^t+\mu}\right)
        \end{split}
    \]
    for $t>0$. Since $(e^t-\mu)^2/(e^t+\mu)\geq(e^t-1)^2/(e^t+1)$ we get
    \[
        \Pr(\log(X\vee1)\geq t)\leq e^{-t^2/2}\,.
    \]
    Since the lower tail is empty, we simply have
    \[
        \lVert\log(X\vee1)-\log(np\vee1)\rVert_{\psi_2}\leq\sqrt{6}\,.
    \]
\end{proof}

The last family we consider is a family of \emph{truncated} binomial distributions, which is simply conditioning on positiveness. This occurs the most often in our analysis, so we introduce a definition.

\begin{definition}\label{def:ztb}
    The \emph{zero-truncated binomial distribution} with parameters $n\in\mathbb{Z}^+$ and $0<p\leq 1$, denoted by $\ZTB(n, p)$, is a distribution with probability mass function
    \[
        \frac{1}{1-(1-p)^n}\cdot\binom{n}{x}p^x(1-p)^{n-x}
    \]
    for $x=1,\cdots,n$.
\end{definition}

We prove a similar result for the zero-truncated distributions.

\begin{lemma}\label{lem:subg-ztb}
    Let $X\sim\ZTB(n,p)$. Then we have
    \[
        \lVert \log X-\log(np\vee 1)\rVert_{\psi_2}\leq\frac{8\sqrt{6}}{\sqrt{np\vee1}}\,.
    \]
    In particular, $\lVert \log X-\log(np\vee 1)\rVert_{\psi_2}$ is bounded by a universal constant.
\end{lemma}
\begin{proof}
    Let $\mu=np$ and first assume $\mu\geq1$. Consider $Y\sim\Binom(n,p)$. For $K=\lVert \log(Y\vee1)-\log\mu\rVert_{\psi_2}$ we have
    \[
        \begin{split}
            2&\geq\E[e^{(\log(Y\vee1)-\log\mu)^2/K^2}]\\
            &= e^{(\log\mu)^2/K^2}\Pr(Y=0)+\E[e^{(\log(Y\vee1)-\log\mu)^2/K^2}\mid Y>0]\Pr(Y>0)\\
            &\geq(1-p)^n+\E[e^{(\log X-\log\mu)^2/K^2}](1-(1-p)^n)
        \end{split}
    \]
    which gives
    \[
        \begin{split}
            \E[e^{\log X-\log(\mu\vee1))^2/K^2}] &\leq \frac{2-(1-p)^n}{1-(1-p)^n}\\
            &=1+\frac{1}{1-(1-p)^n}\\
            &\leq 1+\frac{1}{1-e^{-\mu}}\\
            &< 4\,.
        \end{split}
    \]
    By Jensen's inequality,
    \[
        \E[e^{(\log X-\log(\mu\vee1))^2/(2K^2)}]\leq(\E[e^{(\log X-\log(\mu\vee1))^2/K^2}])^{1/2}\leq2
    \]
    which implies $\lVert \log X-\log(np\vee 1)\rVert_{\psi_2}\leq\sqrt{2}K$. Now we can invoke Lemma~\ref{lem:logbinv-subg} to conclude
    \[
        \lVert \log X-\log(\mu\vee 1)\rVert_{\psi_2}\leq\frac{8\sqrt{6}}{\sqrt{\mu}}\,.
    \]

    Now assume $\mu<1$. Here, we use Lemma~\ref{lem:raw-chernoff} which gives
    \[
        \begin{split}
            \Pr(\log X\geq t) &= \frac{\Pr(X-\mu\geq e^t-\mu)}{1-(1-p)^n}\\
            &\leq\frac{1}{1-e^{-\mu}}\exp\left(e^t-\mu-e^t\log\left(1+\frac{e^t-\mu}{\mu}\right)\right)\\
            &=\exp\left(e^t-\mu-te^t-e^t\log(1/\mu)+\log\left(\frac{1}{1-e^{-\mu}}\right)\right)\\
            &=\exp\left(-(t-1)e^t-(e^t-1)\log(1/\mu)-\log\left(\frac{e^\mu-1}{\mu}\right)\right)\\
            &\leq\exp\left(-(t-1)e^t\right)
        \end{split}
    \]
    for $t>0$. If we only consider $t\geq2$, we get
    \[
        \Pr(\log X\geq t)\leq\exp\left(-(t-1)e^t\right)\leq e^{-t^2/2}\,.
    \]
    Then we can say $\Pr(|\log X|\geq t)\leq2e^{-t^2/2}$ for all $t>0$ since the lower tail does not exist and the case $t<2$ is trivial.
\end{proof}

\subsection{Poisson approximation}\label{apdx:pois-apx}

In many cases, approximating binomial distributions with Poisson distributions simplifies the results and often makes them more interpretable. We first state the following results which can easily be translated from those for binomial distributions in the previous section. We omit the proofs since they are nearly identical.

\begin{lemma}
    Let $Y\sim\Pois(\mu)$. Then for any $0<\delta<1$ we have
    \[
        \Pr(Y\geq(1+\delta)\mu)\leq e^{-\frac{\delta^2\mu}{2+\delta}}
    \]
    and
    \[
        \Pr(Y\leq (1-\delta)np)\leq e^{-\frac{\delta^2\mu}{2}}\,.
    \]
\end{lemma}

\begin{lemma}
    Let $X\sim\Pois(\mu)$. Then for $a>0$ we have
    \[
        \lVert\log(X+a)-\log(\mu+a)\rVert_{\psi_2}\leq\sqrt{\frac{6}{a}}\,.
    \]
\end{lemma}

\begin{lemma}\label{lem:subg-rep}
    Let $X\sim\Pois(\mu)$. Then we have
    \[
        \lVert\log(X\vee1)-\log(\mu\vee1)\rVert_{\psi_2}\leq\frac{8\sqrt{3}}{\sqrt{\mu\vee1}}\,.
    \]
\end{lemma}

The zero-truncated Poisson distributions are also similarly defined.

\begin{definition}\label{def:ztp}
    The \emph{zero-truncated Poisson distribution} with parameter $\mu>0$, denoted by $\ZTP(\mu)$, is a distribution with probability mass function
    \[
        \frac{1}{1-e^{-\mu}}\cdot\frac{\mu^xe^{-\mu}}{x!}
    \]
    for $x=1,2,\cdots$.
\end{definition}

\begin{lemma}\label{lem:subg-ztp}
    Let $X\sim\ZTP(\mu)$. Then we have
    \[
        \lVert\log X-\log(\mu\vee1)\rVert_{\psi_2}\leq\frac{8\sqrt{6}}{\sqrt{\mu\vee1}}\,.
    \]
\end{lemma}

We also extract a powerful result from the literature on on Poisson approximations \cite{10.1214/aop/1176996359, 10.1093/oso/9780198522355.001.0001, 876852c4-3779-3448-b0ee-191ab7613ff7}.

\begin{theorem}[\cite{876852c4-3779-3448-b0ee-191ab7613ff7}]\label{thm:w1-pois}
    There is a universal constant $C>0$ such that
    \[
        W_1(\Binom(n,p),\Pois(np))\leq C\sqrt{n}p^{3/2}\,.
    \]
\end{theorem}

\subsection{LogSumExp functions and empirical measures}\label{apdx:lse}

For a vector $\mathbf{x}=(x_1,\cdots,x_n)\in\R^n$, we define the $m$-wise LogSumExp of $\mathbf{x}$ by
\[
    \mathsf{LSE}_m(\mathbf{x}):=\log\left(\sum_{I\in\binom{[n]}{m}}\exp\left(\sum_{i\in I}x_i\right)\right)\,.
\]
We establish some useful properties of $\mathsf{LSE}_m$ which is used in Section~\ref{sec:hg-conc} to prove uniform concentration of certain log partition functions.

\begin{lemma}\label{lem:lse}
    For a vector $\mathbf{x}=(x_1,\cdots,x_n)\in\R^n$, let $\mu_{\mathbf{x}}$ be its empirical measure on $\R$.
    Then for two vectors $\mathbf{x},\mathbf{y}\in\R^n$, we have
    \begin{equation}\label{eqn:lse-bound}
        \frac{1}{n}|\mathsf{LSE}_m(\mathbf{x})-\mathsf{LSE}_m(\mathbf{y})|\leq W_1(\mu_{\mathbf{x}},\mu_{\mathbf{y}})
    \end{equation}
    for all $1\leq m\leq n$.
\end{lemma}
\begin{proof}
    Since $\mathsf{LSE}_m$ is invariant under permutation, we may assume $x_1\leq\cdots\leq x_n$ and $y_1\leq\cdots\leq y_n$ without loss of generality. We use the following simple inequality that holds for any $a_1,a_2,b_1,b_2>0$:
    \[
        \frac{b_1}{a_1}\wedge\frac{b_2}{a_2}\leq\frac{b_1+b_2}{a_1+a_2}\leq\frac{b_1}{a_1}\vee\frac{b_2}{a_2}
    \]
    which implies
    \[
        \left|\log\left(\frac{b_1+b_2}{a_1+a_2}\right)\right|\leq\left|\log\left(\frac{b_1}{a_1}\right)\right|\vee\left|\log\left(\frac{b_2}{a_2}\right)\right|\,.
    \]
    Applying this repeatedly, we get
    \[
        \begin{split}
            |\mathsf{LSE}_m(\mathbf{x})-\mathsf{LSE}_m(\mathbf{y})| &\leq \max_{I\in\binom{[n]}{m}}\left|\log\left(\frac{\exp\left(\sum_{i\in I}x_i\right)}{\exp\left(\sum_{i\in I}y_i\right)}\right)\right|\\
            &\leq\max_{I\in\binom{[n]}{m}}\left|\sum_{i\in I}x_i-\sum_{i\in I}y_i\right|\\
            &\leq\sum_{i=1}^n|x_i-y_i|\\
            &=nW_1(\mu_{\mathbf{x}},\mu_{\mathbf{y}})\,.
        \end{split}
    \]
\end{proof}

Lemma~\ref{lem:lse} indicates that $\mathsf{LSE}_m(\mathbf{x})$ is related to the empirical distribution of $\mathbf{x}$. As such, we take the following result from \cite{10.3150/19-BEJ1151} (see, e.g., Proposition~3.1 and Corollary~5.2). This also holds for subexponential distributions, but subgaussian assumption suffices for our application.

\begin{lemma}[\cite{10.3150/19-BEJ1151}]\label{lem:empirical-conv}
    Let $\mu$ be a probability measure on $\R$ with subgaussian norm at most $\sigma$. Then there is a constant $C_\sigma$ and a sequence $p_{\sigma,n}=1-o(1)$ such that for an empirical measure $\hat{\mu}_n$ we have
    \[
        \E[W_1(\hat{\mu}_n,\mu)]\leq C_\sigma n^{-1/2}
    \]
    and
    \[
        \Pr\left(W_1(\hat{\mu}_n,\mu)\leq C_\sigma n^{-1/2}\right)\geq p_{\sigma,n}\,.
    \]
    In other words, we have $\E[W_1(\hat{\mu}_n,\mu)]=O(n^{-1/2})$ and $W_1(\hat{\mu}_n,\mu)=O_p(n^{-1/2})$ uniformly over all $\mu$ with subgaussian norm at most $\sigma$.
\end{lemma}

This enhances Lemma~\ref{lem:lse}, leading to the following key lemma.

\begin{lemma}\label{lem:lse2}
    Let $\mu$ be a probability measure on $\R$ with subgaussian norm at most $\sigma$. Then there is a constant $C_\sigma$ such that for any vector $\mathbf{y}=(y_1,\cdots,y_n)\in\R^n$ and $1\leq m\leq n$, we have
    \[
        \left|\mathsf{LSE}_m(\mathbf{y})-\E[\mathsf{LSE}_m(\mathbf{X})]\right|\leq nW_1(\mu, \mu_{\mathbf{y}})+C_{\sigma}\sqrt{n}
    \]
    where the expectation is over i.i.d. samples $X_1,\cdots,X_n\sim\mu$.
\end{lemma}
\begin{proof}
    By Lemma~\ref{lem:lse} we have
    \[
        \begin{split}
            \frac{1}{n}|\mathsf{LSE}_m(\mathbf{X})-\mathsf{LSE}_m(\mathbf{y})| &\leq W_1(\mu_{\mathbf{X}},\mu_{\mathbf{y}})\\
            &\leq W_1(\mu_{\mathbf{X}},\mu)+W_1(\mu,\mu_{\mathbf{y}})\,.
        \end{split}
    \]
    Now Jensen's inequality gives
    \[
        \begin{split}
            \frac{1}{n}\left|\mathsf{LSE}_m(\mathbf{y})-\E[\mathsf{LSE}_m(\mathbf{X})]\right| &\leq \frac{1}{n}\E\left[|\mathsf{LSE}_m(\mathbf{X})-\mathsf{LSE}_m(\mathbf{y})|\right]\\
            &\leq\E[W_1(\mu_{\mathbf{X}},\mu)]+W_1(\mu,\mu_{\mathbf{y}})\,.
        \end{split}
    \]
    Applying Lemma~\ref{lem:empirical-conv} completes the proof.
\end{proof}

\subsection{Large deviations}\label{sec:large-deviations}

In this section, we prove a slight extension of the Cram\'er's theorem in large deviations theory specific to independent not not necessarily identically distributed Bernoulli random variables. We consider the case where the parameters for those Bernoulli distributions are also random, written in terms of a triangular array $\mathbf{p}$ of random variables
\[
    \begin{array}{cccc}
        \mathbf{p}_{11} & & \\
        \mathbf{p}_{21} & \mathbf{p}_{22}\\
        \mathbf{p}_{31} & \mathbf{p}_{32} & \mathbf{p}_{33}\\
        \vdots & \vdots & \vdots & \ddots
    \end{array}
\]
where $\mathbf{p}_{ni}\in(0,1)$ for all $n$ and $1\leq i\leq n$. We also assume that the empirical measures (which are also random)
\[
    \hat{\mu}_{\mathbf{p}_n}:=\frac{1}{n}\sum_{i=1}^n\delta_{\mathbf{p}_{ni}}
\]
asymptotically behaves like a fixed probability distribution $P$ supported on $(0,1)$, in the sense that
\begin{equation}\label{eqn:emp-conv}
    W_1(\hat{\mu}_{\mathbf{p}_n},P)\pto0
\end{equation}
where $W_1$ is the usual $1$-Wasserstein metric. Conditioned on $\mathbf{p}$, we consider a triangular array of conditionally i.i.d. random variables $X_{ni}\sim\Bernoulli(\mathbf{p}_{ni})$. We establish a large deviations behavior of the random variables
\[
    S_n:=\frac{1}{n}\sum_{i=1}^nX_{ni}
\]
conditioned on $\mathbf{p}_n$. To describe the behavior, we define the expected cumulant generating function
\[
    \Lambda(t):=\E_{\mathbf{q}\sim P}[\log(1-\mathbf{q}+\mathbf{q}e^{t})]
\]
and its Legendre transform
\[
    \Lambda^*(x):=\sup_{t\in\mathbb{R}}(tx-\Lambda(t))\,.
\]
Then we have the following result, which can be viewed as a ``random environment'' or a ``quenched'' version of Cram\'er's theorem.

\begin{proposition}\label{prop:ldp}
    Under the assumptions above, we have
    \[
        \frac{1}{n}\log(\Pr(S_n\in \mathcal{I}\mid\mathbf{p}_n))\pto-\inf_{x\in \mathcal{I}}\Lambda^*(x)
    \]
    for any compact interval $\mathcal{I}\subseteq(0,1)$ with nonempty interior.
\end{proposition}

The proof follows very closely along the lines of the proofs of Cram\'er's theorem or the G\"artner--Ellis theorem\footnote{Though we only consider intervals in the interest of simplicity, Proposition~\ref{prop:ldp} easily extends to a statement about any closed and open sets, just like the standard Cram\'er's theorem.} (see, e.g., \cite{dembo2009large} for more on large deviations). A main technical difficulty is that of course the distribution of $S_n\mid\mathbf{p}_n$ itself is random. To deal with this, we begin with a simple but powerful lemma for working with convex functions in random environment.

\begin{lemma}\label{lem:cvx}
    Let $\{f_n\}$ be a sequence of random strictly convex functions $f_n:X\to\mathbb{R}$ defined on a convex open set $X\subseteq\mathbb{R}$ converging pointwise to a strictly convex function $f:X\to\mathbb{R}$ in probability, i.e., for each $x\in X$ we have $f_n(x)\pto f(x)$. The the following hold.

    \begin{enumerate}[label=(\alph*)]
        \item The convergence $f_n\pto f$ is uniform over any compact set $\mathcal{K}\subseteq X$:
        \[
            \sup_{x\in\mathcal{K}}|f_n(x)-f(x)|\pto0\,.
        \]

        \item Suppose that $x_n^*\in X$ is the unique minimizer of $f_n$ for each $n$ and $x^*\in X$ is the unique minimizer of $f$. Then we have
        \[
            x_n^*\pto x^*
        \]
        and
        \[
            f_n(x_n^*)\pto f(x^*)\,.
        \]
    \end{enumerate}

\end{lemma}
\begin{proof}
    \begin{enumerate}[label=(\alph*)]
        \item Let $\epsilon>0$. We use the fact that any convex function $g$ satisfies
        \[
            |g(x)-g(a)|\leq|g(a+\delta)-g(a)|\vee|g(a-\delta)-g(a)|
        \]
        for all $x\in[a-\delta,a+\delta]$. For each $a\in\mathcal{K}$, since $f_n(a+\delta)\pto f(a+\delta)$, $f_n(a)\pto f(a)$, and $f_n(a-\delta)\pto f(a-\delta)$, and since convex functions are locally Lipschitz, we can find a small enough neighborhood $(a-\delta,a+\delta)\subseteq X$ such that
        \[
            \sup_{x\in(a-\delta,a+\delta)}|f_n(x)-f(x)|<\epsilon
        \]
        with probability at least $1-o(1)$. Since $\mathcal{K}$ is compact, we can find finite values of $a$ such that those neighborhoods cover $\mathcal{K}$.

        \item For any $\epsilon>0$ we have $f_n(x^*-\epsilon)\pto f(x^*-\epsilon)>f(x^*)$, $f_n(x^*+\epsilon)\pto f(x^*+\epsilon)>f(x^*)$, and $f_n(x^*)\pto f(x^*)$. Thus, we have $x_n^*\in(x^*-\epsilon,x+\epsilon)$ with high probability, proving the first statement. The second statement follows since convex functions are locally Lipschitz.
    \end{enumerate}
\end{proof}

\begin{proof}[Proof of Proposition~\ref{prop:ldp}]
    Let $\mathcal{I}=[a,b]$ with $0<a<b<1$. We also define the mean cumulant generating function
    \[
        \Lambda(t;\mathbf{p}_n):=\frac{1}{n}\sum_{i=1}^n\log(1-\mathbf{p}_{ni}+\mathbf{p}_{ni}e^t)
    \]
    and its Legendre transform
    \[
        \Lambda^*(x;\mathbf{p}_n):=\sup_{t\in\mathbb{R}}(tx-\Lambda(t;\mathbf{p}_n))\,.
    \]
    Note that by Lemma~\ref{lem:cvx} and the assumption \eqref{eqn:emp-conv} we have
    \[
        \Lambda(t;\mathbf{p}_n)\pto\Lambda(t)
    \]
    pointwise for all $t\in\mathbb{R}$ and also
    \begin{equation}\label{eqn:lbdstar-conv}
        \Lambda^*(x;\mathbf{p}_n)\pto\Lambda^*(x)
    \end{equation}
    pointwise for all $x\in(0,1)$.
    
    We first prove an upper bound
    \begin{equation}\label{eqn:ldp-ub}
        \frac{1}{n}\log(\Pr(S_n\in[a,b]\mid\mathbf{p}_n))\leq-\inf_{x\in[a,b]}\Lambda^*(x)+o_p(1)\,.
    \end{equation}
    Note that if $x=\E[S_n\mid\mathbf{p}_n]=\frac{1}{n}\sum_{i=1}^n\mathbf{p}_{ni}$, then the line $tx$ as a function of $t$ is tangent to $\Lambda(t;\mathbf{p}_n)$, so $\Lambda^*(x)=0$. Thus, if $[a,b]$ contains $\E[S_n\mid\mathbf{p}_n]$ then \eqref{eqn:ldp-ub} is trivial. Thus, we assume that $a>\E[S_n\mid\mathbf{p}_n]$ or $b<\E[S_n\mid\mathbf{p}_n]$. We only consider the former, as the latter can be handled in the same way. Using the Chernoff bound, we have
    \[
        \begin{split}
            \Pr(S_n\in[a,b]\mid\mathbf{p}_n) &\leq\Pr(S_n\geq a\mid\mathbf{p}_n)\\
            &\leq \exp\left(-\sup_{t\geq0}(nta-n\Lambda_n(t;\mathbf{p}_n))\right)\\
            &=\exp\left(-n\Lambda^*(a;\mathbf{p}_n)\right)
        \end{split}
    \]
    where in the last equality we have used $a>\E[S_n\mid\mathbf{p}_n]$. This condition also implies that $\Lambda^*(a)=\inf_{x\in[a,b]}\Lambda^*(x)$, so we have
    \[
        \frac{1}{n}\log(\Pr(S_n\in[a,b]\mid\mathbf{p}_n))\leq-\inf_{x\in[a,b]}\Lambda^*(x;\mathbf{p})\,.
    \]
    Applying Lemma~\ref{lem:cvx} and \eqref{eqn:lbdstar-conv}, \eqref{eqn:ldp-ub} is proved.

    Next, we prove a lower bound
    \[
        \frac{1}{n}\log(\Pr(S_n\in(a,b)\mid\mathbf{p}_n))\geq-\inf_{x\in(a,b)}\Lambda^*(x)-o_p(1)\,.
    \]
    Fix any $y\in(a,b)$ and $\epsilon>0$ such that $a<y-\epsilon<y+\epsilon<b$. Let $t_{n,y}\in\mathbb{R}$ satisfy $\Lambda^*(y;\mathbf{p}_n)=t_{n,y}y-\Lambda(t_{n,y};\mathbf{p}_n)$ and similarly let $t_y\in\mathbb{R}$ satisfy $\Lambda^*(y)=t_yy-\Lambda(t_y)$. Note that
    \[
        |t_{n,y}|\leq |t_{n,a}|\vee |t_{n,b}|
    \]
    which allows us to bound $|t_{n,y}|$ uniformly. We denote the RHS by
    \[
        t_{n,[a,b]}=|t_{n,a}|\vee |t_{n,b}|\,.
    \]
    By Lemma~\ref{lem:cvx}, $t_{n,[a,b]}$ converges in probability, and we denote by $t_{[a,b]}$ its probability limit. Now we consider random variables
    \[
        \tilde{X}_{ni}\sim\Bernoulli\left(\frac{\mathbf{p}_{ni}e^{t_{n,y}}}{1-\mathbf{p}_{ni}+\mathbf{p}_{ni}e^{t_{n,y}}}\right)
    \]
    which are i.i.d. conditioned on $\mathbf{p}$, and also
    \[
        \tilde{S}_n:=\frac{1}{n}\sum_{i=1}^n\tilde{X}_{ni}\,.
    \]
    It is easy to see that $\tilde{S}_n$ follows an exponentially tilted distribution of $S_n$ which satisfies
    \[
        \Pr(\tilde{S}_n=x)=\exp(n(t_{n,y}x-\Lambda_n(t_{n,y};\mathbf{p}_{n})))\cdot\Pr(S_n=x)
    \]
    and has the mean cumulant generating function
    \[
        \tilde{\Lambda}(t;\mathbf{p}_n)=\Lambda(t+t_{n,y};\mathbf{p}_n)-\Lambda(t_{n,y};\mathbf{p}_n)\,.
    \]
    Then we have
    \[
        \begin{split}
            \Pr(S_n\in(y-\epsilon,y+\epsilon)) &= \sum_{x\in(y-\epsilon,y+\epsilon)}\exp(-n(t_{n,y}x-\Lambda_n(t_{n,y};\mathbf{p}_{n})))\cdot\Pr(\tilde{S}_n=x)\\
            &\geq\exp(-n(t_{n,y}y-\Lambda_n(t_{n,y};\mathbf{p}_n)))\exp(-n\epsilon|t_{n,y}|)\Pr(\tilde{S}_n\in(y-\epsilon,y+\epsilon))\\
            &=\exp(-n\Lambda^*(y;\mathbf{p}_n)-n\epsilon|t_{n,y}|)\Pr(\tilde{S}_n\in(y-\epsilon,y+\epsilon))\,.
        \end{split}
    \]
    Put another way,
    \[
        \begin{split}
            \frac{1}{n}\log(\Pr(S_n\in(y-\epsilon,y+\epsilon)\mid\mathbf{p}_n))&\geq-\Lambda^*(y;\mathbf{p}_n)-\epsilon|t_{n,y}|+\frac{1}{n}\log(\Pr(\tilde{S}_n\in(y-\epsilon,y+\epsilon)))\\
            &\geq-\Lambda^*(y;\mathbf{p}_n)-\epsilon t_{n,[a,b]}+\frac{1}{n}\log(\Pr(\tilde{S}_n\in(y-\epsilon,y+\epsilon)))
        \end{split}
    \]
    Due to Lemma~\ref{lem:cvx} the convergences $\Lambda^*(y;\mathbf{p}_n)\pto\Lambda^*(y)$ and $t_{n,[a,b]}\pto t_{[a,b]}$ are uniform over $y$, so we have
    \[
        \frac{1}{n}\log(\Pr(S_n\in(y-\epsilon,y+\epsilon)\mid\mathbf{p}_n))\geq-\Lambda^*(y)-\epsilon t_{[a,b]}+\frac{1}{n}\log(\Pr(\tilde{S}_n\in(y-\epsilon,y+\epsilon)))-o_p(1)
    \]
    where the $o_p(1)$ term does not depend on $y$ or $\epsilon$.
    
    Now it remains to show
    \begin{equation}\label{eqn:tilted-zero}
        \frac{1}{n}\log(\Pr(\tilde{S}_n\in(y-\epsilon,y+\epsilon)\mid\mathbf{p}_n))\pto0\,.
    \end{equation}
    To see this, we observe that $y=\E[\tilde{S}_n\mid\mathbf{p}_n]$ since
    \[
        \tilde{\Lambda}'(0;\mathbf{p}_n)=\Lambda'(t_{n,y};\mathbf{p}_n)=y
    \]
    where the second equality follows from the fact that $t_{n,y}$ maximizes $ty-\Lambda(t;\mathbf{p}_n)$. Now we apply the Chernoff bound to see that
    \[
        \Pr(\tilde{S}_n\geq y+\epsilon\mid\mathbf{p}_n)\leq\exp\left(-n\cdot\sup_{t\in\mathbb{R}}(t(y+\epsilon)-\tilde{\Lambda}(t;\mathbf{p}_n))\right)\,.
    \]
    Here,
    \[
        \begin{split}
            \sup_{t\in\mathbb{R}}(t(y+\epsilon)-\tilde{\Lambda}(t;\mathbf{p}_n)) &= \sup_{t\in\mathbb{R}}(t(y+\epsilon)-\Lambda(t+t_{n,y};\mathbf{p}_n)+\Lambda(t_{n,y};\mathbf{p}_n))\\
            &=\sup_{t\in\mathbb{R}}((t+t_{n,y})(y+\epsilon)-\Lambda(t+t_{n,y};\mathbf{p}_n))+\epsilon t_{n,y}-\Lambda^*(y;\mathbf{p}_n)\\
            &=\Lambda^*(y+\epsilon;\mathbf{p}_n)-\Lambda^*(y;\mathbf{p}_n)+\epsilon t_{n,y}\,.
        \end{split}
    \]
    Hence,
    \[
        \Pr(\tilde{S}_n\geq y+\epsilon\mid\mathbf{p}_n)\leq\exp\left(-n(\Lambda^*(y+\epsilon;\mathbf{p}_n)-\Lambda^*(y;\mathbf{p}_n)+\epsilon t_{n,y})\right)\,.
    \]
    Similar to the above argument, by Lemma~\ref{lem:cvx}, we have that
    \[
        \Lambda^*(y+\epsilon;\mathbf{p}_n)-\Lambda^*(y;\mathbf{p}_n)+\epsilon t_{n,y}=\Lambda^*(y+\epsilon)-\Lambda^*(y)+\epsilon t_{y}+o_p(1)
    \]
    where the $o_p(1)$ term does not depend on $y$ or $\epsilon$. Since $\Lambda^*(y+\epsilon)-\Lambda^*(y)+\epsilon t_{y}>0$, this implies
    \[
        \Pr(\tilde{S}_n\geq y+\epsilon\mid\mathbf{p}_n)=o_p(1)
    \]
    where the $o_p(1)$ term is still unaffected by $y$ or $\epsilon$. Similarly, we have
    \[
        \Pr(\tilde{S}_n\leq y-\epsilon\mid\mathbf{p}_n)=o_p(1)
    \]
    and this proves \eqref{eqn:tilted-zero} as desired.
\end{proof}

\subsection{Useful inequalities}

\begin{lemma}\label{lem:alg1}
    For $0\leq p\leq 1$ and $t\geq1$, we have
    \[
        1-e^{-pt}\leq1-(1-p)^t\leq pt\,.
    \]
\end{lemma}
\begin{proof}
    The left inequality is due to $1-p\leq e^{-p}$. To prove the right inequality, we can set $f(x)=(1-x)^t-1+tx$ and observe that $f(0)=0$ and $f'(x)=-t(1-x)^{t-1}+t\geq0$.
\end{proof}
\begin{lemma}\label{lem:alg2}
    Let $0\leq c\leq 1$. If $0\leq x\leq 2(1-c)$, then $1-e^{-x}\geq cx$.
\end{lemma}
\begin{proof}
    From $0\leq x\leq 2(1-c)$ we have $x-\frac{1}{2}x^2\geq cx$. Then $e^{-x}\leq 1-x+\frac{1}{2}x^2\leq 1-cx$.
\end{proof}

\section{Gibbs measures over trees near the critical temperature}\label{apdx:gibbs-critical}

Recall the Gibbs measure studied in Section~\ref{sec:gibbs}
\[
    \mu_{G_n,\beta}(T)=\frac{1}{Z_{G_n,\beta}}\exp\left(-\beta\log\log(n)\sum_{v\in\overline{V}_n}\mathsf{d}_T(1,v)\right)
\]
over the set of spanning trees of the component $\overline{G}_n$ of $G_n$ containing $1$ with $\overline{V}_n=V(\overline{G}_n)$. We analyzed this model and its phase transition phenomena when $\beta$ crosses the critical temperature $\beta_c=1-\Delta-\kappa^{-1}$. As discussed in Section~\ref{sec:phase-transition}, the model exhibits distinct behaviors depending on the limits of the proxies $\Delta_n\to\Delta\in[0,1]$, $\lambda_n\to\lambda\in[0,1]$, and $\kappa_n\to\kappa\in[0,\infty]$:
\begin{itemize}
    \item if $\Delta<1$ and $\kappa=\infty$, then the system goes through a \emph{discontinuous phase transition};
    \item if $\lambda=1$ (and thus $\Delta=0$) and $\kappa\in(1,\infty)$, then the system goes through a \emph{continuous phase transition}.
\end{itemize}
In terms of kernel size (cf. Section~\ref{sec:energy-landscape}), this means that the typical kernel size changes discontinuously in the former, and it changes continuously in the latter. In the discontinuous case, it is not immediately clear what happens when $\beta$ is extremely close to the critical temperature. The goal of this appendix is to better understand the picture near $\beta_c$.
We will focus on one particular regime
\begin{equation}\label{eqn:bii}
    \Delta_n\to\Delta=0\,,\qquad \lambda_n\to\lambda\in(0,1)
\end{equation}
which has the critical temperature $\beta_c=1$. This corresponds to \ref{item:regime2-2} in Corollary~\ref{cor:logz-phase}.
The behavior in the other regimes with first-order phase transitions is qualitatively the same. 
\subsection{The mean-field approximation}
In this section, we revisit the mean-field spin system \eqref{eqn:mf-cw} as a simplification of our model. 
To look at the near-critical behavior/``finite-size effects''  we slightly modify the definition of our model by introducing a sequence of inverse temperatures $\{\beta_n\}_{n\in\mathbb{Z}_+}$ satisfying
\begin{equation}\label{eqn:beta-lil}
    \liminf_{n\to\infty}\beta_n>0
\end{equation}
and define the Gibbs measure
\[
    \mu_{G_n,\beta_n}(T)=\frac{1}{Z_{G_n,\beta_n}}\exp\left(-\beta_n\log\log (n)\sum_{v\in\overline{V}_n}\mathsf{d}_T(1,v)\right)\,.
\]
Note that with the assumption \eqref{eqn:beta-lil}, our arguments in Section~\ref{sec:gibbs} hold for $\beta_n$ without modification. In particular, an approximate maximum value of the optimization
\begin{equation}\label{eqn:gibbs-opt-recap}
    \max_{1\leq m\leq N_{d^*}}\left\{\Psi(m;N_{d^*},\lambda_n)+(n-m)\log(mq)+\beta_nm\log\log n-\beta_n(d^*+1)n\log\log n\right\}
\end{equation}
gets the value of $\log Z_{G_n,\beta_n}$ correct up $o_p(n)$ error, by Theorem~\ref{thm:main-gibbs}. Here, $m$ represents the kernel size $|\varphi(T)|$ and an approximate optimizer corresponds to an \emph{optimal kernel size} discussed in Section~\ref{sec:energy-landscape}. Also, recall that
\[
    \Psi(m;N_{d^*},\lambda_n):=\E[\mathsf{LSE}_m(\log A_1,\cdots,\log A_{N_{d^*}})\mid N_{d^*}]\,,\quad A_1,\cdots,A_{N_{d^*}}\sim\ZTP(\log(1/\lambda_n))\,.
\]
Similar to our argument in Section~\ref{sec:hg-conc}, by the uniform concentration of $\mathsf{LSE}_m$ established in Lemma~\ref{lem:lse2}, the log partition function of the following mean-field spin system\footnote{This is slightly different from \eqref{eqn:mf-cw} in that $N_{d^*}$ --- the number of indices with nonzero value --- is also quenched (fixed). They are also asymptotically equivalent, but we deliberately chose this quenched version to have better interpretability in terms of the graph $G_n$ in our original system.} on $\mathbf{x}\in\{0,1\}^{N_{d^*}}\times\{0\}^{n-N_{d^*}}$ given $N_{d^*}$ and $\mathbf{A}=(A_1,\cdots,A_{N_{d^*}})$ leads to a variational optimization equivalent to \eqref{eqn:gibbs-opt-recap}:
\begin{equation}\label{eqn:cond-mf}
    \mu_{\mathbf{A},N_{d^*}}(\mathbf{x})=\frac{1}{Z_{\mathbf{A},N_{d^*}}}\exp\left(\sum_{i=1}^{N_{d^*}}(\log A_i-\log(nq)+\beta_n\log\log n)x_i+n\cdot g(\bar{\mathbf{x}})\right)
\end{equation}
where $\bar{\mathbf{x}}=\frac{1}{n}\sum_{i=1}^nx_i$ and $g(x)=(1-x)\log x$. Indeed, we have
\[
    \begin{split}
        Z_{\mathbf{A},N_{d^*}} &= \sum_{m=1}^{N_{d^*}}\sum_{\mathbf{x}:\bar{\mathbf{x}}=m/n}\exp\left(\sum_{i=1}^{N_{d^*}}(\log A_i-\log(nq)+\beta_n\log\log n)x_i+n\cdot g(m/n)\right)\\
        &=\sum_{m=1}^{N_{d^*}}\exp\left(\mathsf{LSE}_m(\log A_1,\cdots,\log A_{N_{d^*}})+(n-m)\log(mq)+\beta_nm\log\log n-n\log(nq)\right)
    \end{split}
\]
so by the uniform concentration of $\mathsf{LSE}_m$
\[
    \log Z_{G_n,\beta_n}=\log Z_{\mathbf{A},N_{d^*}}+n\log(nq)-\beta_n(d^*+1)n\log\log n+o_p(n)\,.
\]
In our new model \eqref{eqn:cond-mf}, $\mathbf{x}$ and $\mathbf{A}$ have the following interpretations. The first $N_{d^*}$ components $x_1,\cdots,x_{N_{d^*}}$ of $\mathbf{x}$ represent the vertices $v_1,\cdots,v_{N_{d^*}}$ in the set $\Gamma_{d^*}$ of vertices at depth $d^*$, and $x_i=1$ if and only if the corresponding vertex $v_i$ is in the kernel $\varphi(T)$. Thus, $\bar{\mathbf{x}}$ represents $|\varphi(T)|/n$, the ratio of the kernel size to $n$. Each $A_i$ represents the number of parent choices of $v_i$, namely, $\mathsf{par}_G(v_i)$.

\paragraph{Phase transitions revisited.} The phase transitions we analyzed in Section~\ref{sec:phase-transition} can also be intuitively seen from \eqref{eqn:cond-mf}. In the main regime of interest \eqref{eqn:bii}, the coefficient of $x_i$ can be approximated by $\log A_i-\log(nq)+\beta_n\log\log n=(-1+\beta_n+o_p(1))\log\log n$. Thus if $\liminf\beta_n>1$, then the coefficients are mostly positive and each $x_i$ roughly contributes $\Omega(\log\log n)$ to the log probability, so every $x_i$ wants to be at $1$. Otherwise, if $\limsup\beta_n<1$, then the coefficients of $x_i$ are negative, but the term $n\cdot g(\bar{\mathbf{x}})$ prevents too many of them from being zero, and the equilibrium is achieved at $\bar{x}=\Theta(1/\log\log n)$.

\subsection{Optimal kernel size at the critical temperature}

Now we describe the behavior of our system in terms of the kernel $\varphi(T)$ of $T$ for a typical tree under $\mu_{G_n,\beta_n}$ when $\beta_n\to1$. We further assume that
\begin{equation}\label{eqn:alpha-conv}
    \alpha_n\to\alpha\in\mathbb{R}_{+}
\end{equation}
and
\begin{equation}\label{eqn:beta-crit}
    \beta_n=1+\frac{b}{\log\log n}\,,\qquad b\in\mathbb{R}\,.
\end{equation}
Then the optimization \eqref{eqn:gibbs-opt-recap}, ignoring $o(n)$ terms, becomes
\[
    \max_{1\leq m\leq N_{d^*}}\{\Psi(m;N_{d^*},\lambda_n)+(n-m)\log(m/n)+(b-\log\alpha)m\}+n\log(nq)-\beta_n(d^*+1)n\log\log n
\]
and the corresponding mean-field model analogous to \eqref{eqn:cond-mf} can be written as
\begin{equation}\label{eqn:cond-mf-crit}
    \mu_{\mathbf{A},N_{d^*}}'(\mathbf{x})=\frac{1}{Z_{\mathbf{A},N_{d^*}}'}\exp\left(\sum_{i=1}^{N_{d^*}}\log(A_i)x_i+n\cdot g_{\alpha,b}(\bar{\mathbf{x}})\right)\,,\qquad\mathbf{x}\in\{0,1\}^{N_{d^*}}\times\{0\}^{n-N_{d^*}}
\end{equation}
where $g_{\alpha,b}(x)=(b-\log\alpha)x+(1-x)\log x$. In the regime \eqref{eqn:bii}, $\log Z_{\mathbf{A},N_{d^*}}'$ can be accurately described with our tools in large deviations (Section~\ref{sec:large-deviations}).

\begin{proposition}\label{prop:optkernel-crit}
    We have
    \begin{equation}\label{eqn:logz-crit}
        \frac{1}{n}\log Z_{\mathbf{A},N_{d^*}}'\pto\sup_{0<y\leq1}\left\{g_{\alpha,b}((1-\lambda)y)-I(y)\right\}+\E[\log(A_i+1)]
    \end{equation}
    where
    \[
        I(y):=\sup_{t\in\mathbb{R}}\left\{ty-\E\left[\log\left(\frac{A_ie^t+1}{A_i+1}\right)\right]\right\}\,.
    \]
\end{proposition}
\begin{proof}
For a constant $K\in\mathbb{Z}_+$, consider the intervals $\mathcal{I}_k=((k-1)/K, k/K]$ and let
\[
    \mathcal{B}_k=\left\{\mathbf{x}:\frac{1}{N_{d^*}}\sum_{i=1}^{N_{d^*}}x_i\in\mathcal{I}_k\right\}\,.
\]
We have
\[
    \begin{split}
        Z_{\mathbf{A},N_{d^*}}'&=\sum_{k=1}^K\sum_{\mathbf{x}\in\mathcal{B}_k}\exp\left(\sum_{i=1}^{N_{d^*}}\log(A_i)x_i+n\cdot g_{\alpha,b}(\bar{\mathbf{x}})\right)\\
        &\leq K\cdot\max_{1\leq k\leq K}\left(\exp\left(n\cdot\sup_{\mathbf{x}\in\mathcal{B}_k} g_{\alpha,b}(\bar{\mathbf{x}})\right)\sum_{\mathbf{x}\in\mathcal{B}_k}\exp\left(\sum_{i=1}^{N_{d^*}}\log(A_i)x_i\right)\right)
    \end{split}
\]
so
\[
    \log Z_{\mathbf{A},N_{d^*}}'\leq\log K+\max_{1\leq k\leq K}\left\{n\cdot\sup_{\mathbf{x}\in\mathcal{B}_k}g_{\alpha,b}(\bar{\mathbf{x}})+\log\left(\sum_{\mathbf{x}\in\mathcal{B}_k}\exp\left(\sum_{i=1}^{N_{d^*}}\log(A_i)x_i\right)\right)\right\}\,.
\]
Similarly, we have a lower bound
\[
    \log Z_{\mathbf{A},N_{d^*}}'\geq\max_{1\leq k\leq K}\left\{n\cdot\inf_{\mathbf{x}\in\mathcal{B}_k}g_{\alpha,b}(\bar{\mathbf{x}})+\log\left(\sum_{\mathbf{x}\in\mathcal{B}_k}\exp\left(\sum_{i=1}^{N_{d^*}}\log(A_i)x_i\right)\right)\right\}\,.
\]
Applying Proposition~\ref{prop:ldp} to $\Bernoulli(\frac{A_i}{A_i+1})$ random variables, we have
\[
    \begin{split}
        \frac{1}{n}\log\left(\sum_{\mathbf{x}\in\mathcal{B}_k}\exp\left(\sum_{i=1}^{N_{d^*}}\log(A_i)x_i\right)\right)&=\frac{1}{n}\sum_{i=1}^{N_{d^*}}\log(A_i+1)-\inf_{y\in\mathcal{I}_k} I(y)+o_K(n)\\
        &=\E[\log(A_i+1)]-\inf_{y\in\mathcal{I}_k}I(y)+o_K(n)\,.
    \end{split}
\]
Using the fact that $g_{\alpha,b}$ is continuous and $N_{d^*}=(1-\lambda)n+o_p(n)$ by Theorem~\ref{thm:main1-tight}, we arrive at the conclusion by taking $K$ large enough.
\end{proof}

\begin{corollary}\label{cor:optkernel-crit}
    The optimization in the RHS of \eqref{eqn:logz-crit} has a unique maximizer $y^*$ and for any constant $\epsilon>0$ we have
    \[
        \mu_{\mathbf{A},N_{d^*}}'(|\bar{\mathbf{x}}-(1-\lambda)y^*|>\epsilon)\pto0\,.
    \]
\end{corollary}
\begin{proof}
    The function $g_{\alpha,b}((1-\lambda)y)-I(y)$ is strictly concave in $y$ and tends to $-\infty$ as $y\to0$, so it has a unique maximizer in $(0,1]$. The uniqueness in turn implies that the log probability is smaller by $\Omega(n)$ outside a neighborhood of $y^*$, yielding the second conclusion.
\end{proof}

Recall from the previous section that $\mathbf{x}$ in the mean-field model \eqref{eqn:cond-mf-crit} corresponds to the proportion of the kernel $|\varphi(T)|/n$. Hence, Corollary~\ref{cor:optkernel-crit} in fact implies that if we independently sample $T_n$ from $\mu_{G_n,\beta_n}$ for each $n$, then the kernel proportion concentrates around a nontrivial value:
\[
    \frac{|\varphi(T_n)|}{n}\pto (1-\lambda)y^*\,.
\]
Note that $y^*$ depends on $\alpha$, $b$, and $\lambda$. This is to some extent complementary to Lemma~\ref{lem:gibbs-low-apx-unif} and Lemma~\ref{lem:gibbs-small-kernel} where it was proved (in a slightly stronger form) that the kernel proportion concentrates around $1-\lambda$ and $0$ in the low and high temperature regimes, respectively. We also note that \eqref{eqn:cond-mf-crit}, being the exact mean-field model, explains the behavior of each vertex at depth $d_n^*$ in the following way. Each vertex $v$ has $p_v$ chance of being included in the kernel, and those probabilities can approximately be described as the maximizer of
\[
    \sum_{v\in\Gamma_{d^*}}(\log|\mathsf{par}_G(v)|+H(p_v))+n\cdot g_{\alpha,b}\left(\frac{1}{n}\sum_{v\in\Gamma_{d^*}}p_v\right)
\]
and Corollary~\ref{cor:optkernel-crit} implies that $\frac{1}{n}\sum_vp_v$ can be approximated by $(1-\lambda)y^*$.

\begin{remark}[Accuracy of naive mean-field approximation]
    In the low and high temperature regimes we studied in Section~\ref{sec:gibbs}, a naive mean-field approximation is accurate in the sense that we can find a vertex-wise product measure that approximates the Gibbs measure within $o_p(n)$ KL divergence. At the critical temperature, the ``mean-field approximation'' is at least accurate in a looser sense that we have a mean-field model \eqref{eqn:cond-mf-crit} that has asymptotically equal log partition function up to $o_p(n)$ error. This does not directly give a product measure approximation to the origin model on trees, because  if each vertex is independently included in the kernel as we described above, then there is a small chance the kernel is invalid (does not correspond to a valid tree). 
\end{remark}

\subsection{Internal energy density as a function of temperature}

\begin{figure}
    \centering

    \begin{tikzpicture}
  \begin{axis}[
    axis lines = middle,
    xmin = 0.2, xmax = 2.4,
    ymin = 0, ymax = 1.2,
    xlabel = {$\beta_n$},
    ylabel = {$\mathcal{E}_{G_n,\beta_n}$},
    samples = 200,
    domain = 0:4,
    ytick = {0,1},
    yticklabels = {0,$1-\lambda$},
    xtick = {1},
    smooth
  ]
    \addplot[thick] {(1-tanh(20*(x-1.02)))/2};
    \addplot[dashed] {1};


    \draw[dashed] (axis cs:1,-0.2) -- (axis cs:1,1);
  \end{axis}
\end{tikzpicture}

    \caption{A plot of the internal energy density $\mathcal{E}_{G_n,\beta_n}$ against the inverse temperature $\beta_n$, for a fixed $n$ (i.e., non-asymptotic). The slope near the critical temperature $1$ is at the order of $-\log\log n$.}
    \label{fig:energy-density}
\end{figure}
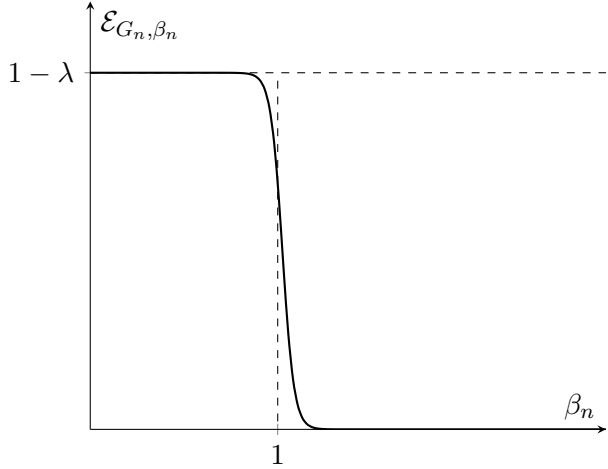

We conclude this appendix with a non-asymptotic landscape of the (relative) \emph{internal energy density} which we define by
\[
    \mathcal{E}_{G_n,\beta_n}:=\frac{1}{n}\E_{\mu_{G_n,\beta_n}}\left[\sum_{v}(\mathsf{d}_T(1,v)-\mathsf{d}_G(1,v))\right]\,.
\]
Theorem~\ref{thm:ground-state-energy} and Proposition~\ref{prop:cond-gibbs-dv} are enough to analyze $\mathcal{E}_{G_n,\beta_n}$ if the kernel proportion $|\varphi(T)|/n$ concentrates. Specifically, if $|\varphi(T)|/n\pto r$, then Proposition~\ref{prop:cond-gibbs-dv} gives
\[
    \sum_{v}\mathsf{d}_T(1,v)=(d^*+1-r)n+o_p(n)
\]
with $\mu_{G_n,\beta_n}$-exponentially high probability, which implies
\[
    \mathcal{E}_{G_n,\beta_n}\pto1-\lambda-r\,.
\]
Along with the results in Section~\ref{sec:gibbs2} (e.g., Lemma~\ref{lem:gibbs-low-apx-unif} and Lemma~\ref{lem:gibbs-small-kernel}), we have the following behavior of $\mathcal{E}_{G_n,\beta_n}$ depending on the temperature of the system.
\begin{itemize}
    \item In the low temperature phase $\liminf\beta_n>1$, the kernel is almost the entire set $\Gamma_{d^*}$ most of the time, so $|\varphi(T)|/n\pto1-\lambda$. This implies $\mathcal{E}_{G_n,\beta_n}\pto0$, i.e., the system is in the ground state.

    \item Near the critical temperature, the limiting kernel proportion depends on the asymptotics of $\alpha_n$ and $\beta_n$ as in \eqref{eqn:alpha-conv} and \eqref{eqn:beta-crit}. For a fixed $\alpha$, as $b$ from \eqref{eqn:beta-crit} moves from $\infty$ to $-\infty$, the limiting kernel proportion $(1-\lambda)y^*$ jumps smoothly from $0$ to $1-\lambda$. As a consequence, the internal energy density $\mathcal{E}_{G_n,\beta_n}$ smoothly increases from $0$ to $1-\lambda$.

    \item In the high temperature phase $\limsup\beta_n<1$, the kernel proportion vanishes, i.e., $|\varphi(T)|/n\pto0$. This gives $\mathcal{E}_{G_n,\beta_n}\pto1-\lambda$.
\end{itemize}
This is summarized in Figure~\ref{fig:energy-density}. The change in internal energy at the phase transition corresponds to what is called the ``latent heat'' for a first-order phase transition (see, e.g., \cite{chaikin1995principles}).

\section{Further discussion}\label{apdx:discussion}
\subsection{Additional analogies: geodesics and extension complexity}
\paragraph{A motivating analogy to geodesics.} In order to better understand the distinction between the OGP in the space of paths vs trees, it is interesting to think about the shortest path problem on other metric spaces, such as Riemannian manifolds. In the Riemannian case, a path between two points $a$ and $b$ is a \emph{geodesic} if it is a critical point of the length functional (i.e., the first variation vanishes) and a shortest path, or minimal geodesic, is a global minimum of the length functional \cite{burago2001course}. Unlike in Euclidean space, on general manifolds the length functional often has multiple geodesics --- or example, on the sphere the geodesics are arcs of great circles, and for generic points $a$ and $b$ there will be two geodesics, with only one of them being the shortest path. 

A natural optimization procedure in the space of paths would be the \emph{curve-shortening flow} (with fixed end points), which would be the gradient flow of the length functional \cite{grayson1989shortening,allen2012dirichlet}. Clearly, for appropriate initialization the curve shortening flow can converge to any geodesic, so it may fail to find the shortest path between two points. It is also straightforward to define analogues of geodesics in general metric spaces \cite{burago2001course} such as graphs, and we can similarly consider a generalization of the curve shortening flow by iteratively modifying a path via local search --- e.g., in the case of graphs, by iteratively moving to the shortest path constrained to a small Hamming ball around the current iterate. For the same reasons as in the continuous case, this procedure can easily converge to a local instead of global optimum. Intuitively, we can view the overlap-gap result for shortest paths \cite{LS2024} as showing that in random graphs, the optimization landscape of paths is particularly adversarial and local search procedures have no hope of succeeding. On the other hand, if we change to a different optimization landscape, where we are allowed to store more information than just a single candidate path, finding shortest path trees is computationally tractable via ``local'' procedures (in a sense made precise by our results).  
\paragraph{Related analogy with extension complexity.} 
As we just discussed, while the result of \cite{LS2024} shows that shortest paths in a random graph exhibits overlap gap, our results show that shortest path \emph{trees}, which encode more information, are nicer. This has a rough analogy with extended formulations. 
As a reminder, some high-complexity polytopes in low dimensions are actually \emph{projections} of much ``simpler'' polytopes in a higher dimensional space \cite{conforti2010extended}, which has important consequences for, e.g., the tractability of convex optimization over these sets. One notable example is the case of the \emph{permutahedron}, the convex hull of $\{(\pi(1),\ldots,\pi(n)) : \pi \in S_n\}$ where $S_n$ is the set of permutations of $[n]$, has $2^n - 2$ many facets. However, 
it can be written as a projection of the set of doubly stochastic $n \times n$ matrices (Birkhoff polytope), which has only $O(n^2)$ facets, or as the projection of a different polytope with $O(n\log n)$ many facets \cite{goemans2015smallest}.

The map from shortest path trees to shortest path is induced by a type of projection map (see Definition~\ref{def:path-uv} below). So as in the case of extended formulations, our result shows that a poorly behaved object (a shortest path in a random graph) can be recovered as the projection of a well-behaved one (a shortest path tree in a random graph).
The field of \emph{extension complexity} includes techniques dedicated to establishing the \emph{non-existence} of small extended formulations of polytopes --- see e.g. \cite{rothvoss2017matching,goos2018extension,braverman2013information,lee2015lower}.
It would be interesting if any related ideas could be applied in the context of overlap gaps and/or MCMC. 

\subsection{The physical interpretation of Dijkstra's algorithm}\label{apdx:dijkstra-continuous}
The content of this section is mostly classical, see textbooks such as \cite{bazaraa2011linear,ahuja1993network}. See also Appendix~\ref{sec:dca} for some related content.

\paragraph{Physical ``derivation'' of Dijkstra's algorithm.} Dijkstra's algorithm is often motivated or explained through the following well-known physical interpretation\footnote{There is also a related physical interpretation of Dijkstra's algorithm in terms of a wave propagating through a system of pipes of different lengths (see Chapter 4.4 of \cite{kleinberg2006algorithm}).}. We imagine the graph as a collection of physical nodes of equal weights, and an edge $(v,w)$ is interpreted as an (infinitely strong) physical string of length $\mathsf{d}_G(v,w)$ between the vertices. Then we consider holding fixing the source vertex $s$ to some fixed height, and then letting all of the other vertices fall down as they are pulled by the force of gravity. If the physical width of the string and nodes are neglected/infinitesimally small, then at equilibrium the \emph{drop in height} between vertex $s$ and vertex $v$ will be exactly the distance $\mathsf{d}_G(s,v)$. In general, the set of strings experiencing tension at equilibrium will contain shortest paths to from the source to every vertex, and if the loads are balanced symmetrically then the set of strings in tension will be exactly the shortest path DAG.

\paragraph{Linear programming, network flows, and duality.}
Modeling the above physical interpretation of Dijkstra's algorithm naturally leads to a network flow/linear programming\footnote{As written, these are convex programs --- they are not literally linear programs since we use the absolute value function --- but they can straightforwardly be transformed into linear programs via a standard procedure.} 
interpretation.
Here the variable $h_v$ corresponds to the \emph{negative} height of the node:
\begin{align}
    \MoveEqLeft \max \sum_v h_v \label{eqn:lp-heights} \\
    \text{s.t.}\quad  &h_s = 0 &\\
    &|h_v - h_w| \le \mathsf{d}_G(v,w) & \text{for all $v \sim w$} \\
    &h_v \in \mathbb R & \text{for all $v$}
\end{align}
The objective of maximizing the total negative height exactly corresponds to minimizing the \emph{potential energy} of the collection of nodes in the physical analogy, which is exactly where the nodes will come to rest if their positions evolve according to the laws of classical mechanics. (See, e.g.,  \cite{goldstein:mechanics}.)

We can also naturally interpret the above linear program as the \emph{dual} of a minimum-cost flow problem. Recall that a shortest path between vertices $s$ and $t$ can be interpreted as an integral \emph{minimum-cost} flow where the cost of pushing one unit of flow through an edge is proportional to the distance traveled, and edges have infinite capacities. Similarly, we can reinterpret the problem of computing shortest paths from vertex $s$ to all of the other vertices in the graph minimizing the cost of sending a unit of flow from source $s$ to every other vertex $v$ in the graph.
This corresponds to the following program:
\begin{align}
\MoveEqLeft \min \sum_{v \sim w} |f(v,w)| \mathsf{d}_G(v,w) \label{eqn:lp-flow} \\
\text{s.t.}\quad & 
f(v,w) = -f(w,v) & \forall v \sim w \\
&\sum_{w : v \sim w} f(w,v) = 1 & \forall v \ne s \label{eqn:flow-conservation} \\
&f(v,w) \in \mathbb R & \forall v \sim w
\end{align}
where the second constraint is the conservation of flow equation (recalling that one unit of flow will flow out at each vertex $v \ne s$). This flow problem is a special case of the uncapacitated transportation problem and is somewhat related to another physical system: differential equation models of slime mold algorithms for shortest paths \cite{bonifaci2012physarum,bonifaci2013physarum}.

The two linear programs are naturally related by convex duality: the flows $f(v,w)$ in the latter program are simply the dual variables to the constraints in the former, and similarly the heights $h_v$ are the dual variables (potentials) for the conservation of flow equations. Concretely, weak duality is the  following consequence of Holder's inequality, which holds for any feasible points of the LPs:
\begin{equation}\label{eqn:weak-duality}
\sum_v h_v = \sum_{v} h_v \sum_w f(w,v) = \frac{1}{2} \sum_{v, w} f(w,v)(h_v - h_w) \le \sum_{v \sim w} |f(v,w)| \mathsf{d}_G(v,w).
\end{equation}
Strong duality corresponds to the fact that at a pair of optimizing solutions, this becomes an equality.

This duality relationship can also be naturally explained in terms of the physical interpretation. The flow $f(v,w)$ corresponds to the net force in the up direction exerted upon vertex $w$ by vertex $v$ (via the tension of the connecting string), and the flow conservation equation \eqref{eqn:flow-conservation} says that at equilibrium, the total upward force exhibit by tension must exactly cancel the downward pull of gravity (so that the net force is zero). 

In terms of the two linear programs, we can understand Dijkstra's algorithm as placing the nodes at their correct heights in a downward scan starting from the origin, which determines the distances from the root, and tracking the saturated constraints in this process corresponds to determining the the union of the supports of all minimum cost flows, or equivalently maintaining the shortest path DAG. This roughly matches the  dynamics one would get by taking the physical analogue, setting all of the vertices to the same initial height, and then letting them fall by the force of gravity.



\subsection{Dijkstra's algorithm and belief propagation}\label{apdx:dijkstra-is-bp}
There are many connections between different versions of belief propagation and dynamic programming methods like Viterbi/min-sum decoding in the context of combinatorial optimization. For example, a simple variant of belief propagation can be used to for computing the maximum a posteriori estimate in various settings --- see Chapter 14.3 of \cite{mezard2009information}. Below, we highlight a direct connection in the special case we are interested in, which is computing shortest paths. See \cite{gamarnik2012belief} for a related work in the context of general max-flow problems as well as for more references into the literature. 

We have already argued that \eqref{eqn:finite-temp} is a natural finite temperature analogue of the shortest path problem. In what follows, we show that starting with this distribution, we can naturally derive Dijkstra's algorithm from belief propagation. Note that for this section, the argument holds for a \emph{general} graph $G$ and does not rely upon properties of random graphs. The key to the derivation is the previously described reinterpretation of the random tree $T$ as a Markov random field (MRF, a.k.a. undirected graphical model or factor model \cite{mezard2009information,wainwright2008graphical,lauritzen1996graphical}) on the graph. 

\paragraph{Belief propagation on the Markov Random Field.}
 From the factor model, we have the following (loopy) Belief Propagation update (see, e.g., \cite{mezard2009information} for the general BP formula): letting $1$ be the root node, for every $u \ne 1$ we have
\begin{equation} 
\nu_{u \to v}^{(t)} = \mathbb{E}^{(t - 1)}_{u \to v}\left[\frac{\sum_{w : w \in N(u) \setminus \{v\}} e^{-\overline\beta(D_w + \ell_{wu})} \delta_{D_w + \ell_{wu}}}{\sum_{w : w \in N(u) \setminus \{v\}} e^{-\overline \beta(D_w + \ell_{wu})}} \right]\label{eqn:finite-temp-bp-update}
\end{equation}
where $\delta_a$ is the Dirac delta measure at $a$,  $N(u)$ denotes the neighborhood of vertex $u$, and $ \mathbb{E}^{(t - 1)}_{u \to v}$ denotes the expectation where independently $D_w \sim \nu_{w \to u}^{(t - 1)}$ for each $w \in N(u) \setminus \{v\}$. 
and for $u = 1$ we have
\[ \nu_{1 \to v}^{(0)}(0) = 1 \]
since $d(1,1) = 0$.
For the initial messages, we can take $\mu_{u \to v}^{(0)}(\infty) = 1$ for all $u \ne 1$ and all neighbors $v$ of $u$. Loopy belief propagation then corresponds to running \eqref{eqn:finite-temp-bp-update} for $t$ steps.

\paragraph{BP converges to Dijkstra in the limit.} The dynamics of loopy belief propagation are nontrivial to analyze. However, if we consider the zero-temperature limit where $\overline \beta \to \infty$, then softmax converges to max so the belief propagation converges to
\[ \eta_{u \to v}^{(t)} =  \mathbb{E}^{(t - 1)}_{u \to v}\left[\delta_{\min_{w \in N(u) \setminus \{v\}} (D_w + \ell_{wu})}\right] \]
where $ \mathbb{E}^{(t - 1)}_{u \to v}$ now denotes the expectation where independently $D_w \sim \eta_{w \to u}^{(t - 1)}$. By applying induction on time $t$, we see that this algorithm only sends delta measures as messages; furthermore, if we track the propagation of messages outward from the root, we see that the message passing is exactly implementing Dijkstra's algorithm. 
\begin{remark}
The belief propagation algorithm above is conceptually related to the success of the ``mean-field approximation'' for the free energy which shows up in Section~\ref{sec:gibbs}. More specifically, the fixed points of belief propagation correspond to the critical points of a related free energy functional called the Bethe free energy \cite{mezard2009information,yedidia2003understanding}, and in the reverse direction the critical points of the na\"ive mean-field free energy correspond to fixed points of a similar message-passing algorithm \cite{wainwright2008graphical}. Because the random graphs we consider have $\omega(1)$ degree and (informally speaking) our system does not exhibit ``frustration'', it is reasonable to guess in this setting that the Bethe and na\"ive mean-field free energy behave very similarly (see, e.g., \cite{yedidia2003understanding,ruozzi2012bethe,koehler2019fast} for some related context). Furthermore, essentially the same derivation as above can be used to show that the na\"ive message-passing algorithm related to mean-field approximation (see, e.g., \cite{wainwright2008graphical}) also converges to Dijkstra.
\end{remark}

\subsection{Instability of the projection map from shortest path trees to shortest paths}
Given that shortest path trees are ``stable'' objects which change in a predictable way under resampling of
the underlying graph, one might consider trying to compute a shortest path from $u$ to $v$ by first computing a shortest path tree rooted at $u$, and then outputting the root-to-leaf path from $u$ to $v$. 
Nevertheless, this procedure cannot be a stable way to construct a shortest path because the OGP rules any such procedure out \cite{LS2024}. 
The reason that it fails to be stable is the following: the projection map $\text{path}_{u,v}$ which takes a shortest path tree from vertex $u$ to the shortest path from $u$ to a vertex $v$ \emph{fails badly to be Lipschitz} in the Hamming metric. 
\begin{definition}[tree-to-path projection map]\label{def:path-uv}
Given any two vertices $u,v$, we can define the map $\text{path}_{u,v}$ which takes as input a tree $T$ containing $u$ and $v$ and outputs the unique shortest path $P$ between $u$ and $v$ in $T$. Note that $\text{path}_{u,v} \circ \text{path}_{u,v} = \text{path}_{u,v}$ so this map is a projector (a.k.a. idempotent \cite[Table of Terminology]{mac2013categories}) in the usual sense. It can be extended to a linear map from the free vector space generated by trees containing $u$ and $v$ onto the subspace generated by paths, in which case it would be a projection in the usual sense for linear operators.
\end{definition}
\begin{example}\label{example:cycle}
For an example in a non-random graph, consider a cycle of length $n$ with $n$ even and label the vertices by elements of $\mathbb{Z}/n\mathbb{Z}$. Then there are two shortest path trees from the vertex $0$: one (call it $T_1$) which is a path formed by removing the edge from $n/2$ to $n/2 + 1$, and a symmetrical tree where the path is given by removing the edge from $n/2 - 1$ to $n/2$. While these trees differ only in $2$ out of $n$ many edges, the corresponding shortest paths from $0$ to $n/2$ are \emph{disjoint!}
\end{example}
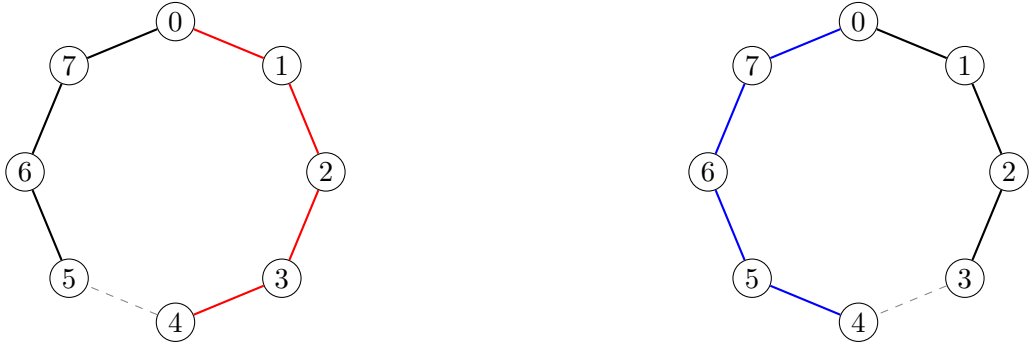
\begin{figure}
  \centering
  \begin{subfigure}{0.45\textwidth}
    \centering
    \begin{tikzpicture}[every node/.style={circle, draw, fill=white, inner sep=2pt}]
      \foreach \i in {0,...,7} {
        \node (n\i) at ({90 - \i*45}:2cm) {$\i$};
      }
      \foreach \i in {0,...,7} {
        \pgfmathtruncatemacro{\j}{mod(\i+1,8)}
        \draw[gray, ultra thin, dashed] (n\i) -- (n\j);
      }
      \draw[red, thick] (n0) -- (n1) -- (n2) -- (n3) -- (n4);
      \draw[black, thick] (n0) -- (n7) -- (n6) -- (n5);
    \end{tikzpicture}
    \caption{Tree with shortest path \(0\rightarrow1\rightarrow2\rightarrow3\rightarrow4\).}
  \end{subfigure}
  \hfill
  \begin{subfigure}{0.45\textwidth}
    \centering
    \begin{tikzpicture}[every node/.style={circle, draw, fill=white, inner sep=2pt}]
      \foreach \i in {0,...,7} {
        \node (n\i) at ({90 - \i*45}:2cm) {$\i$};
      }
      \foreach \i in {0,...,7} {
        \pgfmathtruncatemacro{\j}{mod(\i+1,8)}
        \draw[gray, ultra thin, dashed] (n\i) -- (n\j);
      }
      \draw[blue, thick] (n0) -- (n7) -- (n6) -- (n5) -- (n4);
      \draw[black, thick] (n0) -- (n1) -- (n2) -- (n3);
    \end{tikzpicture}
    \caption{Tree with shortest path \(0\rightarrow7\rightarrow6\rightarrow5\rightarrow4\).}
  \end{subfigure}
  \caption{Illustration of Example~\ref{example:cycle}. The two shortest path trees emanating from vertex 0 are very similar, but the shortest paths from $0 \to 4$ are disjoint. }
\end{figure}
Our results directly imply that the projection is also not average-case stable in the setting of correlated random graphs. This is because we prove the shortest path tree changes in a smooth and predictable way whereas the prior work \cite{LS2024} showed that its projection, the shortest path, will typically change drastically as we resample edges of the graph (see Section~\ref{sec:unstable} as well).


\begin{example}[Uniformly random shortest path tree is not worst-case stable]
It is natural to wonder if  sampling a uniformly random shortest path tree, or computing the shortest path DAG, could also be stable in the worst case (i.e. Lipschitz). This is not the case. For example, onsider a $d$-ary tree $T$ (e.g., with $d = 2$) and depth $\Theta(\log n)$, and form a graph $G$ by identifying the root vertex $v$ and the leaves of $T$ with those of an identical tree $T'$ (so the ``interior'' vertices of $T$ and $T'$ are still disjoint). The initial shortest path DAG is the entire graph $G$, oriented away from the root $v$ and towards the leaves. If we add an edge from $v$ to any of the vertices at distance $2$ from it, the resulting shortcuts leads to a constant proportion of the edges of the graph (those from the ``descendants'' of the symmetrical copy of the vertex in the tree) being removed from the shortest path DAG, and also leads to a similar  change in the uniform distribution over spanning trees --- the set of edges which are lost from the DAG go from appearing with probability $1/(2d)$ to $0$, and their symmetrical pairs go from probability $1/(2d)$ to $1/d$, so the mean of the shortest path tree distribution changes by $\Theta(n)$ in $\ell_1$ distance. 

It also follows from our analysis that in Erdos--R\'enyi graphs, a straightforward analogue of this construction leads to similar unstable (adding a shortcut from the original root vertex to a vertex at distance $2$ from $v$ will significantly change the shortest path DAG and uniformly random shortest path tree). This shows that if insertions/deletions are \emph{worst-case}, then even in relatively well connected graphs the shortest path DAG and uniformly random shortest path tree no longer behave in a Lipschitz manner.
\end{example}

\subsection{Sampling a Gibbs measure is equivalent to a random optimization problem}\label{apdx:gumbel}
Here we recall the classical connection between sampling Gibbs measures and random optimization problems with Gumbel noise \cite[Chapter 3]{train2009discrete}.
The contents of this section are not needed in the proof of our main results, but are helpful for understanding some of the parallels between the OGP and Gibbs measure frameworks.

Let \(\mathcal{X}\) be a finite set and \(f:\mathcal{X}\to\R\) a fixed function.  Fix an inverse temperature \(\beta>0\).  The associated \emph{Gibbs distribution} on \(\mathcal{X}\) is
\[
p_\beta(x)
\;=\;
\frac{\exp\bigl(\beta\,f(x)\bigr)}
     {\displaystyle\sum_{y\in\mathcal{X}}\exp\bigl(\beta\,f(y)\bigr)}
\;=\;
\frac{\exp\bigl(\beta\,f(x)\bigr)}{Z_\beta}\,,
\]
where \(Z_\beta=\sum_{y\in\mathcal{X}}\exp(\beta f(y))\) is the normalizing constant.

\paragraph{Gumbel distribution.}
A random variable \(G\) is said to have the \emph{standard Gumbel distribution} if its cumulative distribution function is
\[
F_G(g)
\;=\;
\Pr[G\le g]
\;=\;
\exp\!\bigl(-e^{-g}\bigr),
\qquad g\in\R.
\]
Equivalently, if \(U\sim\mathrm{Uniform}(0,1)\) then
\[
G \;=\; -\log\bigl(-\log U\bigr)
\]
is standard Gumbel (with mean \(\gamma\approx0.5772\), the Euler–Mascheroni constant).

\begin{proposition}[Gumbel--Max trick]
Let \(\{G(x)\}_{x\in\mathcal{X}}\) be i.i.d.\ standard Gumbel random variables.  Then
\[
\Pr\Bigl(\arg\max_{x\in\mathcal{X}}\bigl\{\beta\,f(x)+G(x)\bigr\}=x\Bigr)
\;=\;
\frac{\exp\bigl(\beta\,f(x)\bigr)}
     {\displaystyle\sum_{y\in\mathcal{X}}\exp\bigl(\beta\,f(y)\bigr)}
\;=\;
p_\beta(x).
\]
In other words, sampling from \(p_\beta\) is equivalent to solving the random optimization
\[
\hat x \;=\;\arg\max_{x\in\mathcal{X}}\{\;\beta\,f(x)\;+\;G(x)\;\}\,.
\]
\end{proposition}

\begin{proof}[Proof sketch.]
By independence,
\[
\Pr\bigl(\arg\max_x[\beta f(x)+G(x)]=x\bigr)
=\int_{-\infty}^\infty 
\Pr\bigl(G(x)\in dt\bigr)\prod_{y\neq x}\Pr\bigl(G(y)\le t+\beta[f(x)-f(y)]\bigr).
\]
Writing each CDF as \(\exp(-e^{-x})\) and simplifying yields exactly
\(\exp(\beta f(x))/\sum_y\exp(\beta f(y))\).
\end{proof}
\begin{proposition}[Log partition function via Gumbel max]
With the same setup, one has
\[
\E\!\Bigl[\max_{x\in\mathcal{X}}\{\beta f(x) + G(x)\}\Bigr]
= \log\!\Bigl(\sum_{y\in\mathcal{X}}e^{\beta f(y)}\Bigr)
\;+\;\gamma
= \log Z_\beta \;+\;\gamma,
\]
where \(\gamma\approx0.5772\) is the Euler–Mascheroni constant.  
\end{proposition}

\begin{proof}[Proof sketch.]
Compute the law of \(M=\max_x(\beta f(x)+G(x))\) via its CDF 
\[ \Pr(M\le t)=\prod_x\exp(-e^{-(t-\beta f(x))}) \]
to show that it is a Gumbel distribution with a correspondingly shifted mean. 
\end{proof}

\subsection{Free energy barriers trap Markov chains}\label{apdx:bottleneck}
It is well known that ``free energy barriers'' as described by the (quenched) Franz--Parisi potential
are rigorous obstructions to the Glauber dynamics and related Markov chains. See Theorem 13.7 of the textbook \cite{mezard2009information} as well as, e.g., \cite{coja2015independent,bandeira2022franz,arous2023free,bandeira2023free}. 
As discussed in the references, this is a direct implication of classical bottleneck theory for MCMC, but for expositional purposes we recall some of the details here.

Recall that for a Markov chain with transition kernel $P$ and stationary measure $\pi$ the bottleneck ratio of a set $S$ is
\[ \Phi(S) = \frac{\sum_{x \in S, y \in S^C}\pi(x) P(x,y)}{\pi(S)}. \]
\begin{lemma}[Proof of Theorem 7.4 of \cite{lpw}]
If $X_0 \sim \pi$ and $X_0,X_1,\ldots$ is a Markov chain evolving according to $P$, then
\[ \Pr(\exists s \le t, X_s \in S^C \mid X_0 \in S) \le t \Phi(s). \]
\end{lemma}
\begin{lemma}
For a reversible Markov chain, for any subset $S$ of the state space and
letting $\partial S = \{ y : \exists x \in S, P(x,y) > 0 \}$,
we have
\[ \Phi(S) \le \frac{\pi(\partial S \cap S^C)}{\pi(S)} \]
\end{lemma}
\begin{proof}
By reversibilty, we have
\[ \Phi(S) =  \frac{\sum_{x \in S, y \in S^C}\pi(x) P(x,y)}{\pi(S)} = \frac{\sum_{x \in S, y \in \partial S \cap S^C}\pi(x) P(x,y)}{\pi(S)} =  \frac{\sum_{x \in S, y \in \partial S \cap S^C}\pi(y) P(y,x)}{\pi(S)} \le \frac{\pi(\partial S \cap S^C)}{\pi(S)} \]
where the last inequality follows since $\sum_x P(y,x) \le 1$ for any $y$.
\end{proof}
If the Franz-Parisi potential fails to be \emph{quasi-convex}, i.e. if the potential has multiple disconnected local minima, then the bottleneck theory tells us that Markov chains which are local and reversible will get trapped from suitable initializations.

To illustrate the basic idea, we explicitly write out statements for a special case which often arises.
If the Franz--Parisi potential (or its nonasymptotic analogue about point $x$) is not monotonically increasing in the overlap at some point $r \in (r^*,1]$, where $r^*$ is the location of the global minimum, then for $\tau = rm$ the above theorem directly implies metastability of states with inner product at least $\tau$ is metastable, i.e. it takes a long time for the dynamics to escape this set. Symmetrical results and arguments cover the case where the potential is not monotonically decreasing from $0$ to $r^*$.

\begin{theorem}[Example metastability from FPP, Homogeneous case]
Let $P$ be a reversible Markov chain on state space $\mathcal X \subset \{x \in \{0,1\}^n : \sum_i x_i = m\}$ for some $m \ge 1$ with stationary measure $\pi$,
and suppose $P$ is $k$-local for some integer $k \ge 0$, in the sense that if $\|x - x'\|_1 > k$ then $P(x,x') = 0$. For any $\tau \ge 0$ and $x \in \mathcal X$, define
\[ S_{x,\tau} = \{ x' : \langle x, x' \rangle \ge \tau \}. \]
Then for any $t \ge 0$,
\[ \Pr_{\pi}(\exists s\le t, X_s \in S_{\tau}^C \mid X_0 \in S) \le t\frac{\pi(\{x' : \langle x, x' \rangle \in (\tau,\tau+k)\})}{\pi(S_{x,\tau})}. \]
\end{theorem}
\begin{proof}
This follows by combining the previous two lemmas.
\end{proof}

\begin{theorem}[Example metastability from FPP, Inhomogeneous case]
 Let $P$ be a reversible Markov chain on state space $\mathcal X \subset \{0,1\}^n$ with stationary measure $\pi$.  Suppose that $P$ is $\epsilon > 0$ stable with respect to overlap with some fixed $x \in \mathcal X$, in the sense that if $P(y,z) > 0$ then
 \[ \frac{\langle z, x \rangle}{\|z\|_2\|x\|_2} \ge \frac{\langle y, x \rangle}{\|y\|_2\|x\|_2} - \epsilon. \]
 For any $\tau \in \mathbb R$ and $x \in \mathcal X$, define
\[ S_{x,\tau} = \{ x' : \frac{\langle x, x' \rangle}{\|x\|_2 \|x'\|_2} \ge \tau \}. \]
Then for any $t \ge 0$,
\[ \Pr_{\pi}(\exists s\le t, X_s \in S_{\tau}^C \mid X_0 \in S) \le t\frac{\pi(\{x' : \frac{\langle x, x' \rangle}{\|x\|_2\|x'\|_2} \in (\tau,\tau+\epsilon)\})}{\pi(S_{x,\tau})}. \]
\end{theorem}
Note that by the following basic lemma, it is straightforward to prove that Markov chains which are $k$-local in the sense of Hamming distance are also local in the sense of normalized overlap, provided that $\|x\|_2 \gg \sqrt{k}$ for all $x \in \mathcal X$.
\begin{lemma}
For any $\Delta \in \{-1,0,1\}^{k}$
\[ \frac{\langle x, x' + \Delta \rangle}{\|x\|_2\|x' + \Delta\|_2} \ge \frac{\langle x, x' \rangle - k}{\|x\|_2(\|x'\|_2 + \sqrt{k})}. \]
\end{lemma}
\subsection{Comment on replica method approach}
In this appendix, we mention the alternative replica method approach to solving the model, which could be interesting to investigate further. The replica method is a non-rigorous but highly reliable way to estimate $\E \log Z$ to leading order in many problems.  See \cite{mezard2009information} for more about the replica method.  In some cases, the replica method can be much easier to use than rigorous methods, but it does not seem so helpful in our model. This is unsurprising since, while the replica method can be applied to sparse problems, it is usually not the preferred method (see \cite{biroli2000variational,braunstein2023cavity} for much more context). To illustrate, we show the basic computations which arise.

In the replica trick we start with the identity
\[ \E[\log Z] = \lim_{k \to 0} \frac{1}{k} \log \E[Z^k] \]
which is straightforward to justify using that $\log(1 + x) \approx x$ for small $x$. (Note that for fixed $z > 0$ $\lim_{k \to 0} z^k = 1$.) 
We then apply this formula and guess that we can swap the order of limits to take the high-dimensional limit $n \to \infty$ before $k \to 0$:
\[ \lim_{n \to \infty} \frac{1}{n \log \log n} \E \log Z = \lim_{k \to 0} \lim_{n \to \infty} \frac{1}{k n \log \log n} \log \E[Z^k]. \]
It is difficult to compute $\E[Z^k]$ for $k \approx 0$, so we instead attempt to guess it from the formula of $\E[Z^k]$ for integer $k$, which is more tractable to compute. By expanding the definition of $Z$, we know that
\[ Z^k = \sum_{T_1,\ldots,T_k} \exp\left(-\beta\log\log(n)\sum_{j = 1}^k \sum_v d_{T_j}(1,v)\right) \prod_{j = 1}^k 1(T_j \subset G) \]
where $G$ is the random graph and the summation ranges over all spanning trees of the complete graph. 

By the definition of the Erdos-Reyni graph $G(n,\frac{\alpha_n \log n}{n})$ we can compute that
\[ \E Z^k =  \sum_{T_1,\ldots,T_k} \exp\left(-\beta\log\log(n)\sum_{j = 1}^k \sum_v d_{T_j}(1,v)\right) (\alpha_n \log(n)/n)^{|T_1 \cup \cdots \cup T_k|} \]
where $|T_1 \cup \cdots \cup T_k|$ counts the number of unique edges among the $k$ spanning trees. In dense spin glass models like the SK model or REM model, at this point the summand would simplify to an expression only depending on the overlap matrix $Q_{ij} = |T_i \cap T_j|/(n - 1)$. In our case, the summand depends only on $E_j = \sum_v d_{T_j}(1,v)$ and $|H| = |T_1 \cup \cdots \cup T_k|$. Computing the formula in the cases $k = 1$ and $k = 2$ would correspond to the computations for the first and second moment method, which is already nontrivial. We did not pursue this direction further.  
\subsection{Simulation details}\label{sec:simulation-details}
We elaborate further on the implementation of the simulation behind Figure~\ref{fig:overlap}, since it is slightly nontrivial. To compute $f(a,b)$, we used a Monte Carlo approximation by generating iid Poisson samples, since we are not aware of a closed form. To generate the plot $100$ values of $\rho_n$ were sampled from $1$ to $0.9$.
To simulate progressive resampling, for each pair of vertices $u,v\in V$ with $u\neq v$, we attach a random variable $T_{u,v}\sim\Unif(0,1)$ which models the ``time'' the pair gets resampled\footnote{In practice, maintaining $\binom{n}{2}$ random variables for large $n$ is infeasible. Instead we leverage sparsity: we first sample the edges of $G_0\cup\cdots\cup G_m$ and assign $T_{u,v}$ only to these edges, and then generate the graphs using sampled subsets of the edges. It is relatively straightforward to achieve the correct distribution for the trajectory this way.}. Namely, for given time indices
\[
    0=t_0<t_1<\cdots<t_m<1
\]
we sample a trajectory of Erd\"os--R\'enyi graphs $(G_0,\cdots,G_m)$ as follows. We first generate $G_0\sim\mathcal{G}(n,q)$, and each $G_k$ for $k\geq1$ is inductively generated from $G_{k-1}$ by resampling $(u,v)$ with $t_{k-1}<T_{u,v}\leq t_k$. Then for each $k$, we compute the overlaps for $G_0$ and $G_k$: one for uniformly random shortest path trees \eqref{eqn:overlap2def} and one for the shortest paths \eqref{eqn:parconc}. Note that each of the shortest paths is determined by the corresponding shortest path tree.


\section{Relationship to discrete convexity}\label{sec:dca}

Discrete convex analysis primarily concerns two natural discrete analogues of convexity, called $L$ and $M$-convexity. Interestingly, while these definitions can be stated in terms of purely geometric properties of a function/set, they have remarkably strong connections to the uncapacitated min-cost flow problem and the shortest path problem.

\subsection{L-convexity and M-convexity}
This part is essentially from chapter 1 of Murota's book
\cite{murota2003discrete}.
\begin{definition}
For any $n \ge 1$, we say that a nonempty set $D \subset \mathbb Z^n$ is L$^{\natural}$-convex if it is closed under discrete midpoints: for all $x,y\in D$,
  \[ \big\lceil\frac{x+y}{2}\big\rceil,\ \big\lfloor\frac{x+y}{2}\big\rfloor\in D \]
 where the ceiling and floor are applied componentwise. We say that such a set $D$ is L-convex if it additionally satisfies translation invariance along the all-ones direction, i.e.
 \[ x \in D \rightarrow x \pm \vec{1} \in D. \]
\end{definition}
The two types of L-convexity are almost equivalent in the sense that any L$^{\natural}$-convex set in $\mathbb Z^n$ can be directly represented as the intersection of an L-convex set in $\mathbb Z^{n + 1}$ with the hyperplane $x_{n + 1} = 0$. There are other natural characterizations of these definitions in terms of submodularity and related properties, see \cite{murota2003discrete}.

For our context, it is useful to know the following alternative representation of $L$-natural convex sets:
\begin{theorem}[Chapter 1.4.1, Equation 1.34, of \cite{murota2003discrete}]
A nonempty set $D \subset \mathbb Z^n$ is $L$-natural convex if and only if there exists a matrix $C \in (\mathbb{R} \cup \{\infty\})^{n \times n}$ with zero-diagonal such that
\[ D = \{ p : p_i - p_j \le C_{ij} \}. \]
\end{theorem}

Given an $L$-natural convex set $D$, we can define its indicator
\[ \delta_D(p) = \begin{cases} 0 & \text{if $p \in D$} \\ \infty & \text{otherwise} \end{cases} \]
and its \emph{discrete convex/Fenchel conjugate} by
\[ \delta^*_D(p) = \sup_{x} [\langle x, p \rangle - \delta_D(p)] = \sup_{x \in D} \langle x, p \rangle, \]
which is a discrete analogue of the support function of a convex body. It is also a prototypical example of an $M$-convex function, as we will see more precisely next.
\begin{definition}
Given a function $f : \mathbb{Z}^n \to \mathbb R \cup \{\infty\}$, we define 
\[ \operatorname{dom}(f) = \{ p : f(p) < \infty \}. \]
\end{definition}
\begin{definition}
Let $f : \mathbb{Z}^n \to \mathbb R \cup \{\infty\}$ be a function with nonempty domain. We say that $f$ is \emph{M-convex} if it satisfies the following \emph{exchange axiom}: for any $x,y \in \operatorname{dom}(f)$ and $i$ such that $x_i > y_i$, there exists $j$ such that $x_j < y_j$ such that
\[ f(x) + f(y) \ge f(x - e_i + e_j) + f(y + e_i - e_j). \]
\end{definition}
\begin{remark}
In many cases in discrete convex analysis, one is interested in integer-valued functions which means that $f$ will be valued in $\mathbb Z \cup \{\infty\}$.
\end{remark}
The domain of such a function $f$ is equivalently an \emph{M-convex set}.
In the special case that $f : \{0,1\}^n \to \R$, the support of $f$ is exactly the set of \emph{bases} of a \emph{matroid}. 
\begin{definition}
We say a function $f : \mathbb{Z}^n \to \mathbb R$ is \emph{positively homogeneous} if for all $x,y \in \operatorname{dom}(f)$ such that $x = \lambda y$ for some $\lambda > 0$,
\[ f(x) = f(\lambda y) = \lambda f(y). \]
\end{definition}
\begin{theorem}[Theorem 1.14 of \cite{murota2003discrete}]
Let $f : \mathbb{Z}^n \to \mathbb Z$ be an arbitrary integer-valued function with nonempty domain.
Then $f$ is a positively homogeneous M-convex function if and only if there exists an L-convex set $D$ such that $f = \delta_D^*$, i.e. $f$ is the discrete convex conjugate of the indicator of $D$. Furthermore, in this case $(\delta_D^*)^* = f^* = \delta_D$, i.e. discrete convex conjugation acts as an involution.
\end{theorem}



\subsection{Graphic matroid and network simplex}\label{sec:simplex}
\paragraph{The support function computes min-cost flows.}
Recall that every $L$-convex set corresponds to an integer zero-diagonal matrix $C$
via the representation
\[ D = \{p \in \mathbb{Z}^n : p_i - p_j \le C_{ij} \} \]
and the discrete convex conjugate is
\[ \delta_D^*(x) = \sup_{p \in D} \langle p, x \rangle. \]
The convex hull of $D$ will be an integral polytope, i.e. a convex set with integral vertices, so $\delta_D^*$ agrees on $\mathbb{Z}^n$ with the restriction of the support function of $conv(D)$, i.e. for all $x \in \mathbb{Z}^n$ we have the equivalent formula
\[ \delta_D^*(x) = \sup_{p \in \overline D} \langle p, x \rangle. \]
where
\[ \overline D =  \{p \in \mathbb{R}^n : p_i - p_j \le C_{ij} \} \]
is the convex hull of $D$. By strong LP duality, this can in turn be rewritten
as an optimization problem over flows:
\[ \delta_D^*(x) = \inf_{f \in \overline F_x} \sum_{i,j} C_{ij} f_{ij} \]
where
\[ \overline F_x = \{ f \in \mathbb{R}_{\ge 0}^{n \times n} : \sum_i C_{ij} f_{ij} - \sum_i C_{ji} f_{ji} = x_i,\qquad f_{jj} = 0 \forall j \}  \]
is the set of real-valued flows satisfying demand vector $x$. 
By complementary slackness from linear programming theory, any vertex of $\overline F_x$ has at least $(n - 1)(n - 1) - (n - 1)$ zero entries; equivalently, any vertex of $\overline F_x$ has at most $n - 1$ nonzero entries.
In fact, any vertex of $\overline F_x$ is supported on a \emph{forest}: this follows by considering connected components of the graph and applying complementary slackness on each component.

The fundamental \emph{simplex algorithm} from linear programming, in this setting, recasts the linear program as an optimization over the elements of the graphic matroid. Each vertex of $\overline F_x$ correspond to one or more bases in the sense of simplex, which is the same as the sense of basis for the graphic matroid, i.e. a spanning tree. In general, degeneracy can occur when many different spanning trees correspond to the same flow (i.e. the same vertex): then pivoting may not strictly improve objective until we find the ``right'' basis for this point. This is called \emph{stalling} in the linear programming literature. 



For tree-like demands $x$, the flow polytope has no nondegenerate vertices\footnote{In other words, given a integer-valued flow from the source to all other vertices, we can read off a unique choice of basis (spanning tree).} so simplex is guaranteed to make progress in every step. Therefore, the network simplex algorithm will necessarily converge in polynomial time (in the size of the numbers in its input, i.e. in psuedopolynomial time) for any standard pivot rule. As a very special case, this tells us that network simplex will converge in polynomial time in the shortest path setting on unweighted graphs considered in our paper.


\begin{proposition}\label{prop:p2-polytime}
1-local search on (P2) converges in polynomial time for any connected graph $G$.
\end{proposition}
\begin{proof}
As discussed above, given the interpretation of local search as the simplex algorithm, this is a direct consequence of basic properties of the simplex algorithm, see e.g. \cite{bazaraa2011linear,ahuja1993network}.
Explicitly, the objective function is integer-valued and can be at most polynomially large, and there are no \emph{degeneracies} in the corresponding linear program --- i.e., every vertex of the linear program is in one-to-one correspondence with a basis (which in turn corresponds to a spanning tree). So the simplex algorithm with any standard pivot rules will always decrease the objective value at every step until it finds the global minimum.
\end{proof}
Recall that (P2) is
\[ \min_{T \in \mathcal{T}} \; \sum_{v \in V} d_T(s,v) \]
where $\mathcal{T}$ is the set of spanning trees of $G$.
By 1-local search we mean any algorithm which, started from a spanning tree $T$, iteratively moves to a tree $T'$ which differs in the addition and removal of one edge, as long as the objective value for $T'$ is smaller. It is also possible to prove Proposition~\ref{prop:p2-polytime} directly by considering breath-first-search from the root $s$ in the underlying graph.

\begin{remark}[Simplex is strongly polynomial time for shortest paths]
Although for worst-case linear programs, all known versions of the simplex algorithm currently require exponential time,
it was proved by Orlin \cite{orlin2009simplex} that the simplex algorithm with Dantzig's pivot rule is \emph{strongly polynomial time for shortest path problems}. See also the improved result of Hansen et al \cite{hansen2014dantzig}. Orlin \cite{orlin1997polynomial} also showed existence of a pivoting approach for general flow problems with polynomial time guarantees.
\end{remark}
\subsection{Dijkstra's algorithm and L-concave function maximization \cite{murota2014dijkstra}}
From the discussion in the previous section, we can understand how the \emph{network simplex} algorithm is a natural
method to optimize linear functions over $L$-convex sets. It was observed by Murota and Shioura \cite{murota2014dijkstra} that Dijkstra's algorithm also naturally arises from considerations in discrete convex analysis.

Let $D$ be as defined in the previous section and interpret $C_{ij}$ as the edge length from vertex $i$ to vertex $j$ in the graph of interest.
Then \emph{shortest path problem} from source vertex $1$ corresponds to computation of
\[ \delta_D^*(-(n - 1), 1, \ldots, 1) = \max_{p \in D} \sum_v (p(v) - p(1)). \]
This is optimization of a linear function over an L-convex set, which is a special case of the 
more general problem of $L$-concave function maximization (the larger class consists of functions with obey
a midpoint inequality similar to the continuous case, and discretized similarly to the definition of $L$-convex set). Any $L$-concave maximization problem can
be solved in polynomial time by the \emph{steepest ascent algorithm} from discrete convex analysis \cite{murota2003discrete}. Consist with the name, the steepest ascent algorithm at iteration $t$ applied to $L$-concave function $g$ starts from a point $p_t \in \mathbb R^n$ and moves to $p_{t + 1} = \arg\max_{y \in \{0,1\}^n, \epsilon \in \{\pm 1\}} g(p + \epsilon y)$. $L$-concave functions have a ``local to global'' property similar to continuous concave functions, i.e., local optima of such functions are always global optima. So steepest ascent is globally convergent. 

In general, steepest ascent can always be implemented by appealing to an oracle for submodular minimization to solve the optimization over $y$. However, in specific cases the steepest ascent algorithm can be simplified. In the case of the shortest path problem, it turns out that the steepest ascent algorithm is \emph{Dijkstra's algorithm} \cite{murota2003discrete}; one always has $\epsilon = 1$ and steepest ascent is increasing the node potentials $p$ just as in Dijkstra's algorithm (the $p$ correspond to the location of the vertices in the ``physical'' interpretation of Dijkstra's algorithm mention before). See the paper for full details.
\subsection{Shortest path tree is not discrete convex}
The support of the shortest path tree objective is on the bases of the graphic matroid, which is an M-convex set, so it is natural to ask if the function itself is M-convex (equivalently, to ask if it is a valuated matroid). The following example shows that this is not the case. (It is direct to see that shortest path tree objective is not $L$-convex as it does not have the needed discrete midpoint property.)
\begin{example}\label{ex:not-mconvex}
For notational convenience, we will root the tree in this example at vertex $0$ instead of vertex $1$.
Let $m \ge 1$, and consider a graph $G$ which is given by starting with a cycle on $2m + 1$ nodes labeled $0,\ldots,2m$, and adjoining a separate path of length $m + 3$ from node $0$ to node $m$ with intermediate vertices labeled $1',2',\ldots,(m + 1)',(m + 2)'$. Let $T_1$ be the optimal shortest path tree, which is formed by dropping the edges from $m$ to $m + 1$ and from $(m + 1)'$ to $(m + 2)'$ from the graph. Let $T_2$ be the spanning tree formed from the original graph by dropping the edge from $m - 1$ to $m$ and from $m$ to $(m + 2)'$. Let $OPT$ be the objective value achieved by $T_1$, and observe that the objective value achieved by $T_2$ is $OPT + 2$ since vertices $m$ and $(m + 2)'$ are now further from the root.  

If we consider an exchange between the two trees by dropping the edge from $(m + 1)'$ to $(m + 2)'$ from $T_2$ and adding it to $T_1$, we must remove the edge from $m$ to $(m + 2)'$ from $T_1$ and add it to $T_2$. This makes the objective value of $T_1$ worsen to $OPT + 1$ and the objective value of $T_2$ stays the same, which violates the exchange axiom.
\end{example}
See Figure~\ref{fig:not-dca} for a diagram illustrating the example.
\begin{figure}[t]
\centering
\begin{tikzpicture}[scale=0.6,
  every node/.style={font=\small},
  vtx/.style={circle, draw, inner sep=1.7pt, fill=white},
  pvtx/.style={circle, draw, inner sep=1.7pt, fill=white},
  root/.style={circle, draw, inner sep=2pt, thick, fill=black!10},
  graph/.style={line width=0.6pt, draw=black!40},
  pathstyle/.style={line width=0.6pt, draw=black!40},
  Tedge/.style={line width=1.5pt, draw=blue!70},
  removed/.style={line width=1.2pt, draw=red!70, dashed}
]

\def\m{4} 
\pgfmathtruncatemacro{\twoM}{2*\m}
\pgfmathtruncatemacro{\mminusone}{\m-1}
\pgfmathtruncatemacro{\mplusone}{\m+1}
\pgfmathtruncatemacro{\mplustwo}{\m+2}
\pgfmathsetmacro{\R}{3.1}     
\pgfmathsetmacro{\Rtwo}{4.4}  
\pgfmathsetmacro{\dtheta}{360/(\twoM+1)} 

\newcommand{\OnePanel}[1]{%
  \begin{scope}[xshift=#1]
    \foreach \k in {0,...,\twoM}{
      \pgfmathsetmacro{\ang}{90 - \k*\dtheta}
      \ifnum\k=0
        \node[vtx, root] (v0) at ({\ang}:\R) {$0$};
      \else
        \node[vtx] (v\k) at ({\ang}:\R) {$\k$};
      \fi
    }

    \pgfmathsetmacro{\angzero}{90}
    \pgfmathsetmacro{\angm}{90 - \m*\dtheta}
    \pgfmathtruncatemacro{\mpTwo}{\m+2}
    \pgfmathtruncatemacro{\mpOne}{\m+1}
    \pgfmathsetmacro{\step}{(\angzero - \angm)/(\m+3)} 
    \foreach \i in {1,...,\mpTwo}{
      \pgfmathsetmacro{\angp}{\angzero - \i*\step}
      \node[pvtx] (p\i) at ({\angp}:\Rtwo) {$\i'$};
    }

    \foreach \k in {0,...,\twoM}{
      \pgfmathtruncatemacro{\knext}{mod(\k+1,\twoM+1)}
      \draw[graph] (v\k) -- (v\knext);
    }
    \draw[pathstyle] (v0) -- (p1);
    \foreach \i in {1,...,\mpOne}{
      \pgfmathtruncatemacro{\j}{\i+1}
      \draw[pathstyle] (p\i) -- (p\j);
    }
    \draw[pathstyle] (p\mpTwo) -- (v\m);

  \end{scope}%
}

\begin{scope}
  \OnePanel{0cm}
\end{scope}

\begin{scope}
  \OnePanel{9.2cm}

  \foreach \k in {0,...,\twoM}{
    \ifnum\k=\m
    \else
      \pgfmathtruncatemacro{\knext}{mod(\k+1,\twoM+1)}
      \draw[Tedge] (v\k) -- (v\knext);
    \fi
  }
  \draw[Tedge] (v0) -- (p1);
  \foreach \i in {1,...,\m}{
    \pgfmathtruncatemacro{\j}{\i+1}
    \draw[Tedge] (p\i) -- (p\j);
  }
  \draw[Tedge] (p\mplustwo) -- (v\m);

  \draw[removed] (v\m) -- (v\mplusone);
  \draw[removed] (p\mplusone) -- (p\mplustwo);

\end{scope}

\begin{scope}
  \OnePanel{18.4cm}

  \foreach \k in {0,...,\twoM}{
    \ifnum\k=\mminusone
    \else
      \pgfmathtruncatemacro{\knext}{mod(\k+1,\twoM+1)}
      \draw[Tedge] (v\k) -- (v\knext);
    \fi
  }
  \draw[Tedge] (v0) -- (p1);
  \foreach \i in {1,...,\mplusone}{ 
    \pgfmathtruncatemacro{\j}{\i+1}
    \draw[Tedge] (p\i) -- (p\j);
  }

  \draw[removed] (v\mminusone) -- (v\m);
  \draw[removed] (p\mplustwo) -- (v\m);

\end{scope}

\begin{scope}[yshift=-4.6cm]
  \draw[Tedge] (-1.2,0) -- (0.8,0); \node[anchor=west] at (1.0,0) {tree edges};
  \draw[removed] (-1.2,-0.6) -- (0.8,-0.6); \node[anchor=west] at (1.0,-0.6) {dropped edges};
  \draw[graph] (-1.2,-1.2) -- (0.8,-1.2); \node[anchor=west] at (1.0,-1.2) {other edges of $G$};
  \node[root, label=right:{root $0$}] at (-1.2,-1.9) {};
\end{scope}

\end{tikzpicture}
\caption{Graph $G$ and the spanning trees $T_1$ and $T_2$ from Example~\ref{ex:not-mconvex} with $m{=}4$.}\label{fig:not-dca}
\end{figure}
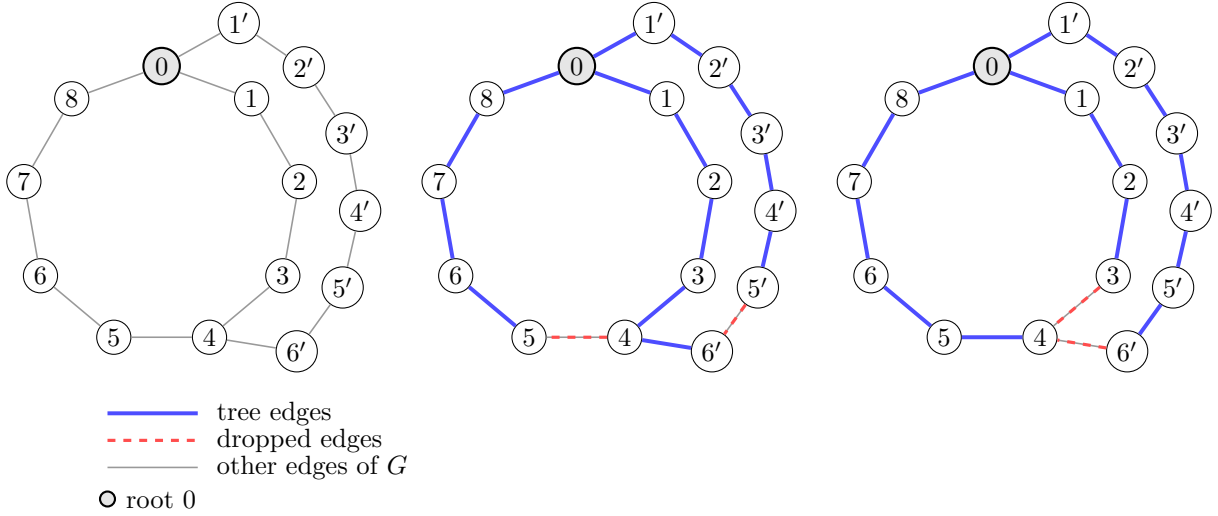

\section{Free energy barrier in the space of paths}\label{apdx:path-barrier}
Fix a target vertex $t$, and
consider for each $n$ the Gibbs measure
\begin{equation}\label{eqn:gibbs-paths}
\nu_{G, \beta}(P) = \frac{1}{Z_{G,\beta}} \exp\left(-\beta \log\log(n) |P|\right) = \frac{1}{Z_{G,\beta}} \log(n)^{-\beta |P|}
\end{equation}
over paths $P \in \mathcal P_{s,t}$ from $1$ to $t$. To be more precise, we let paths in $\mathcal P_{s,t}$ repeat vertices (i.e. they are not required to be simple, sometimes called ``walks''), and we also restrict paths in $\mathcal P_{s,t}$ to have length at most $n - 1$ where $n$ is the number of vertices. We make the latter restriction since the shortest path has at most $n - 1$ edges in any graph, and we can avoid unnecessary technical details by defining the Gibbs measure over a finite space, but the results below do not depend on this choice. 
For a path $P$ we will let $|P|$ denote its length, i.e. the total number of edges it contains. 
\subsection{Basic properties of Gibbs measures on paths}
Below we show that the model is in a low temperature phase when $\beta > 1$ and describe its structure.
\begin{remark}
It is not hard to show with similar ideas that Gibbs measure on paths exhibits a phase transition at $\beta = 1$. Informally, at low temperature ($\beta > 1$) the measure is almost completely supported on the true shortest paths, whereas at high temperature ($\beta < 1$) the measure is close to uniform and typical sample paths are much longer than the shortest path. Since the low temperature regime is most relevant to the results of \cite{LS2024}, we omit the details. 
\end{remark}
\begin{lemma}\label{lem:few-subopt-paths}
Fix $\beta > 1$. With probability $1 - o(1)$, the following holds true uniformly over all $r \ge 1$. 
Let  $\ell^*$ be the length of the shortest path from $1$ to $2$ in the random graph, then
\[ \nu_{G,\beta}(\{P' : |P'| \ge \ell^* + r\}) \le 2/\log(n)^{(\beta - 1)r}. \]
Furthermore,
\[ \sum_{\ell \ge \ell^* + r} |\{P : |P| = \ell \}| \log(n)^{-\beta \ell} \le 2\log(n)^{-\beta \ell^*} \log(n)^{-(\beta - 1) r}. \]
\end{lemma}
\begin{proof}
Since $\ell^*$ be the length of the shortest path from $1$ to $2$, $Z_{G,\beta} \ge \exp(-\beta \ell^* \log\log(n))$.
Observe that
\[ \nu_{G,\beta}(\mathcal P \setminus \arg\min |P|) = \frac{1}{Z_{G,\beta}} \sum_{\ell \ge \ell^* + r} |\{P : |P| = \ell \}| \log(n)^{-\beta \ell} \le \frac{1}{Z_{G,\beta}} \sum_{\ell \ge \ell^* + r} (\Delta/\log(n)^{\beta})^{\ell} \le \sum_{s \ge r}^{\infty} (\Delta/\log(n)^{\beta})^{s}. \]
where $\Delta$ is the maximum degree in the graph. Since $\Delta = O(\log n)$ asymptotically almost surely and $\beta > 1$, the geometric series converges to at most
\[ 2 (\Delta/\log(n)^{\beta})^r \le 2/\log(n)^{(\beta - 1)r}. \]
The second inequality follows from essentially the same argument. 
\end{proof}
\begin{lemma}\label{lem:low-temp-path-is-opt}
Fix $\beta > 1$.
There exists $\epsilon = o(1)$ as $n \to \infty$ such that
\[ \nu_{G,\beta} = (1 - \epsilon) Uni(\arg\min |P|) + \epsilon \]
with probability $1 - o(1)$.
Equivalently, $\nu_{G,\beta}$ is equal to the uniform distribution
over the set of shortest paths from $1$ to $2$ up to $o(1)$ Huber
contamination \cite{huber1964robust}. 
\end{lemma}
\begin{proof}
Since the restriction of $\nu$ to $P \in \arg\min |P|$ is the uniform measure, the conclusion follows from the previous lemma applied with $r = 1$.
\end{proof}
\subsection{Key structural result from Li--Schramm \cite{LS2024}}
Li and Schramm showed that the set of approximate optima for the shortest path problem satisfy a very strong structural condition \cite{LS2024}.
\begin{definition}
For $\epsilon > 0$, let
$\mathcal{P}_{\epsilon}(G)$ be the set of paths from vertex $1$ to $2$ whose length
is at most $(1 + \epsilon)$ times longer than the shortest path from $1$ to $2$.
\end{definition}
\begin{theorem}[Special case of Theorem 2.2 of \cite{LS2024}]\label{thm:ls-structural}
There exists a constant $C > 0$ such that for any $\epsilon > 0$ sufficiently small,
the following is true. With probability at least $1 - O\left(\frac{\log\log n}{\log n}\right)$,
all pairs $p,p' \in \mathcal{P}_{\epsilon}$ are nearly disjoint in the following sense:
\[ \frac{|p \cap p'|}{\sqrt{|p||p'|}} \le C \epsilon. \]
\end{theorem}
They also obtained an estimate of the number of $\epsilon$-shortest paths:
\begin{lemma}[Lemma 2.1 of \cite{LS2024}]
For all $\epsilon > 0$ sufficiently small, it holds that with probability $1 - O\left(\frac{\log n}{\log \log n}\right)$ that
\[ |\mathcal P_{\epsilon}| = (1 \pm o(1))n^{\epsilon}. \]
\end{lemma}

\subsection{Free energy barriers and Franz--Parisi Potential}
Recall that the (asymptotic) Franz--Parisi potential \cite{franz1995recipes} (see, e.g., \cite{el2025shattering} for the formal mathematical definition below, up to the conventional factor of $-1/\beta$) is given, assuming the limit exists, by
\[ \mathcal{F}_{\beta}(q) = \lim_{\epsilon \to 0} \operatorname{p-lim}_{n \to \infty} \frac{-1}{\beta \log n} \log \sum_{P' : |P\cap P'|/\sqrt{|P||P'} \in (q - \epsilon, q + \epsilon)} e^{-\beta \log\log(n) |P'|}  \]
where $P$ is a sample from $\nu_{G,\beta}$, the Gibbs measure over paths at inverse temperature $\beta$. 
With this sign convention, if the FPP is not \emph{quasiconvex}, i.e. if it has multiple local minima then we say there is a free energy barrier. In many cases, including this example, quasiconvexity fails because the potential is not monotonically increasing from the global minimum. Glauber dynamics and similar processes can only ``roll downward'' from $q = 1$, so they need the monotone increasing property to be able to escape from their initialization in a small amount of time. See Appendix~\ref{apdx:bottleneck} for more.

For our purposes, we will work with the corresponding nonasymptotic version
\[ \mathcal{F}_{\beta}(q,n,\epsilon, P) =  \frac{-1}{\beta \log n} \log \sum_{P' : |P\cap P'|/\sqrt{|P||P'} \in (q - \epsilon, q + \epsilon)} e^{-\beta \log\log(n) |P'|}   \]
which has the same consequences for bottlenecks, dynamics, etc. This is for convenience since the existence of the limit is irrelevant here.

\begin{proposition}
There exists $c > 0$ such that the following is true for all sufficiently small $\epsilon > 0$ and sufficiently large $n$.
For any $\beta > 1$, for all $q \in (0.01,0.99)$, 
\[ \mathcal{F}_{\beta,n,\epsilon,P}(q) \ge \mathcal{F}_{\beta,n,\epsilon,P}(1) + c(\beta - 1) \]
asymptotically almost surely for $P \sim \nu_{G,\beta}$.
\end{proposition}
\begin{proof}
By Lemma~\ref{lem:low-temp-path-is-opt}, $P$ sampled from the Gibbs measure will be a true shortest path from $1$ to $2$ asymptotically almost surely. By Theorem~\ref{thm:ls-structural}, any path $P' \ne P$ from $1$ to $2$ with length at most $|P|(1 + \epsilon)$ will have
\[ \frac{|P \cap P'|}{\sqrt{|P| |P'|}} \le C\epsilon. \]
Therefore by Lemma~\ref{lem:few-subopt-paths},
\[ \frac{1}{\log(n)^{\beta \ell^*}} \sum_{P' \ne P : |P \cap P'|/\sqrt{|P| |P'|} > C\epsilon} \le 2/\log(n)^{(\beta - 1)\epsilon|P|}. \]
Since $|P|$ is asymptotically almost surely of order $\log(n)/\log\log(n)$ (e.g. by Theorem~\ref{thm:main1}), 
and $\mathcal{F}_{\beta}(1) \ge \log(n)^{-\beta\ell^*}$, the conclusion follows.
\end{proof}

\begin{remark}[Alternative free energy barrier]
We can also show a free energy barrier from the structural result in a slightly different way, by arguing there are no paths (sufficiently) close to an optimal $P$. Because for them to have large overlap, they need to have similar length $|P \cap P'|/\sqrt{|P| |P'|} \le \sqrt{|P|/|P'|}$ so if $|P'| > (1 + \epsilon) |P|$ then the overlap is at most $\sqrt{1/(1 + \epsilon)} \le 1 - c\epsilon$. In terms of the Franz--Parisi potential, this alternative argument gives a taller but less ``wide'' gap than the proposition, so it is slightly weaker in terms of the locality of chains it applies to.
\end{remark}
\section{An analogous example: local search on multilinear vs Lov\'asz extension}\label{sec:example-lovasz}
In this section, we consider the consequences of the fact that minimizing the \emph{multilinear} and \emph{Lov\'asz} extensions of a submodular function $f$ are obviously polynomial time equivalent --- the global minimum of either is simply the global minimum of $f$ on sets --- but their optimization landscapes are very different. This illustrates the point that (1) the landscape of two polynomial time equivalent optimization problems can be very different, and (2) this difference is meaningful/useful from an algorithm design perspective. To rephrase point (2), it is helpful to know that the optimization landscape of the Lov\'asz extension is benign whereas that of the multilinear extension is not, so we can know to perform local search on the former rather than the latter.  

We also illustrate how the landscape of the multilinear extension is reflected in the Franz--Parisi potential for a related high-dimensional Ising model, which implies rigorous metastability results for the Glauber dynamics at low temperature. 

 \paragraph{Optimization of quadratic polynomials, supermodularity, and Lov\'asz extension.} Any quadratic polynomial of the form $f(x) = \langle x, M x \rangle/2 + \langle x, h \rangle$ where $M_{ij}$ is an $n \times n$ matrix with zero diagonal and nonnegative entries can in fact be maximized over $[-1,1]^n$ in polynomial time. This is because the maximum is achieved at the vertices, and the maximization over the vertices is a supermodular maximization problem\footnote{Here we consider $\{\pm 1\}^n$ in bijective correspondence with $2^{[n]}$ by viewing $-1$ in coordinate $i$ as indicating that element $i$ is not present in the set, and $+1$ as indicating that the element is included.}, equivalently a submodular minimization of $-f$, so it can be solved in polynomial time. Explicitly, we can define the Lov\'asz extension \cite{lovasz1983submodular}
\[ \hat f(x) = \frac{1}{2} \int_{-1}^1 f((-1)^{1[x_1 < t]}, \ldots, (-1)^{1[x_n < t]}) dt. \]
which agrees with $f$ on $\{ \pm 1\}^n$. The fundamental property of the Lov\'asz extension is that because $-f$ is supermodular, $-\hat f$ is convex, i.e. $\hat f$ is a concave function on $[-1,1]^n$. 

\subsection{Local search in a simple example}
\paragraph{An example.} Define
\[
f(x,y,z,w)=x+y-z-w+xy+zw+\tfrac{3}{4}\,(x + y)(z + w),
\qquad (x,y,z,w)\in[-1,1]^4.
\]
We consider minimizing $-f$, equivalently maximizing $f$. A direct check over vertices shows $f(1,1,1,1)=f(-1,-1,-1,-1)=5$ are the global maxima, while $f(1,1,-1,-1)=3$. Gradient flow on $-f$ goes to the suboptimal point $(1,1,-1,-1)$ rather than the global optima. This happens because $f$ is the \emph{multilinear} extension of its restriction to $\{\pm 1\}^n$, which is not suitable for direct optimization.

On the other hand, as we will discuss afterwards, optimizing $f$ is equivalent to maximizing the Lov\'asz extension $\hat f$ of $f$ restricted to $\{\pm 1\}^4 \to \mathbb R$, which is a concave optimization problem and thus can succcessfully be solved via local search \emph{on the correct choice of objective/landscape.}

\paragraph{Failure of gradient flow on multilinear extension.} 
Then
\[
\nabla f(x,y,z,w)=\bigl(1+y+\tfrac{3}{4}(z + w),\; 1+x+\tfrac{3}{4}(z + w),\; -1+w+\tfrac{3}{4}(x + y),\; -1+z+\tfrac{3}{4}(x + y)\bigr).
\]
The unconstrained gradient (ascent) flow is
\begin{equation}\label{eq:flow}
\frac{d}{dt}(x,y,z,w) = \nabla f(x,y,z,w)
\end{equation}
On the box $[-1,1]^4$ we consider the \emph{projected} gradient flow, i.e.\ we follow \eqref{eq:flow} while in the interior and, upon hitting the boundary, we project the velocity onto the tangent cone of $[-1,1]^4$ at the current point.

\begin{proposition}
For the projected ascent flow with initial condition $(x(0),y(0),z(0),w(0))=(0,0,0,0)$, the trajectory hits the point $(1,1,-1,-1)$ at time $t^*=\ln 4$ and remains there for all $t\ge t^*$. In particular, the flow started at the origin converges to $(1,1,-1,-1)$ even though the global maximizers of $f$ on $[-1,1]^4$ are $(1,1,1,1)$ and $(-1,-1,-1,-1)$.
\end{proposition}

\begin{proof}
The ODE \eqref{eq:flow} is invariant under the symmetry $x\leftrightarrow y$, $z\leftrightarrow w$, and $(x,y,z,w)\mapsto (y,x,w,z)$. With the symmetric initial condition, the solution satisfies
\[
x(t)=y(t)=:s(t),\qquad z(t)=w(t)=-s(t)
\]
for as long as it remains in the interior. Substituting into \eqref{eq:flow} gives a single scalar ODE
\[
\dot s=1+s+\tfrac{3}{2}(-s)=1-\tfrac{1}{2}s,\qquad s(0)=0,
\]
whose solution is
\[
s(t)=2\bigl(1-e^{-t/2}\bigr).
\]
The first contact with the box occurs when $s(t)=1$, i.e.\ $2(1-e^{-t/2})=1$, which yields $t^*=\ln 4$. At $t=t^*$ we reach $(x,y,z,w)=(1,1,-1,-1)$.

At this point,
\[
\nabla f(1,1,-1,-1)=(1.5,\,1.5,\,-0.5,\,-0.5).
\]
The active constraints are $x=1$, $y=1$, $z=-1$, $w=-1$. The tangent cone therefore consists of velocities with $\dot x\le 0$, $\dot y\le 0$, $\dot z\ge 0$, $\dot w\ge 0$. Projecting the ascent direction $(1.5,1.5,-0.5,-0.5)$ onto this cone gives the zero vector, so the projected flow has $\dot x=\dot y=\dot z=\dot w=0$ there. Thus the trajectory stays at $(1,1,-1,-1)$ for all $t\ge t^*$.
\end{proof}

\paragraph{Success of local search on Lov\'asz extension in the same example.} As a consequence of the general theory, we know that $\hat f$ is concave and that projected gradient ascent (i.e., subgradient descent on $- \hat f$), as well as other convex optimization procedures like the ellipsoid method, will converge to a global optimum. Explicitly, in the example of zero initialization, we have $\hat f(0,0,0,0) = (1/2) f(1,1,1,1) + (1/2) f(-1,-1,-1,-1) = 5$ so the all-zeros initialization is already a global optimum. From an algorithmic perspective, the fact that $\hat f$ is concave and locally maximized at $(0,0,0,0)$ tells us that it is a global optima of $\hat f$, and by the definition of the Lov\'asz extension this exactly tells us that both $f(1,1,1,1)$ and $f(-1,-1,-1,-1)$ are global maxima of $f$.

\subsection{High-dimensional version}
We can define a high-dimensional version of the previous example by, for any $m \ge 1$, defining
\[ f_m(x,y,z,w) = f(\overline x, \overline y, \overline z, \overline w) = f\left(\frac{1}{m} \sum_{i = 1}^m x_i, \frac{1}{m} \sum_{i = 1}^m y_i, \frac{1}{m} \sum_{i = 1}^m z_i, \frac{1}{m} \sum_{i = 1}^m w_i\right). \]
Here $\overline{x} = \frac{1}{m} \sum_{i = 1}^m x_i$ is a shorthand for the mean of the $x$ vector. 
\subsubsection{Ising Model and Franz--Parisi Potential}
We can define a corresponding mean-field Ising model on $\{\pm 1\}^{4m}$ at inverse temperature $\beta >J 0$ by its Gibbs measure
\[ p_{\beta,m}(x,y,z,w) \frac{1}{Z_{\beta,m}} \exp(\beta m f_m(x,y,z,w)). \]
Let $Z = Z_{\beta,m}$ be the normalizing constant. From the definition, we can compute\footnote{The computations here are similar to those in the Curie-Weiss/mean-field Ising model, so we omit unnecessary details. See, e.g., \cite{ellis2007entropy} for the details of the calculation in the case of the Curie-Weiss model.} that
\[ Z = \sum_{\overline x, \overline y, \overline z, \overline w} e^{\beta m f(\overline x, \overline y, \overline z, \overline w) + S_m(x,y,z,w)} \]
where
\begin{align*} 
S_m(\overline x, \overline y, \overline z, \overline w) 
&= \#\{x : \frac{1}{m} \sum_i x_i = \overline x, \ldots, \frac{1}{m} \sum_i w_i = \overline w\} \\
&= m\left(h(\frac{1 + \overline x}{2}) + h(\frac{1 + \overline y}{2}) + h(\frac{1 + \overline z}{2}) + h(\frac{1 + \overline w}{2})\right) + O(\log m) 
\end{align*}
where $h(p) = H(Ber(p)) = p \log(1/p) + (1 - p)\log(1/(1 - p))$ is the binary entropy function and the second equality follows by applying Stirling's approximation to the binomial coefficients. Defining the Franz--Parisi potential as the limit
\[ \mathcal{F}_{\beta}(r) = -\frac{1}{\beta} \lim_{\epsilon \to 0} \plim_{m \to \infty} \frac{1}{m} \log \sum_{x',y',z',w'} e^{\beta mf_m(x,y,z,w)} 1(\langle (x,y,z,w), (x',y',z',w') \rangle/4m \in [r - \epsilon, r + \epsilon])  \]
where $(x,y,z,w)$ are sampled from the Gibbs measure at inverse temperature $\beta$, we can therefore compute that in the zero-temperature limit $\beta \to \infty$,
\[ \mathcal{F}_{\infty}(r) := \lim_{\beta \to \infty} \mathcal F_{\beta}(r) = (-1)\sup_{\overline x, \overline y, \overline z, \overline w \in K(r)} f(\overline x, \overline y, \overline z, \overline w)  \]
where
\[ K(r) = \{ (\overline x, \overline y, \overline z, \overline w) \in [-1,1]^4 : (\overline x + \overline y + \overline z + \overline w)/4 = r \} \]
because the global maxima of $f$ are at $(-1,-1,-1,-1)$ and $(1,1,1,1)$. 

\subsubsection{Analytical computation of Franz--Parisi Potential}
We analytically compute $\mathcal{F}_{\infty}(r)$, which corresponds to solving a constrained maximization of $f(x,y,z,w)$ with a linear constraint. Since $\mathcal{F}_{\beta}(r) = \mathcal{F}_{\infty}(r) + o(1)$ as $\beta \to \infty$, we know that the finite but low-temperature Franz--Parisi potential will be a perturbation of the zero-temperature one.

\paragraph{Symmetrization.}
By symmetry within the pairs \((x,y)\) and \((z,w)\), we may assume the maximizer satisfies
\[
x=y=:a,\qquad z=w=:b.
\]
Indeed, for fixed \(x+y\) the product \(xy\) is maximized at \(x=y\), and similarly \(zw\) at \(z=w\). Hence it suffices to optimize over
\[
a,b\in[-1,1],\qquad a+b=2r.
\]

\paragraph{Reduction to one variable.}
With \(x=y=a\), \(z=w=b\),
\[
f(a,b)=2a-2b+a^2+b^2+3ab.
\]
Let \(s:=a+b=2r\) and \(d:=a-b\). Then \(ab=\tfrac{s^2-d^2}{4}\), and the box constraints imply
\[
a,b\in[-1,1]\quad\Longleftrightarrow\quad |d|\le 2-|s|=2-2|r|.
\]
A short calculation gives
\[
f=2d-\tfrac14 d^2+\tfrac54 s^2
= 2d-\tfrac14 d^2+5r^2.
\]
The function \(d\mapsto 2d-\tfrac14 d^2\) is concave with unconstrained maximizer at \(d=4\), so under the constraint
\(|d|\le 2-2|r|\) the maximum is attained at
\[
d^\star=2-2|r|.
\]

\paragraph{Value and maximizers.}
Plugging \(s=2r\) and \(d^\star\) back,
\[
\sup_{K(r)} f
= 2(2-2|r|)-\tfrac14(2-2|r|)^2+5r^2
= 4r^2-2|r|+3.
\]
Thus
\[
 \mathcal F_\infty(r)=-(4r^2-2|r|+3),\qquad r\in[-1,1].
\]
One maximizing configuration is
\[
(r\ge 0):\ (x,y,z,w)=(1,\,1,\,2r-1,\,2r-1),\qquad
(r\le 0):\ (x,y,z,w)=(1+2r,\,1+2r,\,-1,\,-1),
\]
which satisfies \((x+y+z+w)/4=r\) and \([-1,1]\)-bounds.
\begin{figure}
    \centering
    \includegraphics[width=0.75\linewidth]{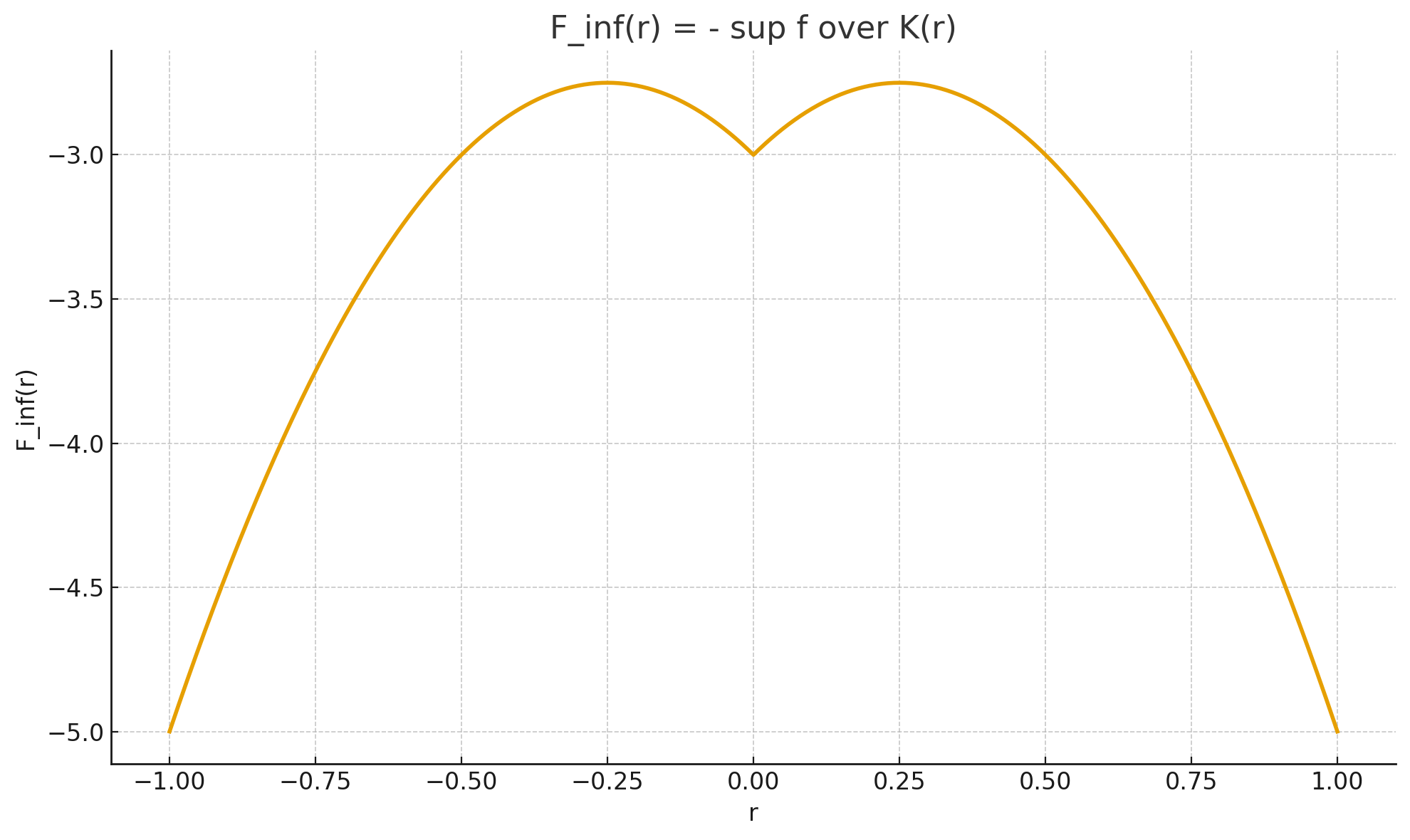}
    \caption{Plot of Franz--Parisi potential $\mathcal{F}_{\infty}(r)$ in the Ising model example. There are three local minima, so there is an (exponentially large) free energy barrier.}
    \label{fig:placeholder}
\end{figure}
\subsubsection{Metastability and contrast to log-concave setting}
From the analytical solution at zero temperature, we see that the Franz--Parisi potential is not quasi-convex but rather has 2 global minima and 1 sub-optimal local minima at all sufficiently large $\beta > 0$. So we conclude that this model does have exponentially large free energy barriers. 

This rigorously rules out the possibility that low-temperature Glauber dyanmics can escape from its initialization in $e^{o(n)}$ time, and we also see visually the fact that the optimization initialized at some states can converge to the suboptimal local minima, as in the explicit analysis of gradient flow in the low-dimensional example.  

On the other hand, while local search fails we again remark that this is not a hard optimization problem, since its global minimum of $f_m$ can be found by minimizing the Lov\'asz extension of the restriction of $-f$ to $\{\pm 1\}^{4m}$. Because the Lov\'asz extension of $-f_m$ is convex, there will not be any free energy barriers in the corresponding landscape, and in fact we can sample from the corresponding log-concave Gibbs measures via MCMC. For example, the continuous-time Langevin diffusion in $n$ dimensions is known to have relaxation time $polylog(n)$ for log-concave measures in isotropic position due to recent progress on the KLS conjecture \cite{klartag2022bourgain}. This also implies Lipschitz concentration which shows that the overlap distribution is highly concentrated.
See \cite{bandeira2023free} for more references and discussion about the Franz--Parisi potential in the log-concave case. 


\end{document}